\documentclass[10pt,journal,compsoc]{IEEEtran}
\ifCLASSOPTIONcompsoc
  \usepackage[nocompress]{cite}
\else
  \usepackage{cite}
\fi
\ifCLASSINFOpdf

\else

\fi

\usepackage{float}
\usepackage{subfig}
\usepackage{graphicx}
\usepackage{algorithm}
\usepackage[noend]{algpseudocode}
\usepackage{mathrsfs}
\usepackage{enumitem}
\usepackage{array}
\newcolumntype{C}[1]{>{\centering\arraybackslash}m{#1}}
\usepackage{mathtools}
\usepackage{amsmath, amsthm, amssymb}
\usepackage{cuted}
\usepackage{bm}
\usepackage{multirow}

\let\S\relax
\DeclareMathOperator{\S}{\mathcal{S}}

\let\E\relax
\let\define\relax

\DeclareMathOperator{\R}{\mathcal{R}}
\DeclareMathOperator{\E}{\mathbb{E}}
\DeclareMathOperator{\I}{\mathbb{I}}

\DeclareMathOperator{\U}{\mathcal{U}}

\DeclareMathOperator{\out}{out}
\DeclareMathOperator{\e}{end}
\DeclareMathOperator{\f}{final}

\DeclareMathOperator*{\argmax}{arg\,max}
\DeclareMathOperator{\define}{\coloneqq}
\theoremstyle{definition}
\newtheorem{problem}{\textbf{Problem}}
\newtheorem{remark}{\textbf{Remark}}

\newtheorem{lemma}{\textbf{Lemma}}
\newtheorem*{lemma*}{\textbf{Lemma}}

\newtheorem{theorem}{\textbf{Theorem}}
\newtheorem{definition}{Definition}

\newtheorem{property}{Property}
\newtheorem{example}{Example}

\newcommand{\RNum}[1]{\uppercase\expandafter{\romannumeral #1\relax}}

\begin{document}
%
\title{On Adaptive Influence Maximization under General Feedback Models}
%
%
%
%

\author{Guangmo (Amo)~Tong,~\IEEEmembership{Member,~IEEE}
	and Ruiqi Wang
\IEEEcompsocitemizethanks{\IEEEcompsocthanksitem G. Tong and R. Wang are with the Department
of Computer and Information Sciences, University of Delaware, Newark,
DE, USA, 19715.\protect\\
E-mail: \{amotong, wangrq\}@udel.edu}
\thanks{Manuscript received xx, xx, 2019; revised xx, xx.}}

%
%

\markboth{IEEE TRANSACTIONS ON KNOWLEDGE AND DATA ENGINEERING,~Vol.~?, No.~?, ?~2019}%
{Shell \MakeLowercase{\textit{et al.}}: Bare Demo of IEEEtran.cls for Computer Society Journals}
%



\IEEEtitleabstractindextext{%
\begin{abstract}
Influence maximization is a prototypical problem enabling applications in various domains, and it has been extensively studied in the past decade. The classic influence maximization problem explores the strategies for deploying seed users before the start of the diffusion process such that the total influence can be maximized. In its adaptive version, seed nodes are allowed to be launched in an adaptive manner after observing certain diffusion results. In this paper, we provide a systematic study on the adaptive influence maximization problem, focusing on the algorithmic analysis of the scenarios when it is not adaptive submodular. We introduce the concept of regret ratio which characterizes the key trade-off in designing adaptive seeding strategies, based on which we present the approximation analysis for the well-known greedy policy. In addition, we provide analysis concerning improving the efficiencies and bounding the regret ratio. Finally, we propose several future research directions. 
\end{abstract}

}

\maketitle

\IEEEdisplaynontitleabstractindextext

\IEEEpeerreviewmaketitle

\IEEEraisesectionheading{\section{Introduction}\label{sec: introduction}}

\IEEEPARstart{I}{nformation} diffusion is an essential function of today's online social network. An information cascade is typically initiated by a seed set, and it then spreads from users to users stochastically. Influence maximization (IM),  proposed by Kempe, Kleinberg and Tardos \cite{kempe2003maximizing}, investigates the approaches for selecting seed nodes such that the resulted influence can be maximized. This problem has become one of the core problems in social computing, and it has drawn tremendous attentions due to its potential for enabling applications in different domains \cite{li2018influence,guille2013information,aslay2018influence}. 

As a natural variant as well as an important generalization of the IM problem, adaptive influence maximization (AIM) problem allows the seed nodes to be selected after observing certain diffusion results, and correspondingly it investigates adaptive seeding strategies for maximizing the influence. Compared to the non-adaptive seeding strategy, an adaptive one can better utilize the budget because it makes decisions adapted to the observations. In addition, an adaptive seeding strategy can take account of the dynamic features of a social network, such as the change of the network topology due to the frequent user join and leave. For example, while limiting the spread of rumor by launching a positive cascade, selecting positive seed nodes in an adaptive manner can react against the new actions of the rumor cascade. In viral marketing, one would prefer to observe the feedback from the customer and launch the campaign step by step. In this paper, we study the AIM problem focusing on algorithm design and approximation analysis.

\textbf{Model and Problem Formulation.} This paper considers the prominent Independent Cascade (IC) model.\footnote{The results in this paper apply to other submodular diffusion models, namely the Linear Threshold model.} Under the IC model, the diffusion process is defined by the propagation probability between the users, and there is one chance of activation between each pair of users. The seed users are the first who are active, and in each round the users activated in the last round attempt to activate their inactive neighbors. The cascade spreads round by round and it terminates when no node can be further activated. In the AIM problem, we may first select some seed nodes, observe the diffusion for a certain number of rounds, and then deploy more seed nodes according to the diffusion results. Consequently, an adaptive seeding strategy consists of two parts: \textit{feedback model} and \textit{seeding policy}. A feedback model specifies how many diffusion rounds we would wait for before selecting the next seed node, and a seeding policy indicates which node should be selected.\footnote{For simplicity, we assume only one seed node is selected each time, and our analysis applies to the batch mode as discussed later in Sec. \ref{sec: general}.} The ultimate goal is to maximize the number of active nodes.

\textbf{General Feedback Model.} While designing feedback models, intuitively we wish for more observations and therefore should delay the seeding action as much as possible. Therefore, given a budget $k$, the optimal design is to consume one budget each time and always wait for the diffusion to terminate before selecting the next seed node, which is called Full Adoption feedback model \cite{golovin2011adaptive}. However, this optimal feedback model is not feasible for the applications where the seed nodes are required to be deployed after a fixed number of rounds. For example, a company would prefer to post online advertisements every Monday, or that one would propagate a certain information cascade in a timely manner and therefore would like to utilize all the budget after every one diffusion round. For such purpose, we propose a generalized feedback model where one seed node is selected after every $d$ diffusion rounds with $d \in \mathbb{Z}^+\cup \{\infty\}$, where $d=\infty$ denotes the case when we always wait for the diffusion to terminate before selecting the next seed node. When $d=1$ and $d=\infty$, it reduces to the Myopic feedback model \cite{golovin2011adaptive} and Full Adoption feedback model  \cite{golovin2011adaptive}, respectively. 

\textbf{The State-of-the-art.} A series of literature has considered the AIM problem under different settings (e.g., \cite{golovin2011adaptive,tong2017adaptive,vaswani2016adaptive,chen2013near}) in which the main technique is \textit{adaptive submodular maximization} invented by Golovin and Krause \cite{golovin2011adaptive}. This optimization technique shows that when the considered the problem is adaptive submodular, the greedy policy yields a $(1-1/e)$-approximation. The existing works have utilized this result to investigate the AIM problem for special feedback models under which the AIM problem is adaptive submodular. In particular, these feedback models satisfy the condition that \textit{the seeding decisions are always made after the current diffusion has terminated}. However, the AIM problem is not adaptive submodular under general feedback models, and it remains unknown that how to bound the performance of the greedy policy for the general cases (i.e., $d\neq \infty$), which is the primary motivation of our work. 

\textbf{This Paper.} The main results of this paper are briefly summarized as follows.
\begin{itemize}
\item We provide a systematic model for the AIM problem under general feedback settings and formulate the considered problem as an optimization problem. Our model naturally fits the AIM problem and generalizes the existing ones.
\item We introduce the concept of regret ratio which describes the trade-off between \textit{waiting} and \textit{seeding}. This ratio is not only intuitively meaningful but also measures the performance bound of the greedy policy.
\item We analyze the greedy policy from the view of \textit{decision tree}, and show that the greedy policy gives a $(1-e^{-\frac{1}{\alpha(T_g)}})$-approximation under the general feedback model, where $\alpha(T_g)$ is the regret ratio associated with the greedy policy. \item We show how to generalize the reverse sampling technique for the AIM problem to improve the efficiency of the greedy policy.
\item We propose several directions of future work taking account of more realistic applications.
\item We design simulations to experimentally evaluate the greedy policy under general feedback models and further examine the effect of the feedback models on the diffusion process. 
\end{itemize}

\textbf{Road Map.} We survey the related work in Sec. \ref{sec: relate}. The preliminaries are provided in Sec. \ref{sec: pre}. The main analysis is given in Sec. \ref{sec: aim}. In Sec. \ref{sec: exp}, we present the experiments. The future work is discussed in Sec. \ref{sec: general}. Sec. \ref{sec: con} concludes this paper. The missing proofs and additional experimental results are provided in the supplementary material.

\section{Related Work.} 
\label{sec: relate}
Due to space limitation, we focus on the AIM problem and will not survey the literature concerning the classic IM problem. The interested reader is referred to the recent surveys \cite{li2018influence,guille2013information,aslay2018influence}.

The AIM problem was first studied by Golovin and Krause \cite{golovin2011adaptive} in the investigation on adaptive submodular maximization. In \cite{golovin2011adaptive}, the authors proposed two special feedback models: Full Adoption feedback model and Myopic feedback model. Before selecting the next seed node, one always waits for the diffusion process to terminate under the Full Adoption feedback model while waits for one diffusion round under the Myopic feedback model. For the Full Adoption feedback model, it is shown in \cite{golovin2011adaptive} that greedy policy has an approximation ratio of $1-1/e$ by using the technique of adaptive submodularity. Combining the result in \cite{tong2017adaptive} that the Full Adoption feedback model is optimal, greedy policy plus the Full Adoption feedback model has the approximation ratio of $1-1/e$ to any adaptive seeding strategy for the AIM problem under a budget constraint. Following the optimization framework given in \cite{golovin2011adaptive}, Sun \textit{et al.} \cite{sun2018multi}, Chen \textit{et al. }\cite{chen2013near} and Vaswani \textit{et al. }\cite{vaswani2016adaptive} have studied the AIM problem under the feedback models which are variants of the Full Adoption feedback model. One common setting in these works is that we always wait until the diffusion terminates before making the next seeding decision, which is critical for the AIM problem to be adaptive submodular. As noted in \cite{golovin2011adaptive} and \cite{vaswani2016adaptive}, when we wait for a fixed number of rounds (e.g., Myopic feedback model) the AIM problem is unfortunately not adaptive submodular anymore, and to the best of our knowledge there is no analysis technique available for such general feedback models. We in this paper make an attempt to fill this gap by providing a metric for quantifying the approximation ratio when the AIM problem is not adaptive submodular. 

There exist several other research lines concerning the AIM problem from perspectives different from ours in this paper. One research branch focuses on online influence maximization where the seeding process also consists several stages, but the primary consideration therein is to overcome the incomplete knowledge of the social network (\cite{lei2015online, chen2016combinatorial,vaswani2017model, wen2017online}). In another issue, Seeman \textit{et al.} \cite{seeman2013adaptive} studied a two-stage adaptive seeding process where the neighbors of the first-stage nodes are candidates for the second seeding stage, which is essentially different from the AIM problem considered in our paper. Later in \cite{salha2018adaptive}, the authors considered the problem of modifying the IC model so that the AIM problem becomes adaptive submodular under the Myopic feedback model.  Recently, Han \textit{et al.} \cite{han2018efficient} investigated the issue of speeding up an adaptive seeding strategy. 

{\renewcommand{\arraystretch}{1.5}
	\begin{table*}[t]
		\centering
		{\begin{tabular}{ |p{7cm} ||p{6cm} || p{2cm} |}
				\hline 
				\textbf{Symbol}& \textbf{Keyword} &\textbf{Reference} \\
				\hline 
				$G=(V, E)$ and $p_e \in (0,1]$ & IC model&Def. \ref{def: ic}    \\ 
				\hline
				$\phi=\big(L(\phi),D(\phi)\big)\in 2^E \times 2^E=\Phi$ and $\Pr[\phi] \in [0,1]$ &realization &Def. \ref{def: realization}    \\ 
				\hline
				$\phi_1 \prec \phi_2, \Pr[\phi_2|\phi_1]$ &sub-realization, super-realization & Def. \ref{def: sub-realization}    \\ 
				\hline
				$\phi_{\emptyset}$ &empty realization& Def. \ref{def: partial-realization}    \\ 
				\hline
				$\Psi$ &the set of full realizations & Def. \ref{def: partial-realization} \\ 
				\hline
				$\phi_1 \sim \phi_2$ & realization compatibility & Def. \ref{def: compati}  \\ 
				\hline
				$\phi_1 \oplus \phi_2$& realization concatenation& Def. \ref{def: realization_con}  \\ 
				\hline
				$U=(\dot{S}(U), \dot{\phi}(U)) \in 2^V\times \Phi$ &status& Def. \ref{def: status} \\
				\hline
				$\dot{\U}_d(U) \in 2^{\Phi}$& d-round-status &Def. \ref{def: d-status}\\ 
				\hline 
				$U_1 \cup U_2$&  status union & Def. \ref{def: status-union}\\ 
				\hline 
				$\pi:  2^{V}\times \Phi \rightarrow 2^V$&policy &  Def. \ref{def: policy}.\\ 
				\hline
				$\pi_g$& greedy policy &  Def. \ref{def: greedy}.\\ 
				\hline
				$T$&decision tree &  Def. \ref{def: decision_tree}\\ 
				\hline
				$T|U$& decision tree conditioned on $U$ &  Def. \ref{def: condition-tree}\\ 
				\hline
				$T_1 \oplus T_2$ & concatenation of trees $T_1$ and $T_2$ &  Def. \ref{def: con_tree}\\ 
				\hline
				$S_i^T$ & the tree-nodes in $T$ in level $i$ &  Def. \ref{def: decision_tree_notation}\\ 
				\hline
				$U_i^T$ & the tree-edges in $T$ from level $i$ to $i+1$ &  Def. \ref{def: decision_tree_notation}\\ 
				\hline
				$\dot{S}_{\e}(U)$ & the endpoint of a tree-edge $U$  &  Def. \ref{def: decision_tree_notation}\\ 
				\hline
				$\dot{U}_{\out}(S)$ & out edges of a tree-node $S$ &  Def. \ref{def: decision_tree_notation}\\ 
				\hline
		\end{tabular}}
		\caption{Notations. }
		\label{table:symbol}
\end{table*}}

\section{Preliminaries}
\label{sec: pre}
In this section, we introduce a collection of definitions some of which are extended from that in the work \cite{golovin2011adaptive} of Golovin and Krause. We intend to spare more space in this section for explaining the definitions in order to make the analysis in Sec. \ref{sec: aim} smooth. The references of the notations are given in Table \ref{table:symbol}.

\begin{definition}[\textbf{Convenient Set Notations}]
	For an element $v$ and a set $S$, we use $v$ and $v+S$ in replace of $\{v\}$ and $\{v\}\cup S$, respectively.
\end{definition}

\subsection{IC Model and Adaptive Seeding Process} 

\begin{definition}[\textbf{IC Model}]
\label{def: ic}
A social network is represented by a directed graph $G=(V,E)$. For each edge $(u,v)$, we say $u$ is an in-neighbor of $v$, and $v$ is an out-neighbor of $u$. An instance of IC model is given by a directed graph $G=(V, E)$ and the probability $p_e \in (0,1]$ on each edge $e \in E$.\footnote{We assume the probability $p_e$ is strictly larger than 0 for otherwise we can remove it from the graph.}
\end{definition}

\begin{definition}[\textbf{Round}]
	In one round, each node $v$ activated in the last round attempts to activate each of $v$'s inactive neighbor $u$, with the success probability of $p_{(v,u)}$. We assume the observations are made round by round.
\end{definition}

\begin{definition}[\textbf{Adaptive Seeding Process}]
An adaptive seeding process alternates between the following two steps:
\begin{itemize}
	\item (\textbf{seeding-step}) Select and activate a certain set of seed nodes.
	\item (\textbf{observing-step}) Observe the diffusion for a certain number of rounds.
\end{itemize}
\end{definition}

\subsection{Realizations and Status}
The definitions in this section are used to describe an intermediate stage during a diffusion process. 

\begin{definition}[\textbf{States of Edges}]
	Following \cite{kempe2003maximizing}, we speak of each edge $(u, v)$ as being \textit{live} or \textit{dead} to indicate that if $u$ can activate $v$ once $u$ becomes active. 
\end{definition}

\begin{definition}[\textbf{Realization}]
	\label{def: realization}
	A \textit{realization} $\phi=(L(\phi),D(\phi)) \in 2^E \times 2^E$ is an ordered two-tuple where $L(\phi) \subseteq E ,D(\phi) \subseteq E$, and $L(\phi)\cap D(\phi) =\emptyset$, specifying the states of the edges that have been observed. In particular, $e\in L(\phi)$ (resp., $e\in D(\phi)$) means $e$ is a live (resp., dead) edge, and, the state of an edge $e$ is unknown in $\phi$ if $e \notin L(\phi)\cup D(\phi)$. We use $\Phi$ to denote the set of all realizations. For each realization $\phi$, we define $\Pr[\phi]$ as $\Pr[\phi]\define \prod_{e\in L(\phi)}p_e\prod_{e\in D(\phi)}(1-p_e)$, which is the probability that $\phi$ can be realized (i.e., sampled). 
\end{definition}

\begin{definition}[\textbf{$t$-live-path}]
	\label{def: t-path}
	For two nodes $u, v \in V$ and a realization $\phi$, a \textit{$t$-live path} from $u$ to $v$ in $\phi$ is a path of at most $t$ edges which are all in $L(\phi)$. When there is no limit on the length of the path, we use the notation \textit{$\infty$-live-path}. 
\end{definition}

\begin{definition}[\textbf{Sub-realization}]
	\label{def: sub-realization}
	For two realizations $\phi_1$ and $\phi_2$, we say $\phi_1$ is a \textit{sub-realization} of $\phi_2$ if $L(\phi_1) \subseteq L(\phi_2)$ and $D(\phi_1) \subseteq D(\phi_2)$, and denote it as $\phi_1 \prec \phi_2$. If $\phi_1 \prec \phi_2$, we also say $\phi_2$ is a \textit{super-realization} of $\phi_1$. Intuitively, $\phi_2$ is one possible outcome if we continue to observe the states of the edges after observing $\phi_1$. We use $\Pr[\phi_2|\phi_1]$ to denote the probability that $\phi_2$ can be realized conditioned on $\phi_1$, and therefore we have 
	\[\Pr[\phi_2|\phi_1]= \prod_{e\in L(\phi_2)\setminus L(\phi_1)}p_e\prod_{e\in D(\phi_2)\setminus D(\phi_1)}(1-p_e).\] 
\end{definition}

\begin{definition}[\textbf{Full and Partial Realization}]
	\label{def: partial-realization}
	We say $\phi$ is a \textit{full realization} if $L(\phi)\cup D(\phi)=E$, which indicates that all the edges have been observed. We use $\Psi$ to denote the set of all full realizations. Note that $\Psi$ also represents the basic event space of the IC model. We use $\phi_{\emptyset}=(\emptyset,
	\emptyset)$ to denote the empty realization. A realization is a \textit{partial realization} if it is not a full realization.
\end{definition}

We use the following concepts to describe the observations during the seeding process.
\begin{definition}[\textbf{Status}]
\label{def: status}
A \textit{status} $U$ is a two-tuple $U=\big(\dot{S}(U), \dot{\phi}(U)\big) \in 2^V\times \Phi$, where $\dot{S}(U)\subseteq V$ is the set of the current active nodes and realization $\dot{\phi}(U) \in \Phi$ shows the state of the edges that have been observed. We use $U_\emptyset \define (\emptyset,\phi_\emptyset)$ to denote the status where there is no active node and no edge has been observed. 
\end{definition}

\begin{remark}
We use $\phi$ to denote a realization while use $\dot{\phi}()$ with an over-dot to denote a realization associated with an object. For example, $\dot{\phi}(U)$ is a realization associated with a status $U$. The similar rule of the use of over-dot applies to other notations in this paper.
\end{remark}

\begin{definition}[\textbf{Final Status}]
	We say a status $U$ is final if there is no $\infty$-live-path from any node in $\dot{S}(U)$ to $V\setminus \dot{S}(U)$ in any full realization $\psi$ where $\dot{\phi}(U) \prec \psi$. 
\end{definition}

\begin{remark}
	\label{remark: final}
	A status is not final if and only if it is possible to have more active nodes in the future rounds even if we do not select any new seed node. When the diffusion process terminates, it reaches a final status. 
\end{remark}

\begin{definition}[\textbf{d-round Status}]
	\label{def: d-status}
	For a status $U$ and $d\in \mathbb{Z}^+$, we use $\dot{\U}_d(U)$ to denote the set of the possible statuses after $d$ diffusion rounds following $U$. In addition, we use $\dot{\U}_\infty(U)$ to denote the outcomes after the diffusion terminates following $U$.
\end{definition}

\begin{figure}[t]
	\begin{center}
		\includegraphics[width=0.45\textwidth]{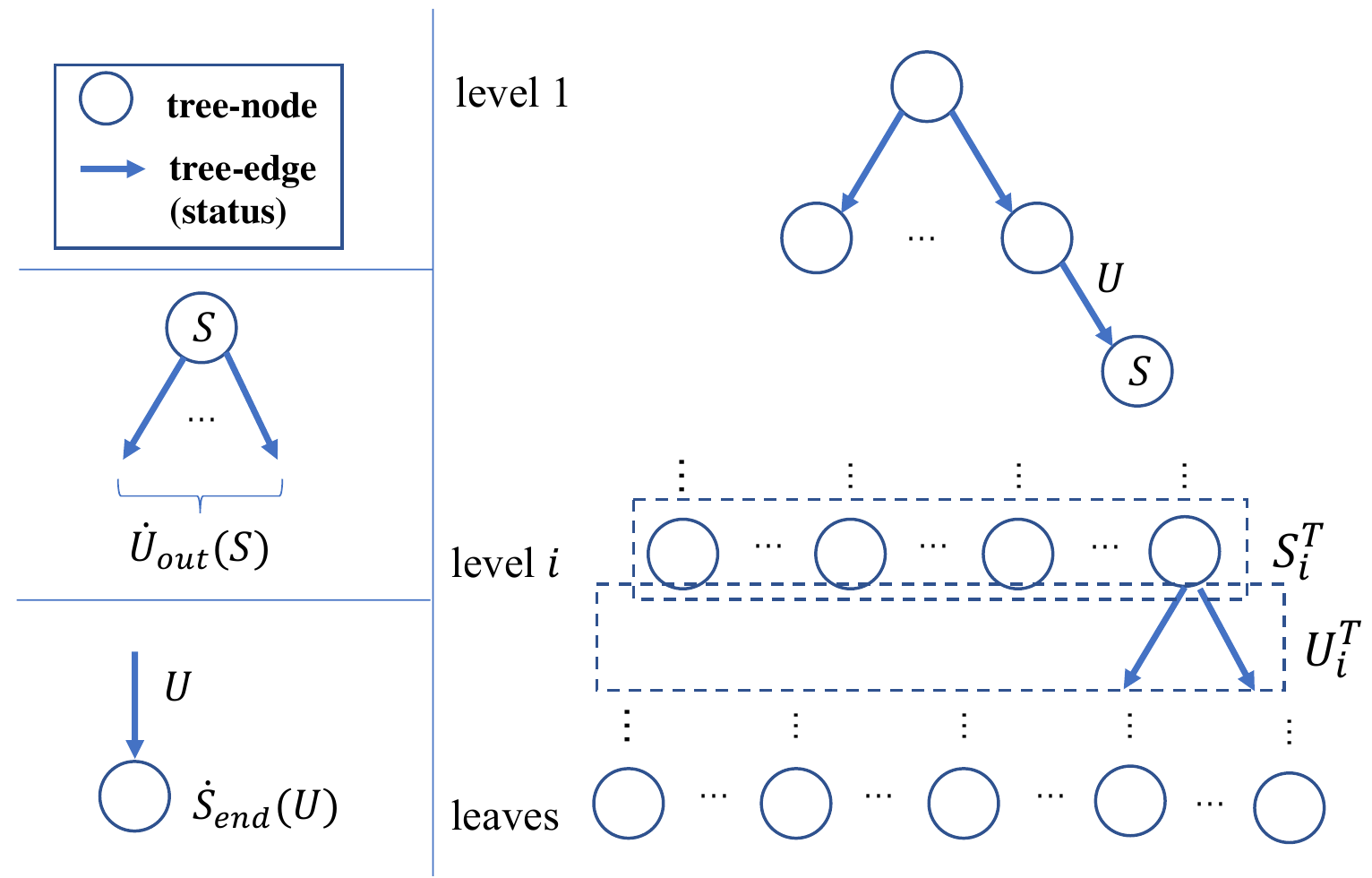} 
	\end{center} 
	\caption{\textbf{Decision tree.}}
	\label{fig: tree}
\end{figure}

\subsection{Seeding Process}
The definitions in this section are used to describe an adaptive seeding process.

\begin{definition}[\textbf{Policy}]
\label{def: policy}
A \textit{policy} $\pi:  2^{V}\times \Phi \rightarrow 2^V$ maps a status $U$ to a node-set $S \subseteq V$, which indicates that $\pi$ will select $S$ as the seed set if the observed status is $U$.
\end{definition}

\begin{remark}
Note that a realization together with the current active nodes determines the rest of the diffusion process. Thus, a policy makes decisions according to the current status rather than the current realization. 
\end{remark}

An adaptive seeding process can be viewed as a decision tree. An illustration of the following definitions is shown in Fig. \ref{fig: tree}. 

\begin{definition}[\textbf{Decision tree}]
\label{def: decision_tree}
A \textit{decision tree} $T$ of an adaptive seeding process is an arborescence, where each tree-edge is associated with a status $U$ which corresponds to an observing-step showing what have been observed, and each tree-node is a node-set $S \subseteq V$ which corresponds to a seeding-step showing the seed nodes that are selected and activated.\footnote{We reserve the term \textit{edge} for the social network graph and the term \textit{tree-edge} for the decision tree.}  
\end{definition}

\begin{definition}[\textbf{Decision Tree Notations}]
\label{def: decision_tree_notation}
For each tree-node $S$, let $\dot{U}_{\out}(S)$ be the set of the tree-edges out of $S$, showing all possible different observations after selecting $S$.\footnote{Two observations are different if and only if there is at least one edge with different observed states.} Since two edges in a tree cannot have the same realization, by abusing the notation, we also use $U$ to denote the tree-edge of which the status is $U$. For a tree-edge $U$, we use $\dot{S}_{\e}(U) \subseteq V$ to denote the end-node of $U$. The tree-nodes can be grouped by levels. For a decision tree $T$ and $i \in \{1, 2,3, ...\}$, we use $S_i^T \in 2^{2^V}$ to denote the set of the tree-nodes in the $i$-th level. For $i \in \{2,3,...\}$, we use $U_i^T \subseteq \Phi$ to denote the set of the statuses of the tree-edges from the tree-nodes in $S_{i-1}^T$ to those in $S_i^T$, and define $U_{1}^T$ as $\{U_{\emptyset}\}$. In addition, we use $S_{\infty}^T$ to denote the set of the nodes in the lowest level (i.e., the leaves) and use $U_{\infty}^T$ to denote the set of the tree-edges connecting to the leaves. 
\end{definition}

\begin{remark}
\label{remark: decision-tree}
When a decision tree represents an adaptive seeding process, level $i$ shows all possible scenarios in the $i$-th seeding step. For each pair $U$ and $\dot{S}_{\e}(U)$, it means that $\dot{S}_{\e}(U)$ is selected by $\pi$ as a seed set when $U$ is observed. The edges out of a tree-node $S$ indicate the possible observations after selecting $S$. For each sequence of statuses ($U_1, U_2,...$) from the root to a leaf, we have $\dot{\phi}(U_i)\prec \dot{\phi}(U_{i+1})$ and $\dot{S}_{\e}(U_i)\subseteq \dot{S}_{\e}(U_{i+1})$.
\end{remark}

\subsection{Process Concatenation}
The analysis of the policy requires to measure the effect of the \textit{union} of two seeding processes. In the work \cite{golovin2011adaptive} of Golovin and Krause, it was stated as: \textit{running one policy to completion and then running another policy as if from a fresh start, ignoring the information gathered during the running of the first policy}. To make this concept mathematically tractable, we adopt the decision tree perspective and employ the following definitions.

\begin{definition}[\textbf{Realization Compatibility}]
	\label{def: compati}
	For two realizations $\phi_1, \phi_2 \in \Phi$, we say they are \textit{compatible} if $L(\phi_1) \cap D(\phi_2)=\emptyset$ and $D(\phi_1) \cap L(\phi_2)=\emptyset$. That is, there is no conflict observations. We denote this relationship by $\phi_1 \sim \phi_2$. 
\end{definition}

\begin{definition}[\textbf{Realization Concatenation}]
	\label{def: realization_con}
	For a set of finite realizations $\{\phi_1, \phi_2,..., \phi_m\}$ where \[\big(L(\phi_1)\cup ...\cup L(\phi_m)\big) \cap \big(D(\phi_1)\cup ...\cup D(\phi_m)\big)=\emptyset, \]we define $\phi_1 \oplus... \oplus \phi_m$ as a new realization with \[L(\phi_1 \oplus... \oplus \phi_m)=L(\phi_{1})\cup...\cup L(\phi_{m})\] and \[D(\phi_1 \oplus ... \oplus \phi_m)=D(\phi_{1})\cup...\cup D(\phi_{m}).\]
\end{definition}

\begin{definition}[\textbf{Status Union}]
	\label{def: status-union}
	For two status $U_1$ and $U_2$ where $\dot{\phi}(U_1)$ and $\dot{\phi}(U_2)$ are compatible, we define $U_1 \cup U_2$ as a new status with $\dot{S}(U_1 \cup U_2)=\dot{S}(U_1)\cup \dot{S}(U_2)$ and $\dot{\phi}(U_1 \cup U_2)=\dot{\phi}(U_1) \oplus \dot{\phi}(U_2)$.
\end{definition}

\begin{figure}[tp]
	\begin{center}
		\includegraphics[width=0.4\textwidth]{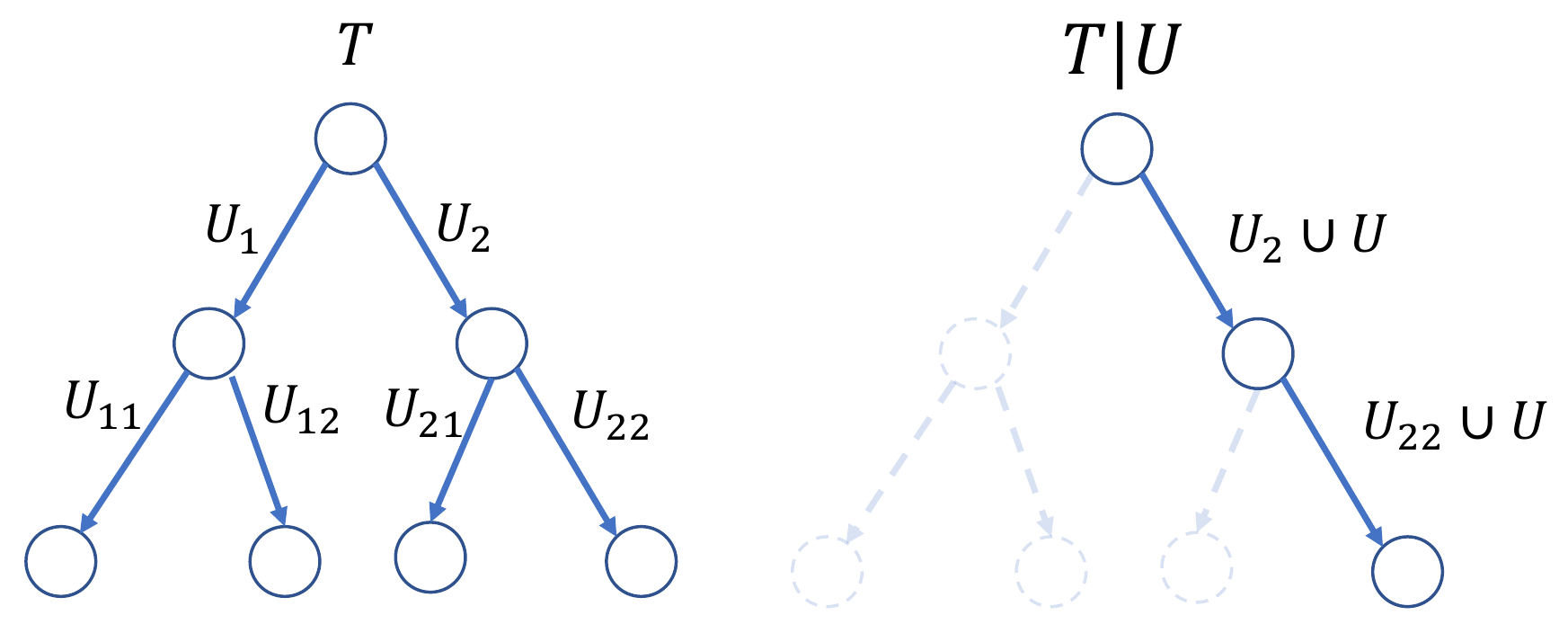} 
	\end{center} 
	\vspace{-4mm}
	\caption{\textbf{Decision tree conditioned on a status.} {\small Suppose $U_1$ and $U_{21}$ are not compatible with $U$.}}
	\vspace{-2mm}
	\label{fig: condition-tree}
\end{figure}

\begin{definition}[\textbf{Decision Tree Conditioned on a Status}]
	\label{def: condition-tree}
	Given a decision tree $T$ and a status $U$, we construct another decision tree by modifying $T$ as follows. For each tree-edge $U_*$ in $T$ such that $\dot{\phi}(U_*)$ is not compatible with $\dot{\phi}(U)$, we remove the tree-edge $U_*$ as well as the tree-node $\dot{S}_{\e}(U_*)$. For each tree-edge $U_*$ in $T$ such that $\dot{\phi}(U_*)$ is compatible with $\dot{\phi}(U)$, we replace the status $U_*$ by $U_*\cup U$. We denote the resulted decision tree as $T|U$ named as \textit{the decision tree of $T$ conditioned on status $U$}. An illustrative example is given in Fig. \ref{fig: condition-tree}.
\end{definition}

\begin{remark}
One can see that if we remove one tree-edge $U$, we must also remove all of its following tree-edges due to Remark \ref{remark: decision-tree}.  In the tree $T|U$, each realization is a super-realization of $\dot{\phi}(U)$. Note that we only modify the status but do not change the tree-node. When $T$ is an adaptive seeding process of a certain policy, the tree $T|U$ shows the adaptive seeding process when the states of the edges in $L(\dot{\phi}(U))\cup D(\dot{\phi}(U))$ have been fixed and the nodes in $\dot{S}(U)$ are activated, but the policy does not have such information. 
\end{remark}

\begin{definition}[\textbf{Decision Tree Concatenation}]
	\label{def: con_tree}
	Given two decision trees $T_1$ and $T_2$, we construct another decision tree by modifying $T_1$, as follows. For each tree-edge $U$ in $T_1$ where $\dot{S}_{\e}(U)$ is a leaf, we replace $\dot{S}_{\e}(U)$ by the tree $T_2|U$. We denote the new tree as $T_1 \oplus T_2$. An example is shown in Fig. \ref{fig: tree-con}. 
\end{definition}

\begin{remark}
In our context, a decision tree can be given by an adaptive seeding process, or its modification according to Defs. \ref{def: condition-tree} and \ref{def: con_tree}. The notations in Def. \ref{def: decision_tree_notation} also apply to a concatenation of two decision trees. For two seeding process, the concatenation of their decision trees provides a tractable representation of the process that first running the first process and then running the second one without considering the obtained observations. 
\end{remark}

\begin{figure}[t]
	\begin{center}
		\includegraphics[width=0.45\textwidth]{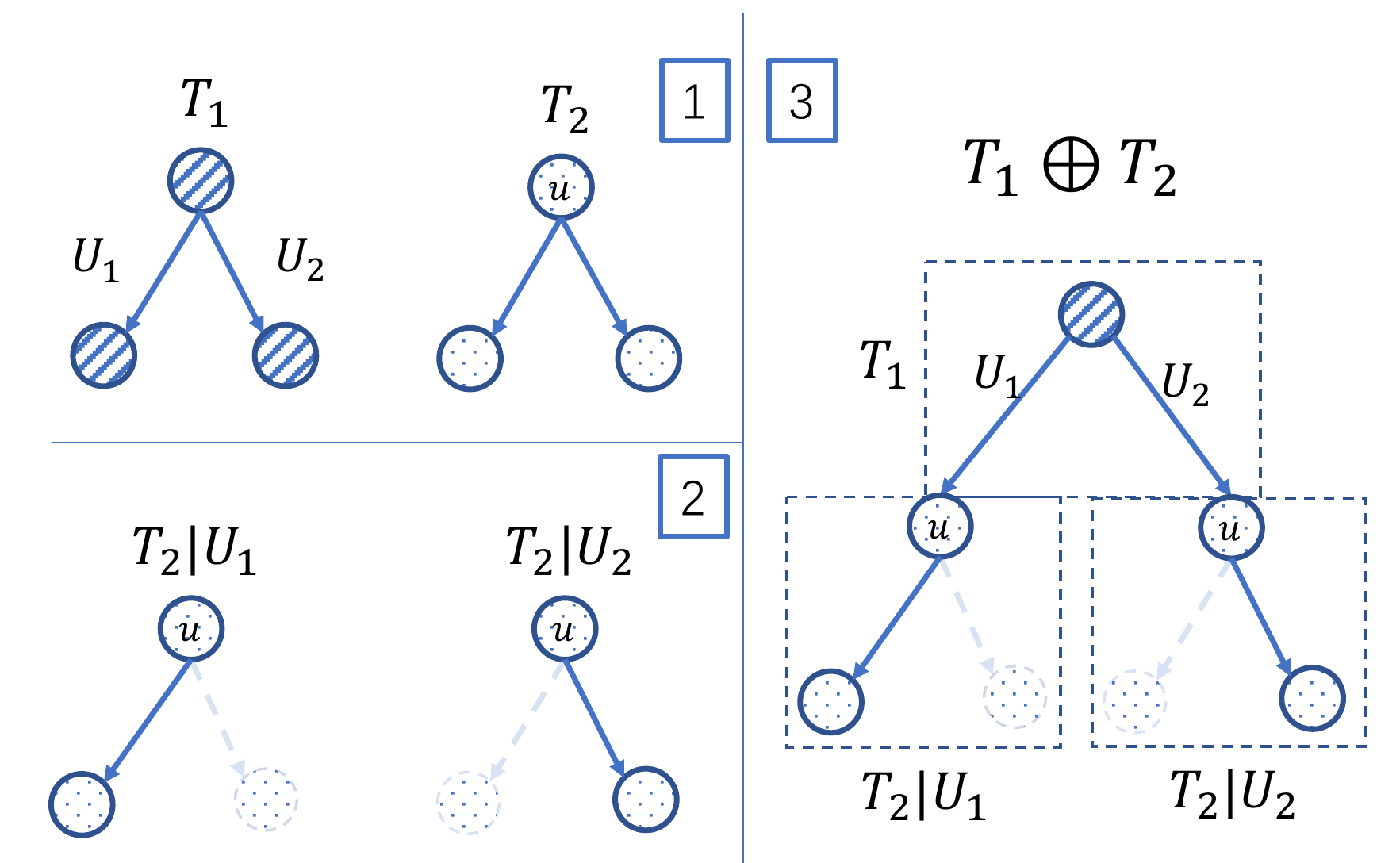} 
	\end{center} 
	\vspace{-2mm}
	\caption{\textbf{Concatenation of two decision trees.} {\small Suppose we have two trees $T_1$ and $T_2$ where $T_1$ has two tree-edges $U_1$ and $U_2$ connecting to the leaves. Supposing the trees $T_2|U_1$ and $T_2|U_2$ are as those given in the figure, we have the tree $T_1 \oplus T_2$ as shown therein.}}
	\label{fig: tree-con}
	\vspace{-2mm}
\end{figure}

A decision tree naturally has the following property.

\begin{property}
	For each tree-edge $U_1$ where $\dot{S}_{\e}(U_1)$ is not a leaf, we have 
	\begin{equation}
	\sum_{U_2 \in \dot{U}_{\out}(\dot{S}_{\e}(U_1))} \Pr[\dot{\phi}(U_2)|\dot{\phi}(U_1)]=1.   
	\end{equation}
	This is because (a) $\Pr[\dot{\phi}(U_2)|\dot{\phi}(U_1)]$ is the probability that $\dot{\phi}(U_2)$ happens conditioned on $U_1$ and (b) $\dot{U}_{\out}(\dot{S}_{\e}(U_1))$ is the set of \textit{all} possible different observations after observing $U_1$ and selecting $\dot{S}_{\e}(U_1)$ as the seed set. Note that this is also true for a tree which a concatenation of other two trees.
\end{property}

\subsection{Counting Active Nodes}

\begin{definition}[$A_t(S,\psi)$]
\label{def: A_t}
For each $S\in 2^V$ and $\psi\in \Psi$, we use $A_t(S,\psi)$ to denote the number of the active nodes in $\psi$ after $t$ diffusion rounds when $S$ is set of the current active nodes. We use $t=\infty$ for the case when the diffusion terminates. For each $S \subseteq V$, $V^* \subseteq V$, $\phi \in \Phi$ and $\psi \in \Psi$, we define that 
\begin{equation}
\Delta_t(S,V^* ,\psi)\define |A_t(S \cup V^* ,\psi)|-|A_t(S,\psi)|,
\end{equation}
as the marginal increase resulted by $V^* $, and we use 
\begin{equation}
\label{eq: f_t}
\Delta f_t(S,V^*  ,\phi)=\sum_{ \phi \prec \psi, \psi \in \Psi} \Pr[\psi|\phi]\cdot \Delta _t(S,V^*  ,\psi)
\end{equation} 
to denote the expected marginal profit when the current status is $(S,\phi)$.
\end{definition}

The following lemma is a start point to calculating the influence.

\begin{lemma}[Kempe \textit{et al.} \cite{kempe2003maximizing}]
For a full-realization $\psi$, according to \cite{kempe2003maximizing}, $v \in A_t(S,\psi)$ iff there exists a $t$-live-path in $\psi$ from a node $u \in S$ to $v$ .  
\end{lemma}

\subsection{Discussions}
Before proceeding with the analysis, we briefly discuss the issues regarding formulating the AIM problem. The existing works model the AIM following the seminal work \cite{golovin2010adaptive} of Golovin and Krause where we select elements from a ground set and the elements may have random states. To utilize this framework, under the Full Adoption model, the random states of a node $u$ are defined as the states of the edges reachable from $u$, in which way the selections are made among the nodes and meanwhile the randomness is also forced to be associated with the nodes. For other feedback models, we need to modify the definitions of the random states in order to apply the framework in \cite{golovin2010adaptive}. However, such a method does not naturally fit the AIM problem because in the AIM problem the selections are made among nodes while the randomness comes from the edges which are \textit{independent} of our selection of the nodes. Therefore, in this paper the randomness is directly defined on the edges and we model the node selections and edge observations in a separate manner, which brings us higher flexibility in the analysis as shown later in Sec. \ref{sec: aim}, as well as in defining more general and natural AIM problems as later discussed in Sec. \ref{sec: general} . 

\section{AIM Problem and Greedy Strategy}
\label{sec: aim}
This section gives the main analysis of this paper. We first formally formulate the considered problem and introduce the concept of regret ratio,  and then show how to bound the performance of the greedy strategy. In addition, we discuss how to improve the efficiency with the reverse sampling technique. Finally, we provide an upper bound of the regret ratio.
\subsection{Problem Statement}
In this paper, we consider the AIM problem formulated as follows.
\begin{problem}[\textbf{$(k,d)$-AIM Problem}]
\label{problem: k-d}
Given a budget $k \in \mathbb{Z}^+$ and an integer $d \in \mathbb{Z}^+$, we consider the feedback model where one seed node is selected after every $d$ rounds of diffusion until $k$ seed nodes are selected. We aim to design a policy for selecting nodes under this pattern such that the influence can be maximized.
\end{problem}

\begin{remark}
When $d\geq n-1$, it is equivalent that we wait for the diffusion process to terminate before selecting the next seed node, and in this case, we denote it as the $(k,\infty)$-AIM Problem which reduces to the problem studied in \cite{golovin2011adaptive}. When $d=1$, it reduces to the Myopic feedback model studied in \cite{golovin2011adaptive} and \cite{salha2018adaptive}.
\end{remark}

\begin{remark}
For Problem \ref{problem: k-d}, as we always select one node in each seeding step, we restrict our attention to the policy $\pi$ where $|\pi(S,\phi)|=1$ for each $S \subseteq V$ and $\phi \in \Phi$.
\end{remark}
The adaptive seeding process in the $(k,d)$-AIM problem under a policy $\pi$ is described as follows:

\begin{definition}[\textbf{$(\pi, k, d)$-process}]~
\begin{itemize}
\item Set $(S,\phi)$ as $(\emptyset, \phi_\emptyset)$. Repeat the following process for $k$ times. 
\begin{itemize}
\item (\textbf{seeding-step}) Select and activate the node $\pi(S,\phi)$.
\item (\textbf{observing-step}) Observe the diffusion for $d$ rounds. Update $(S,\phi)$ by setting $S$ as the set of the current active nodes and $\phi$ as the current realization.
\end{itemize}
\item Wait for the diffusion to terminate and output the number of active nodes.
\end{itemize}
\end{definition}

We define that we can wait for one diffusion round even if there is no node can activate their neighbors, which conceptually allows us to wait for any number of rounds. The output of the above diffusion process is nondeterministic since the diffusion process is stochastic. We use $F(\pi, k, d)$ to denote the expected number of the active nodes produced by the $(\pi, k,d)$-process. Problem \ref{problem: k-d} can be restated as follows.

\begin{problem}
\label{problem: k-d-2}
Given $k \in \mathbb{Z}^+$ and $d \in \mathbb{Z}^+$, find a policy $\pi$ such that $F(\pi, k, d)$ is maximized.
\end{problem}

\begin{definition}[\textbf{Greedy Policy $\pi_g$}]
\label{def: greedy}
Given a status $(S,\phi)$, the greedy policy $\pi_g$ always select the node that can maximize the marginal gain conditioned on $(S,\phi)$, and therefore, 
\begin{equation}
\pi_g(S,\phi)= \argmax_v \Delta f_\infty(S,v ,\phi)
\end{equation} 
\end{definition}

The greedy strategy is the most popular strategy due to its excellent performance and simple implementation. It has been adopted by most of the research regarding the AIM problem and has been shown to have an approximation ratio of $1-1/e$ for several special cases where this ratio is tight. In this paper, we are particularly interested in the performance bound of the greedy policy for Problem \ref{problem: k-d}.

\subsection{Regret Ratio}
\label{sec: ratio}
As aforementioned, Problem \ref{problem: k-d} is not adaptive submodular in general, which is essentially caused by not waiting for the diffusion to terminate before making seeding decisions. In this section, we provide a metric dealing with such scenarios. 

Suppose that the current status is $U$, the budget is one, and we aim at maximizing the number of active nodes after $t$ rounds. Let us consider two ways to deploy this seed node.
\begin{itemize}
\item \textbf{Method 1}: We select the seed node immediately based on $U$. In this case, the best marginal profit we can achieve is $\max_v \Delta f_t\big(\dot{S}(U),v ,\dot{\phi}(U)\big)$, by selecting $\argmax_v \Delta f_t\big(\dot{S}(U),v ,\dot{\phi}(U)\big)$ as the seed node. 

\item \textbf{Method 2}: We wait for $d<t$ rounds of diffusion and then select the seed node. After $d$ rounds, for each possible status $U_* \in \dot{\U}_d(U)$, the best marginal profit would be $\max \Delta f_{t-d}\big(\dot{S}(U_*),v ,\dot{\phi}(U_*)\big)$. Thus, the total marginal profit would be $\sum_{U_* \in \dot{\U}_d(U)} \Pr[\dot{\phi}(U_*)|\dot{\phi}(U)] \cdot \max \Delta f_{t-d}\big(\dot{S}(U_*),v ,\dot{\phi}(U_*)\big)$.
\end{itemize}

\begin{definition}[\textbf{Regret Ratio}]
\label{def: ratio}
For a status $U$ and two integers $t, d \in \mathbb{Z}^+$ with $d \leq t$, we define the regret ratio $\alpha_{t,d}(U) $ as 
{\small
\begin{eqnarray*}
&\alpha_{t,d}(U)\define \\ & \dfrac{\sum_{U_* \in \dot{\U}_d(U)} \Pr[\dot{\phi}(U_*)|\dot{\phi}(U)] \cdot \max \Delta f_{t-d}\big(\dot{S}(U_*),v ,\dot{\phi}(U_*)\big)}{\max_v \Delta f_t\big(\dot{S}(U),v ,\dot{\phi}(U)\big)},
\end{eqnarray*}}
which measures the ratio of the marginal profits resulted by those two methods. When there is no time constraint (i.e., $t=\infty$), we denote it as 
{\small
\begin{eqnarray*}
&\alpha_{\infty,d}(U) \define\\
&\dfrac{\sum_{U_* \in \dot{\U}_d(U)} \Pr[\dot{\phi}(U_*)|\dot{\phi}(U)] \cdot \max \Delta f_{\infty}\big(\dot{S}(U_*),v ,\dot{\phi}(U_*)\big)}{\max_v \Delta f_\infty\big(\dot{S}(U),v ,\dot{\phi}(U)\big)}.
\end{eqnarray*}}
Furthermore, if we wait until the diffusion terminates before selecting the next seed node (i.e., $d=\infty$), we have 
{\small
\begin{eqnarray*}
&\alpha_{\infty,\infty}(U)\define \\
&\dfrac{\sum_{U_* \in \dot{\U}_\infty(U)} \Pr[\dot{\phi}(U_*)|\dot{\phi}(U)] \cdot \max \Delta f_{\infty}\big(\dot{S}(U_*),v ,\dot{\phi}(U_*)\big)}{\max_v \Delta f_\infty\big(\dot{S}(U),v ,\dot{\phi}(U)\big)}.
\end{eqnarray*}}
For a decision tree, we define that \[\alpha(T)\define \max_{U \in \bigcup {U}_i^T} \alpha_{\infty,\infty}(U)\] which is the largest $\alpha_{\infty,\infty}(U)$ among all the statuses in the tree.
\end{definition}

\begin{remark}
\label{remark: ratio_final}
We can see that if $U$ is a final status, then the ratio $\alpha_{t,d}(U)$ is always no larger than $1$ as no more observation can be made, indicating that the first method is a dominant strategy. Because the root status is always final in each decision tree $T$, $\alpha(T)$ is no less than 1. Furthermore, if $U$ is a final status and $t=\infty$, we have $\alpha_{\infty,d}(U)=1$ for each $d$.
\end{remark}

\begin{remark}
When $t=\infty$, the ratio $\alpha_{\infty,d}(U)$ is always no less than $1$, and therefore the second method is a dominant strategy. In this case, $\alpha_{\infty,d}(U)$ shows the penalty incurred by not waiting for the diffusion to terminate before determining the seed node. Furthermore, $\alpha_{\infty,d}(U)$ will not decrease with the increase of $d$. That is, the more rounds we wait for, the more profit we can gain. We will use $\alpha_{\infty,\infty}$ to bound the approximation ratio of the greedy policy. 
\end{remark}

\begin{remark}
\label{remark: trade-off}
When $t<n-1$ and $U$ is not final, there is a trade-off determined by the number of the rounds $d$ we would wait for. If we waited for more rounds, we would have more observations and had a better chance to explore high-quality seed nodes, but meanwhile we would lose more diffusion rounds. 
\end{remark}

\subsection{Main Result}
The main result of our paper is shown below.
\begin{theorem}
\label{theorem: main}
For each policy $\pi_*$, \[F(\pi_g, k, d) \geq (1-e^{-1/\alpha(T_g)})\cdot F(\pi_*, k, d),\] where $T_g$ is the decision tree of the $(\pi_g,k,d)$-process
\end{theorem}
\begin{remark}
The $(k,d)$-AIM problem has been studied for the case $d=\infty$. An early discussion was given by Golovin and Krause in \cite{golovin2010adaptive} where it was confirmed that the greedy policy provides a $(1-1/e)$-approximation to the $(k,\infty)$-AIM problem based on the concept of adaptive submodularity. According to Remark \ref{remark: ratio_final}, $\alpha(T_g)=1$ when $d=\infty$, because all the status in  $T_g$ are now final status, and therefore we have $1-1/e$ again for this special case. \footnote{In \cite{tong2017adaptive}, we claimed that the greedy algorithm gives a $(1-1/e)$-approximation for Problem \ref{problem: k-d} when $d< n-1$. We now retract that claim.}
\end{remark}

\subsection{Proof of Theorem \ref{theorem: main}}  
This section gives the main line of the proof where the details are deferred to the supplementary material. The main idea is to introduce another seeding process which is (a) equivalent to the seeding process we consider and (b) comparable to the greedy policy via the regret ratio.

Before delving into the proof, we first introduce the necessary preliminaries. For the purpose of analysis, let us consider the following L-$(\pi, k, d)$-process.

\begin{definition}[\textbf{L-$(\pi, k, d)$-process}]
Set $(S,\phi)$ as $(\emptyset, \phi_\emptyset)$. 
\begin{itemize}
\item \textbf{Step 1}. Repeat the following process for \textit{$k-1$ times}. 
\begin{itemize}
\item (\textbf{seeding-step}) Select and activate the node $\pi(S,\phi)$.
\item (\textbf{observing-step}) Observe the diffusion for $d$ rounds. Update $(S,\phi)$ by setting $S$ as the set of the current active nodes and $\phi$ as the current realization.
\end{itemize}
\item \textbf{Step 2}. 
\begin{itemize}
\item (\textbf{seeding-step}) Decide the seed node $v^*=\pi(S,\phi)$, but do not activate $v^*$. This is the $k$-th seeding step.
\item (\textbf{observing-step}) Wait for the diffusion to terminate.
\end{itemize}
\item \textbf{Step 3}. 
\begin{itemize}
\item (\textbf{seeding-step}) Activate the node $v^*$. This is the $(k+1)$-th seeding step.
\end{itemize}
\item Wait for the diffusion to terminate and output the number of active nodes.
\end{itemize}
\end{definition}

The L-$(\pi, k, d)$-process is identical to the $(\pi, k, d)$-process except that the last seed node is activated with a delay. However, the total profit remains the same. Let $F_L(\pi,k,d)$ be the expected number of the active nodes output by the L-$(\pi, k, d)$-process. We have the following result.

\begin{lemma}
\label{lemma: lazy_same}
$F(\pi,k,d)=F_L(\pi,k,d)$.
\end{lemma}
\begin{proof}
For a fixed full realization, the $(\pi,k,d)$-process and L-$(\pi,k,d)$-process always select the same seed nodes and the only difference is that the last seed in the L-$(\pi,k,d)$-process may be seeded with a delay. However, with respect to the number of active nodes, it does not matter that when we make the seed node activated, as long as we allow the diffusion process to finally terminate. The idea of lazy seeding was also seen early in \cite{mossel2007submodularity}. 
\end{proof}

\begin{figure}[t]
	\begin{center}
		\includegraphics[width=3in]{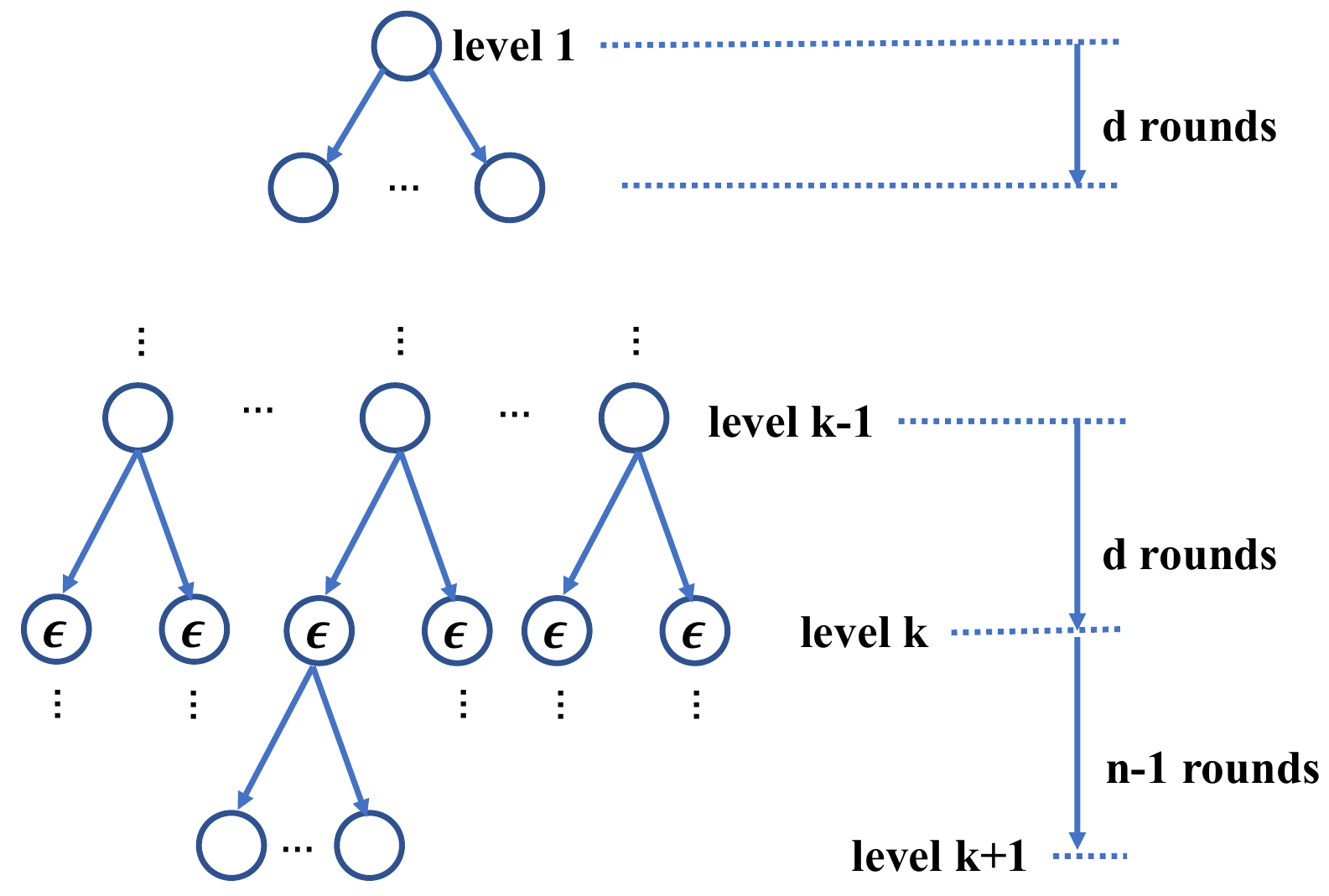} 
	\end{center} 
	\vspace{-2mm}
	\caption{\textbf{L-$(\pi, k, d)$-tree}}
	\vspace{-2mm}
	\label{fig: k_tree}
\end{figure}

\begin{definition}[\textbf{L-$(\pi, k, d)$-tree}]
\label{def: L-tree}
We denote the decision tree of the L-$(\pi, k, d)$-process as the L-$(\pi, k, d)$-tree. In an L-$(\pi, k, d)$-process, there are totally $k$ observing-steps and $k+1$ seeding-steps. Note that in the $k$-th seeding step no node is activated and therefore we label the tree-node by a special character $\epsilon$. An illustration is given in Fig. \ref{fig: k_tree}. 
\end{definition}

\begin{remark}
Different from the $(\pi, k, d)$-process, the leaves in a  L-$(\pi, k, d)$-tree are now final status, which is the key to establishing the performance bound based on the regret ratio.
\end{remark}
\begin{definition}[\textbf{Decision Tree Profit}]
\label{def: tree_profit}
For a decision tree $T$, we define that
{\small
\begin{align}
\label{eq: F(T)}
&F(T) \nonumber  \\
&\define\sum_{U \in U_{\infty}^T} \Pr[\dot{\phi}(U)]\sum_{\substack{\psi \in \Psi \\ \dot{\phi}(U) \prec \psi}} \Pr[\psi|\dot{\phi}(U)] \cdot |A_\infty(\dot{S}(U)+\dot{S}_{\e}(U),\psi)| \nonumber \\
&=\sum_{U \in U_{\infty}^T} \sum_{\substack{\psi \in \Psi \\ \dot{\phi}(U) \prec \psi}} \Pr[\psi]\cdot |A_\infty(\dot{S}(U)+\dot{S}_{\e}(U),\psi)|
\end{align}}
as the profit of the decision tree\footnote{We slightly abuse the notation $F$ by allowing it to have different definitions for different types of inputs.}, which is the immediate sum of the profit among all possible outcomes. 
\end{definition}

\begin{remark}
\label{remark: profit}
As one can see, if $T$ is the decision tree of a certain adaptive seeding process, then $F(T)$ is the expected number of active users resulted by the process. Thus, we have $F(\pi,k,d)= F_L(\pi,k,d)=F(T)$ where $T$ is the L-$(\pi, k, d)$-tree. 
\end{remark}

Note that $F(T_1 \oplus T_2)$ is also well-defined for a concatenation of two trees, and we have the following result showing that a concatenation will not decrease the total profit.
\begin{lemma}[\cite{golovin2011adaptive}]
	\label{lemma: monotone}
	$F(T_1 \oplus T_2)\geq F(T_2)$
\end{lemma}
\begin{proof}
One can obtain a proof using the view of policy similar to that given in  \cite{golovin2011adaptive}. We provide a proof from the view of decision tree in the supplementary material. 
\end{proof}

\begin{remark}
According to Remark \ref{remark: profit}, Eq. (\ref{eq: F(T)}) gives an explicit formula of $F_L(\pi,k,d)$. However, this formula is only used for analysis and it is not feasible to compute Eq. (\ref{eq: F(T)}) directly as there are exponential number of terms to sum. If $d=\infty$, L-$(\pi, k, d)$-tree is the same as $(\pi, k, d)$-tree.
\end{remark}

Now we are ready to prove Theorem \ref{theorem: main}. By Lemma \ref{lemma: lazy_same}, it is sufficient to analyze the L-$(\pi_g, k, d)$-process. In the rest of this section, we assume $\pi_*$, $d$ and $k$ are fixed, where $\pi_*$ can be an arbitrary policy. For each $i \in \{1,..., k\}$, we use $T_g^i$ to denote the decision tree of the L-$(\pi_g, i, d)$-process, and use $T_*^i$ to denote the decision tree of the $(\pi_*, i, d)$-process.   By Remark \ref{remark: profit}, we have $F(T_g^k)=F(\pi_g, k, d)$ and $F(T_*^k)=F(\pi_*, k, d)$, and therefore to prove Theorem \ref{theorem: main}, it suffices to show that 
\begin{eqnarray}
\label{eq: tree_ratio}
F(T_g^k) \geq (1-e^{-1/\alpha(T_g)}) \cdot F(T_*^k),
\end{eqnarray}
where $\alpha(T_g)$ is defined in Def. \ref{def: ratio}. For two integers $i, j \in \{1,...,k\}$, let us consider the decision tree $T_g^i \oplus T_*^j$. We define that $T_g^0 \oplus T_*^j\define T_*^j$ and $T_g^i \oplus T_*^0\define T_g^i$. For conciseness, we denote $T_g^i \oplus T_*^j$ by $T_{i,j}$. Furthermore, we define that $F(T_{0, 0})\define 0$ but do not define the tree $T_{0, 0}$.\footnote{One can imagine $T_{0, 0}$ as an empty tree.} By Lemma \ref{lemma: monotone}, we have
\begin{equation}
F(T_{i,k})\geq F(T_{0,k}).
\end{equation}

The following is a key lemma for proving Theorem \ref{theorem: main}. 
\begin{lemma}
\label{lemma: key}
For each $i, l \in \{1,...,k\}$, we have
\[F(T_{i-1,l})-F(T_{i-1,l-1}) \leq \alpha(T_g)\cdot \big(F(T_{0,i})-F(T_{0,i-1})\big).\]
\end{lemma}
\begin{proof}
The proof is based on subtle arrangements of the equations given by the decision tree so that we can analyze the marginal profit in a fine-grained manner. This is the main theoretical result of this paper, but for the sake of continuity we put this part in the supplementary material.
\end{proof}
With Lemma \ref{lemma: key}, the rest of the proof follows the standard analysis of submodular maximization. Summing the inequality in Lemma \ref{lemma: key} over $l \in \{1,...,k\}$, we have \[F(T_{i-1,k})-F(T_{i-1,0}) \leq k\cdot \alpha(T_g)\cdot \big(F(T_{0,i})-F(T_{0,i-1})\big),\] 
and, due to Lemma \ref{lemma: monotone} we have \[F(T_{0,k})-F(T_{i-1,0}) \leq k\cdot \alpha(T_g)\cdot  \big(F(T_{0,i})-F(T_{0,i-1})\big).\] Define that $ \bigtriangleup_i \define F(T_{0,k})-F(T_{i-1,0})$, and we therefore have $\bigtriangleup_i \leq k\cdot \alpha(T_g)\cdot (\bigtriangleup_i-\bigtriangleup_{i+1})$, implying $\bigtriangleup_{i+1} \leq (1-\frac{1}{\alpha(T_g)\cdot k})\cdot \bigtriangleup_{i}$ and consequently 
\[\bigtriangleup_{k+1} \leq (1-\frac{1}{\alpha(T_g)\cdot k})^k \cdot \bigtriangleup_{1} \leq \exp{(-\frac{1}{\alpha(T_g)})}\cdot \bigtriangleup_{1}.\] 
As a result, we have $F(T_{0,k})-F(T_{k,0}) \leq \exp{(-\frac{1}{\alpha(T_g)\cdot})} \cdot \big(F(T_{0,k})-F(T_{0,0})\big)$, and therefore, $F(T_g^k)=F(T_{k,0})\geq (1-e^{-1/\alpha(T_g)}) \cdot F(T_{0,k})=(1-e^{-1/\alpha(T_g)}) \cdot F(T_*^k)$, which completes the proof.

\subsection{Generalized RR-set}
\label{subsec: rrset}
In this section, we discuss one technique that can be used to improve the efficiency of the algorithms concerning the AIM problem. 

In each seeding step, the greedy policy selects the node with the highest marginal profit conditioned on the current status, and this process demands to calculate $ \Delta f_\infty(S,v, \phi)$ which is a \#P-hard problem \cite{chen2013near}. One straightforward method is to utilize Monte Carlo simulation which is unfortunately not efficient as widely discussed. Alternatively, with the idea of reverse sampling \cite{borgs2014maximizing}, an efficient estimating approach is obtainable, and this technique has been extensively studied for the IM problem or its variants (e.g., \cite{tang2015influence,sun2018multi,han2018efficient, tong2017efficient, tong2019beyond, wang2017activity}). In particular, The previous work has shown that how to use this technique for the case when $S=\emptyset$ and $\phi=\phi_\emptyset$, by utilizing the so-called RR-set. In what follows, we will show that an analogous approach can be obtained to estimate $\Delta f_t(S, V^*, \phi)$ given in Eq. (\ref{eq: f_t}) for the general case. With the concept inherited from \cite{tang2015influence}, we redefine the RR-set as follows.

\begin{definition}[\textbf{RR-set}]
\label{def: rrset}
Given $S \subseteq V$, $\phi \in \Phi$ and $t \in \mathbb{Z}^+$, an RR-set $\R \subseteq V$ is a set of nodes generated randomly as follows.
\begin{itemize}
\item \textbf{Step 1.} For an edge $e$ not in $L(\phi) \cup D(\phi)$, we sample its state according to the probability $p_e$. For an edge $e$ in $L(\phi) \cup D(\phi)$ we keep its state as given by $\phi$. After this step, we obtain a full realization $\psi^*$.
\item \textbf{Step 2.} Select a node $v$ which is not in $S$ uniformly at random, and let $R_v \subseteq V$ be the set of the nodes from which $v$ is reachable via a $t$-live-path in $\psi^*$. If $R_v \cap S= \emptyset$, we return  $\R=R_v$ as the RR-set. Otherwise, we return $\R=\emptyset$.
\end{itemize}
\end{definition}
Following the standard analysis of reverse sampling \cite{tang2015influence}, the following result can be readily derived.
\begin{lemma}
\label{lemma: rrset}
For each $V^* \subseteq V$, we have 
\begin{equation}
\label{eq: lemma_rreset}
(|V|-|S|) \cdot \E [\I(V^*\cap \R)]=\Delta f_t(S, V^* ,\phi),
\end{equation}
where 
\[
\I(V^* \cap \R)\define
\begin{cases}
1 &  \hspace{0mm} \hspace{-0.5mm} \text{if $S \cap V^* \neq \emptyset$} \\
0 & \hspace{0mm} \hspace{-0.5mm} \text{else } 
\end{cases}
\]
and $\R$ is a RR-set generated with the input $S, \phi$ and $t$.
\end{lemma}
\begin{proof}
When a node $v$ is selected in Step 2, $\E [\I(V^* \cap R_v)]$ is the marginal increase of the probability that $v$ can be activated. There are two cases depending on whether or not $R_v \cap S$ is empty. If $R_v \cap S \neq \emptyset$, it indicates that $v$ can be activated by the current active nodes and therefore it cannot contribute to $\Delta f_t(S, V^* ,\phi)$ for any $V^*$. Accordingly, $R_v$ is set as $\emptyset$ so that $\I(V^* \cap R_v)=0$ for any $V^*$. If $R_v \cap S = \emptyset$, $R_v$ consists of the nodes that activate $v$ within $t$ rounds, and thus, $\E [\I(V^* \cap R_v)]$  shows the marginal contribution of $V^*$ on node $v$. Since $\Delta f_t(S, V^* ,\phi)$ is sum of the marginal increase of the activation probability over all the nodes not in $S$, calculating the mean of $\I(V^*\cap \R)$ immediately yields Eq. (\ref{eq: lemma_rreset}).
\end{proof}
\begin{remark}
\label{remark: rrset}
Note that there is no need to determine the state of each edge in advance, and instead one can collect the reachable nodes $R_v$ from $v$ in step 2 in a reverse direction.  Such an idea was invented by Borgs \textit{et al.} \cite{borgs2014maximizing}.
\end{remark}
According to Lemma \ref{lemma: rrset}, one can estimate $\Delta f_t(S, V^* ,\phi)$ by utilizing the sample mean of $\I(V^* \cap \R)$, and the estimation can be arbitrarily accurate provided sufficient samples are used. With the generalized RR-set, the greedy policy can be implemented in the similar manner as that in \cite{tang2015influence}. Precisely, the policy with such implementation is in fact a $(1-e^{-1/\alpha(T_g)}-\epsilon)$-approximation where $\epsilon$ is the error incurred by sampling and it can be arbitrarily small. We do not provide detailed analysis for the relationship between $\epsilon$ and the number of used RR-sets, as it is out of the scope of this paper, and the interested reader is referred to \cite{sun2018multi} and \cite{han2018efficient}. 

\subsection{An Upper Bound of $\alpha_{t,d}(U)$} 
\label{subsec: upper_bound}
Calculating the regret ratio $\alpha_{t,d}(U)$ is computationally hard due to the \#P-hardness in computing the function value of  $\Delta f_t(S, V^* ,\phi)$ and furthermore to the fact that it requires to consider an exponential number of statuses after $d$ diffusion rounds.  In this section, we provide an upper bound of $\alpha_{t,d}(U)$. Throughout this part, we assume the function value can be obtained with an arbitrary high accuracy, and therefore the denominator $\max_v \Delta f_t(\dot{S}(U),v ,\dot{\phi}(U))$ of $\alpha_{t,d}(U)$  is obtainable. Let us denote the numerator as 
\begin{eqnarray*}
&N(U)\define \\
&\sum_{U_* \in \dot{\U}_d(U)} \Pr[\dot{\phi}(U_*)|\dot{\phi}(U)] \cdot \max \Delta f_{t-d}(\dot{S}(U_*),v ,\dot{\phi}(U_*)).
\end{eqnarray*}
In addition, let us define a new status $U_{\f}$ where $\dot{S}(U_{\f})=\dot{S}(U)$, $L\big(\dot{\phi}(U_{\f}) \big) \define L(\dot{\phi}(U))$ and 
{\small
\begin{eqnarray*}
&D(\dot{\phi}(U_{\f})) \define \\
&D(\dot{\phi}(U)\cup \{(u,v): (u,v)\notin L(\dot{\phi}(U))\cup D(\dot{\phi}(U)), u\in \dot{\S}(U)\}.
\end{eqnarray*}}
Note that $U_{\f}$ is a super-realization of $U$, and $U_{\f}$ sets each undetermined edge out of an active node in $\dot{S}(U)$ as a dead edge. Informally, it is the worst outcome in $\dot{\U}_d(U)$ with respect to the number of total active nodes. The following lemma gives an upper bound of the numerator.
\begin{lemma}
\label{lemma: upper_bound}
For each status $U$, we have $N(U_{\f}) \geq N(U)$.
\end{lemma}
\begin{proof}
See supplementary material.
\end{proof}
\begin{remark}
One can see that $N(U_{\f})$ can be easily computed as it is a final status. We are particularly interested in the upper bound of $\alpha_{t,d}(U)$ due to the interests in the lower bound of the approximation ratio in Theorem \ref{theorem: main}.
\end{remark}

\section{Experiments}
\label{sec: exp}
The primary goal of our experiments is to evaluate the performance of the greedy algorithm with respect to Problem \ref{problem: k-d} by (a) comparing it with other baseline algorithms and (b) examining the effect of the feedback model on the influence pattern.


\begin{figure*}[!pt]
	\centering
	\subfloat[{[$k=5$, non-daptive]}]{\label{fig: higgs5_3_0}\includegraphics[trim = 0.5in 0in 0.5in 0in, clip, width=0.24\textwidth]{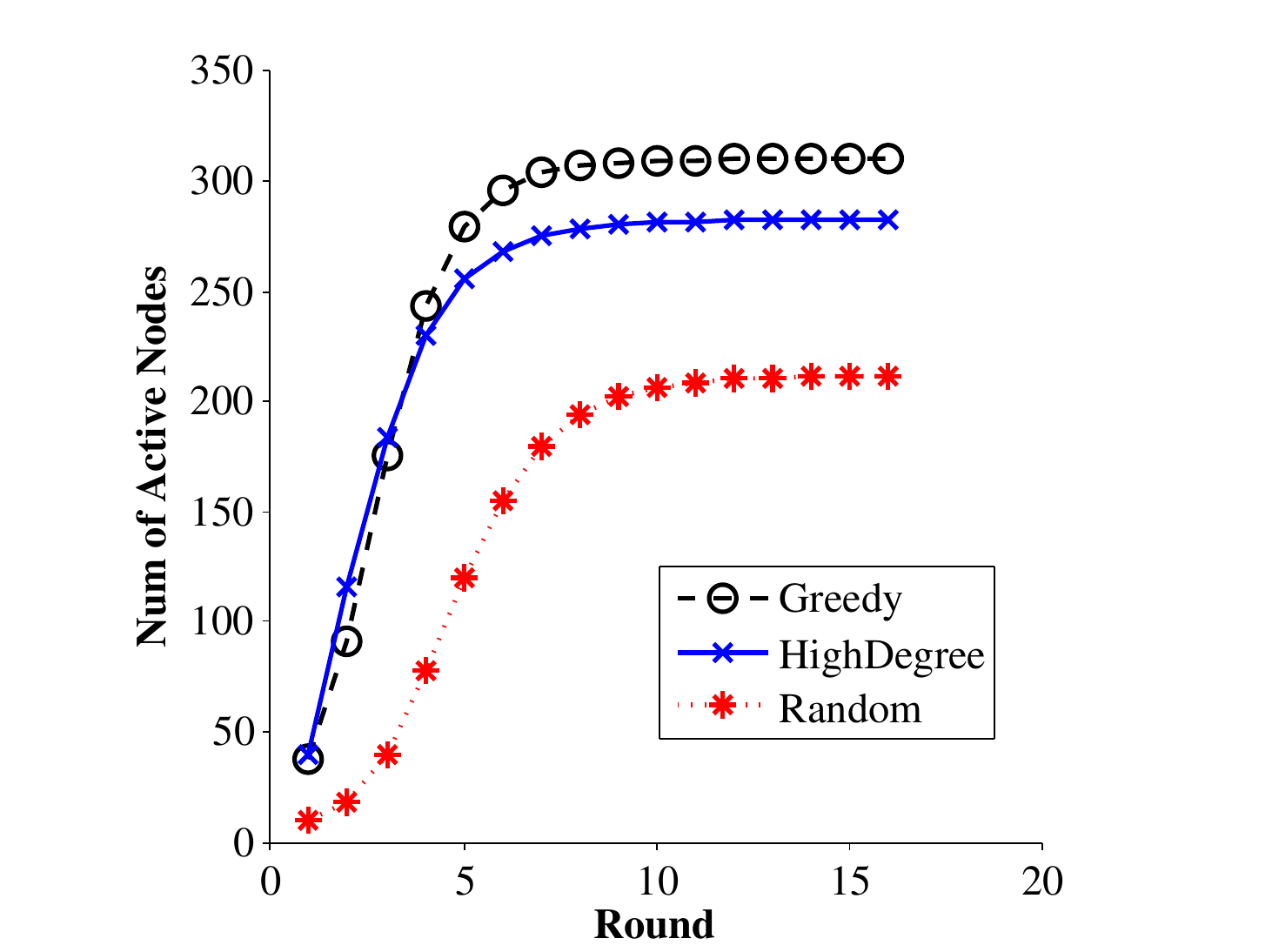}}
	\subfloat[{[$k=5, d=1$]}]{\label{fig: higgs5_3_3}\includegraphics[trim = 0.5in 0in 0.5in 0in, clip,width=0.24\textwidth]{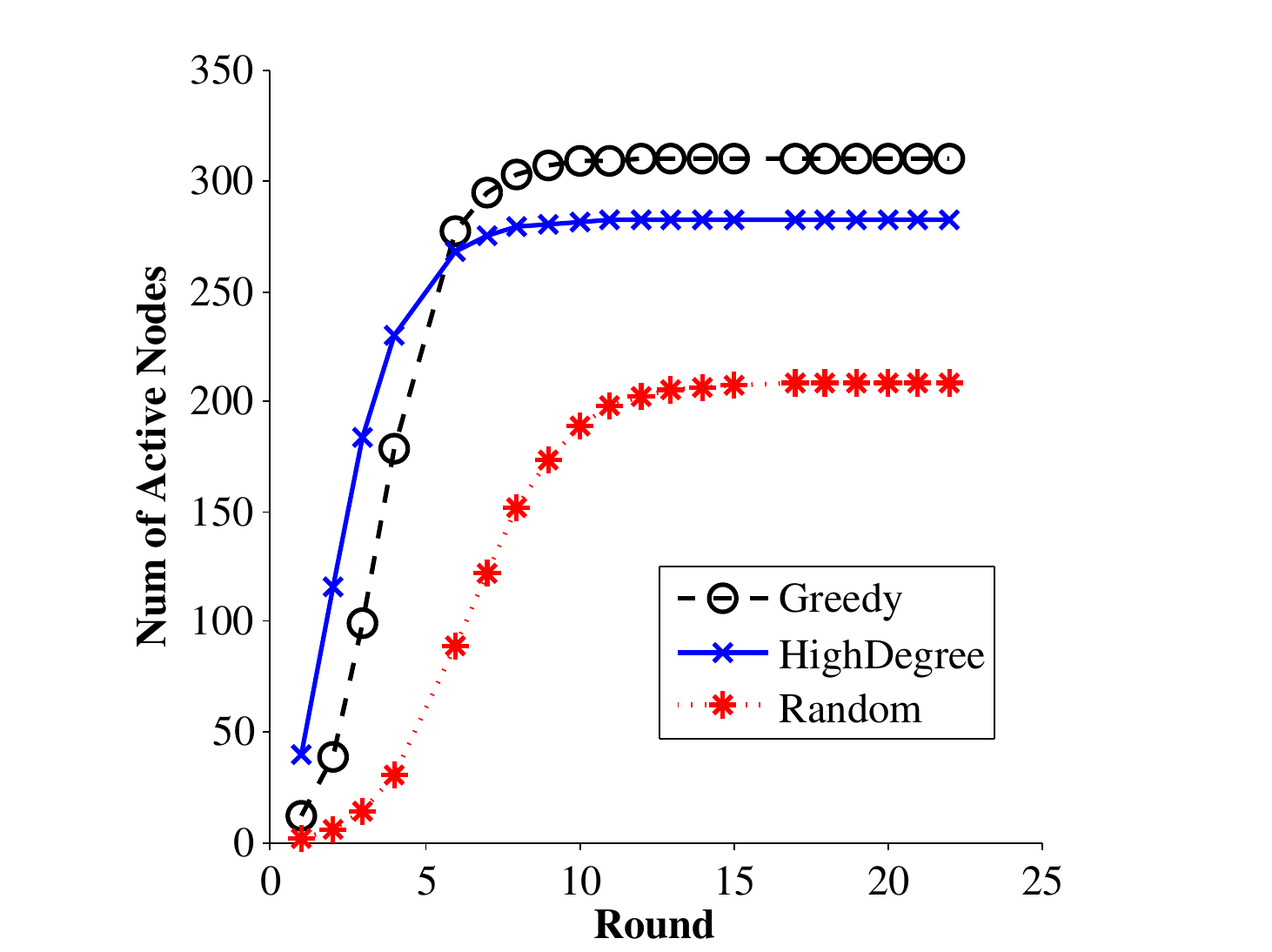}}
	\subfloat[{[$k=5, d=8$]}]{\label{fig: higgs5_3_9}\includegraphics[trim = 0.5in 0in 0.5in 0in, clip,width=0.24\textwidth]{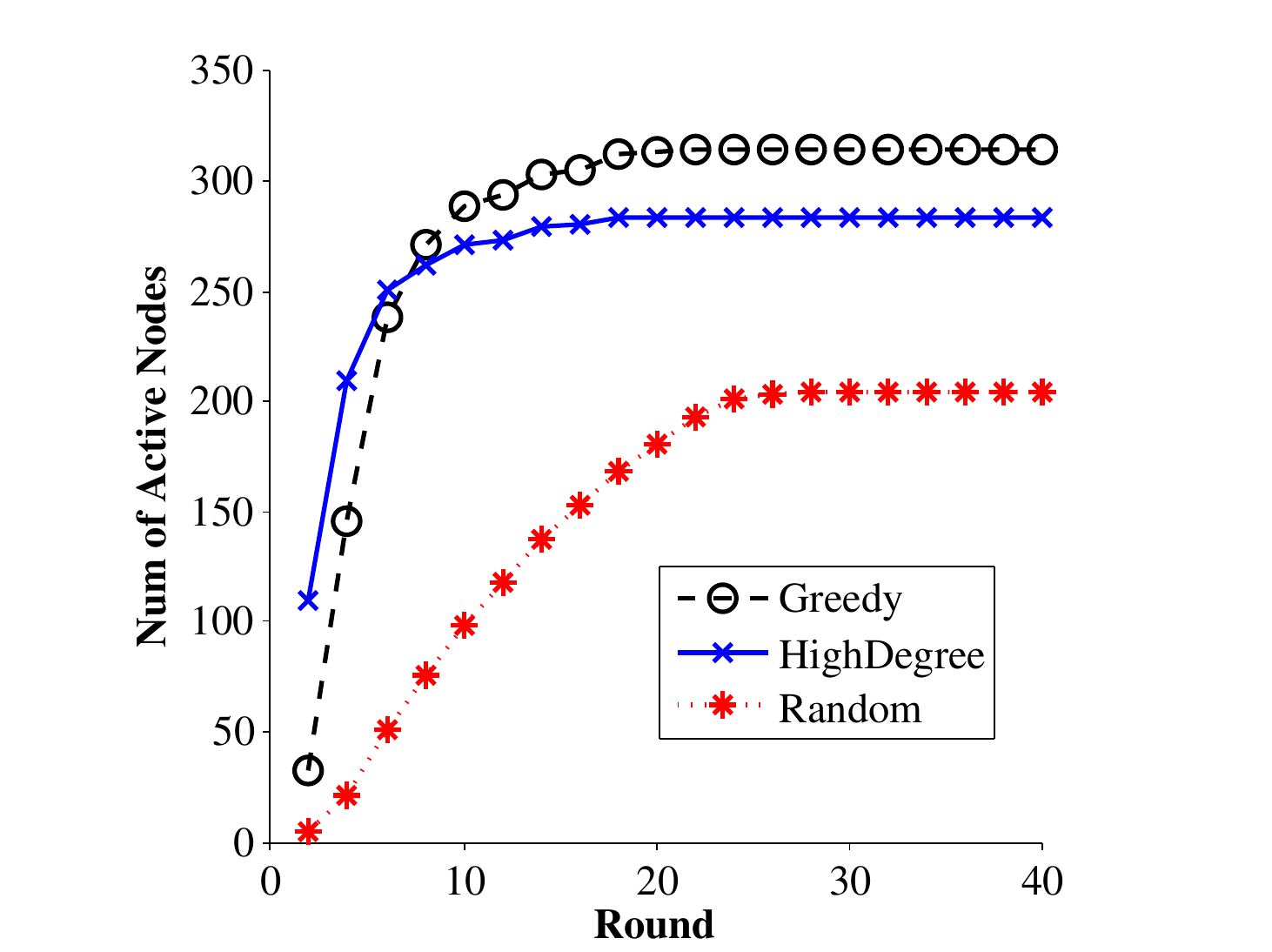}}
	\subfloat[{[$k=5, d=\infty$]}]{\label{fig: higgs5_3_15}\includegraphics[trim = 0.5in 0in 0.5in 0in, clip,width=0.24\textwidth]{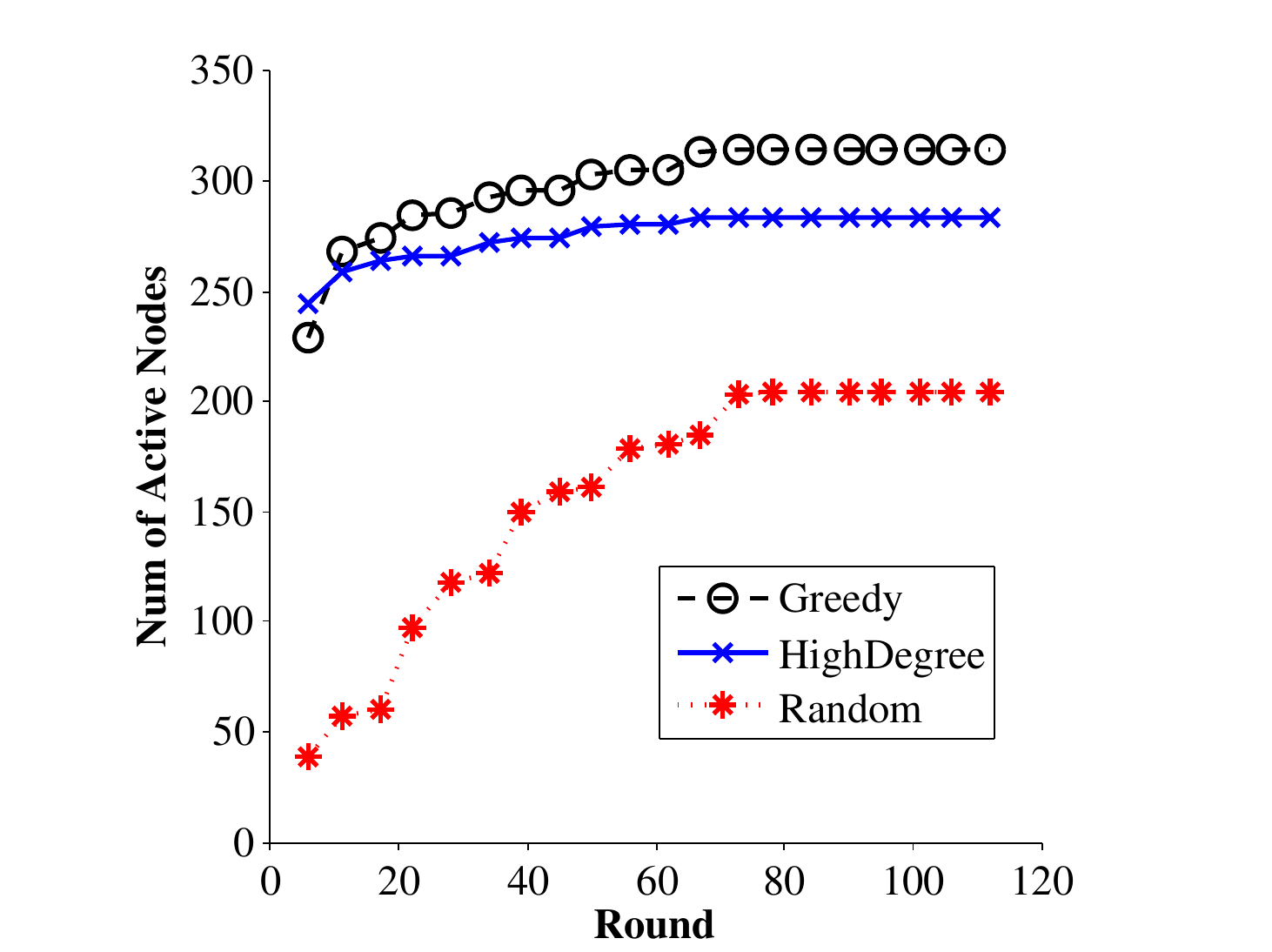}}
	
	\subfloat[{[$k=50$, non-daptive]}]{\label{fig: higgs50_3_0}\includegraphics[trim = 0.5in 0in 0.5in 0in, clip,width=0.24\textwidth]{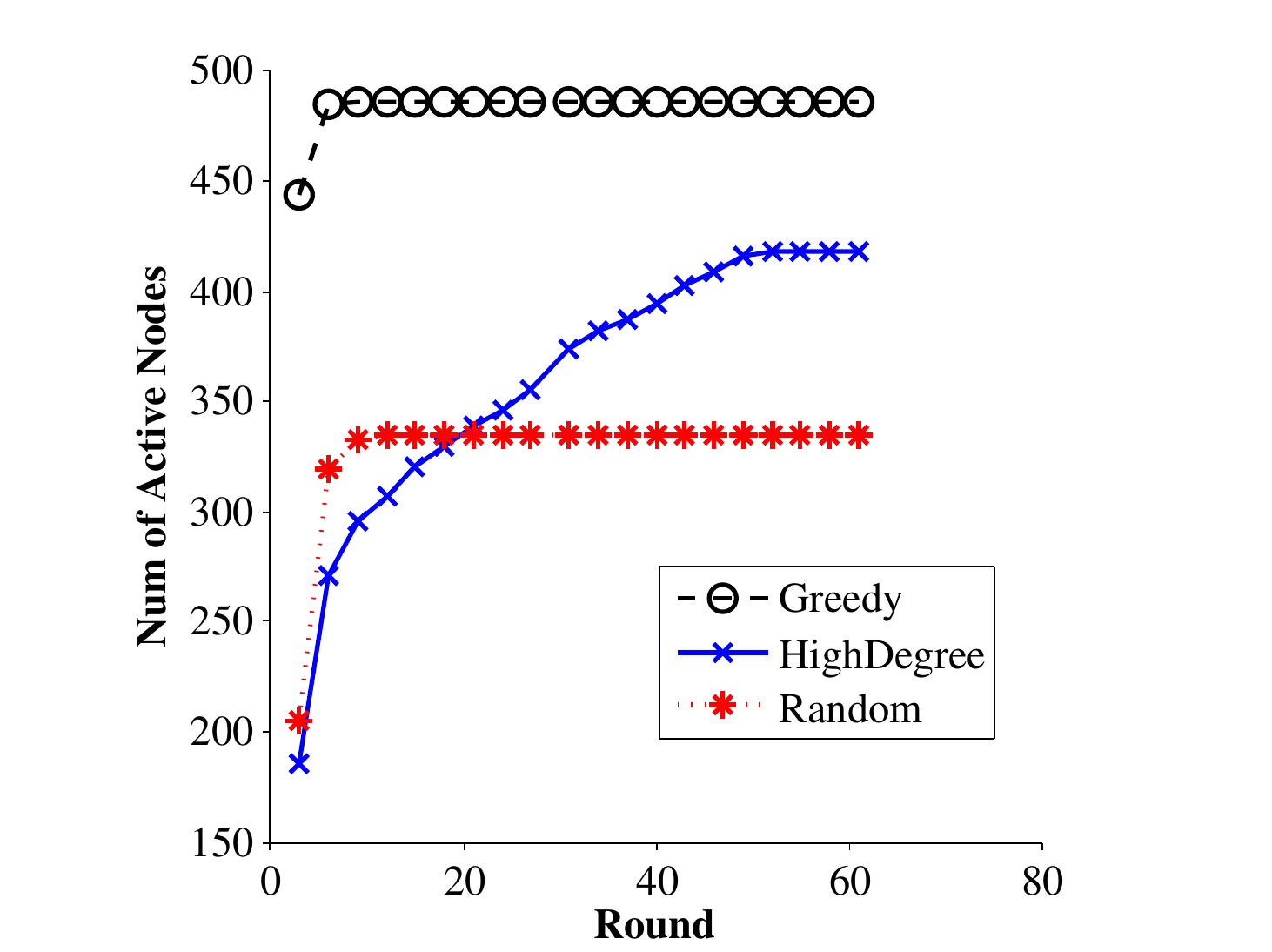}}
	\subfloat[{[$k=50, d=1$]}]{\label{fig: higgs50_3_3}\includegraphics[trim = 0.5in 0in 0.5in 0in, clip,width=0.24\textwidth]{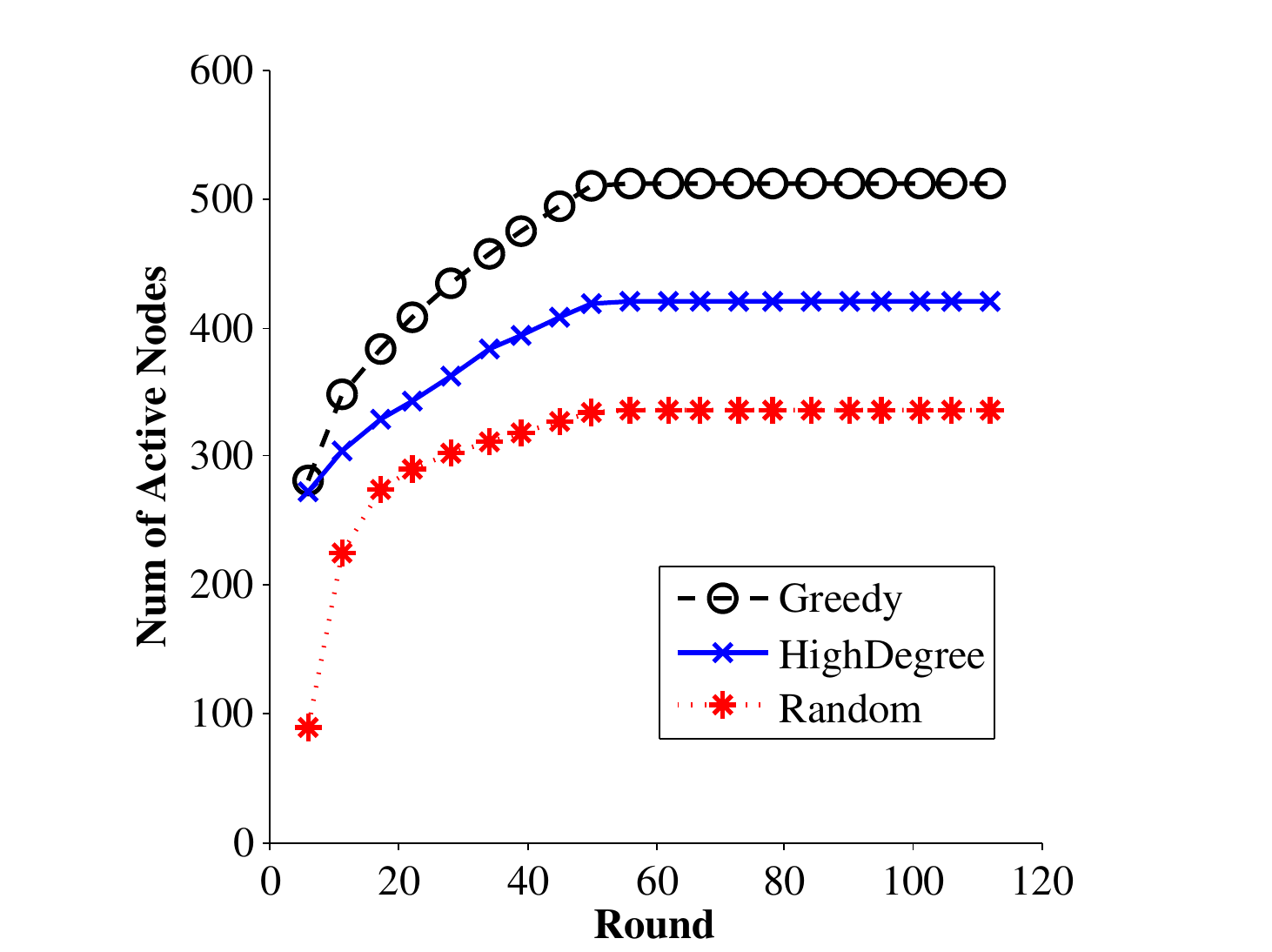}}
	\subfloat[{[$k=50, d=8$]}]{\label{fig: higgs50_3_9}\includegraphics[trim = 0.5in 0in 0.5in 0in, clip,width=0.24\textwidth]{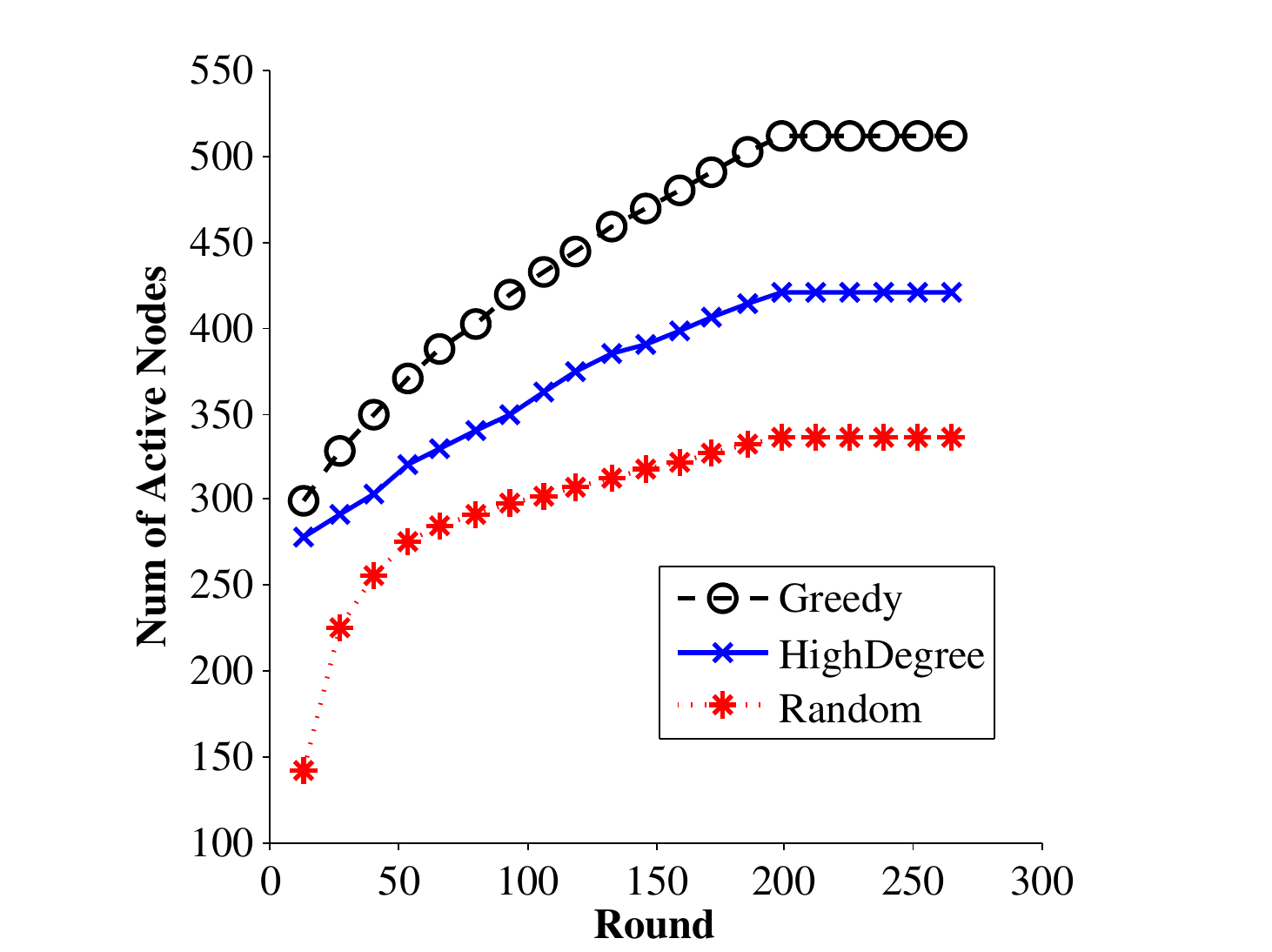}}
	\subfloat[{[$k=50, d=\infty$]}]{\label{fig: higgs50_3_15}\includegraphics[trim = 0.5in 0in 0.5in 0in, clip,width=0.24\textwidth]{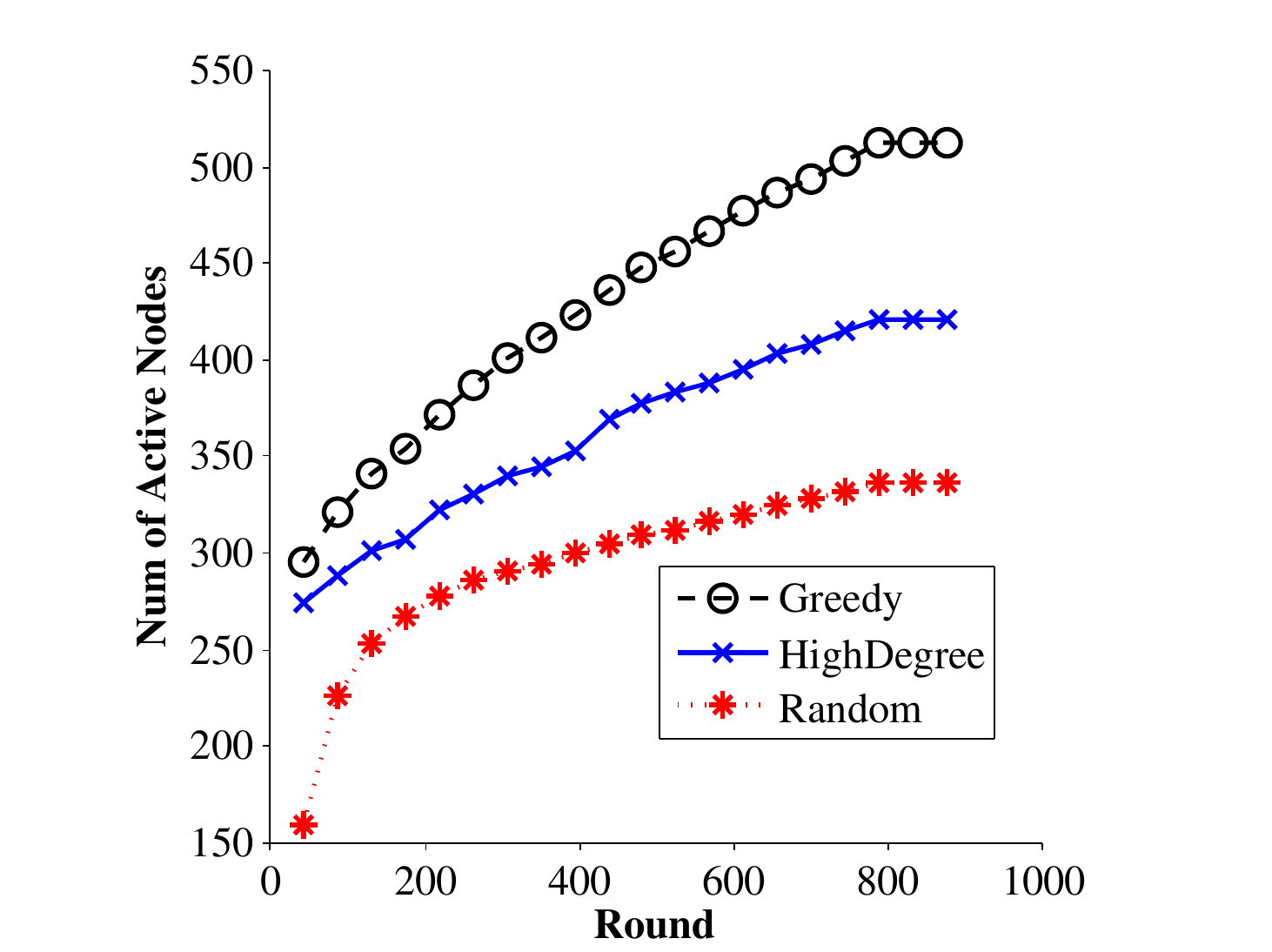}}
	\vspace{-1mm}
	\caption{Results of Experiment \RNum{1} on Higgs}
	\vspace{-5mm}
	\label{fig: exp1_higgs}
\end{figure*}

\begin{figure*}[!pt]
	\centering
	\subfloat[{[$k=5$, non-daptive]}]{\label{fig: hepph5_3_0}\includegraphics[trim = 0.5in 0in 0.5in 0in, clip, width=0.24\textwidth]{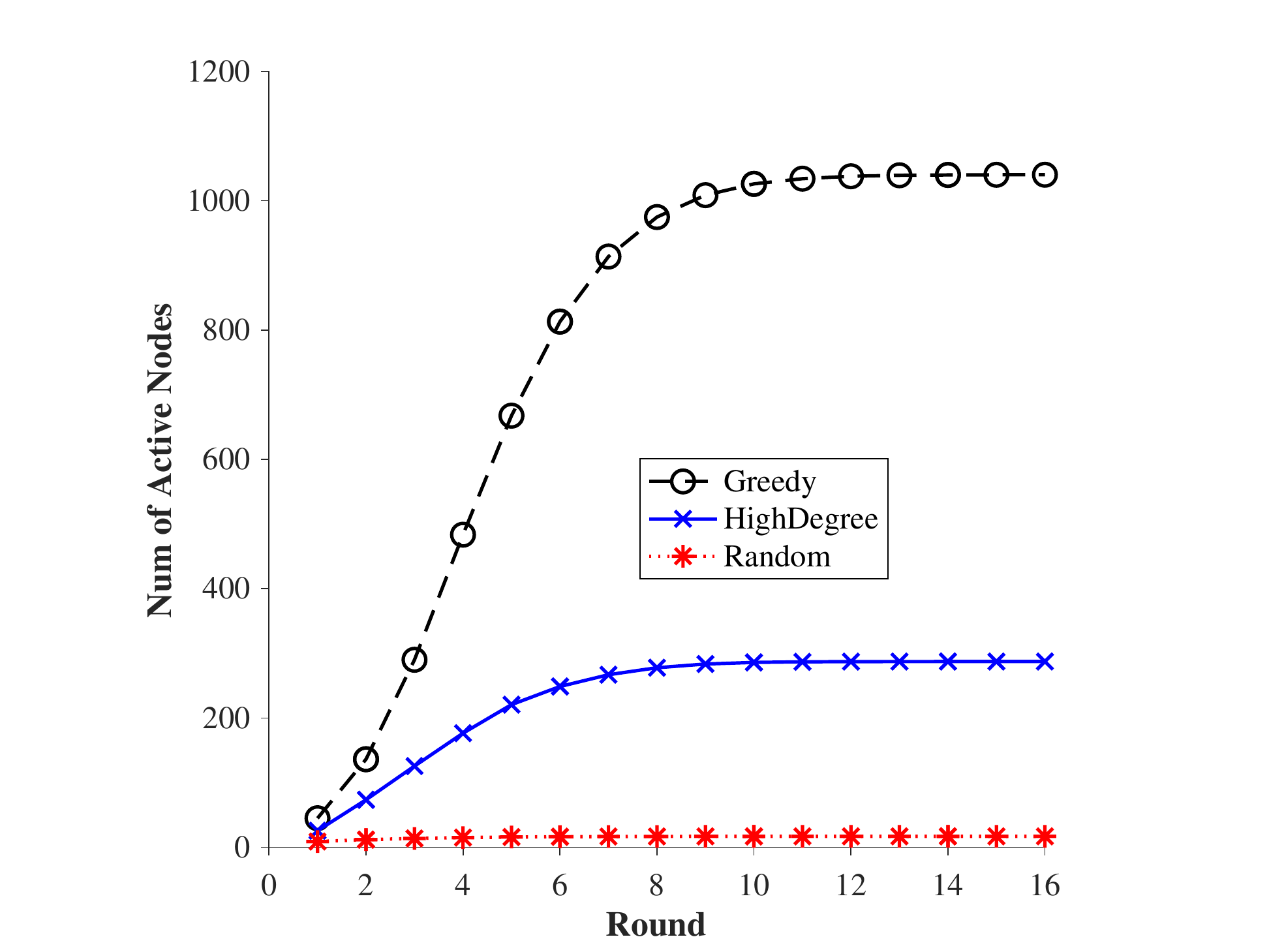}}
	\subfloat[{[$k=5, d=1$]}]{\label{fig: hepph5_3_3}\includegraphics[trim = 0.5in 0in 0.5in 0in, clip,width=0.24\textwidth]{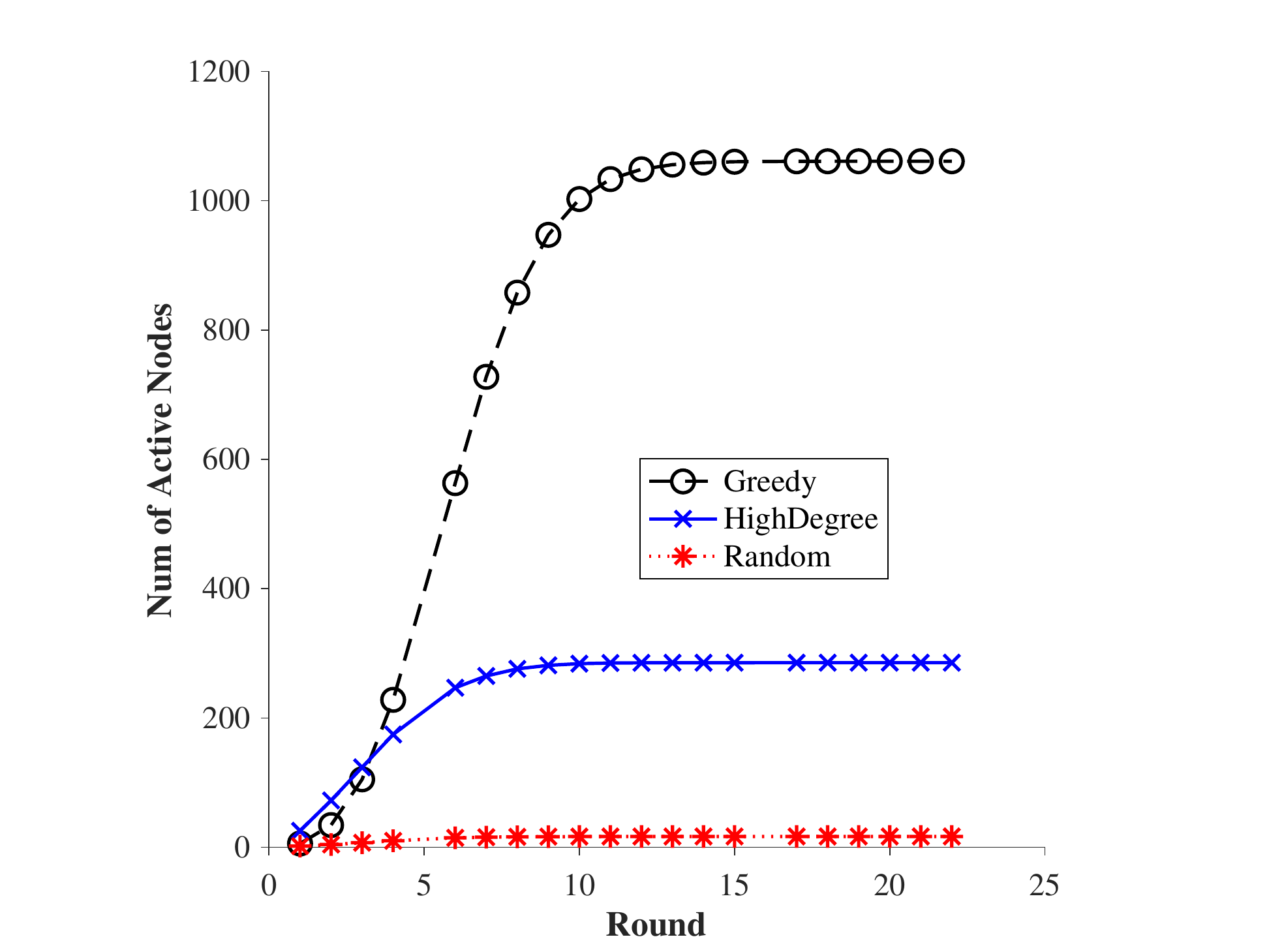}}
	\subfloat[{[$k=5, d=8$]}]{\label{fig: hepph5_3_9}\includegraphics[trim = 0.5in 0in 0.5in 0in, clip,width=0.24\textwidth]{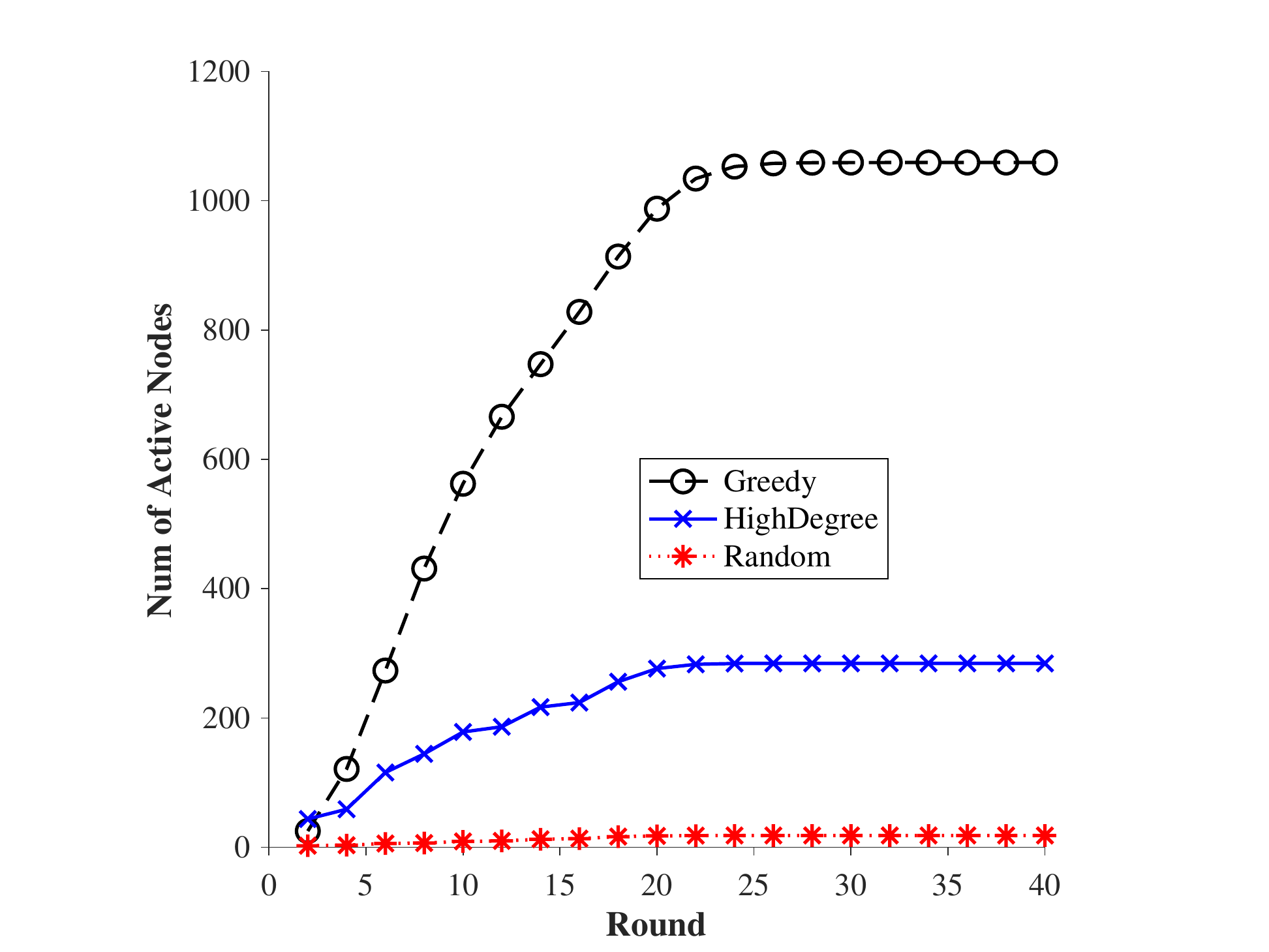}}
	\subfloat[{[$k=5, d=\infty$]}]{\label{fig: hepph5_3_15}\includegraphics[trim = 0.5in 0in 0.5in 0in, clip,width=0.24\textwidth]{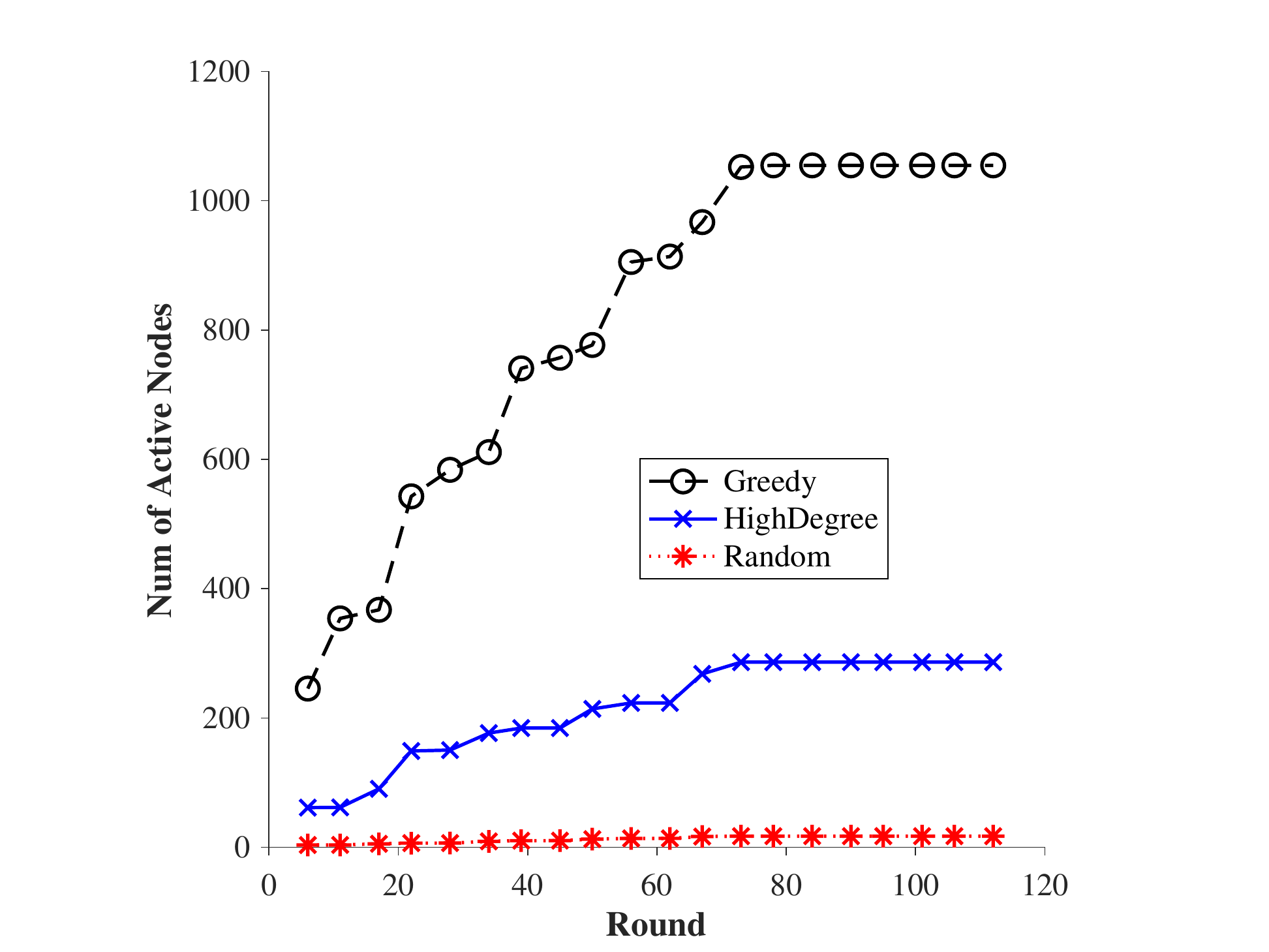}}
	
	\subfloat[{[$k=50$, non-daptive]}]{\label{fig: hepph50_3_0}\includegraphics[trim = 0.5in 0in 0.5in 0in, clip,width=0.24\textwidth]{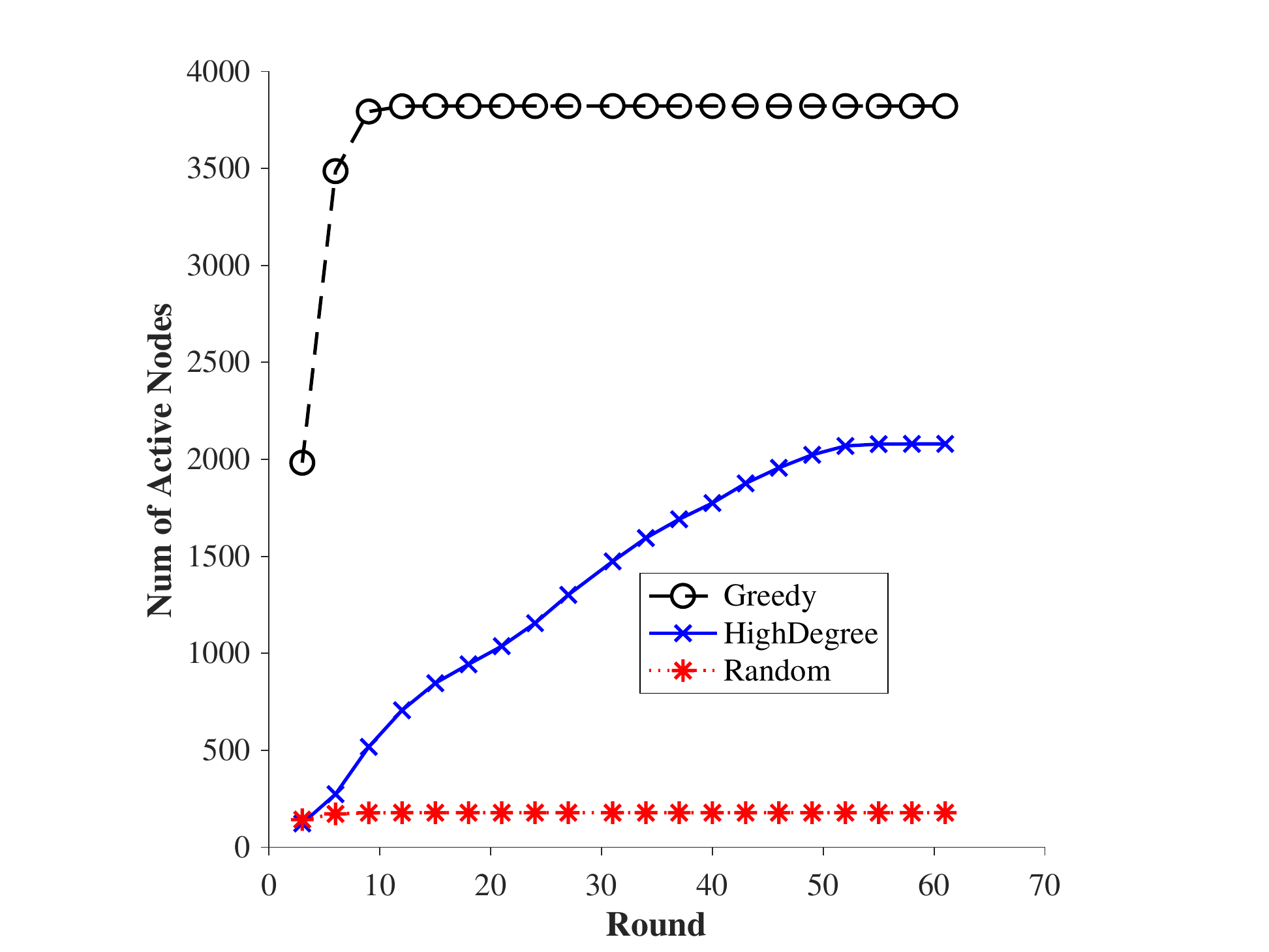}}
	\subfloat[{[$k=50, d=1$]}]{\label{fig: hepph50_3_3}\includegraphics[trim = 0.5in 0in 0.5in 0in, clip,width=0.24\textwidth]{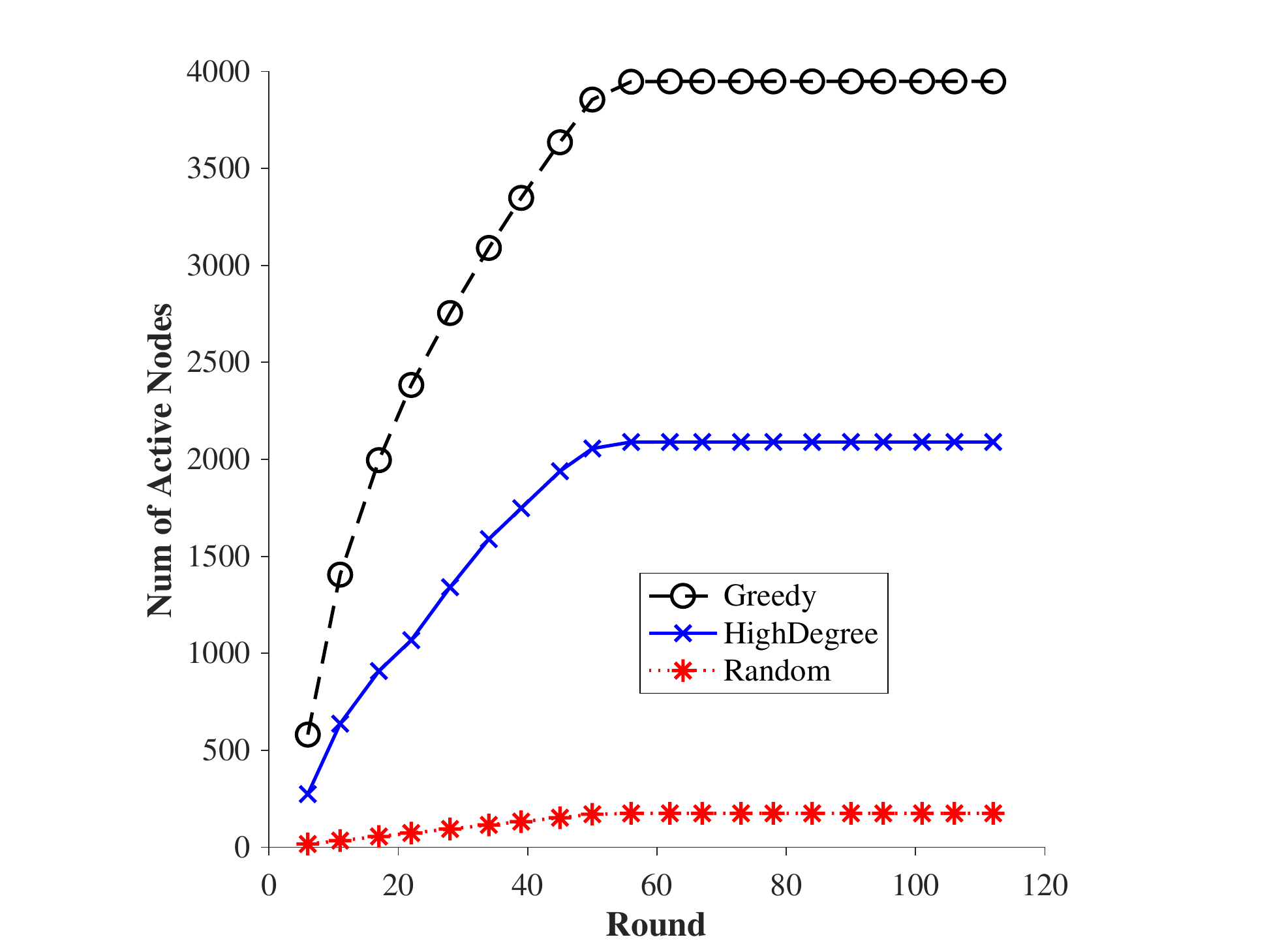}}
	\subfloat[{[$k=50, d=8$]}]{\label{fig: hepph50_3_9}\includegraphics[trim = 0.5in 0in 0.5in 0in, clip,width=0.24\textwidth]{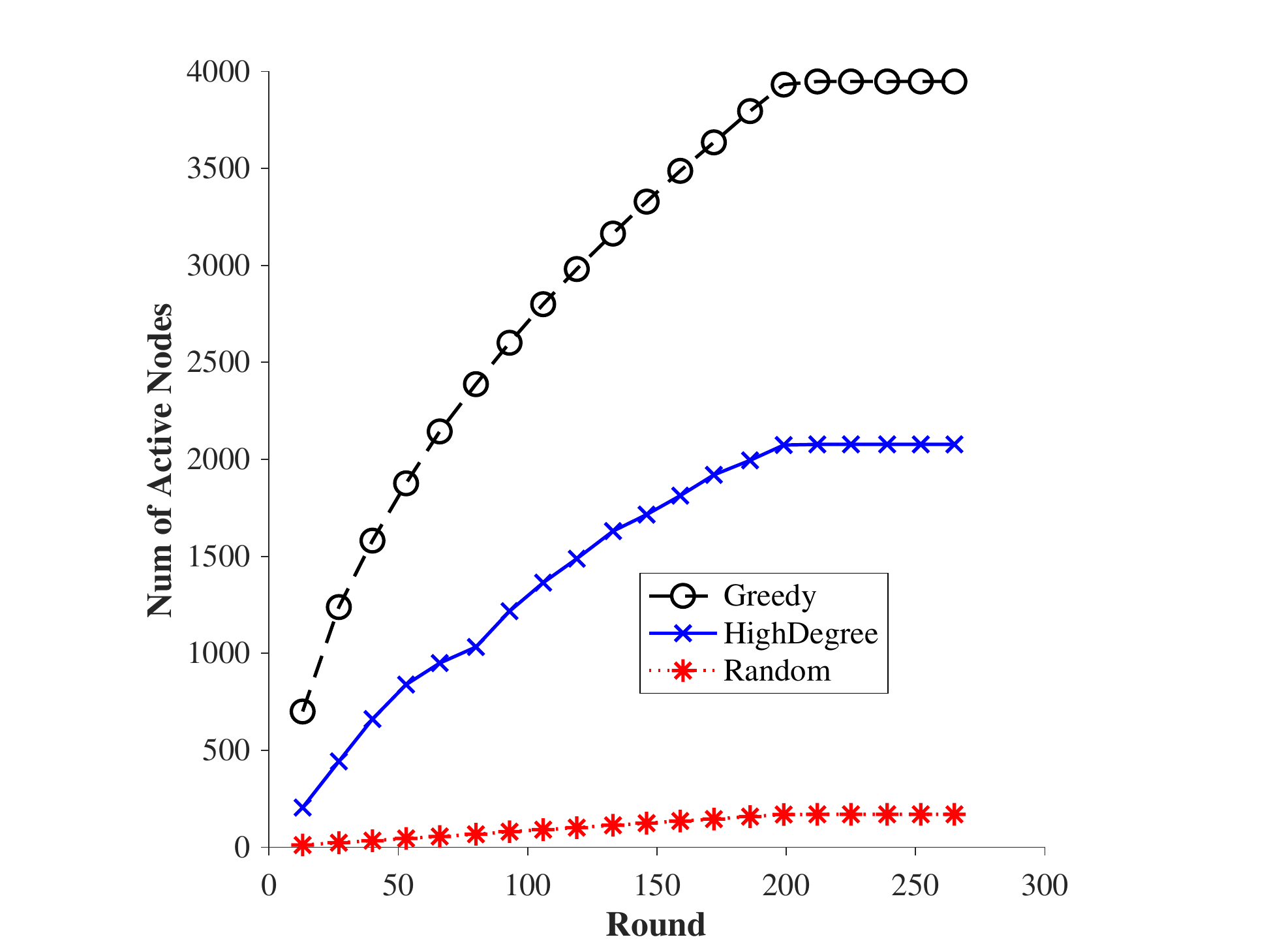}}
	\subfloat[{[$k=50, d=\infty$]}]{\label{fig: hepph50_3_15}\includegraphics[trim = 0.5in 0in 0.5in 0in, clip,width=0.24\textwidth]{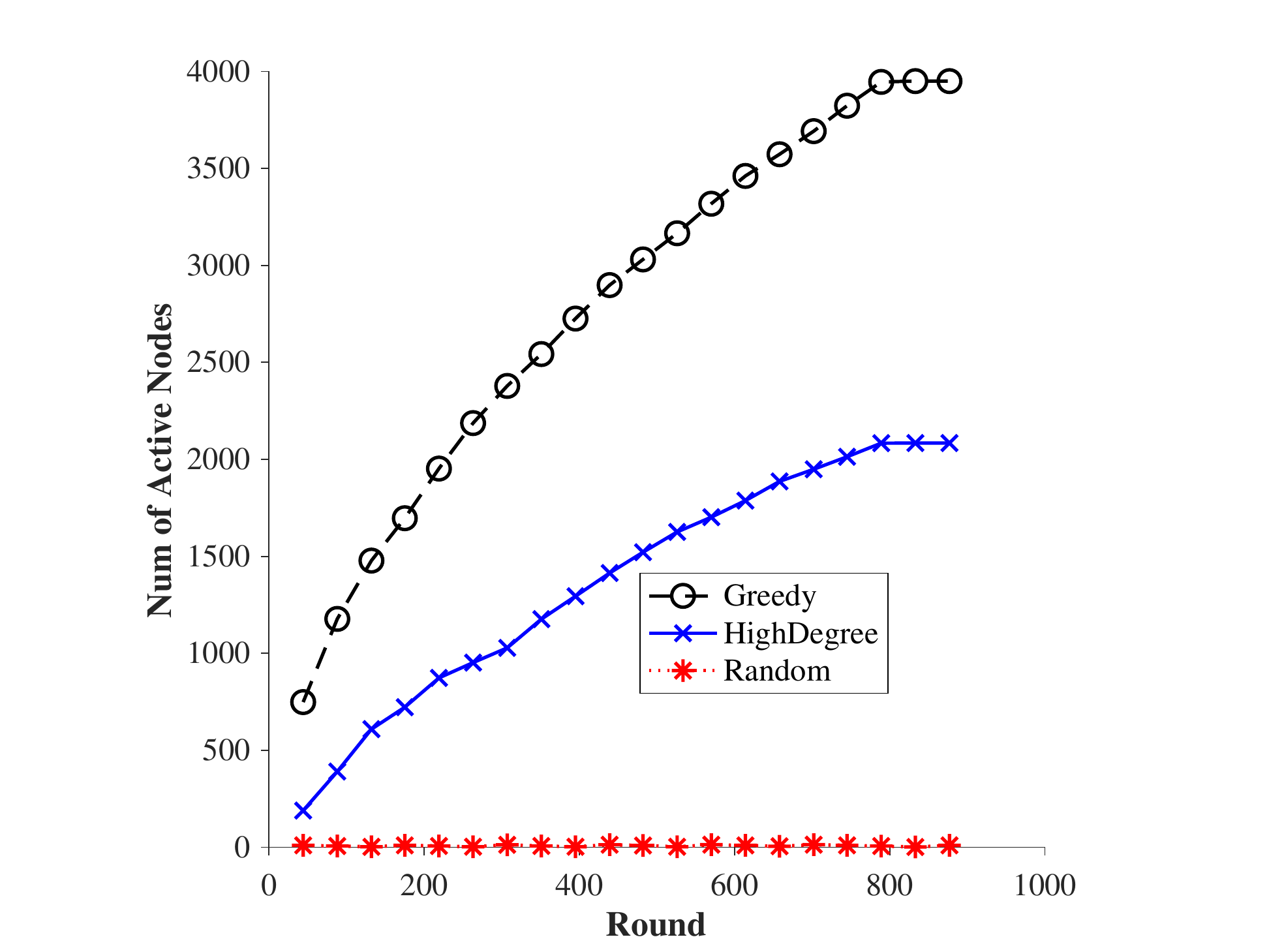}}
	\vspace{-1mm}
	\caption{Results of Experiment \RNum{1} on Hepph}
	\vspace{-1mm}
	\label{fig: exp1_hepph}
\end{figure*}

\subsection{Experimental Settings.} 

\textbf{Dataset.} We employ seven datasets, from small to large, denoted as Power, Wiki, Higgs, Hepth, Hepph, DBLP and Youtube, where Power is a power-law graph generated by the LFT benchmark \cite{lancichinetti2008benchmark}, Higgs is a Twitter dataset, and the rest of the datasets are borrowed from SNAP \cite{leskovec2015snap}. The details of the datasets are provided in the supplementary material. Due to space limitation, in the main paper we only show and discuss the results on Higgs, Hepph and DBLP, and the complete experimental results can be found in the supplementary material. The source code is maintained online \cite{code}.

\textbf{Propagation Probability.} Higgs consists of a collection of activities between users, including re-tweeting action, replying action, and mentioning action. We follow the setting in \cite{tong2018misinformation} so that the propagation probability between users is proportional to the frequency of the actions between them. For the other datasets, we adopt either the uniform setting with $p_e=0.1$  or the well-known weighted cascade setting where $p_{(u,v)}=1/deg(v)$ where $deg(v)$ is the in-degree of $v$. The setting of the propagation probability primarily affects the scale of the influence, and we have similar observations under different settings. 

\textbf{Problem Setting.} We set the budget $k$ as either $5$ or $50$, and select the parameter $d$ from $\{0, 1, 2, 4 ,8, \infty\}$ where $d=0$ and $d=\infty$ denote the non-adaptive case and the Full Adoption feedback model, respectively. Because the diffusion process typically terminates within $8$ rounds without new seed nodes, it is identical to the Full Adoption feedback model when $d>8$. Thus, we do not test the case for $d > 8$. 

\textbf{Policies.} Besides the greedy policy, we implemented two baseline algorithms, HighDegree and Random. HighDegree adaptively selects the node with the highest degree as the seed node, and Random selects the seed nodes randomly. 

\textbf{Simulation Setting.} Whenever the reverse sampling method in Sec. \ref{subsec: rrset} is called, the number of used RR-sets is at least 100,000 which is sufficient for an accurate single node selection, as shown in \cite{nguyen2016stop, huang2017revisiting, nguyen2018revisiting}. For each set of dataset and policy, 500 simulations were performed and we report the average result.

\begin{figure*}[!pt]
	\centering
	\subfloat[{[$k=5$, non-daptive]}]{\label{fig: dblp5_3_0}\includegraphics[trim = 0.5in 0in 0.5in 0in, clip, width=0.24\textwidth]{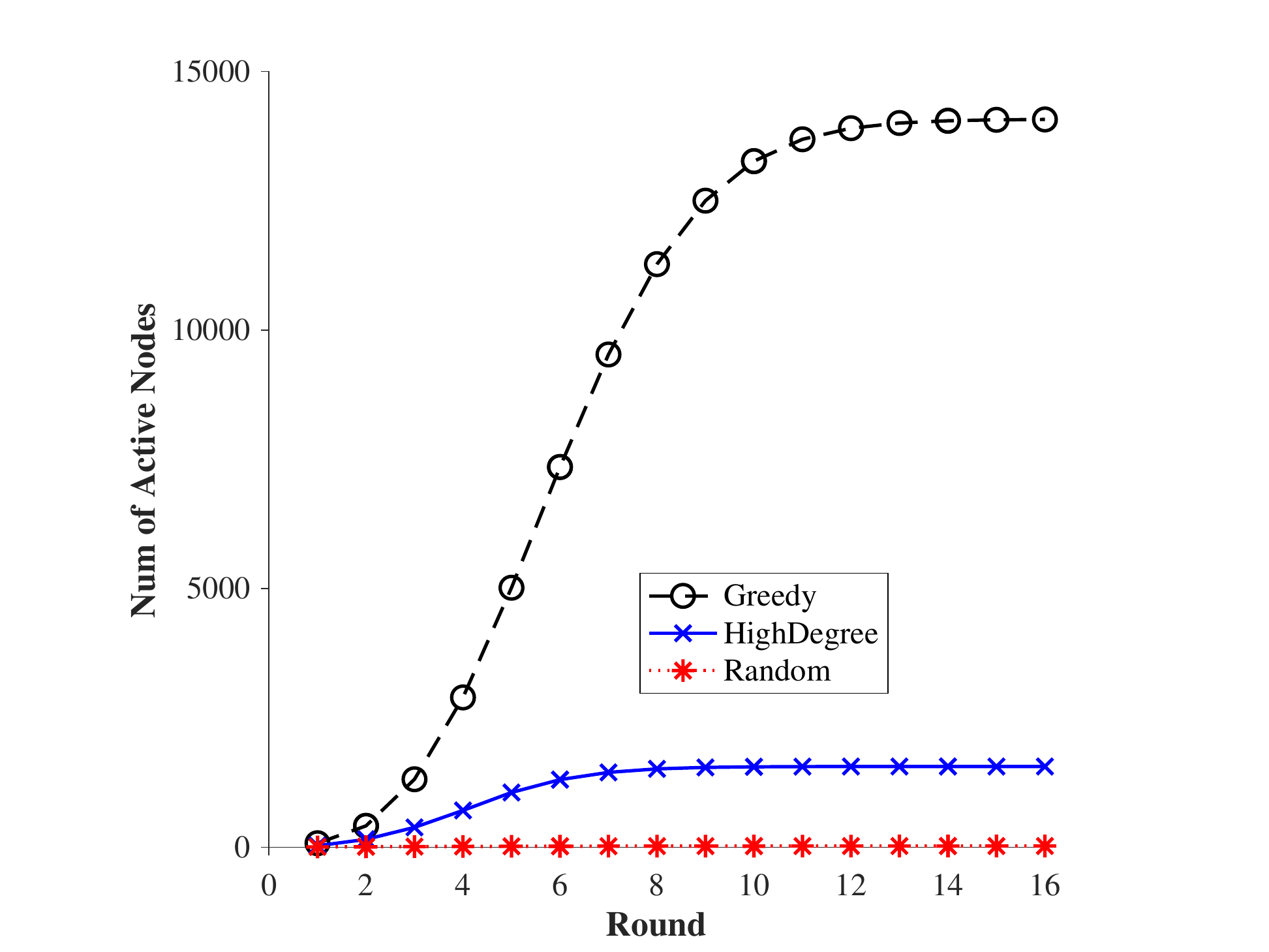}}
	\subfloat[{[$k=5, d=1$]}]{\label{fig: dblp5_3_3}\includegraphics[trim = 0.5in 0in 0.5in 0in, clip,width=0.24\textwidth]{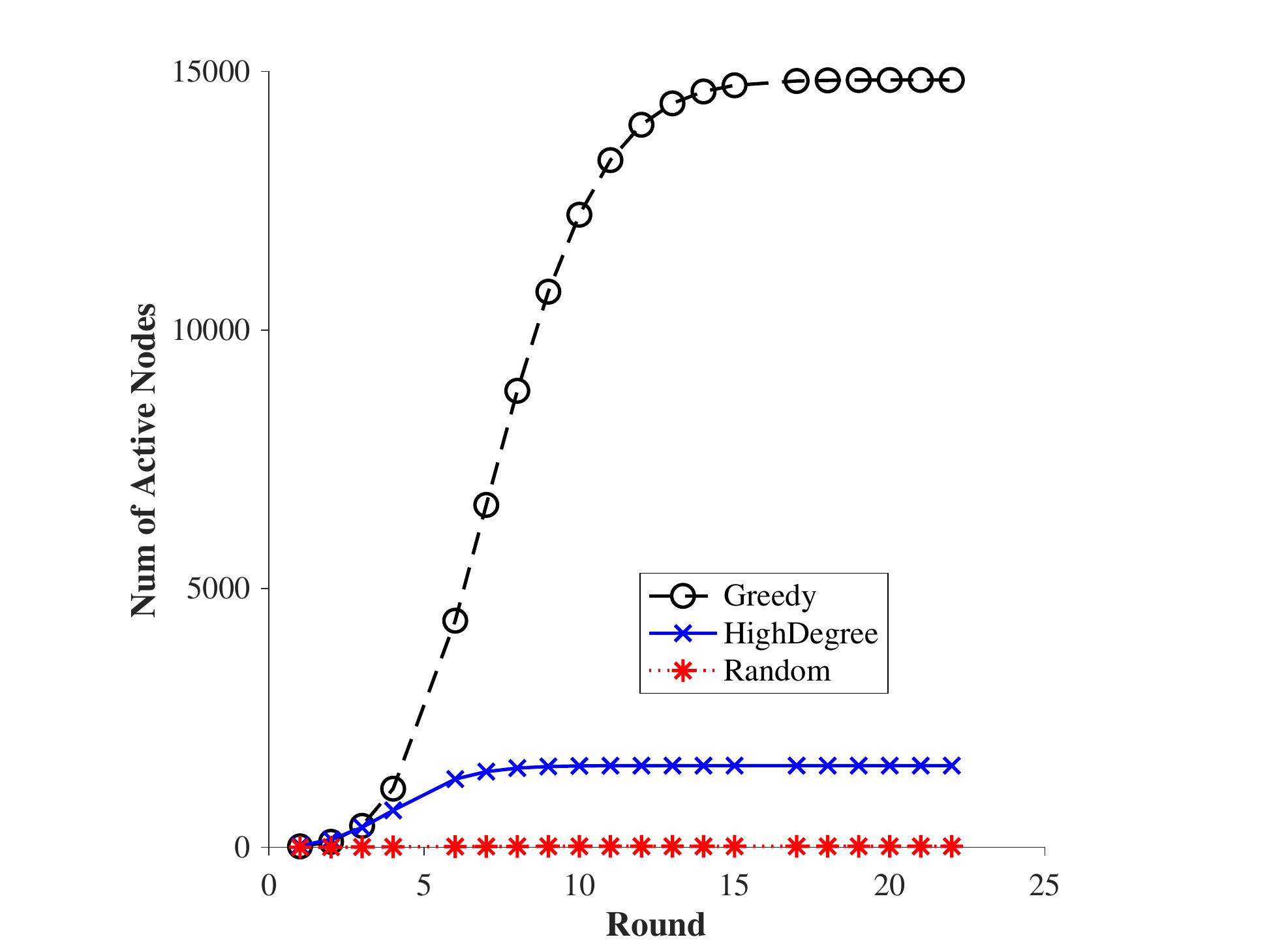}}
	\subfloat[{[$k=5, d=8$]}]{\label{fig: dblp5_3_9}\includegraphics[trim = 0.5in 0in 0.5in 0in, clip,width=0.24\textwidth]{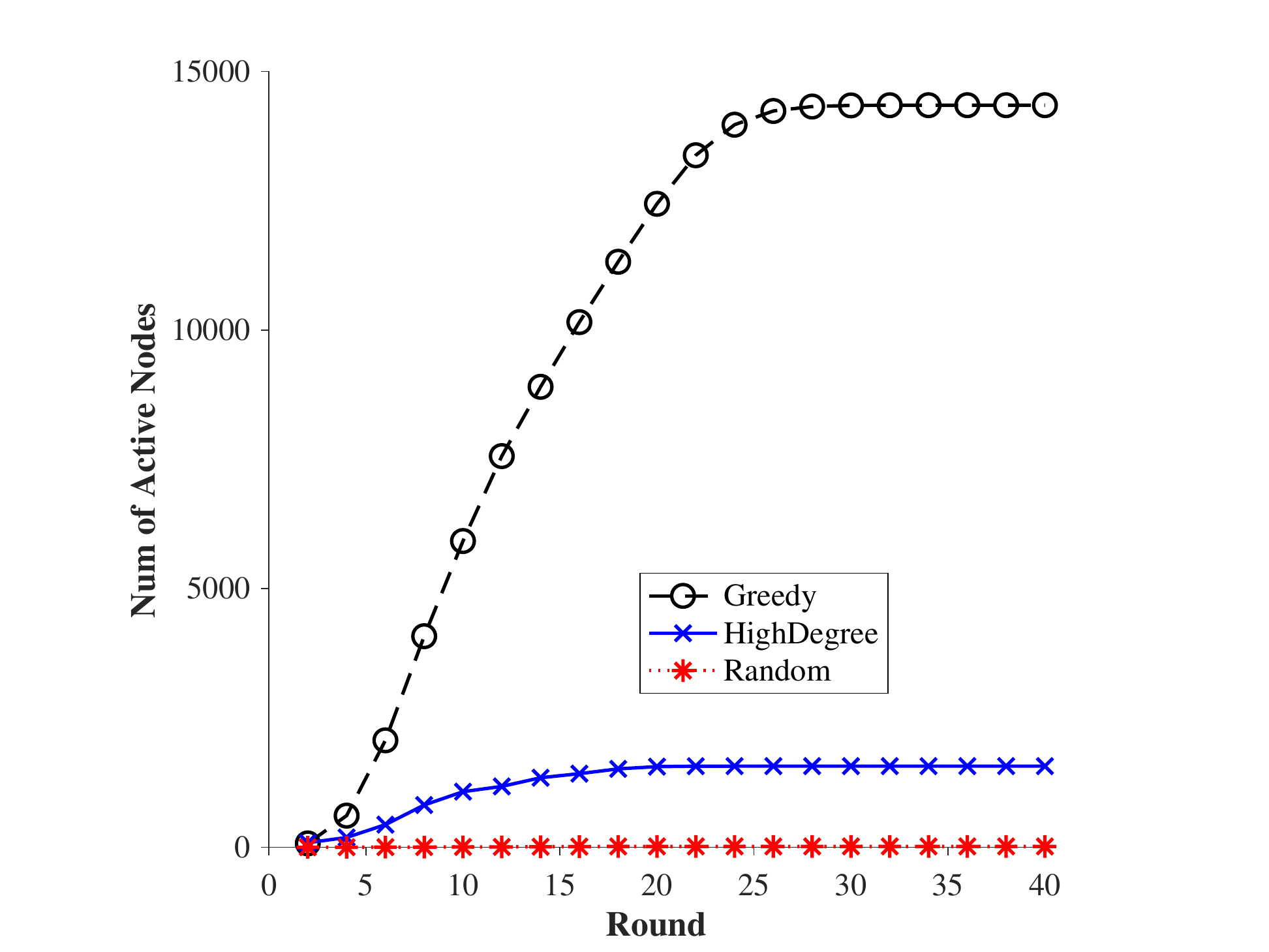}}
	\subfloat[{[$k=5, d=\infty$]}]{\label{fig: dblp5_3_15}\includegraphics[trim = 0.5in 0in 0.5in 0in, clip,width=0.24\textwidth]{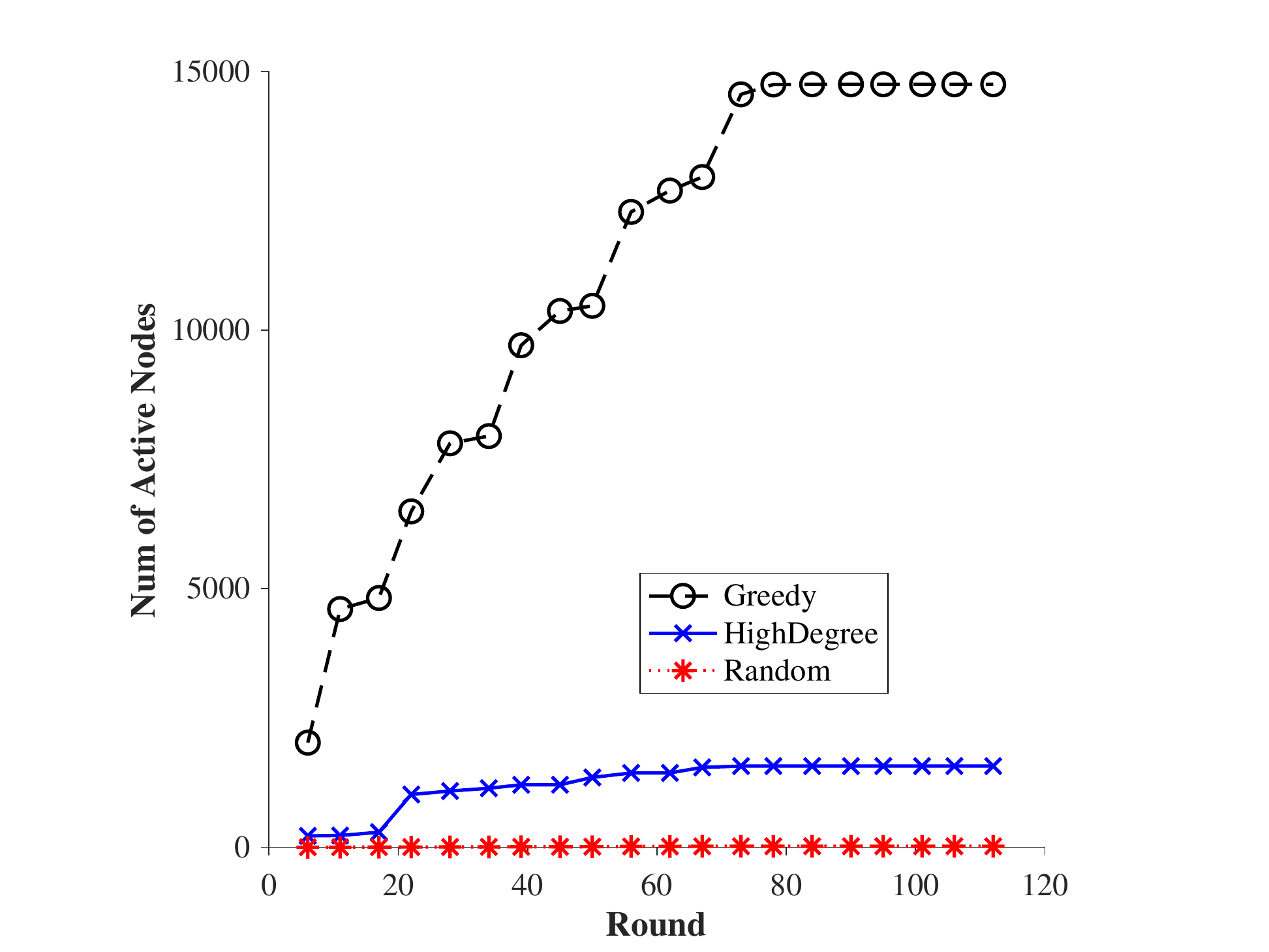}}
	
	\subfloat[{[$k=50$, non-daptive]}]{\label{fig: dblp50_3_0}\includegraphics[trim = 0.5in 0in 0.5in 0in, clip,width=0.24\textwidth]{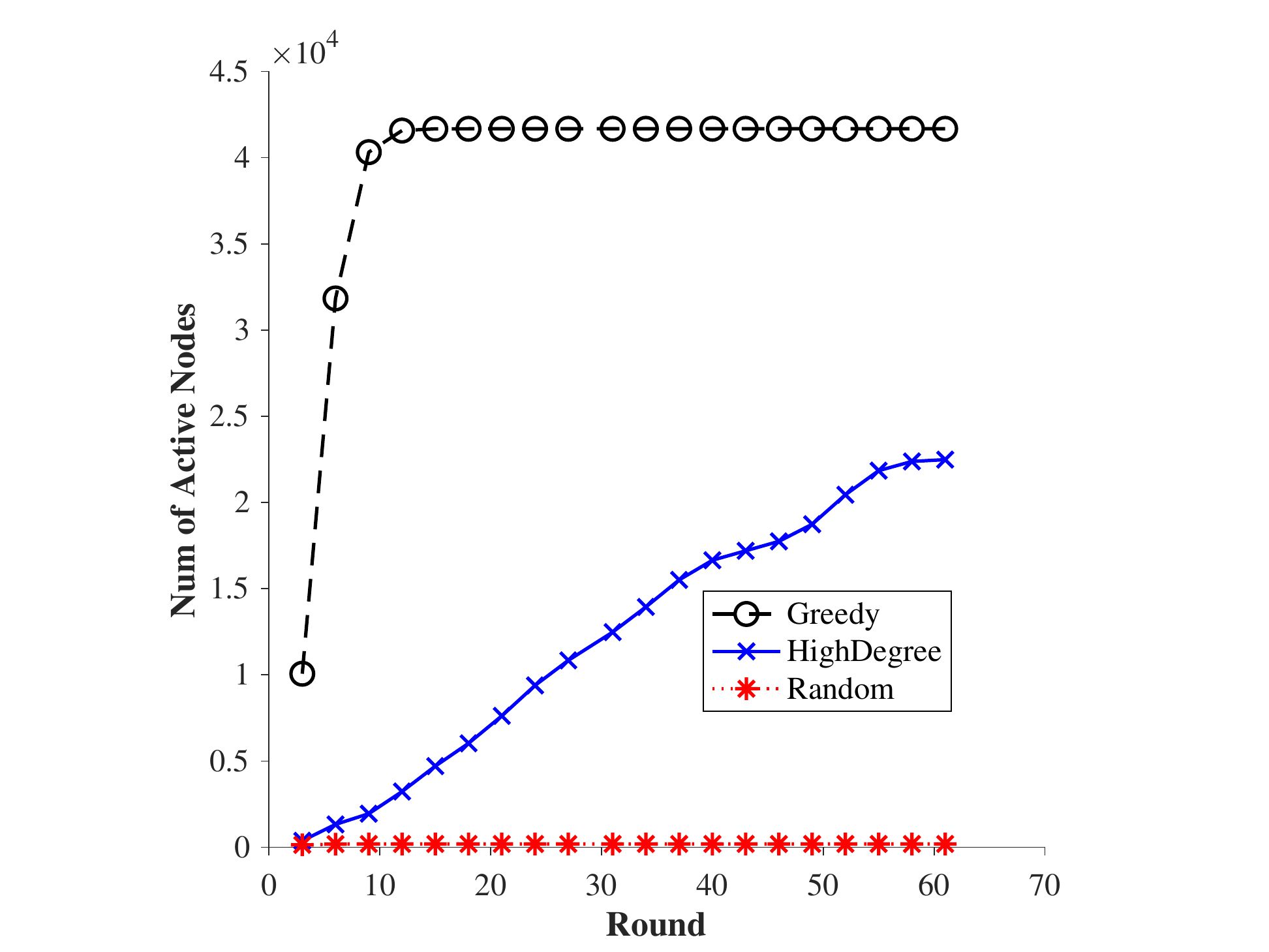}}
	\subfloat[{[$k=50, d=1$]}]{\label{fig: dblp50_3_3}\includegraphics[trim = 0.5in 0in 0.5in 0in, clip,width=0.24\textwidth]{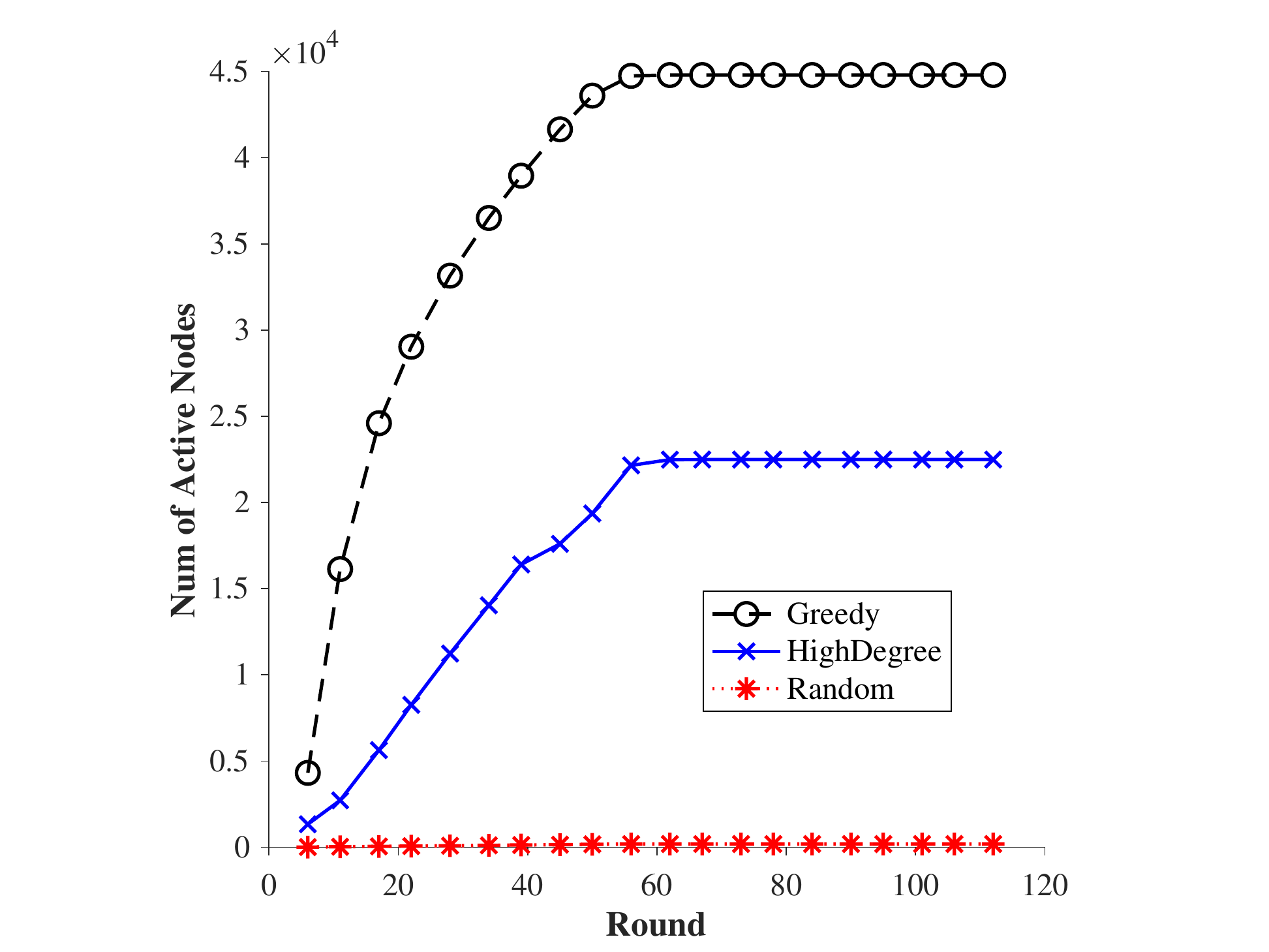}}
	\subfloat[{[$k=50, d=8$]}]{\label{fig: dblp50_3_9}\includegraphics[trim = 0.5in 0in 0.5in 0in, clip,width=0.24\textwidth]{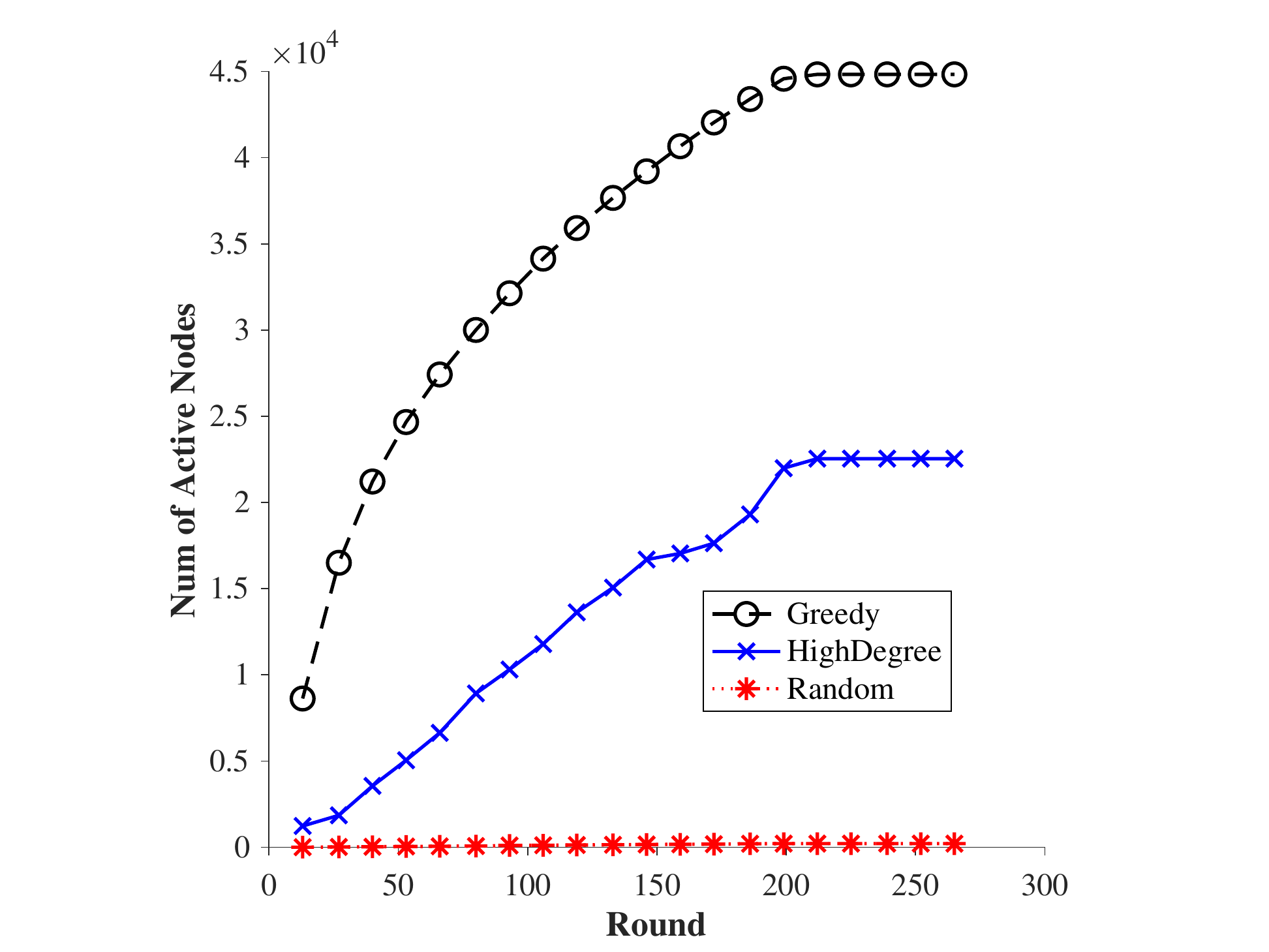}}
	\subfloat[{[$k=50, d=\infty$]}]{\label{fig: dblp50_3_15}\includegraphics[trim = 0.5in 0in 0.5in 0in, clip,width=0.24\textwidth]{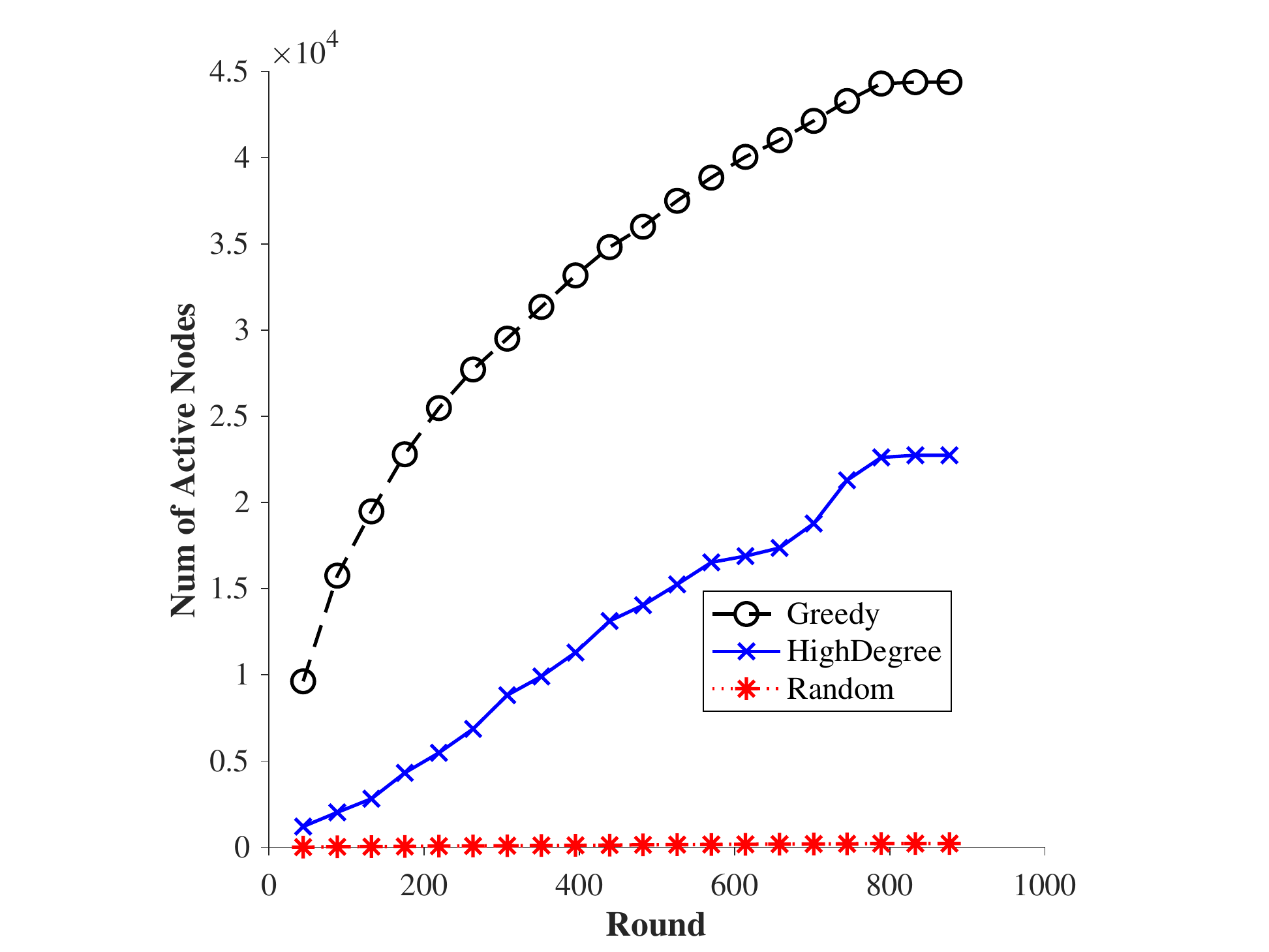}}
	\vspace{0mm}
	\caption{Results of Experiment \RNum{1} on DBLP}
	\vspace{-3mm}
	\label{fig: exp1_dblp}
\end{figure*}

\begin{figure*}[!t]
	\centering
	\subfloat[Higgs with $k=5$]{\label{fig: higgs5}\includegraphics[width=0.33\textwidth]{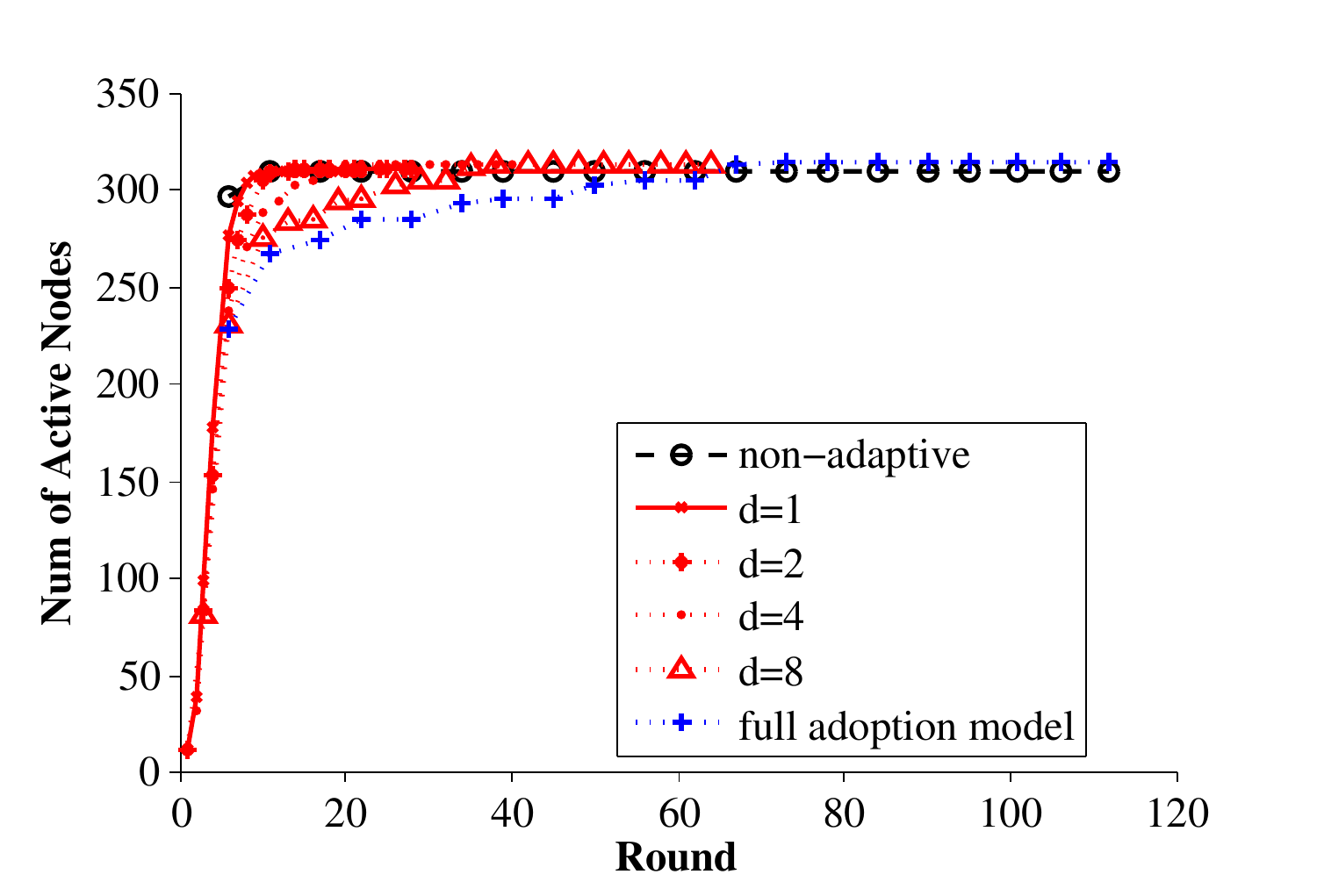}} 
	\subfloat[Hepph with $k=5$]{\label{fig: wiki}\includegraphics[width=0.33\textwidth]{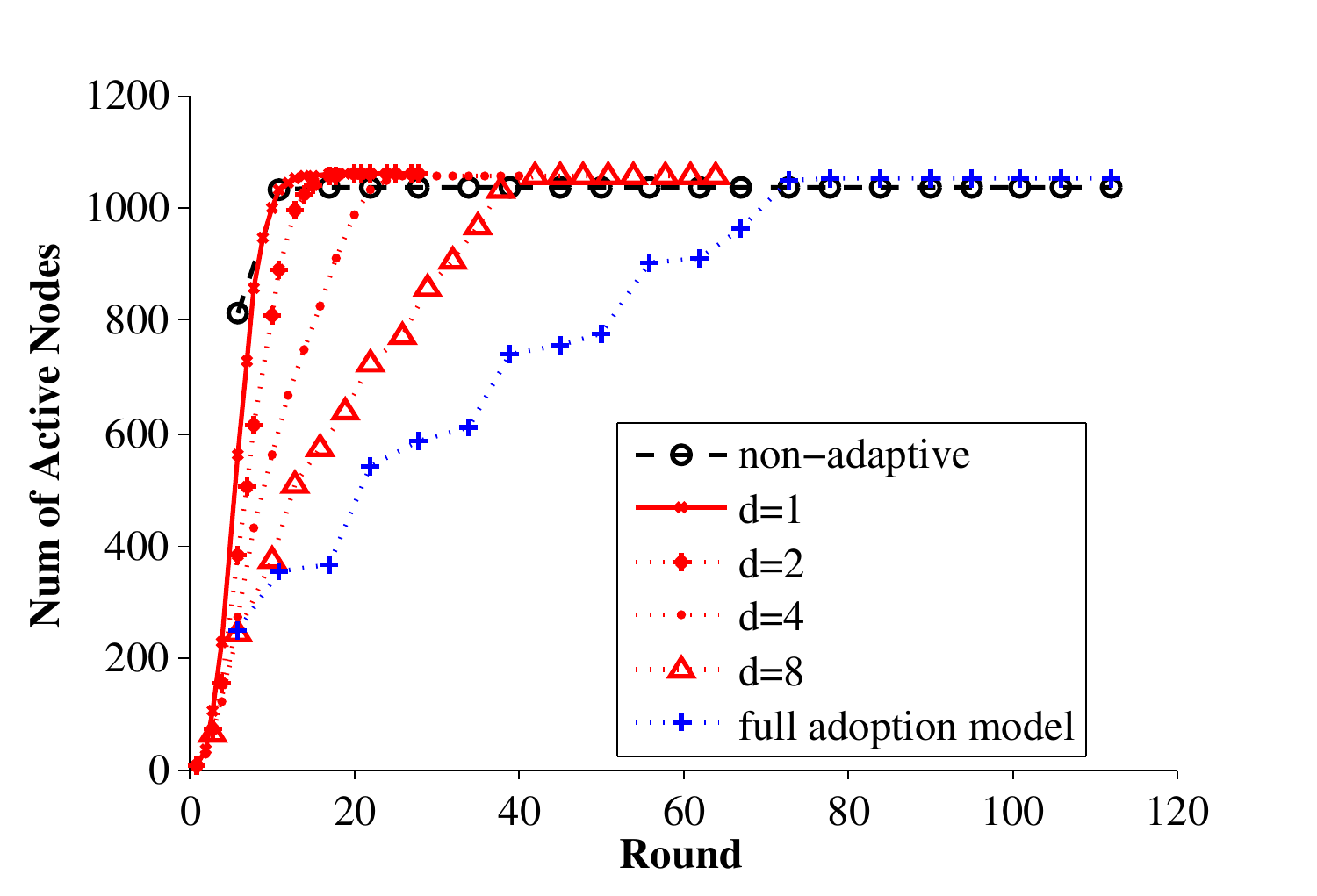}} 
	\subfloat[DBLP with $k=5$]{\label{fig: dblp5}\includegraphics[width=0.33\textwidth]{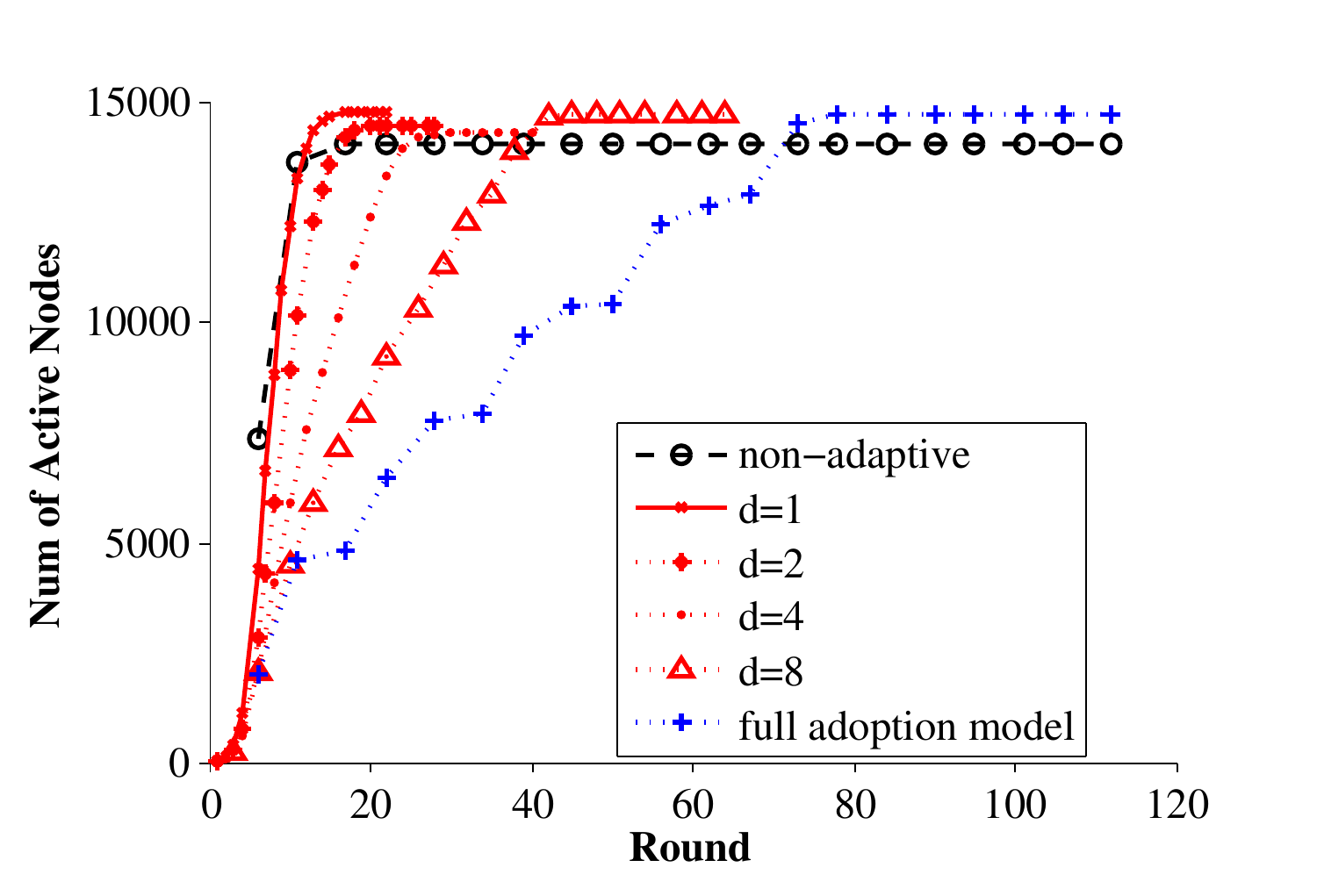}} 
	\vspace{-0mm}
	\subfloat[Higgs with $k=50$]{\label{fig: higgs50}\includegraphics[width=0.33\textwidth]{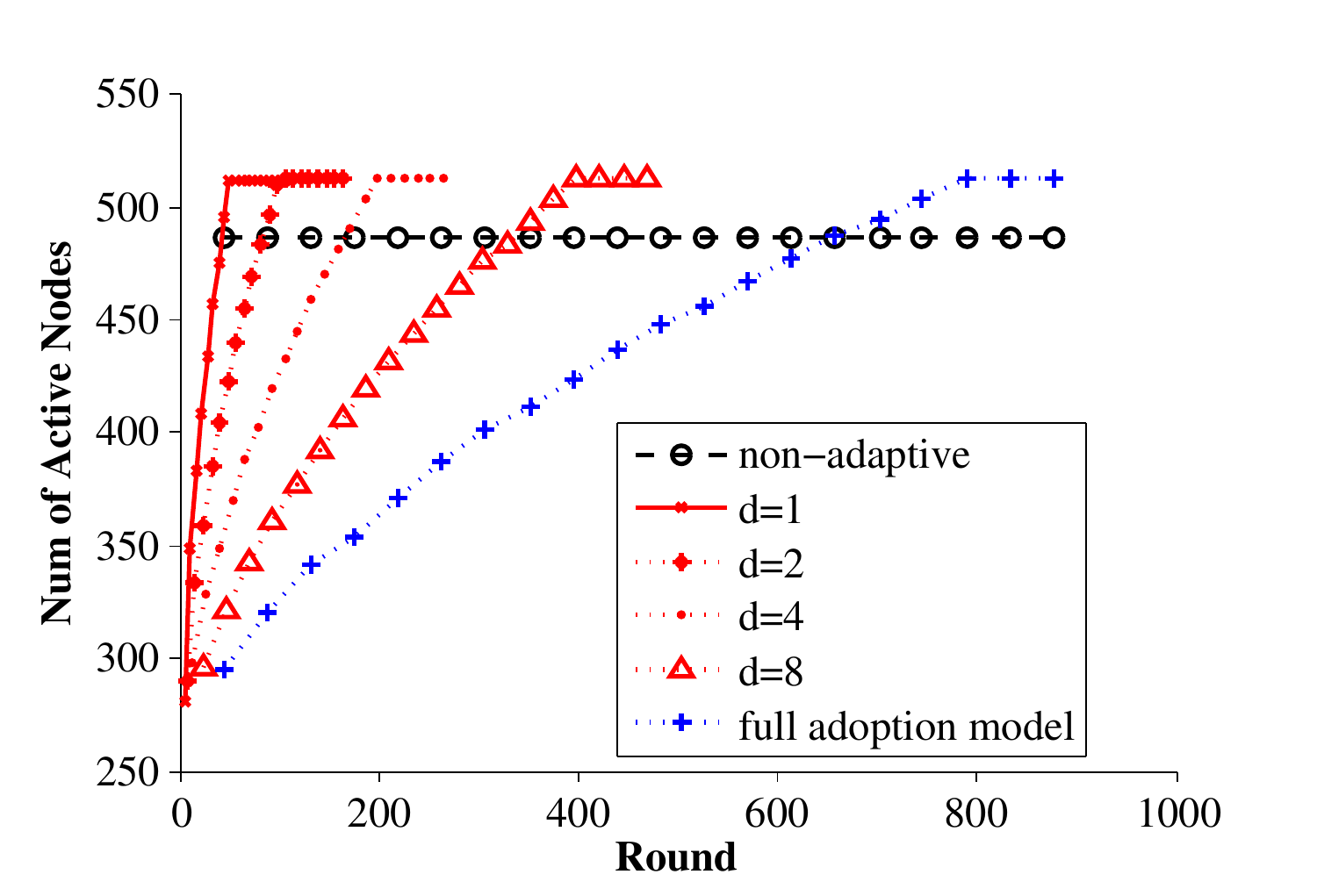}}
	\subfloat[Hepph with $k=50$]{\label{fig: hepph50}\includegraphics[width=0.33\textwidth]{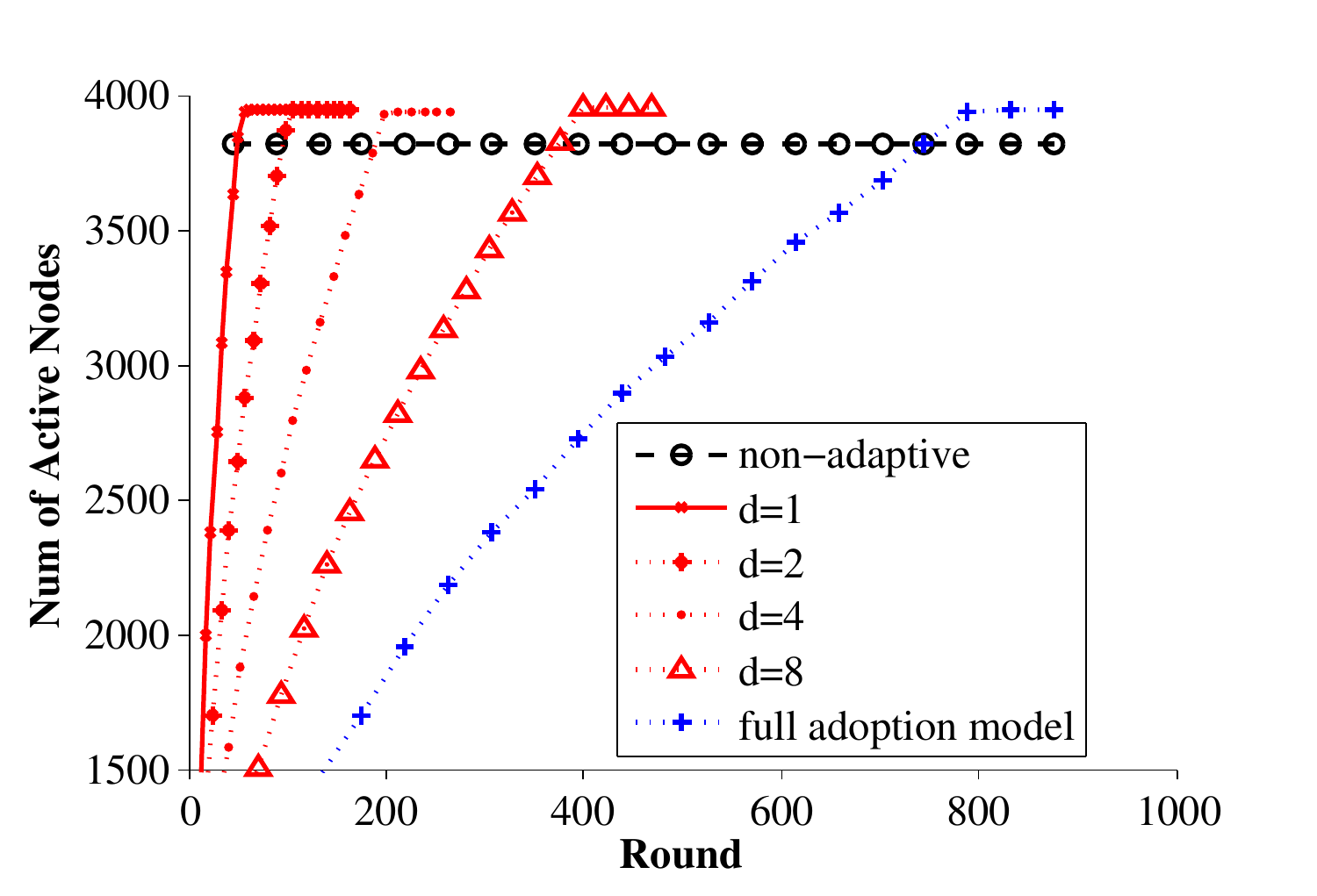}}
	\subfloat[DBLP with $k=50$]{\label{fig: dblp50}\includegraphics[width=0.33\textwidth]{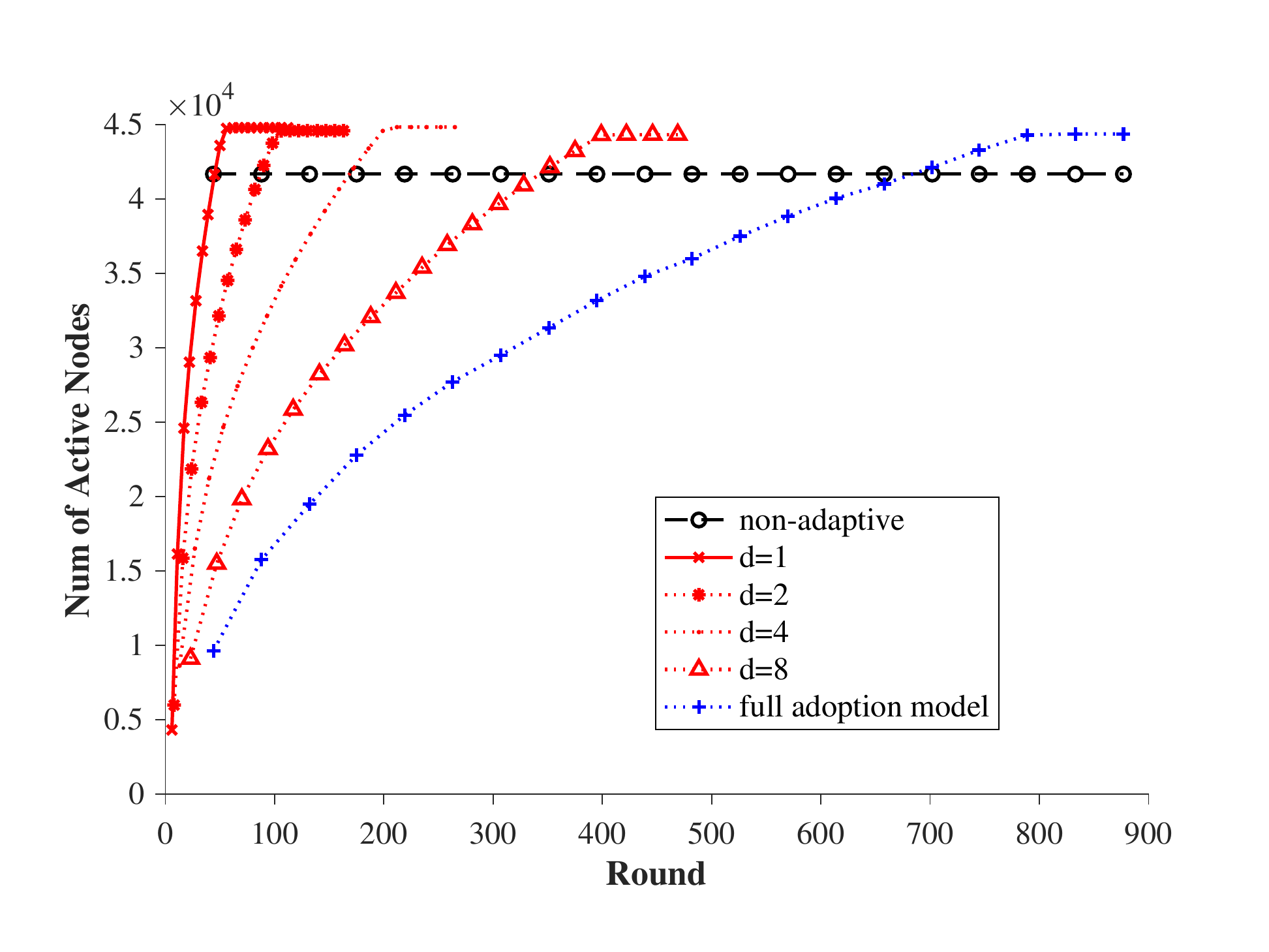}}
	\caption{Results of Experiment \RNum{2}}
	\vspace{0mm}
	\label{fig: exp2}
\end{figure*}
\subsection{Experiment \RNum{1}}
In the first experiment we compare the seeding processes under the different policies. The result of this part on Higgs, Hepph and DBLP are shown in Figs. \ref{fig: exp1_higgs}, \ref{fig: exp1_hepph} and \ref{fig: exp1_dblp}, respectively. We have the following observations.

Even though the greedy policy does not have a constant approximation ratio under general feedback models, it still dominates HighDegree which is an effective heuristic method. The superiority of Greedy becomes more significant with the increase in the network size. As we can see, on small graphs such as Higgs, HighDegree performs slightly worse than Greedy does, and Random even produces comparable results. However, on large graphs (e.g., Hepph and DBLP) Greedy outperforms the baseline methods by a significant gap. For example, as shown in Fig. \ref{fig: exp1_hepph}, Random and HighDegree can hardly bring 2,000 active nodes on Hepph with $k=50$, whereas Greedy achieves at least 4,000 active nodes under all cases. 

Another observation is that HighDegree is occasionally better than Greedy if we restrict our attention to the first several rounds. For example on Higgs with $k=5$, Figs. \ref{fig: higgs5_3_0}, \ref{fig: higgs5_3_3} and \ref{fig: higgs5_3_9}, the influence resulted by HighDegree is higher than that of Greedy in the first two or three diffusion rounds. This is intuitive as HighDegree is a heuristic targeted on the influence right after the next round, and it also suggests that HighDegree is effective for maximizing the influence within few rounds.

\subsection{Experiment \RNum{2}}
In the second experiment, we compare the diffusion pattern under different feedback models. The diffusion pattern herein is characterized by the increase in the  number of active nodes after each diffusion round. The result of this part is shown in Fig. \ref{fig: exp2}.

\textbf{Non-adaptive vs Adaptive.} The main question we are interested in is how much the adaptivity can help in resulting in a higher influence. The first we can observe is that the adaptive setting it is not always significantly better the non-adaptive one, especially when the budget is small, which can be seen by comparing the cases with $k=5$ and $k=50$ on the same graph. For example, as shown in Fig. \ref{fig: higgs5}, on Higgs with $k=5$ the final influence is the same regardless of the feedback models. Recall that one of the main advantages of an adaptive policy is that better selections can be made if the nodes that were optimal have been activated. Because the optimal seed nodes are relatively sparse when $k$ is small, they are likely to remain inactive after other seed nodes are selected, resulting in that the adaptivity cannot provide better options for node selection and thus the final influence remains the same. 

\textbf{Comparing Adaptive Patterns.} According to Fig. \ref{fig: exp2}, the Full Adoption feedback model always produces the best result, coinciding with that it is theoretically optimal, but the difference between the patterns with different $d$ is not significant. One plausible reason is that the diffusion process terminates very fast due to the setting of the propagation probability in our experiments. In such cases, the seeding decisions are always made when the diffusion process terminates, and thus a fixed $d$ is equivalent to the Full Adoption feedback model. Such an observation suggests that $d=1$ can be the best choice if one also considers a time constraint. As we can see, namely in Figs. \ref{fig: higgs50}, \ref{fig: hepph50} and \ref{fig: dblp50}, different adaptive patterns have almost the same final influence but it reaches its maximum much faster when $d=1$. Therefore, even though the Full Adoption feedback model is optimal, adopting a fixed small $d$ is practically sufficient for certain datasets.

\section{Future Work}
\label{sec: general}
In this section, we present the future work. 

\textbf{Feedback Model with Further Generalizations.} From the Full Adoption feedback model in \cite{guille2013information} to the $(k,d)$-feedback model proposed in this paper, we see that the manner in which we observe the diffusion results is generalized further and further. In a more general case, we can consider any feedback function which maps a status $U$ to a subset $E^* \subseteq E$ of edge, indicating that the state of the edges in $E^*$ will be observed from status $U$. For example, under the Myopic feedback model, $E^*$ is the set of out-edges of the newly activated nodes. Furthermore, the observable edges can be nondeterministic. For example, under the Full Adoption feedback model, the edge we can observe in each step depends on the diffusion process which is stochastic. Therefore, the most general feedback model maps a status $U$ to a distribution over the super-realizations of $\dot{\phi}(U)$. This general setting enables us to make observations independent of the diffusion process and it admits extra flexibility for modeling complex real applications. Several examples are shown below.

\begin{example}[\textbf{Limited Observation}]
Considering the large scale of the social network, one often has a limit in observing the diffusion results. For example, we are given a subset $E^*$ of edges which are the only the edges we can observe. Combining the Full Adoption feedback model, in each observing step only the edges in $E^*$ reachable from the active nodes can be observed. Recall that the key issue in the AIM problem is to utilize the observations to make the next seeding decision. Due to the limited observation, there is an interesting trade-off between the quality of the observations and the quality of the seed nodes. On the one hand, we prefer the seed node which can result in more observations in $E^*$ but such a node may not necessarily be the influential node. Conversely, the node with a high influence may be distant from $E^*$ and therefore brings no observation. 
\end{example}

\begin{example}[\textbf{Flexible Observation}]
Suppose that we aim at select $k \in \mathbb{Z}^+$ seed nodes for a certain cascade in an IC social network. Due to privacy issues, the states of the edges can only be observed by probing, and we can probe at most $p \in \mathbb{Z}^+$ edges. Under such a setting, we have a chance to decide the edges to observe, and an adaptive seeding process consists of probing-step and seeding-step. The problem asks for the co-design of the seeding strategy and the probing strategy such that the total influence can be maximized, which is another interesting future work. 
\end{example}

\begin{example}[\textbf{Observations Beyond Round-by-round}]
One typical setting in the AIM problem is that we always make observations round by round. However, this is not realistic in practice because one can hardly synchronize the diffusion process by round in real social networks namely Facebook and Twitter because the activations may take different periods. Instead, the seeding action can be made immediately once an important event has been observed. For example, in online advertising with several target users, a company would prefer to start advertising another one if the current target has been influenced. In another issue, for time-sensitive tasks, we can immediately deploy another seed node once we have a sufficient number of new active users. In such cases, the observations are not necessarily or even not allowed to be made round by round. It is promising to investigate how to model the AIM problem in such cases and design seeding policies accordingly.
\end{example}

\textbf{Batch Mode.} The policy considered in this paper selects one node each time while it is possible to generalize it to a batch mode where more than one nodes are selected. Under such setting, the node selection in each seeding-step becomes an NP-hard problem in general but one can obtain a $(1-1/e)$-approximation. For such cases, one can define the regret ratio for the batch model and the techniques in \cite{chen2013near} and \cite{sun2018multi} are potentially applicable.

\section{Conclusion}
\label{sec: con}
In this paper, we study the AIM problem under the general feedback models applying to many real-world applications. We show that the performance of the greedy policy under the considered feedback models can be bounded by the regret ratio which quantifying the trade-off between waiting and seeding for the general case. The proposed analysis is the first applies to the AIM when it is not adaptive submodular. We design experiments to examine the performance of the greedy policy under general feedback models, as well as the effect of the feedback model on the final influence. In particular, the conducted experiments show that the adaptive settings are supreme in most cases. Finally, we discuss the future work.

\ifCLASSOPTIONcompsoc
\section*{Acknowledgments}
\else
\section*{Acknowledgment}
\fi

The authors would like to thank Wei Chen for his comments on the AIM problem.

\ifCLASSOPTIONcaptionsoff
\newpage
\fi

\bibliographystyle{IEEEtran}
\bibliography{sample-bibliography}

\begin{thebibliography}{10}
\providecommand{\url}[1]{#1}
\csname url@samestyle\endcsname
\providecommand{\newblock}{\relax}
\providecommand{\bibinfo}[2]{#2}
\providecommand{\BIBentrySTDinterwordspacing}{\spaceskip=0pt\relax}
\providecommand{\BIBentryALTinterwordstretchfactor}{4}
\providecommand{\BIBentryALTinterwordspacing}{\spaceskip=\fontdimen2\font plus
\BIBentryALTinterwordstretchfactor\fontdimen3\font minus
  \fontdimen4\font\relax}
\providecommand{\BIBforeignlanguage}[2]{{%
\expandafter\ifx\csname l@#1\endcsname\relax
\typeout{** WARNING: IEEEtran.bst: No hyphenation pattern has been}%
\typeout{** loaded for the language `#1'. Using the pattern for}%
\typeout{** the default language instead.}%
\else
\language=\csname l@#1\endcsname
\fi
#2}}
\providecommand{\BIBdecl}{\relax}
\BIBdecl

\bibitem{kempe2003maximizing}
D.~Kempe, J.~Kleinberg, and {\'E}.~Tardos, ``Maximizing the spread of influence
  through a social network,'' in \emph{Proceedings of the ninth ACM SIGKDD
  international conference on Knowledge discovery and data mining}.\hskip 1em
  plus 0.5em minus 0.4em\relax ACM, 2003, pp. 137--146.

\bibitem{li2018influence}
Y.~Li, J.~Fan, Y.~Wang, and K.-L. Tan, ``Influence maximization on social
  graphs: A survey,'' \emph{IEEE Transactions on Knowledge and Data
  Engineering}, vol.~30, no.~10, pp. 1852--1872, 2018.

\bibitem{guille2013information}
A.~Guille, H.~Hacid, C.~Favre, and D.~A. Zighed, ``Information diffusion in
  online social networks: A survey,'' \emph{ACM Sigmod Record}, vol.~42, no.~2,
  pp. 17--28, 2013.

\bibitem{aslay2018influence}
C.~Aslay, L.~V. Lakshmanan, W.~Lu, and X.~Xiao, ``Influence maximization in
  online social networks,'' in \emph{Proceedings of the Eleventh ACM
  International Conference on Web Search and Data Mining}.\hskip 1em plus 0.5em
  minus 0.4em\relax ACM, 2018, pp. 775--776.

\bibitem{golovin2011adaptive}
D.~Golovin and A.~Krause, ``Adaptive submodularity: Theory and applications in
  active learning and stochastic optimization,'' \emph{Journal of Artificial
  Intelligence Research}, vol.~42, pp. 427--486, 2011.

\bibitem{tong2017adaptive}
G.~Tong, W.~Wu, S.~Tang, and D.-Z. Du, ``Adaptive influence maximization in
  dynamic social networks,'' \emph{IEEE/ACM Transactions on Networking (TON)},
  vol.~25, no.~1, pp. 112--125, 2017.

\bibitem{vaswani2016adaptive}
S.~Vaswani and L.~V. Lakshmanan, ``Adaptive influence maximization in social
  networks: Why commit when you can adapt?'' \emph{arXiv preprint
  arXiv:1604.08171}, 2016.

\bibitem{chen2013near}
Y.~Chen and A.~Krause, ``Near-optimal batch mode active learning and adaptive
  submodular optimization.'' \emph{ICML (1)}, vol.~28, pp. 160--168, 2013.

\bibitem{sun2018multi}
L.~Sun, W.~Huang, P.~S. Yu, and W.~Chen, ``Multi-round influence
  maximization,'' in \emph{Proceedings of the 24th ACM SIGKDD International
  Conference on Knowledge Discovery \& Data Mining}.\hskip 1em plus 0.5em minus
  0.4em\relax ACM, 2018, pp. 2249--2258.

\bibitem{lei2015online}
S.~Lei, S.~Maniu, L.~Mo, R.~Cheng, and P.~Senellart, ``Online influence
  maximization,'' in \emph{Proceedings of the 21th ACM SIGKDD International
  Conference on Knowledge Discovery and Data Mining}.\hskip 1em plus 0.5em
  minus 0.4em\relax ACM, 2015, pp. 645--654.

\bibitem{chen2016combinatorial}
W.~Chen, Y.~Wang, Y.~Yuan, and Q.~Wang, ``Combinatorial multi-armed bandit and
  its extension to probabilistically triggered arms,'' \emph{The Journal of
  Machine Learning Research}, vol.~17, no.~1, pp. 1746--1778, 2016.

\bibitem{vaswani2017model}
S.~Vaswani, B.~Kveton, Z.~Wen, M.~Ghavamzadeh, L.~V. Lakshmanan, and
  M.~Schmidt, ``Model-independent online learning for influence maximization,''
  in \emph{International Conference on Machine Learning}, 2017, pp. 3530--3539.

\bibitem{wen2017online}
Z.~Wen, B.~Kveton, M.~Valko, and S.~Vaswani, ``Online influence maximization
  under independent cascade model with semi-bandit feedback,'' in
  \emph{Advances in Neural Information Processing Systems}, 2017, pp.
  3022--3032.

\bibitem{seeman2013adaptive}
L.~Seeman and Y.~Singer, ``Adaptive seeding in social networks,'' in
  \emph{Foundations of Computer Science (FOCS), 2013 IEEE 54th Annual Symposium
  on}.\hskip 1em plus 0.5em minus 0.4em\relax IEEE, 2013, pp. 459--468.

\bibitem{salha2018adaptive}
G.~Salha, N.~Tziortziotis, and M.~Vazirgiannis, ``Adaptive submodular influence
  maximization with myopic feedback,'' in \emph{2018 IEEE/ACM International
  Conference on Advances in Social Networks Analysis and Mining
  (ASONAM)}.\hskip 1em plus 0.5em minus 0.4em\relax IEEE, 2018, pp. 455--462.

\bibitem{han2018efficient}
K.~Han, K.~Huang, X.~Xiao, J.~Tang, A.~Sun, and X.~Tang, ``Efficient algorithms
  for adaptive influence maximization,'' \emph{Proceedings of the VLDB
  Endowment}, vol.~11, no.~9, pp. 1029--1040, 2018.

\bibitem{golovin2010adaptive}
D.~Golovin, A.~Krause, and E.~CH, ``Adaptive submodularity: Theory and
  applications in active learning and stochastic optimization,'' \emph{arXiv
  preprint arXiv:1003.3967}, 2010.

\bibitem{mossel2007submodularity}
E.~Mossel and S.~Roch, ``On the submodularity of influence in social
  networks,'' in \emph{Proceedings of the thirty-ninth annual ACM symposium on
  Theory of computing}.\hskip 1em plus 0.5em minus 0.4em\relax ACM, 2007, pp.
  128--134.

\bibitem{borgs2014maximizing}
C.~Borgs, M.~Brautbar, J.~Chayes, and B.~Lucier, ``Maximizing social influence
  in nearly optimal time,'' in \emph{Proceedings of the twenty-fifth annual
  ACM-SIAM symposium on Discrete algorithms}.\hskip 1em plus 0.5em minus
  0.4em\relax SIAM, 2014, pp. 946--957.

\bibitem{tang2015influence}
Y.~Tang, Y.~Shi, and X.~Xiao, ``Influence maximization in near-linear time: A
  martingale approach,'' in \emph{Proceedings of the 2015 ACM SIGMOD
  International Conference on Management of Data}.\hskip 1em plus 0.5em minus
  0.4em\relax ACM, 2015, pp. 1539--1554.

\bibitem{tong2017efficient}
G.~Tong, W.~Wu, L.~Guo, D.~Li, C.~Liu, B.~Liu, and D.-Z. Du, ``An efficient
  randomized algorithm for rumor blocking in online social networks,''
  \emph{IEEE Transactions on Network Science and Engineering}, 2017.

\bibitem{tong2019beyond}
G.~Tong and D.-Z. Du, ``Beyond uniform reverse sampling: A hybrid sampling
  technique for misinformation prevention,'' \emph{arXiv preprint
  arXiv:1901.05149}, 2019.

\bibitem{wang2017activity}
Z.~Wang, Y.~Yang, J.~Pei, L.~Chu, and E.~Chen, ``Activity maximization by
  effective information diffusion in social networks,'' \emph{IEEE Transactions
  on Knowledge and Data Engineering}, vol.~29, no.~11, pp. 2374--2387, 2017.

\bibitem{lancichinetti2008benchmark}
A.~Lancichinetti, S.~Fortunato, and F.~Radicchi, ``Benchmark graphs for testing
  community detection algorithms,'' \emph{Physical review E}, vol.~78, no.~4,
  p. 046110, 2008.

\bibitem{leskovec2015snap}
J.~Leskovec and A.~Krevl, ``$\{$SNAP Datasets$\}$:$\{$Stanford$\}$ large
  network dataset collection,'' 2015.

\bibitem{code}
G.~Tong and R.~Wang,
  https://github.com/New2World/Computational-Data-Science-Lab/tree/multithread/AIM-TKDE/.

\bibitem{tong2018misinformation}
G.~Tong, D.-Z. Du, and W.~Wu, ``On misinformation containment in online social
  networks,'' in \emph{Advances in Neural Information Processing Systems},
  2018, pp. 339--349.

\bibitem{nguyen2016stop}
H.~T. Nguyen, M.~T. Thai, and T.~N. Dinh, ``Stop-and-stare: Optimal sampling
  algorithms for viral marketing in billion-scale networks,'' in
  \emph{Proceedings of the 2016 International Conference on Management of
  Data}.\hskip 1em plus 0.5em minus 0.4em\relax ACM, 2016, pp. 695--710.

\bibitem{huang2017revisiting}
K.~Huang, S.~Wang, G.~Bevilacqua, X.~Xiao, and L.~V. Lakshmanan, ``Revisiting
  the stop-and-stare algorithms for influence maximization,'' \emph{Proceedings
  of the VLDB Endowment}, vol.~10, no.~9, pp. 913--924, 2017.

\bibitem{nguyen2018revisiting}
H.~T. Nguyen, T.~N. Dinh, and M.~T. Thai, ``Revisiting of ‘revisiting the
  stop-and-stare algorithms for influence maximization’,'' in
  \emph{International Conference on Computational Social Networks}.\hskip 1em
  plus 0.5em minus 0.4em\relax Springer, 2018, pp. 273--285.

\end{thebibliography}

\begin{IEEEbiography}[{\includegraphics[width=0.9in,clip,keepaspectratio]{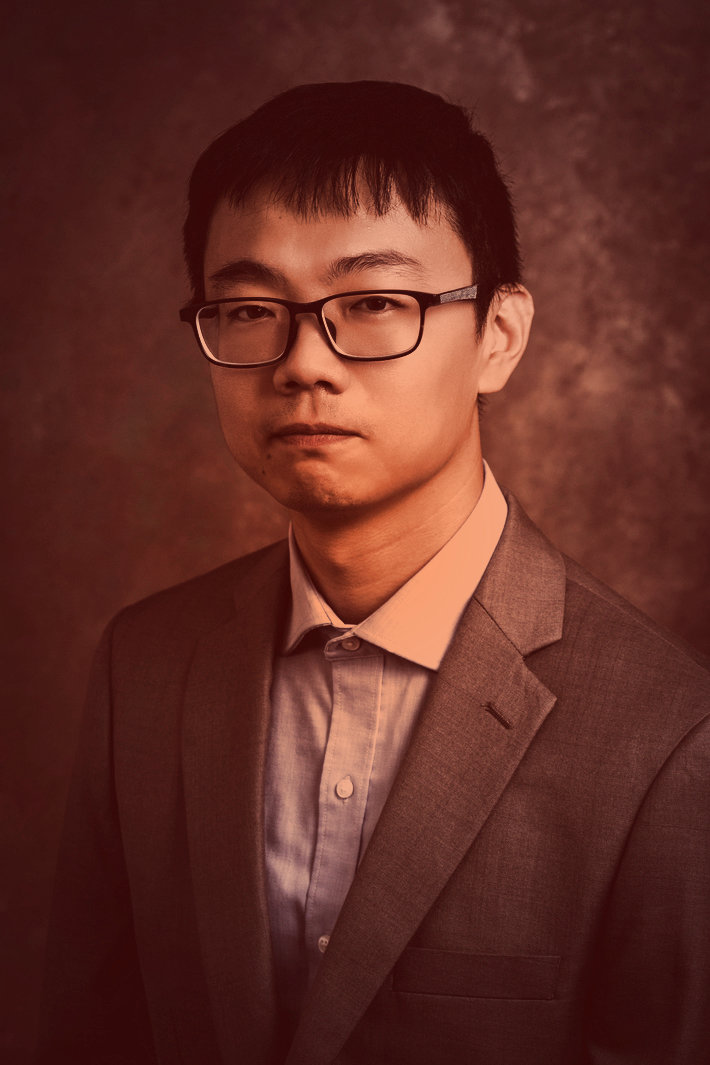}}]{Guangmo (Amo)Tong} is an Assistant Professor in the Department of Computer and Information Sciences at the University of Delaware. He received a Ph.D. in the Department of Computer Science at the University of Texas at Dallas in 2018. He received his BS degree in Mathematics and Applied Mathematics from Beijing Institute of Technology in July 2013. His research interests include computational social systems, data science and theoretical computer science. 
\end{IEEEbiography}

\begin{IEEEbiography}[{\includegraphics[width=0.9in,clip,keepaspectratio]{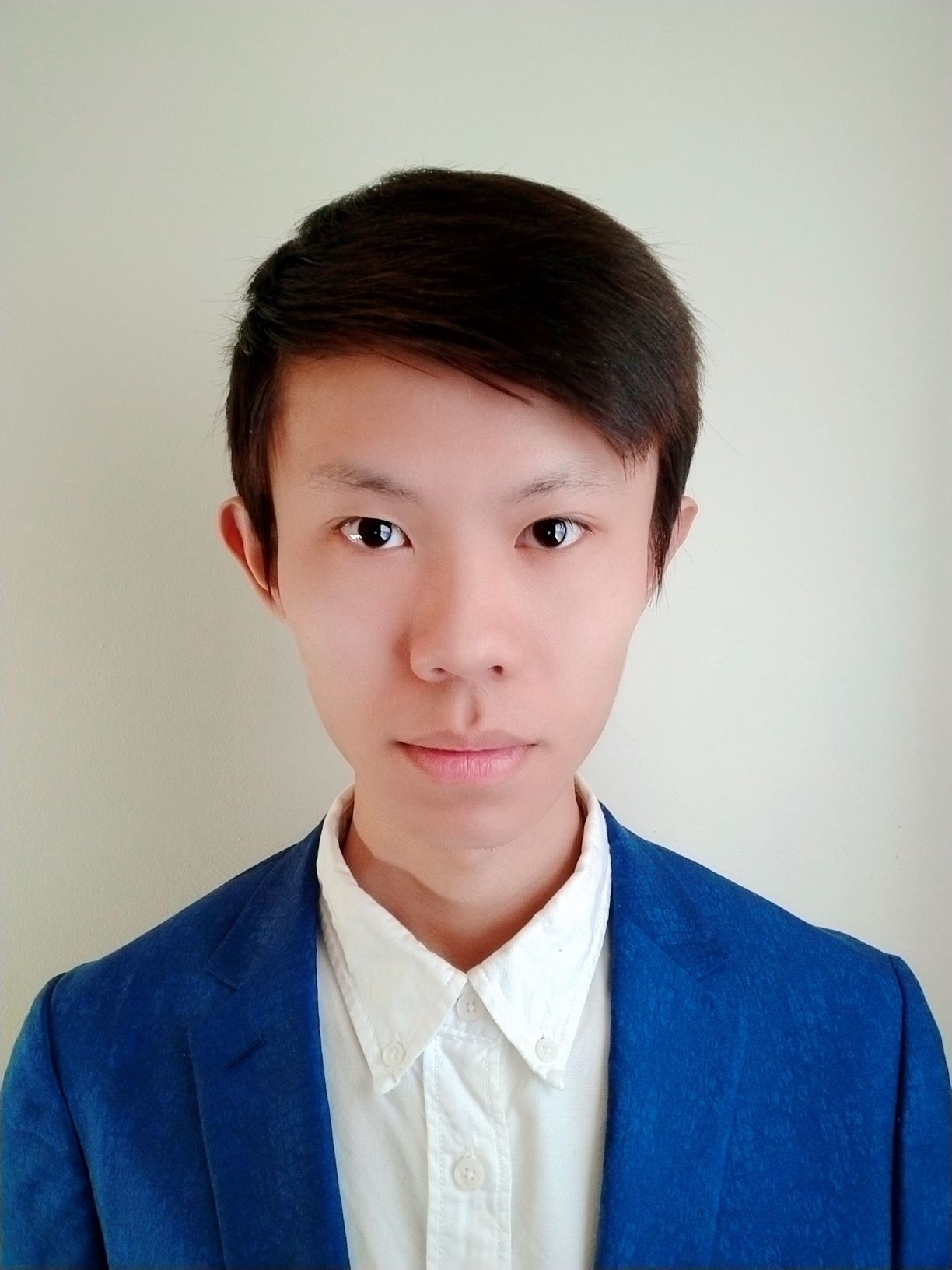}}]{Ruiqi Wang} received his B.E. degree in Information and Software Engineering from University of Electronic Science and Technology of China, in 2018. He is currently pursuing a Master degree in Computer Science at University of Delaware. His current research interests are in the area of social networks and information diffusion.

\end{IEEEbiography}

\vfill

\newpage



\section*{On Adaptive Influence Maximization under General Feedback Models (Supplementary Material)}

\section{Missing Proofs}
\begin{figure*}[!t]
	\begin{center}
		\frame{	\includegraphics[width=0.99\textwidth]{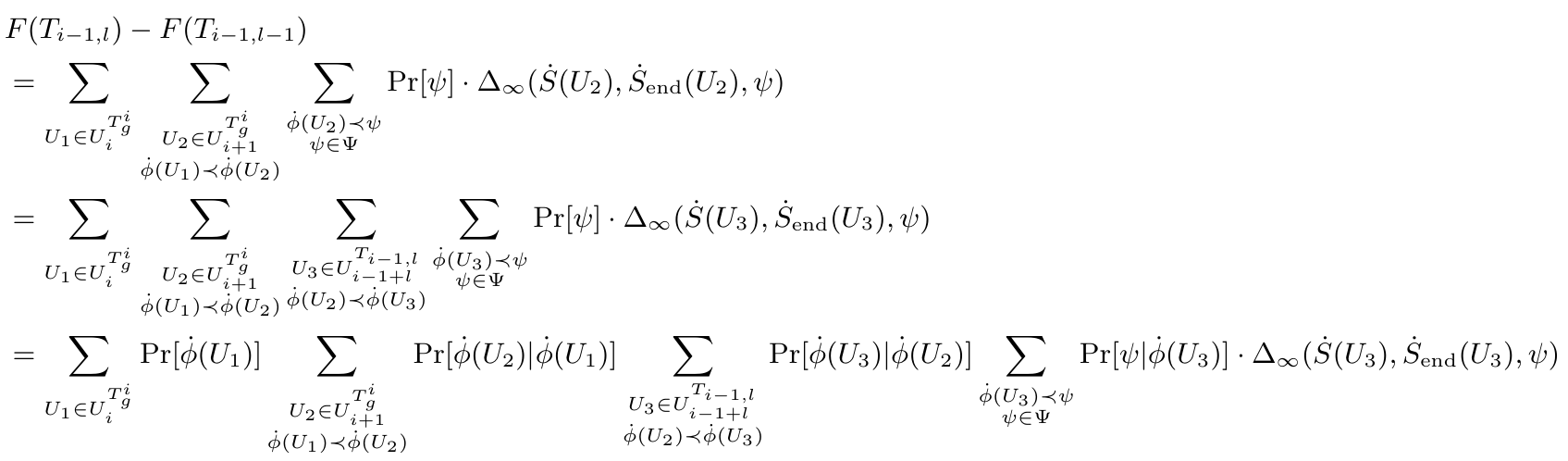} }
	\end{center} 
	\vspace{-3mm}
	\caption{Analysis with Long Equations \RNum{1}.}
	\label{fig: longeq1}
	\vspace{-0mm}
\end{figure*}

\subsection{Proof of  Lemma \ref{lemma: monotone}}
\label{appendix: proof_lemma: monotone}
{\small
	\begin{align*}
	&F(T_1 \oplus T_2) \label{eq: xxx}\\
	&\{\small \text{By Eq. (\ref{eq: F(T)})}\}  \\
	&=\sum_{U \in U^{T_1 \oplus T_2}_{\infty}}\sum_{\substack{\psi \in \Psi \\ \dot{\phi}(U) \prec \psi}} \Pr[\psi]\cdot |A_\infty(\dot{S}(U)+\dot{S}_{\e}(U),\psi)|  \\
	&\{\small \text{By Def. \ref{def: con_tree}}\}  \\
	&=\sum_{U_1 \in U^{T_1}_{\infty}}\sum_{\substack{U_2 \in U^{T_2}_{\infty} \\ \dot{\phi}(U_1) \sim \dot{\phi}(U_2)}} \sum_{\substack{\psi \in \Psi \\ \dot{\phi}(U_1 \cup U_2) \prec \psi}} \\
	&\hspace{+2cm} \Pr[\psi]\cdot |A_\infty(\dot{S}(U_1 \cup U_2)+\dot{S}_{\e}(U_2),\psi)| \\
	&\{\small \text{Step 3}\}  \\
	&=\sum_{\substack{U_2 \in U^{T_2}_{\infty} }} \sum_{\substack{U_1 \in U^{T_1}_{\infty} \\ \dot{\phi}(U_1) \sim \dot{\phi}(U_2) }} \sum_{\substack{\psi \in \Psi \\ \dot{\phi}(U_1 \cup U_2) \prec \psi}}\\
	&\hspace{+2cm} \Pr[\psi]\cdot |A_\infty(\dot{S}(U_1 \cup U_2)+\dot{S}_{\e}(U_2),\psi)| \\
	&\{\small \text{Step 4}\}  \\
	&\geq\sum_{\substack{U_2 \in U^{T_2}_{\infty} }} \sum_{\substack{U_1 \in U^{T_1}_{\infty} \\ \dot{\phi}(U_1) \sim \dot{\phi}(U_2) }} \sum_{\substack{\psi \in \Psi \\ \dot{\phi}(U_1 \cup U_2) \prec \psi}}\\ 
	&\hspace{+2cm} \Pr[\psi]\cdot |A_\infty(\dot{S}(U_2)+\dot{S}_{\e}(U_2),\psi)| \\
	&\{\small \text{Step 5}\}  \\
	&=\sum_{\substack{U_2 \in U^{T_2}_{\infty} }}  \sum_{\substack{\psi \in \Psi \\ \dot{\phi}(U_2) \prec \psi}} \\
	&\hspace{2cm}\Pr[\psi]\cdot |A_\infty(\dot{S}(U_2)+\dot{S}_{\e}(U_2),\psi)| \nonumber\\
	&\{\small \text{By Eq. (\ref{eq: F(T)})}\} \nonumber \\
	&=F(T_2)\nonumber
	\end{align*}}
The third step follows from the fact that the collection of $\{\psi \in \Psi: \dot{\phi}(U_1 \cup U_2) \prec \psi\}$ among all possible pairs $U_1 \in U^{T_1}_{\infty}$ and $U_2 \in U^{T_2}_{\infty}$ such that $\dot{\phi}(U_1) \sim \dot{\phi}(U_2)$ forms a partition of $\Psi$. The fourth step follows from the fact that $A_\infty(S,\psi)$ is monotone with respect to $S$. The fifth step follows from the fact that for each $U_2 \in U^{T_2}_{\infty}$ we have
{\small
 \[\bigcup_{\substack{U_1 \in U^{T_1}_{\infty} \\ \dot{\phi}(U_1) \sim \dot{\phi}(U_2) }}\{\psi \in \Psi: \dot{\phi}(U_1 \cup U_2) \prec \psi\}=\{\psi \in \Psi: \dot{\phi}(U_2) \prec \psi\}.\]}

\subsection{Proof of Lemma \ref{lemma: key}}
\label{subsubsec: lemma_key}
There will be two parts of analysis with long equations which are shown in Figs. \ref{fig: longeq1} and \ref{fig: longeq2}. By Def. \ref{def: tree_profit}, we have
{\small
	\begin{align*}
	F(T_{0,i-1})=\sum_{U \in U_{i}^{T_g^{i}}}\sum_{\substack{\dot{\phi}(U) \prec \psi \\ \psi \in \Psi}} \Pr[\psi] \cdot  |A_{\infty}(\dot{S}(U),\psi)|
	\end{align*}}
and
{\small
	\begin{align*}
	F(T_{0,i})=\sum_{U \in U_{i}^{T_g^{i}}}\sum_{\substack{\dot{\phi}(U) \prec \psi \\ \psi \in \Psi}} \Pr[\psi] \cdot  |A_{\infty}(\dot{S}(U)+\dot{S}_{\e}(U),\psi)|.
	\end{align*}}
Therefore, $F(T_{0,i})-F(T_{0,i-1})$ is equal to 
{\small
	\begin{align}
	\label{eq: lemma_key_right}
	&\sum_{U \in U_{i}^{T_g^{i}}}\sum_{\substack{\dot{\phi}(U) \prec \psi \\ \psi \in \Psi}} \Pr[\psi] \cdot  \Delta_{\infty}(\dot{S}(U), \dot{S}_{}(U), \psi)\nonumber \\
	&=\sum_{U \in U_i^{T_g^i}}\Pr[\dot{\phi}(U)] \sum_{ \substack{\dot{\phi}(U) \prec \psi \\ \psi \in \Psi}} \Pr[\psi|\dot{\phi}(U)]\cdot \Delta_{\infty}(\dot{S}(U), \dot{S}_{\e}(U), \psi)\nonumber \\
	&\{\text{By the greedy policy}\}\nonumber \\
	&=\sum_{U \in U_i^{T_g^i}}\Pr[\dot{\phi}(U)] \cdot \max_v  \Big(\sum_{ \substack{\dot{\phi}(U) \prec \psi \\ \psi \in \Psi}} \Pr[\psi|\dot{\phi}(U)]\cdot \Delta_{\infty}(\dot{S}(U), v, \psi)\Big)
	\end{align}}
Similarly, $F(T_{i-1,l})-F(T_{i-1,l-1})$ is equal to 
{
	\begin{equation*}
	\sum_{U \in U_{i-1+l}^{T_{i-1,l}}}\sum_{ \substack{\dot{\phi}(U) \prec \psi \\ \psi \in \Psi}} \Pr[\psi]\cdot \Delta_{\infty}(\dot{S}(U), \dot{S}_{\e}(U), \psi)   
	\end{equation*}}
Grouping the realizations in $U_{i-1+l}^{T_{i-1,l}}$ according to that if they are the super-realizations of that in $U_{i}^{T_g^{i}}$, we have
{\small
	\begin{align}
	\label{eq: lemma_key_left}
	&\sum_{U \in U_{i-1+l}^{T_{i-1,l}}}\sum_{ \substack{\dot{\phi}(U) \prec \psi \\ \psi \in \Psi}} \Pr[\psi]\cdot \Delta_{\infty}(\dot{S}(U), \dot{S}_{\e}(U), \psi) \\
	&=\sum_{U_1 \in U_{i}^{T_g^{i}}}\sum_{U_2 \in U_{i-1+l}^{T_{i-1,l}}}\sum_{ \substack{\dot{\phi}(U_2) \prec \psi \\ \psi \in \Psi}}\Pr[\psi]\cdot \Delta_{\infty}(\dot{S}(U_2), \dot{S}_{\e}(U_2), \psi) \nonumber
	\end{align}}
Now to prove Lemma \ref{lemma: key}, it suffices to show that \[\text{Eq}. (\ref{eq: lemma_key_left}) \leq \alpha(T_g) \cdot \text{Eq}. (\ref{eq: lemma_key_right}).\] To this end, let us further group the super-realizations of $U \in U_{i}^{T_g^{i}}$ by that if they are the super-realizations of that in $U \in U_{i+1}^{T_g^{i}}$, and Eq. (\ref{eq: lemma_key_left}) can be further represented as shown in Fig. \ref{fig: longeq1}. Finally, we have the following lemma as the last ingredient.

\begin{lemma}
	\label{lemma: submodular}
	For each $U_2 \in U_{i+1}^{T_g^{i}}$,  and $U_3 \in U_{i-1+l}^{T_{i-1,l}}$ such that $\dot{\phi}(U_2)\prec \dot{\phi}(U_3)$, consider two functions over $V$,
	\begin{equation*}
	g_1(v)\define \sum_{\substack{\dot{\phi}(U_2) \prec \psi \\ \psi \in \Psi}} \Pr[\psi|\dot{\phi}(U_2)] \cdot  \Delta_{\infty}(\dot{S}(U_2), v, \psi)
	\end{equation*}
	and
	\begin{equation*}
	g_2(v)\define \sum_{ \substack{\dot{\phi}(U_3) \prec \psi \\ \psi \in \Psi}}\Pr[\psi|\dot{\phi}(U_3)]\cdot \Delta_{\infty}(\dot{S}(U_3), v, \psi).
	\end{equation*}
	We have $g_1(v) \geq g_2(v)$.
\end{lemma}
\begin{proof}
	See Sec. \ref{appendix: proof_lemma: submodular}.
\end{proof}

Supposing Lemma \ref{lemma: submodular} is true, we have the result in Fig .\ref{fig: longeq2}, which completes the proof of Lemma \ref{lemma: key}.

\begin{figure*}[t]
	\begin{center}
		\frame{\includegraphics[width=0.99\textwidth]{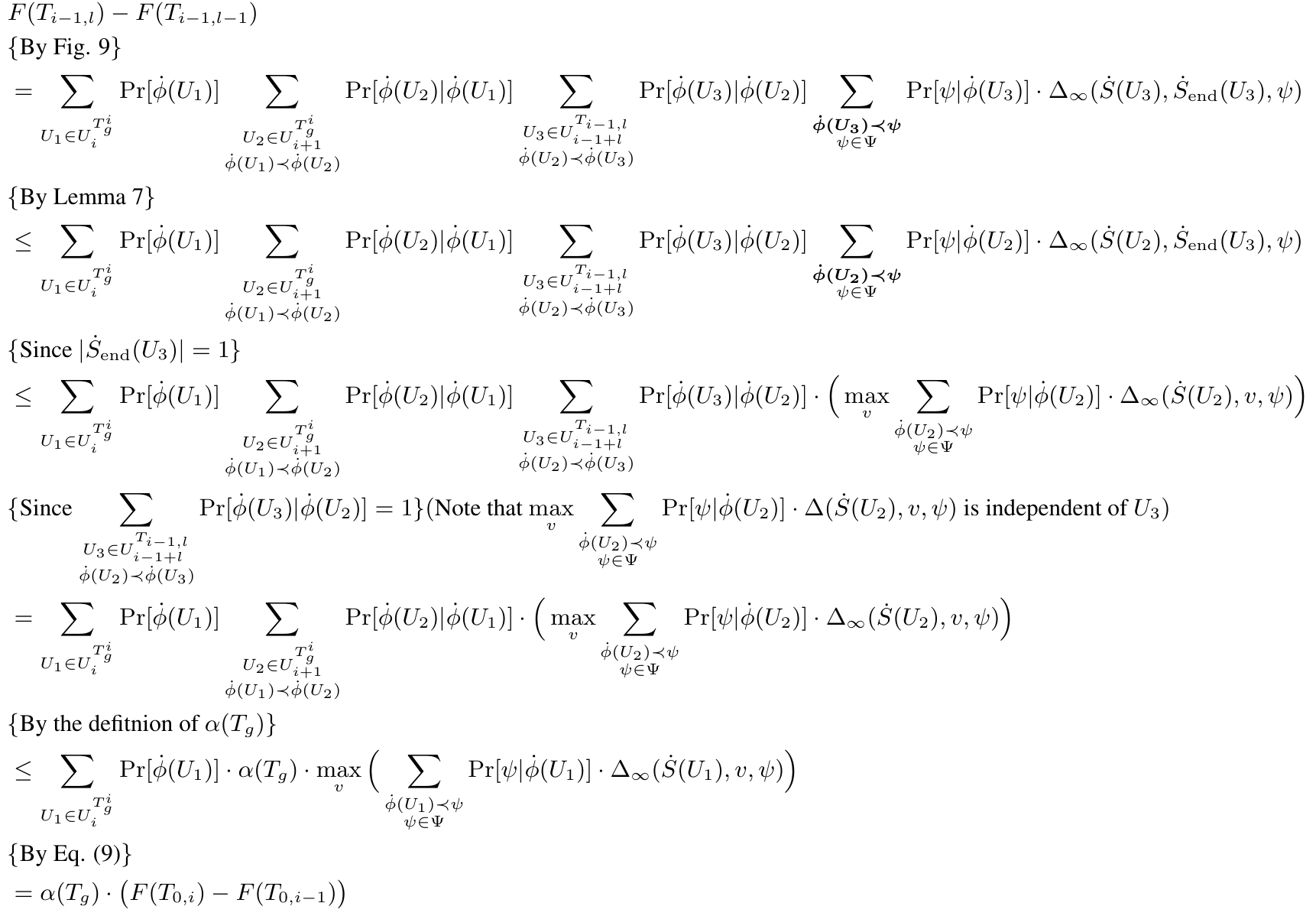}}
	\end{center} 
	\vspace{-3mm}
	\caption{Analysis with Long Equations \RNum{2}.}
	\vspace{-3mm}
	\label{fig: longeq2}
\end{figure*}

\subsection{Proof of  Lemma \ref{lemma: submodular}}
\label{appendix: proof_lemma: submodular}
We prove the following general lemma of which Lemma \ref{lemma: submodular} is a special case.
\begin{lemma}
	\label{lemma: general_submodular}
	For each pair of status $U_1$ and $U_2$, such that $U_1$ is final, $\dot{S}(U_1) \subseteq \dot{S}(U_2)$, and $L\big(\dot{\phi}(U_1)\big) \cup D\big(\dot{\phi}(U_1)\big) \subseteq L\big(\dot{\phi}(U_2)\big) \cup D\big(\dot{\phi}(U_2)\big)$. Let us consider two functions: 
	\begin{equation}
	g_1(v)\define \sum_{\substack{\dot{\phi}(U_1) \prec \psi \\ \psi \in \Psi}} \Pr[\psi|\dot{\phi}(U_1)] \cdot  \Delta_{t}(\dot{S}(U_1), v, \psi)
	\end{equation}
	and
	\begin{equation}
	g_2(v)\define \sum_{ \substack{\dot{\phi}(U_2) \prec \psi \\ \psi \in \Psi}}\Pr[\psi|\dot{\phi}(U_2)]\cdot \Delta_{t}(\dot{S}(U_2), v, \psi).
	\end{equation}
	We have $g_1(v) \geq g_2(v)$ for each $v \in V$ and $t\in \mathbb{Z}^+$.
\end{lemma}

\begin{table*}[tp!]
	\renewcommand{\arraystretch}{1.5}
	\centering
	{\begin{tabular}{| C{1.7cm} || C{1.7cm}| C{1.7cm}| C{9.7cm}| }
			\hline
			\textbf{} 	& \textbf{Node \#} 		& \textbf{Edge \# }& \textbf{Description}  \\   
			\hline
			\hline
			\textbf{Power} 	& 2,500	& 26,000	& A synthetic lower-law network  \\
			\hline
			\textbf{Wiki} 	& 8,300	& 103,000	& A who-votes-on-whom network from Wikipedia\\
			\hline 
			\textbf{Higgs} 	& 10,000	& 22,482	& A Twitter cascade regarding regarding the discovery of a new particle \\
			\hline 
			\textbf{Hepph} 	& 35,000	& 421,482	& Arxiv High Energy Physics paper
			citation network  \\
			\hline 
			\textbf{Hepth} 	& 28,000	& 353,482	& Arxiv High Energy Physics Theory paper citation network \\
			\hline 
			\textbf{DBLP} 	& 317,000	& 1,040,482	& DBLP collaboration network \\
			\hline 
			\textbf{Youtube} 	& 1,100,000	& 6,000,000	& Youtube online social network\\
			\hline
	\end{tabular}}
	\caption{Datasets}
	\label{table: dataset}
	\vspace{-3mm}
\end{table*}

To prove Lemma \ref{lemma: general_submodular}, let us consider the edge set \[E_1 \define E\setminus \Big(L\big(\dot{\phi}(U_1)\big)\cup D\big(\dot{\phi}(U_1)\big) \Big)\] which consists of the edges that are to be determined in each full-realization $\dot{\phi}(U_1) \prec \psi$. Similarly, let us define 
\[E_2\define E\setminus \Big(L\big(\dot{\phi}(U_2)\big)\cup D\big(\dot{\phi}(U_2)\big) \Big)\] 
with respect to $U_2$. Since $E_2 \subseteq E_1$, the edge set $E$ can be partitioned into three parts, $\{L(\dot{\phi}(U_1))\cup D(\dot{\phi}(U_1))\}$, $E_1 \setminus E_2$ and $E_2$. Consider two sets of realizations: $\Phi_{2} =\{\phi\in \Phi: L(\phi)\cup D(\phi)=E_2\}$ and $\Phi_{1|2} =\{\phi\in \Phi: L(\phi)\cup D(\phi)=E_1\setminus E_2\}$. With Def. \ref{def: realization_con}, we have
\begin{equation*}
\{\psi \in \Psi: \dot{\phi}(U_2) \prec \psi \}=\{\dot{\phi}(U_2) \oplus \phi_2: \phi_2 \in \Phi_{2}\},
\end{equation*} 
and
{\small
	\begin{eqnarray*}
		\{\psi \in \Psi: \dot{\phi}(U_1) \prec \psi \}=\{\dot{\phi}(U_1)\oplus \phi_1 \oplus \phi_2: \phi_1 \in \Phi_{1|2}, \phi_2 \in \Phi_{2}\}.
\end{eqnarray*}}
With these notations, we have
{\small
	\begin{align*}
	&g_1(v)=\sum_{ \substack{\dot{\phi}(U_1) \prec \psi \\ \psi \in \Psi}} \Pr[\psi|\dot{\phi}(U_1)]\cdot \Delta_{t}(\dot{S}(U_1), v, \psi)\\
	&= \sum_{ \phi_2 \in \Psi_2} \sum_{ \phi_1 \in \Phi_{1|2}}\Pr[\dot{\phi}(U_1)\oplus \phi_1 \oplus \phi_2|\dot{\phi}(U_1)] \\
	&\hspace{3.5cm}\cdot \Delta_{t}(\dot{S}(U_1), v, \dot{\phi}(U_1)\oplus \phi_1 \oplus \phi_2)\\
	&= \sum_{ \phi_2 \in \Psi_2} \Pr[\phi_2] \sum_{ \phi_1 \in \Phi_{1|2}}\Pr[\phi_1] \\
	&\hspace{3.5cm}\cdot \Delta_{t}(\dot{S}(U_1), v, \dot{\phi}(U_1)\oplus \phi_1 \oplus \phi_2)
	\end{align*}}
and
{\small
	\begin{align*}
	&g_2(v)=\sum_{ \substack{\dot{\phi}(U_2) \prec \psi \\ \psi \in \Psi}} \Pr[\psi|\dot{\phi}(U_2)]\cdot \Delta_{t}(\dot{S}(U_2), v, \psi)\\
	&= \sum_{ \phi_2 \in \Psi_2} \Pr[\dot{\phi}(U_2)\oplus \phi_2|\dot{\phi}(U_2)] \cdot \Delta_{t}(\dot{S}(U_3), v, \dot{\phi}(U_2)\oplus \phi_2) \\
	&= \sum_{ \phi_2 \in \Psi_2} \Pr[\phi_2] \cdot \Delta_{t}(\dot{S}(U_2), v, \dot{\phi}(U_2)\oplus \phi_2)
	\end{align*}}
Since $\sum_{ \phi_1 \in \Phi_{1|2}}\Pr[\phi_1]=1$, to prove $g_1(v) \geq g_2(2)$, it suffices to prove the following lemma.
\begin{lemma}
	For each $\phi_1 \in \Phi_{1|2}$ and $\phi_2 \in \Phi_2$, we have $\Delta_{t}(\dot{S}(U_1), v, \dot{\phi}(U_1)\oplus \phi_1 \oplus \phi_2) \geq \Delta_{t}(\dot{S}(U_2), v,  \dot{\phi}(U_2)\oplus \phi_2)$ .
\end{lemma}
\begin{proof}
	For conciseness, let us define that
	{\small
		\begin{align*}
		\begin{split}
		&\hspace{0cm}V_1 \define\\
		&|A_{t}(\dot{S}(U_1)+v,\dot{\phi}(U_1)\oplus \phi_1 \oplus \phi_2)\hspace{0cm}\setminus A_{t}(\dot{S}(U_1),\dot{\phi}(U_1)\oplus \phi_1 \oplus \phi_2)|
		\end{split}
		\end{align*}}
	and
	{\small
		\begin{align*}
		&\hspace{0cm}V_2 \define|A_t(\dot{S}(U_2)+v,\dot{\phi}(U_2)\oplus  \phi_2)\setminus A_{t}(\dot{S}(U_2), \dot{\phi}(U_2) \oplus \phi_2)|,
		\end{align*}}
	By Def. \ref{def: A_t}, we have 
	{
		\begin{align*}
		\hspace{0cm}\Delta_{t}(\dot{S}(U_1), v, \dot{\phi}(U_1)\oplus \phi_1 \oplus \phi_2)=|V_1|
		\end{align*}}
	and
	{
		\begin{align*}
		\hspace{0cm}\Delta_t(\dot{S}(U_2), v,  \dot{\phi}(U_2)\oplus \phi_2) =|V_2|.
		\end{align*}}
	Thus, it is sufficient to prove that $V_2$ is subset of $V_1$. Suppose that $u$ is in $V_2$. It implies that $u$ is in $A_t(\dot{S}(U_2)+v,\dot{\phi}(U_2)\oplus  \phi_2)$ but not in $ A_t(\dot{S}(U_2),\dot{\phi}(U_2) \oplus \phi_2)$. That is, in the full realization $\dot{\phi}(U_2)\oplus \phi_2$, we have
	\begin{itemize}
		\item (a) there exists a $t$-live-path from $S(U_2)+v$ to $u$, and
		\item (b) there is no $t$-live-path from $S(U_2)$ to $u$.
	\end{itemize}
	To prove $u$ is in $V_1$, we have to prove that (a) $u$ is in $A_t(\dot{S}(U_1)+v, \dot{\phi}(U_1)\oplus \phi_1 \oplus \phi_2)$ and (b) $u$ is not in $A_t(\dot{S}(U_1), \dot{\phi}(U_1)\oplus \phi_1 \oplus \phi_2)$.
	
	First, we prove that $u$ is in $A_t(\dot{S}(U_1)+v, \dot{\phi}(U_1)\oplus \phi_1 \oplus \phi_2)$. Since $u$ is not in $ A_t(\dot{S}(U_2), \dot{\phi}(U_2) \oplus \phi_2)$, we have $u \in V\setminus \dot{S}(U_2)$. By (a), there is a $t$-live-path from $v$ to $u$ in $\dot{\phi}(U_2)\oplus \phi_2$. Furthermore, by (b), this path cannot use any edge in $L(\dot{\phi}(U_2))$, and therefore this path only uses the edges in $L(\phi_2)$. Thus, we have this live path in $\dot{\phi}(U_1)\oplus \phi_1 \oplus \phi_2$ as well, and therefore, $u \in A_t(\dot{S}(U_1)+v, \dot{\phi}(U_1)\oplus \phi_1 \oplus \phi_2)$.
	
	Second, since $U_1$ is a final status and $u \notin \dot{S}(U_1)$, we have $u \notin A_t(\dot{S}(U_1), \dot{\phi}(U_1)\oplus \phi_1 \oplus \phi_2)$.
\end{proof}

\subsection{Proof of Lemma \ref{lemma: upper_bound}}
\label{appendix: proof_lemma: upper_bound}
Because $\sum_{U_* \in \dot{\U}_d(U)} \Pr[\dot{\phi}(U_*)|\dot{\phi}(U)]=1$, it suffices to prove that for each $U_* \in  \dot{\U}_d(U)$ we have \[\Delta f_{t-d}(\dot{S}(U_{\f}),v ,\dot{\phi}(U_{\f}))\geq \Delta f_{t-d}(\dot{S}(U_*),v ,\dot{\phi}(U_*)),\] which is in fact a special case of Lemma \ref{lemma: general_submodular} given in Sec. \ref{appendix: proof_lemma: submodular}.

\section{Additional Experimental Results}
The details of the adopted datasets are given in Table. \ref{table: dataset}. Additional results from Experiment \RNum{1} and \RNum{2} are provided in Figs. \ref{fig: exp1_more_1}, \ref{fig: exp1_more_2}, \ref{fig: exp1_more_3} and \ref{fig: exp2_more}.

The observations on other datasets are similar to those discussed in the main paper. One minor point is that when $k=5$ HighDegree occasionally has the same performance as that of Greedy. In addition, the superiority of the Full Adoption feedback model is significant on certain graphs, e.g., Fig. \ref{fig: poweric50}.

\begin{figure*}[!pt]
	\centering
	\subfloat[{[Power, IC, $k=5$, $d=0$]}]{\label{fig: poweric5_3_0}\includegraphics[trim = 0.5in 0in 0.5in 0in, clip, width=0.24\textwidth]{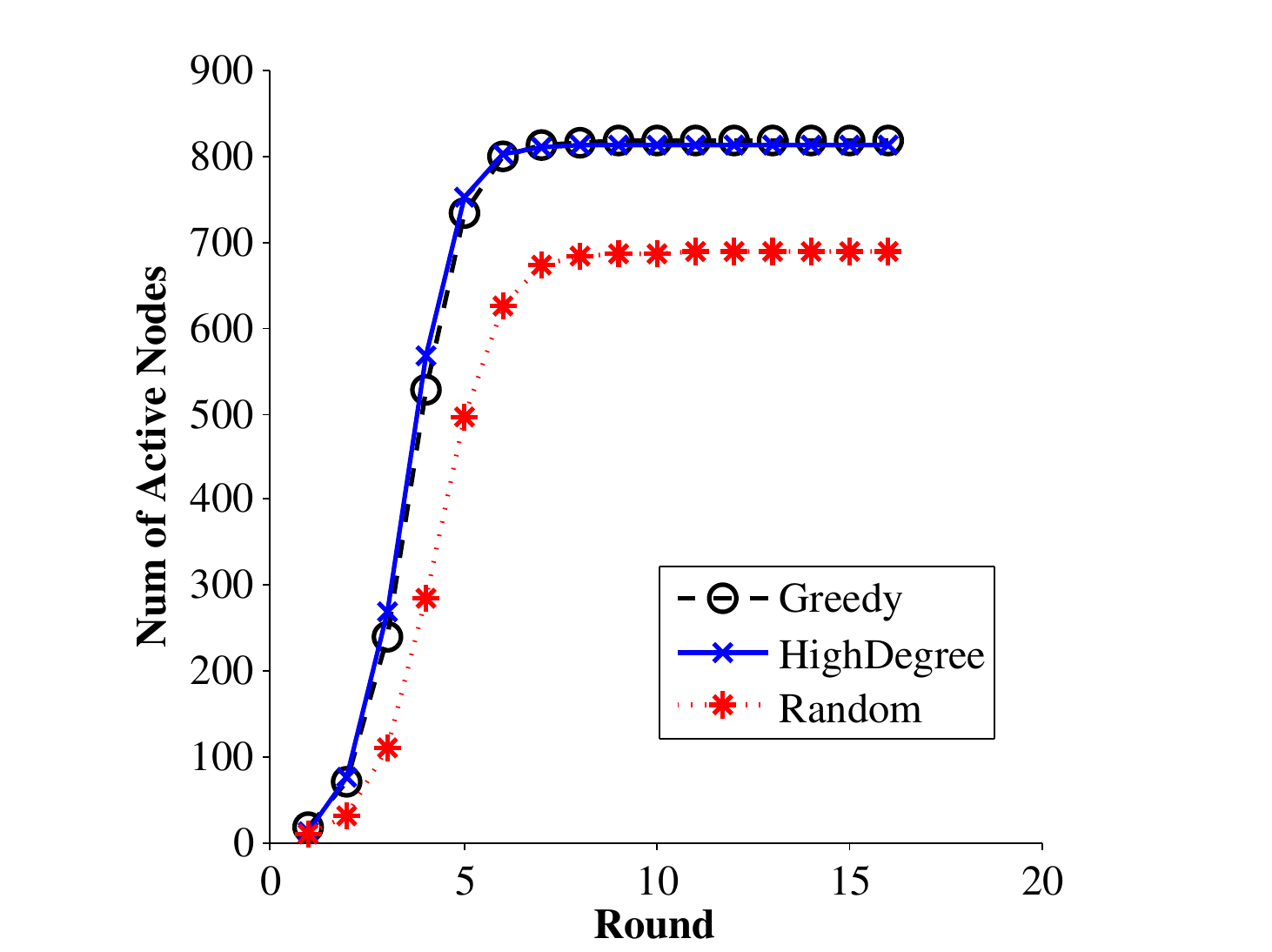}}
	\subfloat[{[Power, IC, $k=5$, $d=1$]}]{\label{fig:poweric5_3_3}\includegraphics[trim = 0.5in 0in 0.5in 0in, clip, width=0.24\textwidth]{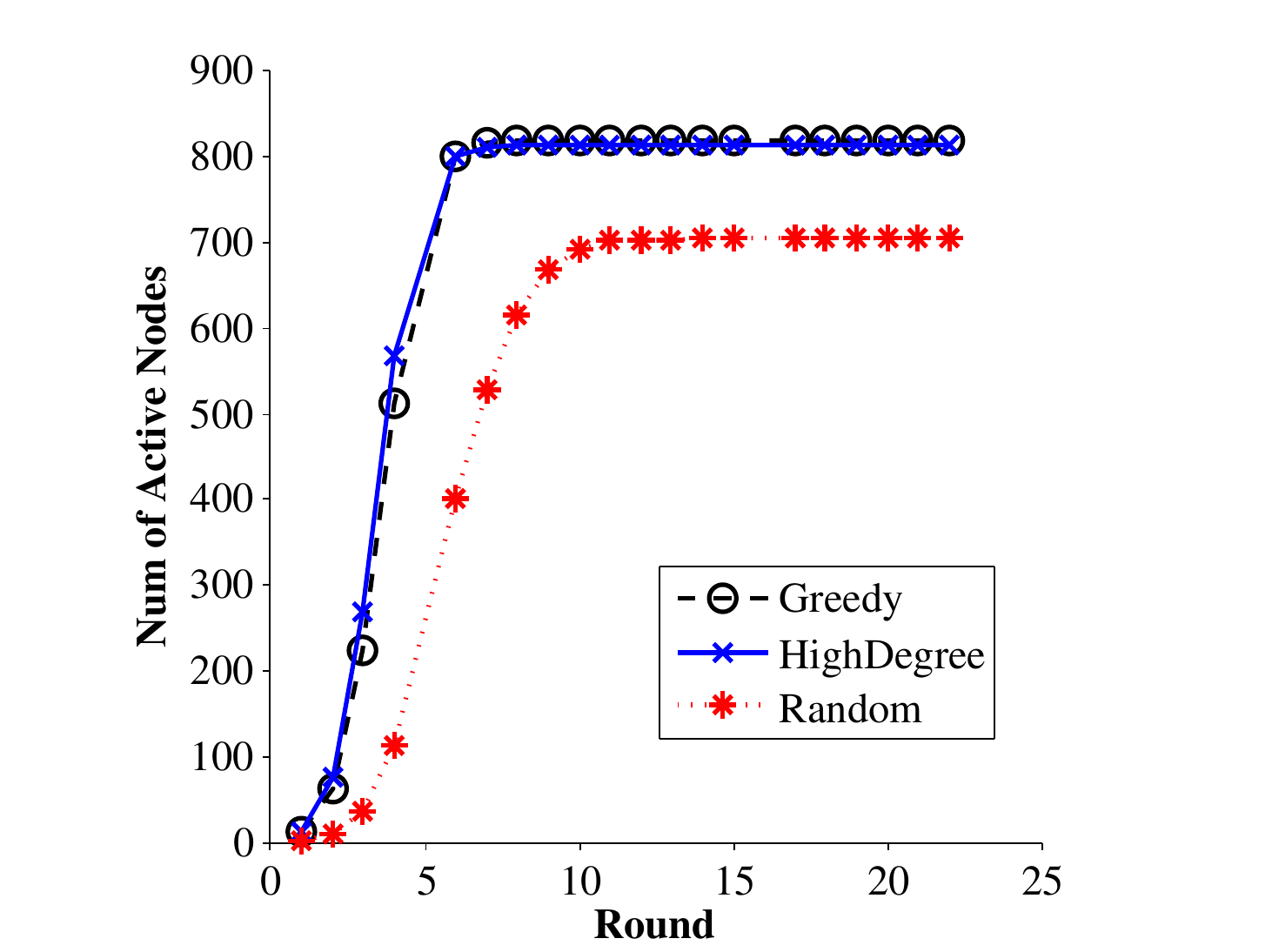}}
	\subfloat[{[Power, IC, $k=5$, $d=8$]}]{\label{fig: poweric5_3_9}\includegraphics[trim = 0.5in 0in 0.5in 0in, clip, width=0.24\textwidth]{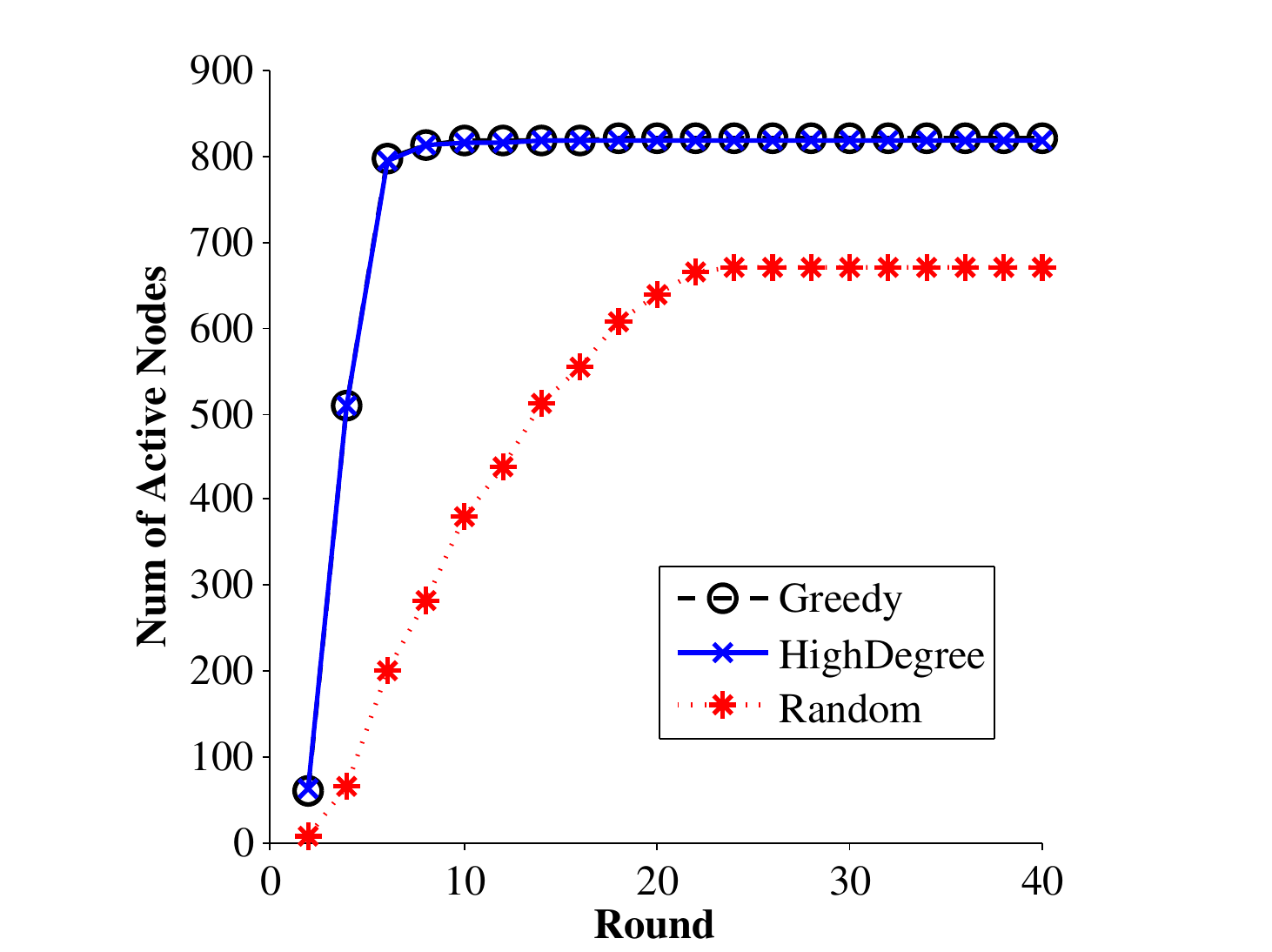}}
	\subfloat[{[Power, IC, $k=5$, $d=\infty$]}]{\label{fig: poweric5_3_15}\includegraphics[trim = 0.5in 0in 0.5in 0in, clip, width=0.24\textwidth]{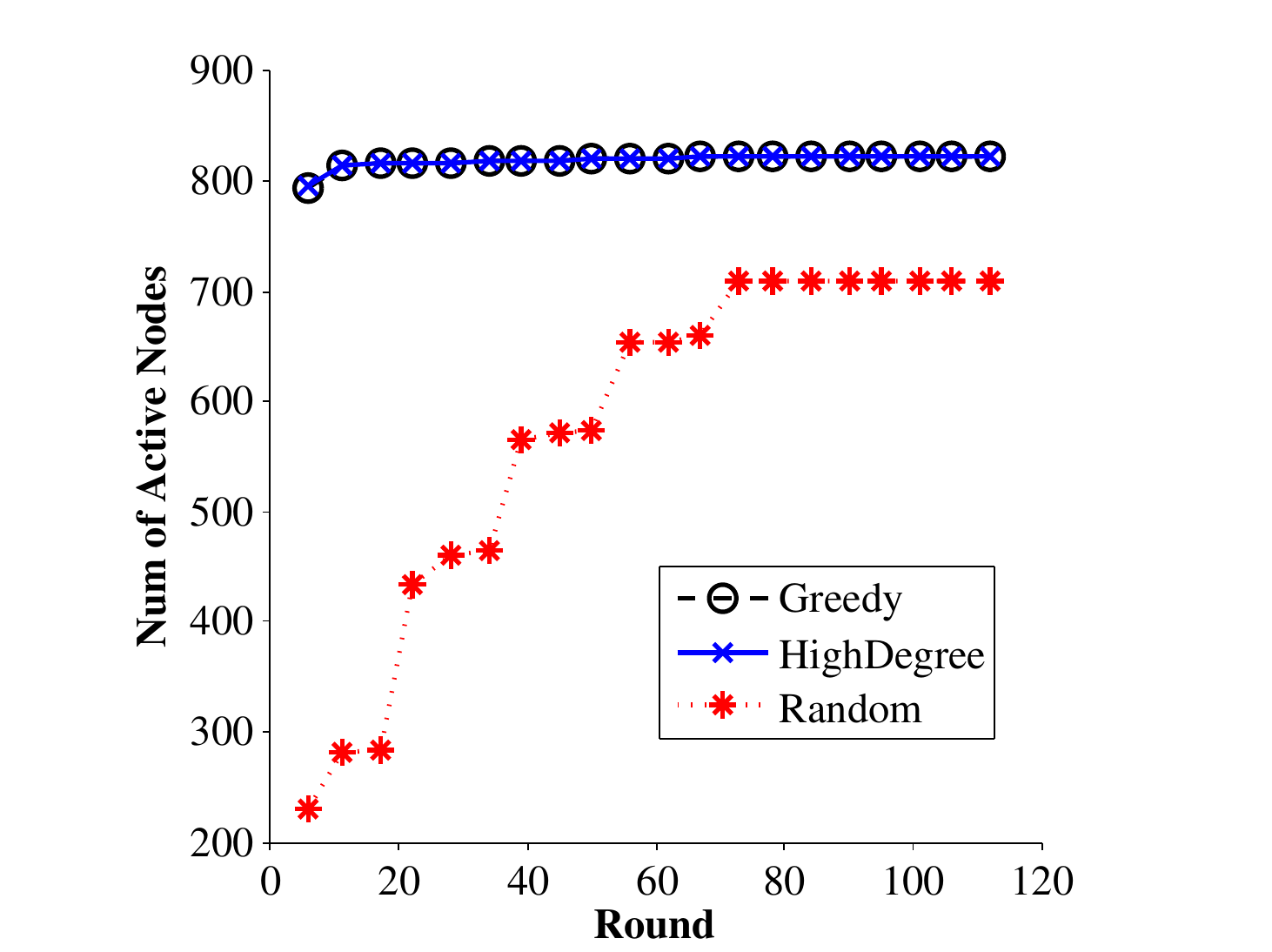}}
	
	\subfloat[{[Power, IC, $k=10$, $d=0$]}]{\label{fig: poweric10_3_0}\includegraphics[trim = 0.5in 0in 0.5in 0in, clip, width=0.24\textwidth]{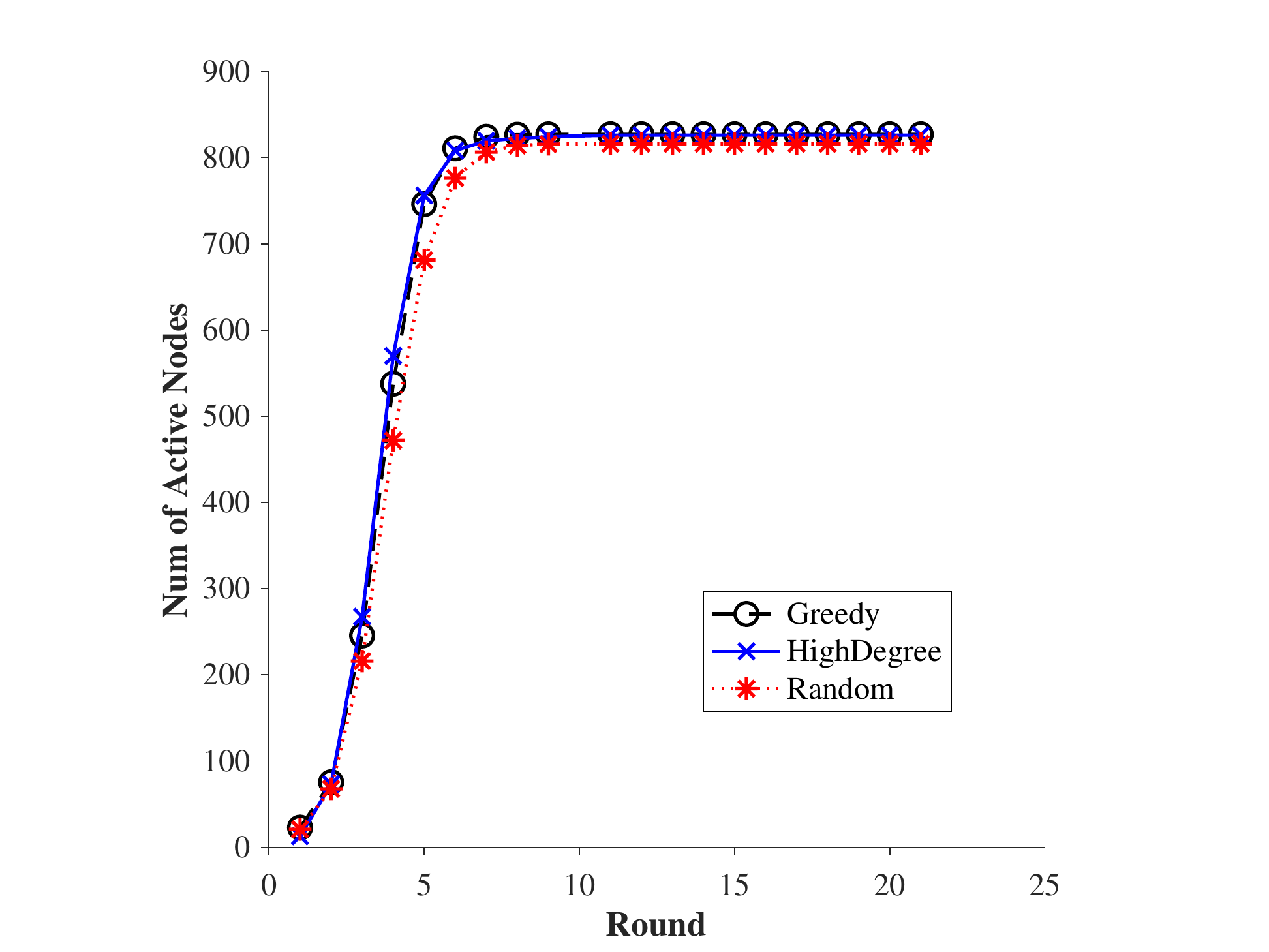}}
	\subfloat[{[Power, IC, $k=10$, $d=1$]}]{\label{fig:poweric10_3_3}\includegraphics[trim = 0.5in 0in 0.5in 0in, clip, width=0.24\textwidth]{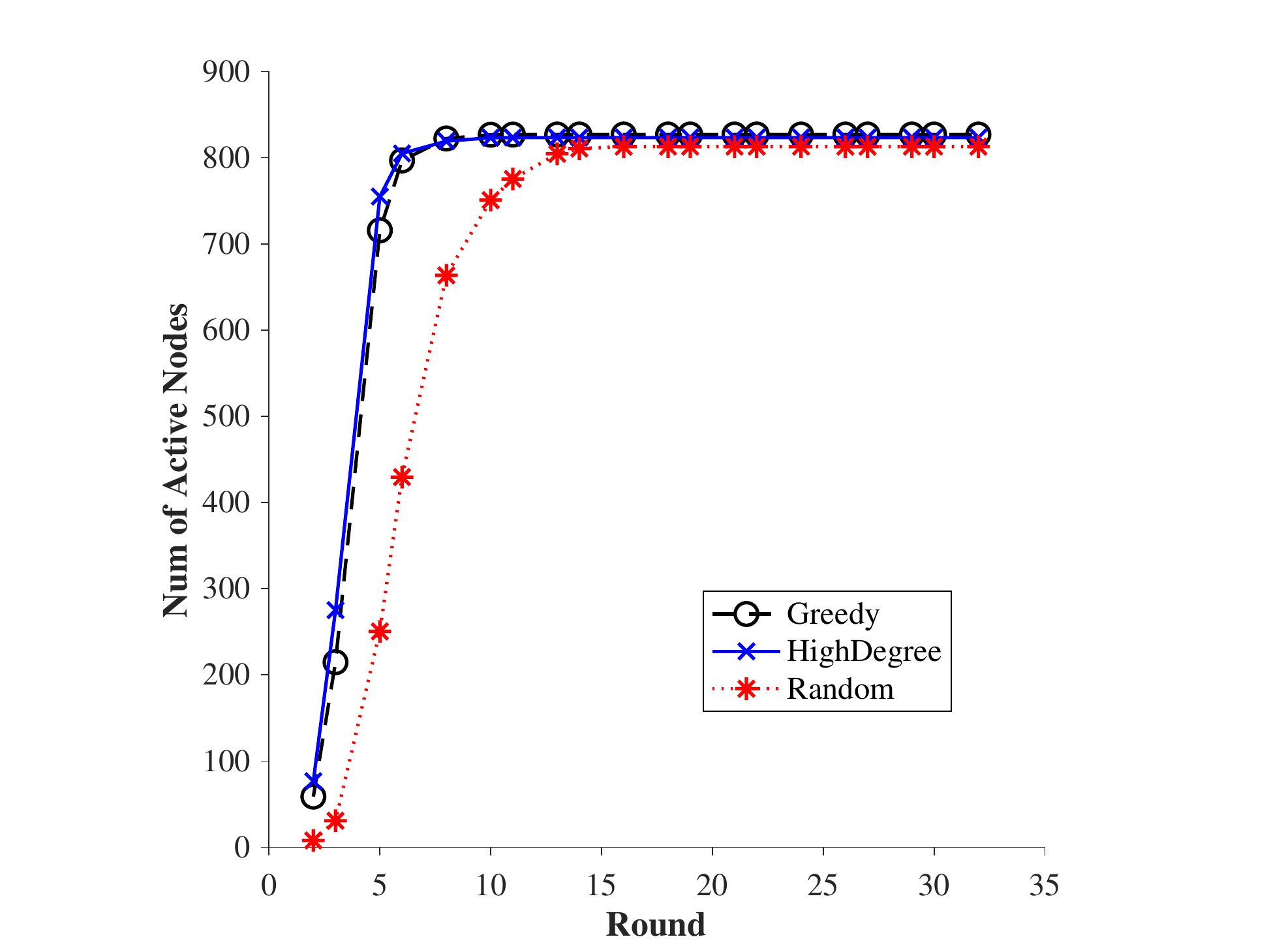}}
	\subfloat[{[Power, IC, $k=10$, $d=8$]}]{\label{fig: poweric10_3_9}\includegraphics[trim = 0.5in 0in 0.5in 0in, clip, width=0.24\textwidth]{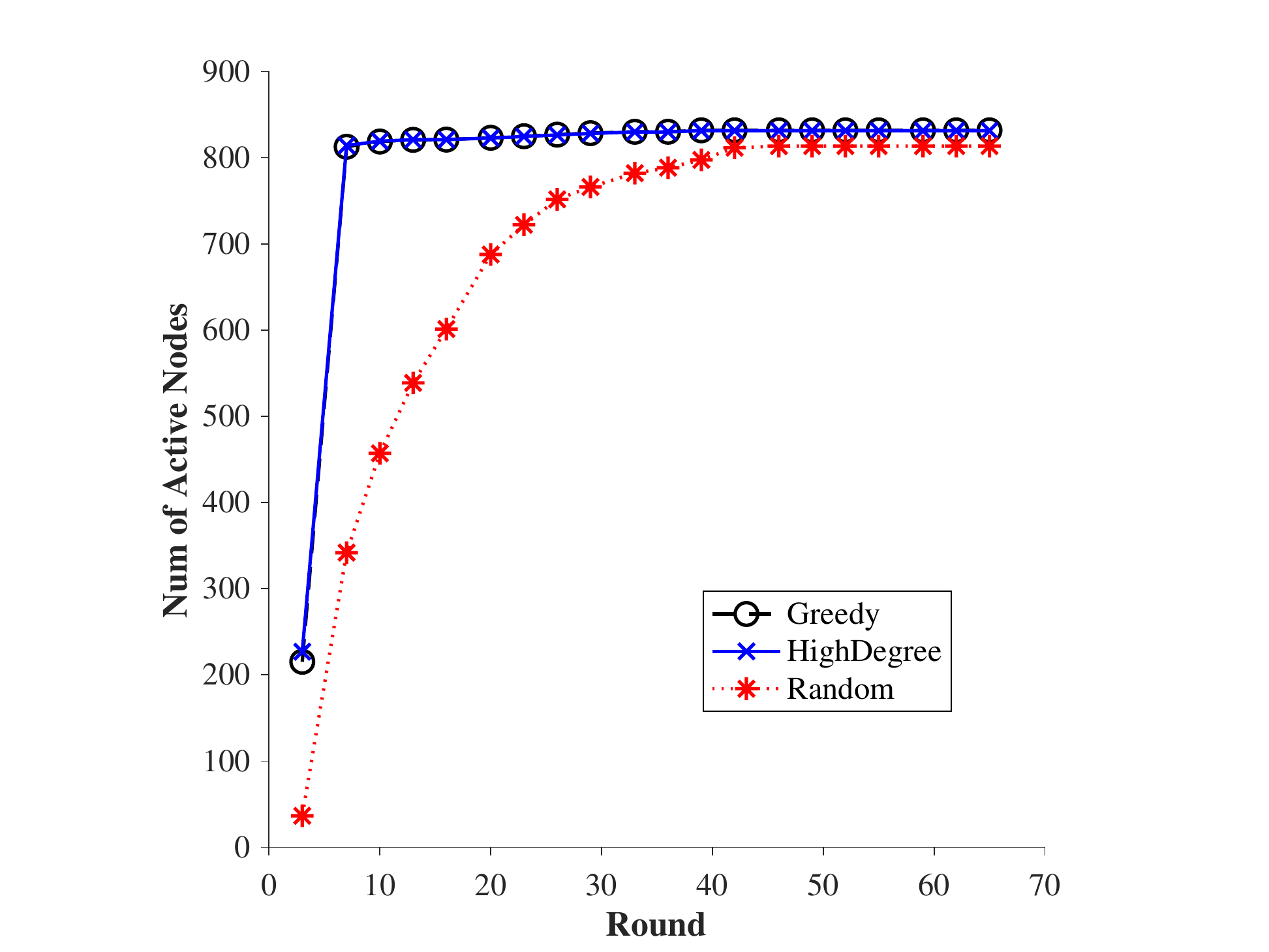}}
	\subfloat[{[Power, IC, $k=10$, $d=\infty$]}]{\label{fig: poweric10_3_15}\includegraphics[trim = 0.5in 0in 0.5in 0in, clip, width=0.24\textwidth]{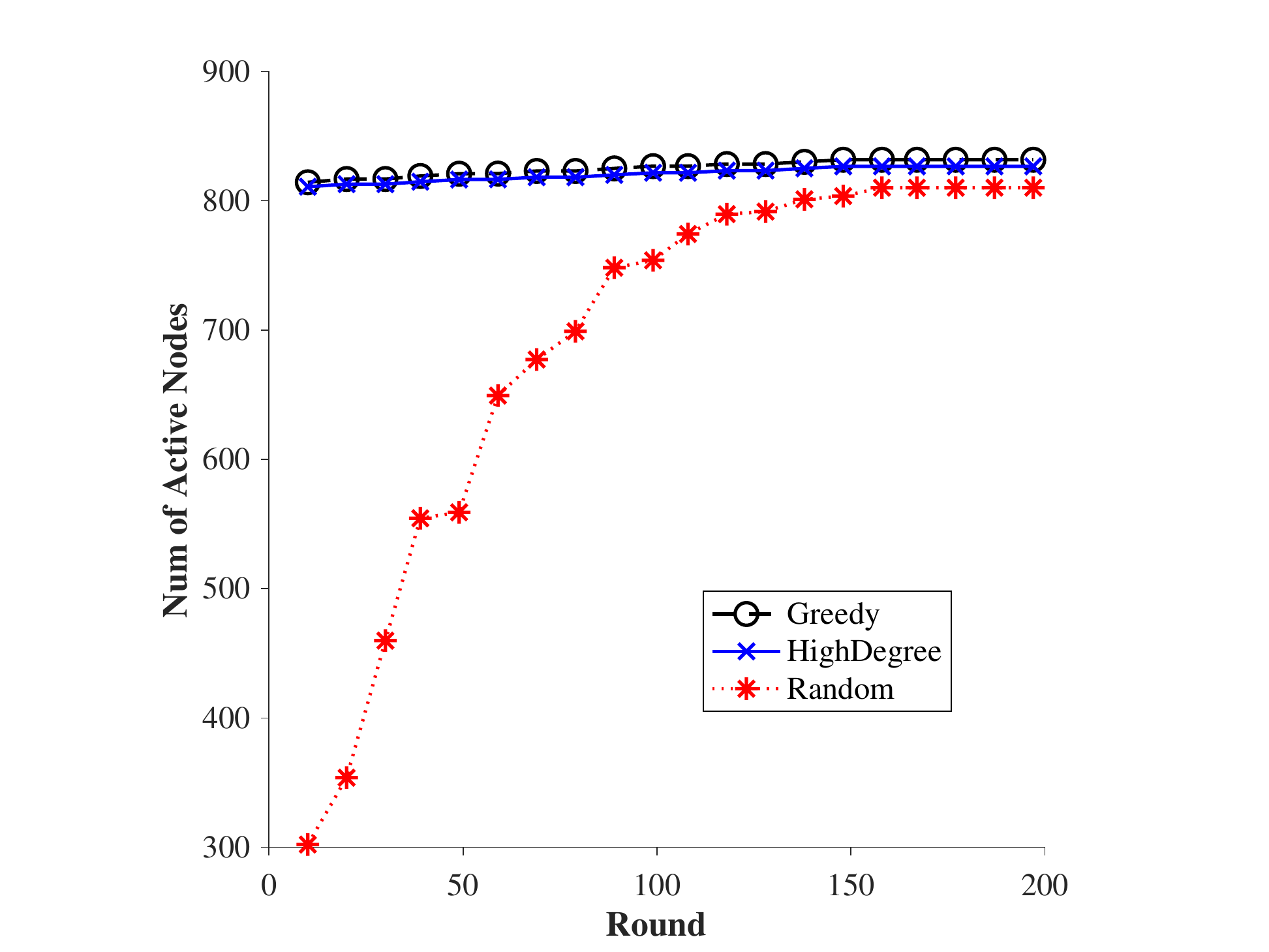}}

	\subfloat[{[Power, IC, $k=20$, $d=0$]}]{\label{fig: poweric20_3_0}\includegraphics[trim = 0.5in 0in 0.5in 0in, clip, width=0.24\textwidth]{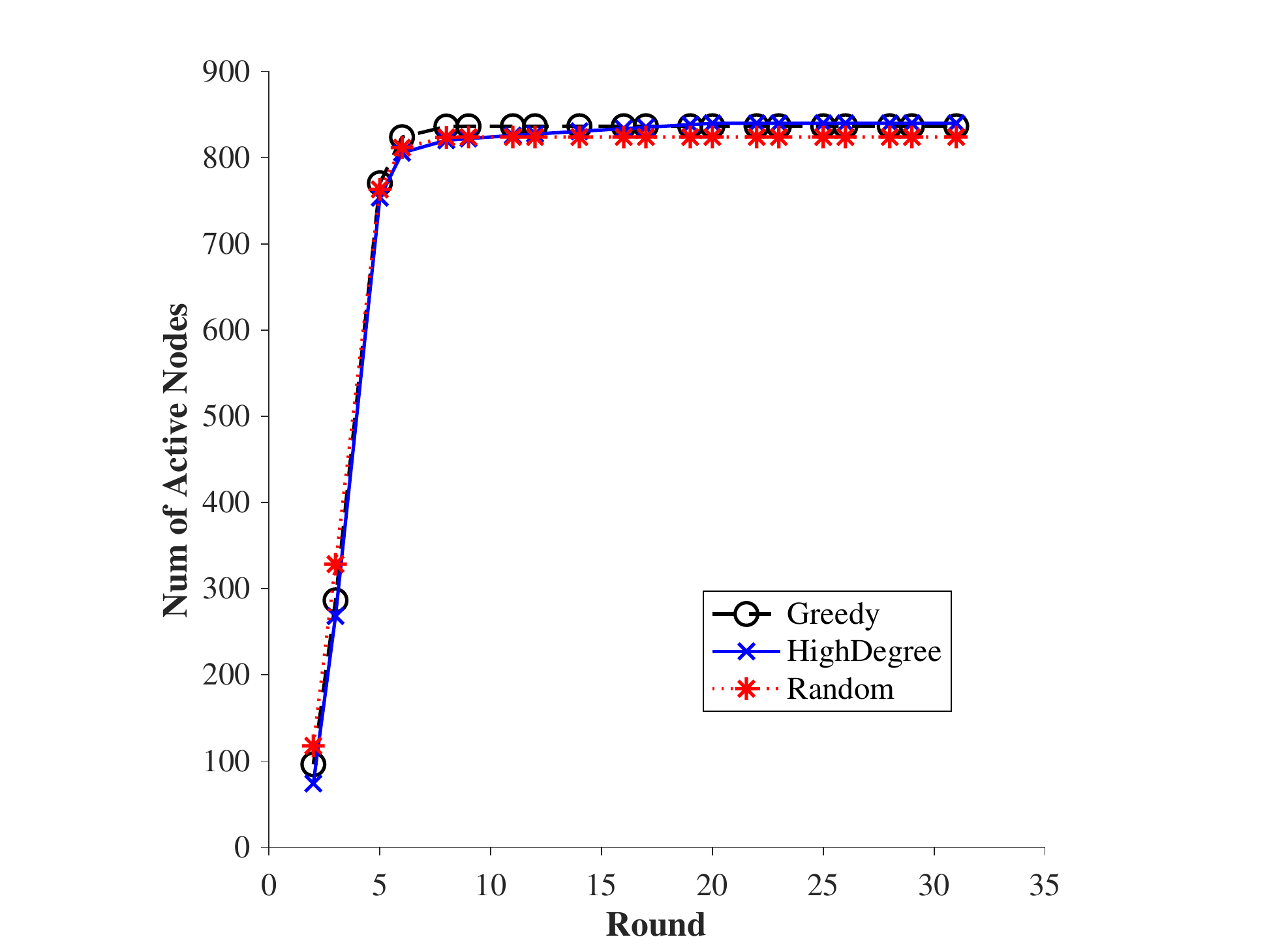}}
	\subfloat[{[Power, IC, $k=20$, $d=1$]}]{\label{fig:poweric20_3_3}\includegraphics[trim = 0.5in 0in 0.5in 0in, clip, width=0.24\textwidth]{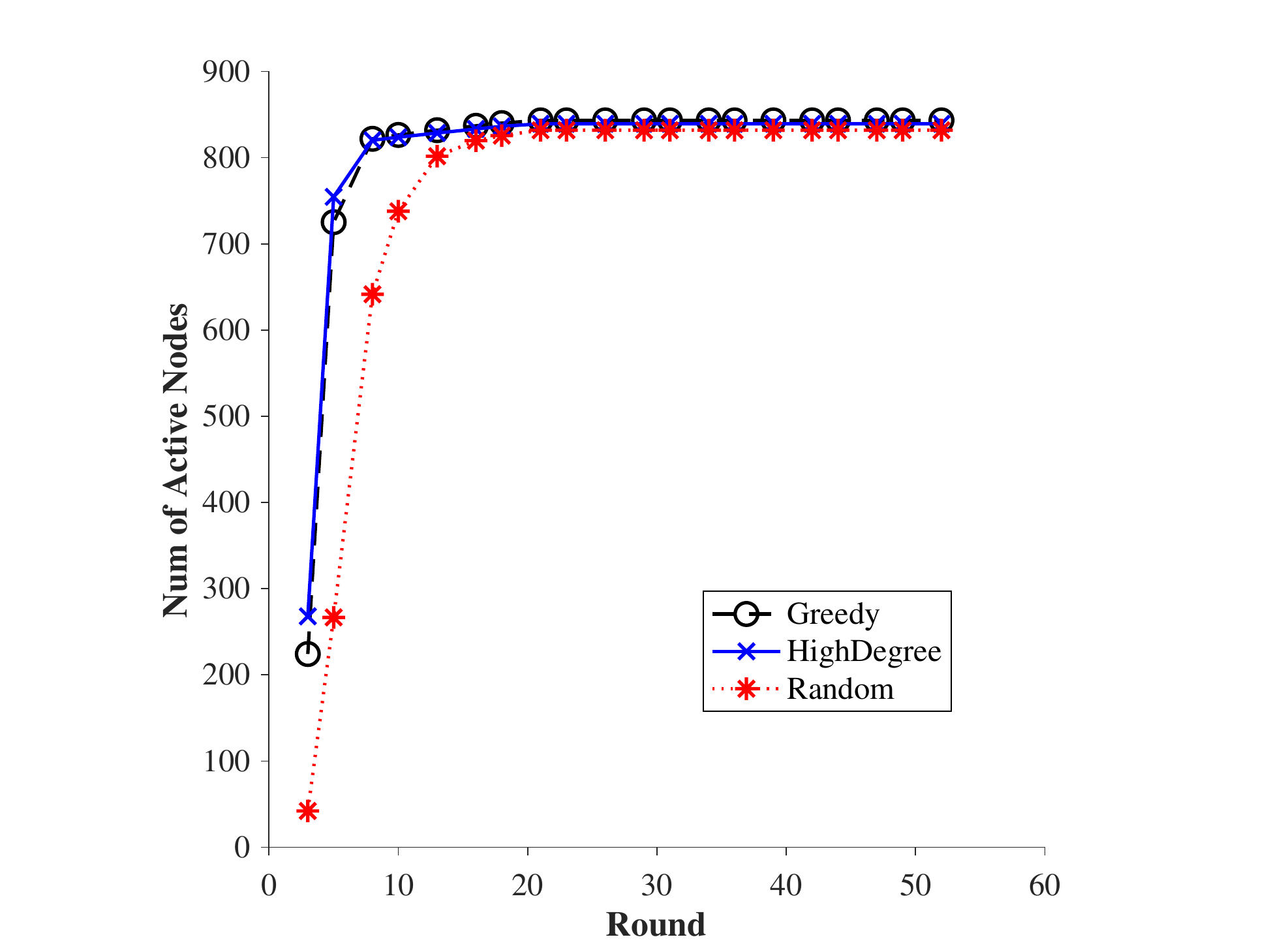}}
	\subfloat[{[Power, IC, $k=20$, $d=8$]}]{\label{fig: poweric20_3_9}\includegraphics[trim = 0.5in 0in 0.5in 0in, clip, width=0.24\textwidth]{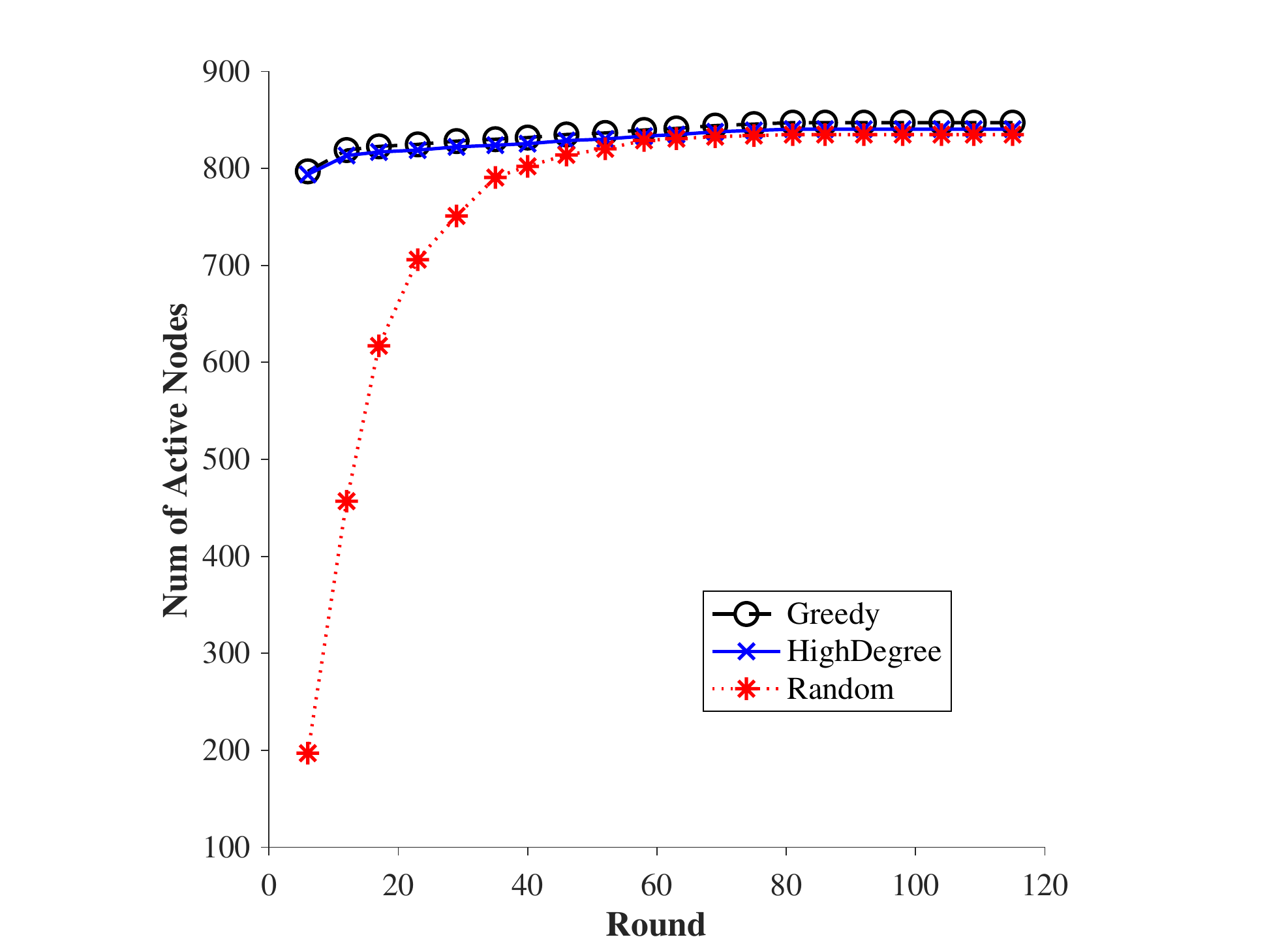}}
	\subfloat[{[Power, IC, $k=20$, $d=\infty$]}]{\label{fig: poweric20_3_15}\includegraphics[trim = 0.5in 0in 0.5in 0in, clip, width=0.24\textwidth]{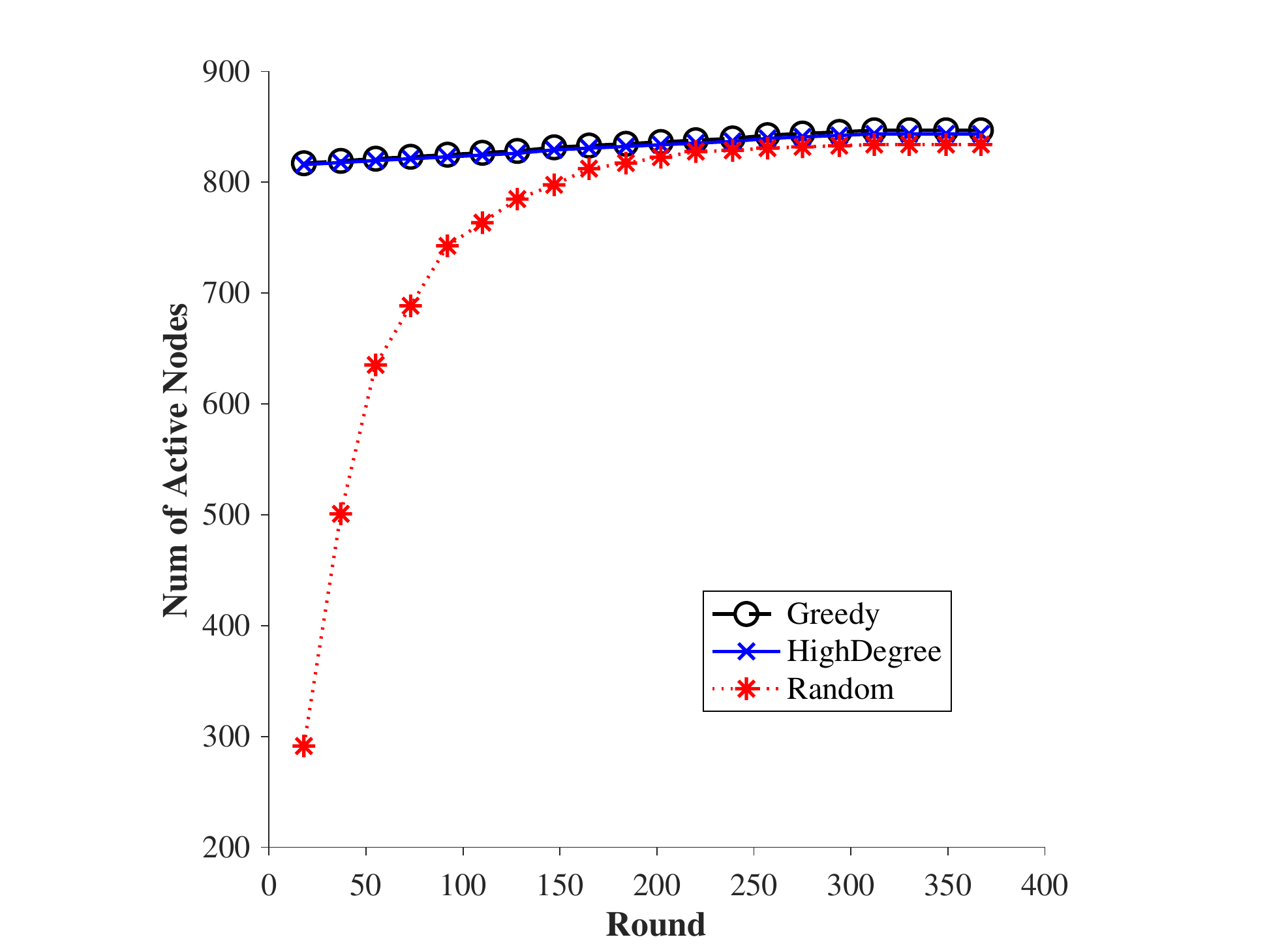}}
	
	\subfloat[{[Power, IC, $k=50$, $d=0$]}]{\label{fig: poweric50_3_0}\includegraphics[trim = 0.5in 0in 0.5in 0in, clip, width=0.24\textwidth]{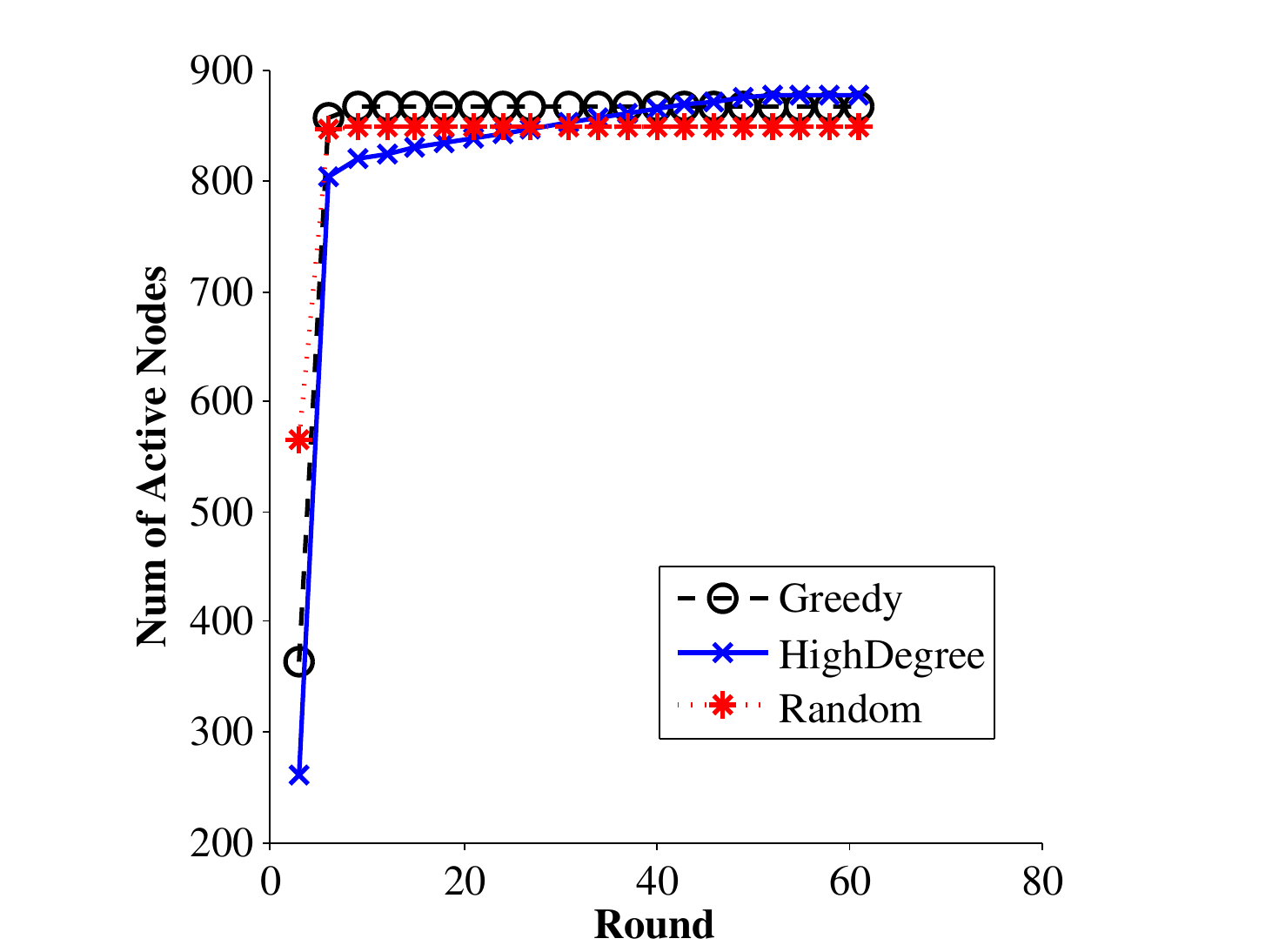}}
	\subfloat[{[Power, IC, $k=50$, $d=1$]}]{\label{fig:poweric50_3_3}\includegraphics[trim = 0.5in 0in 0.5in 0in, clip, width=0.24\textwidth]{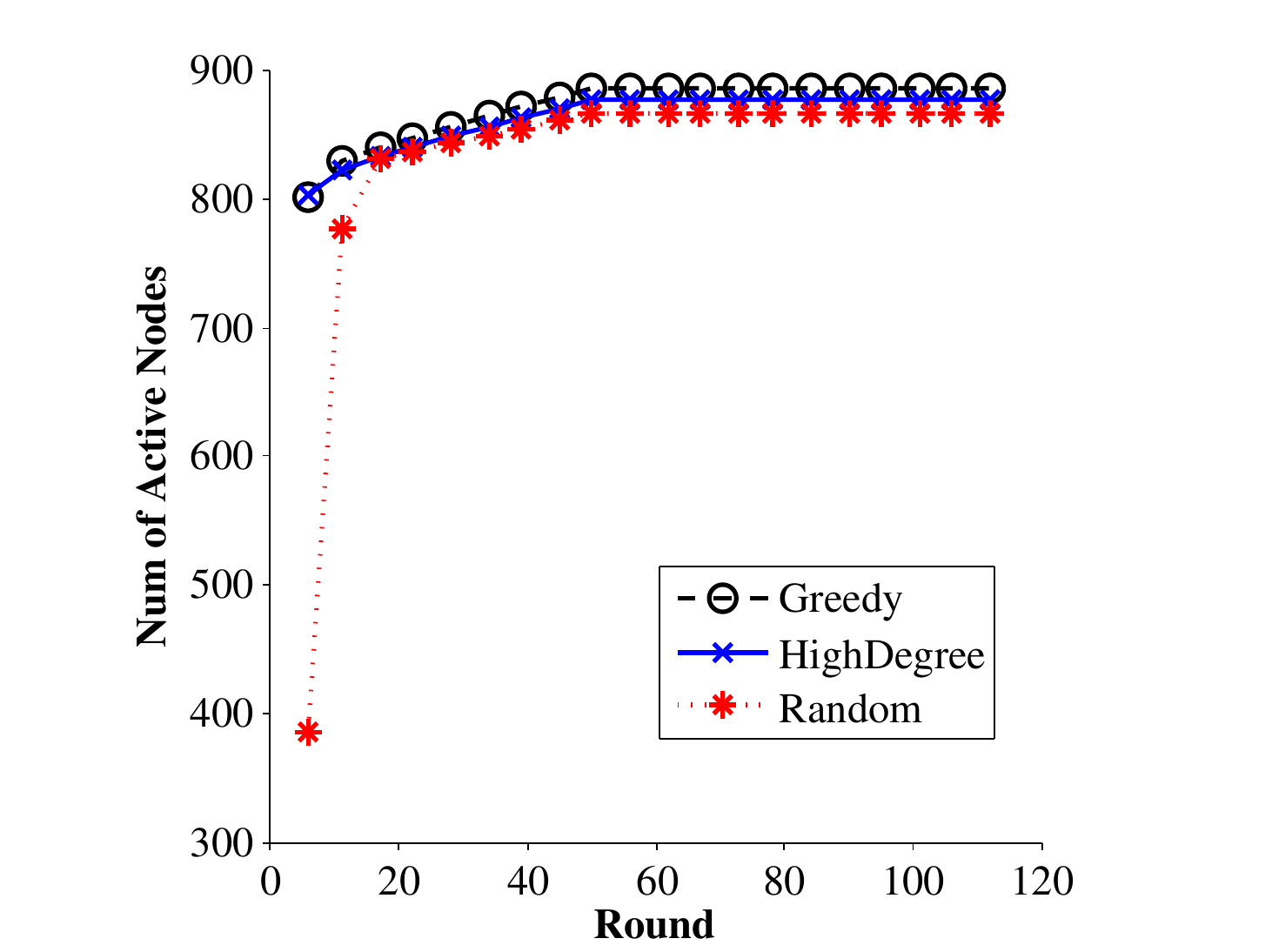}}
	\subfloat[{[Power, IC, $k=50$, $d=8$]}]{\label{fig: poweric50_3_9}\includegraphics[trim = 0.5in 0in 0.5in 0in, clip, width=0.24\textwidth]{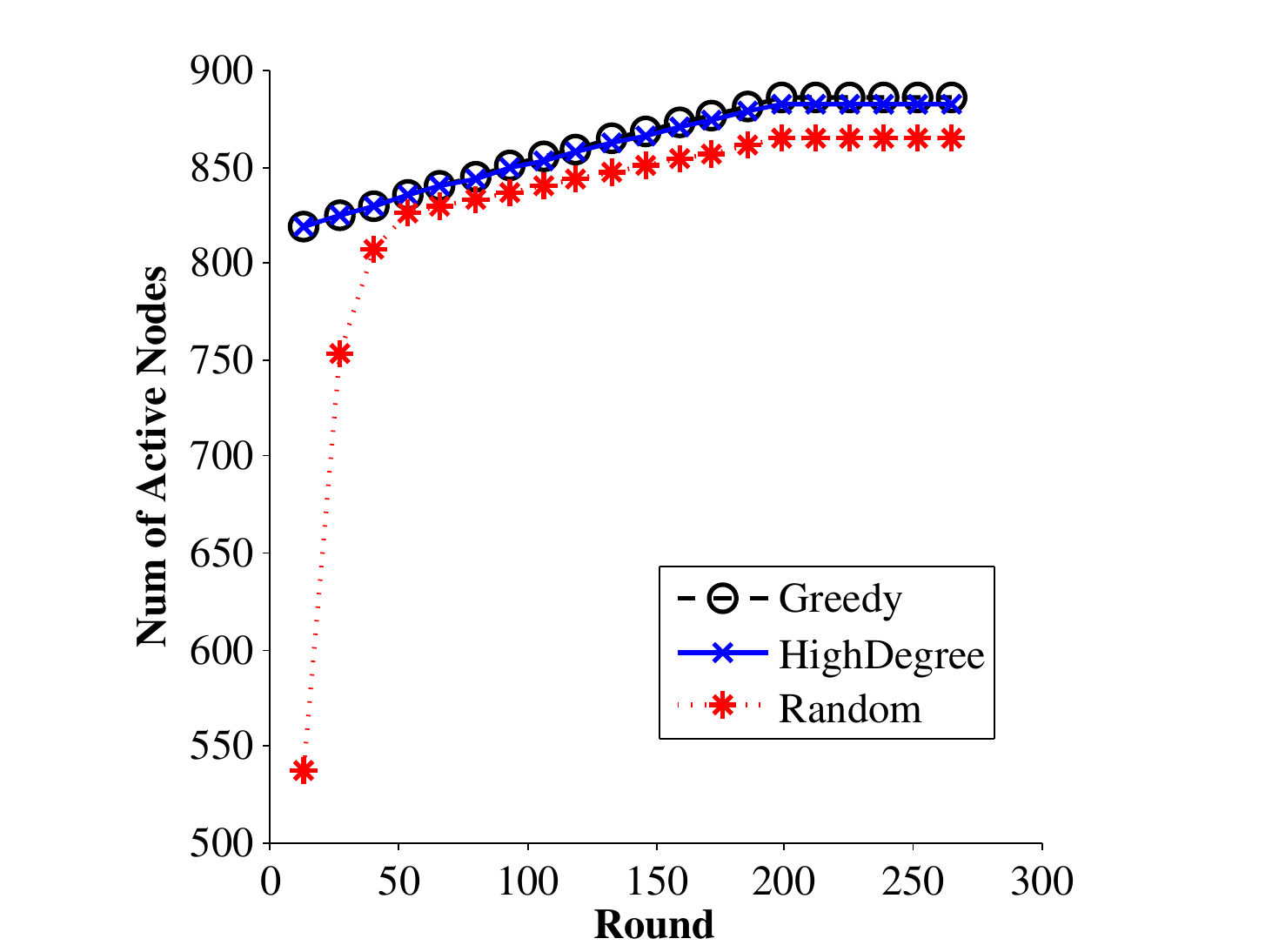}}
	\subfloat[{[Power, IC, $k=50$, $d=\infty$]}]{\label{fig: poweric50_3_15}\includegraphics[trim = 0.5in 0in 0.5in 0in, clip, width=0.24\textwidth]{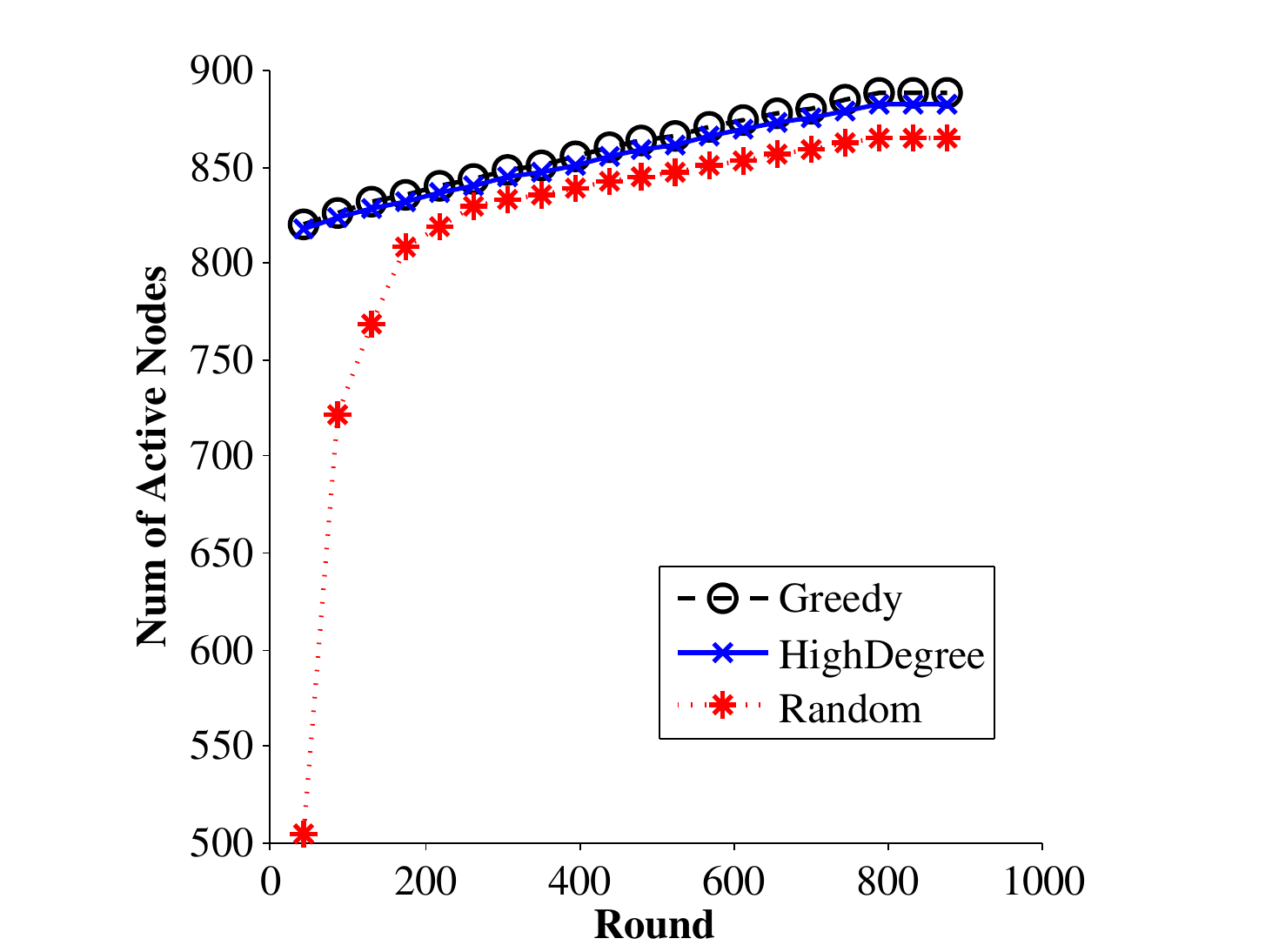}}
	
	\subfloat[{[Wiki, WC, $k=5$, $d=0$]}]{\label{fig: wikiwc5_3_0}\includegraphics[trim = 0.5in 0in 0.5in 0in, clip, width=0.24\textwidth]{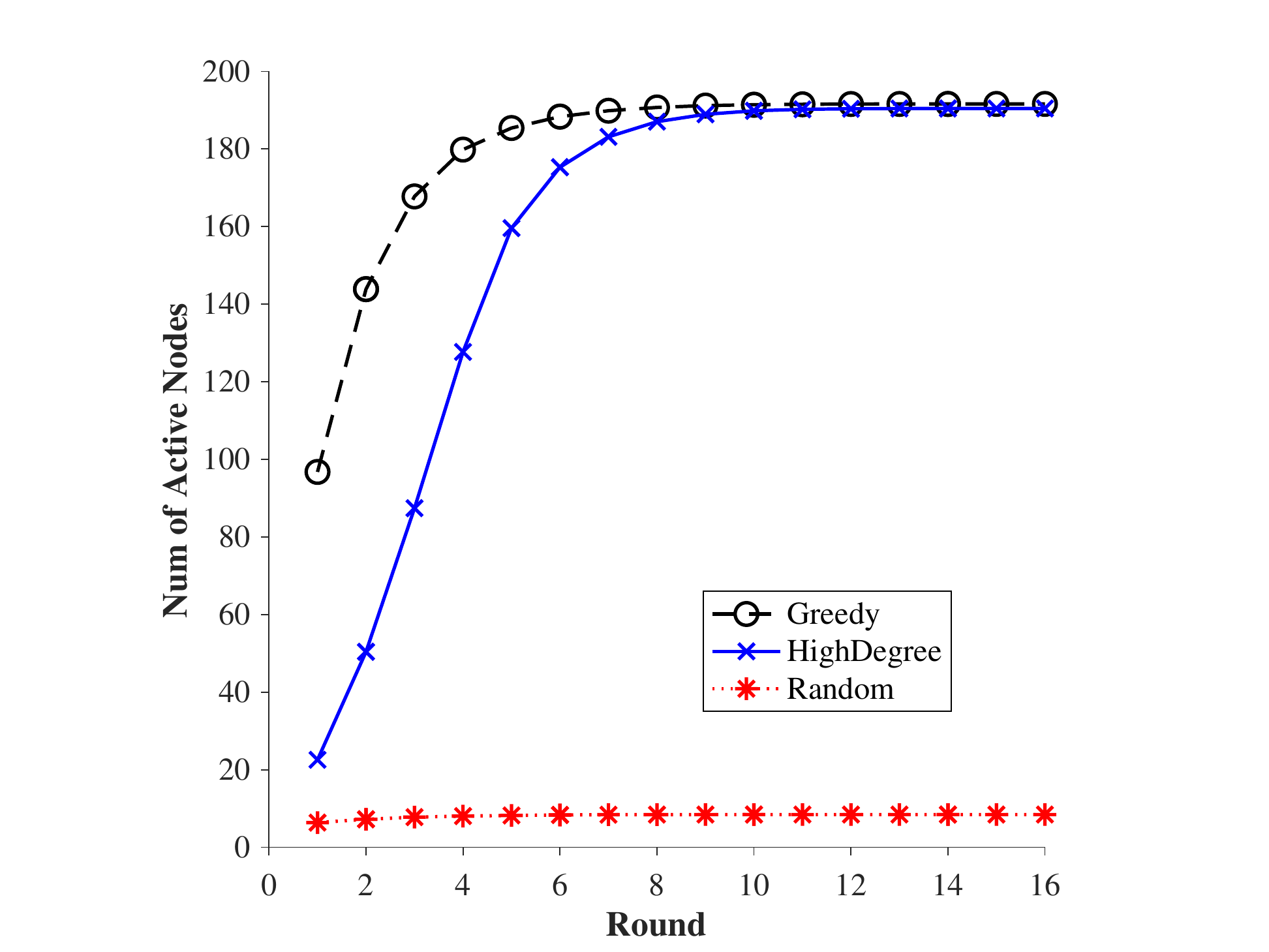}}
	\subfloat[{[Wiki, WC, $k=5$, $d=1$]}]{\label{fig: wikiwc5_3}\includegraphics[trim = 0.5in 0in 0.5in 0in, clip, width=0.24\textwidth]{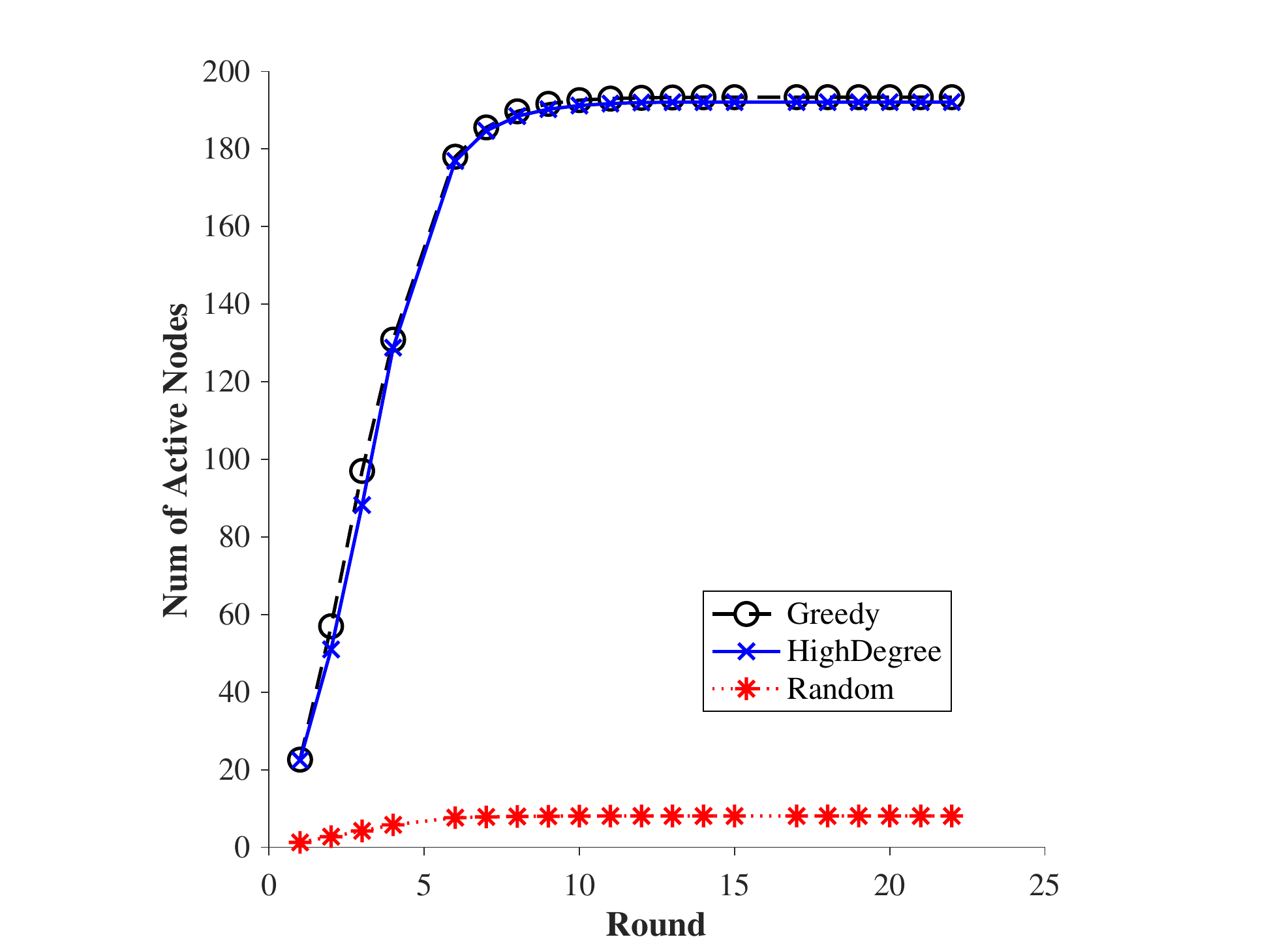}}
	\subfloat[{[Wiki, WC, $k=5$, $d=8$]}]{\label{fig: wikiwc5_3_9}\includegraphics[trim = 0.5in 0in 0.5in 0in, clip, width=0.24\textwidth]{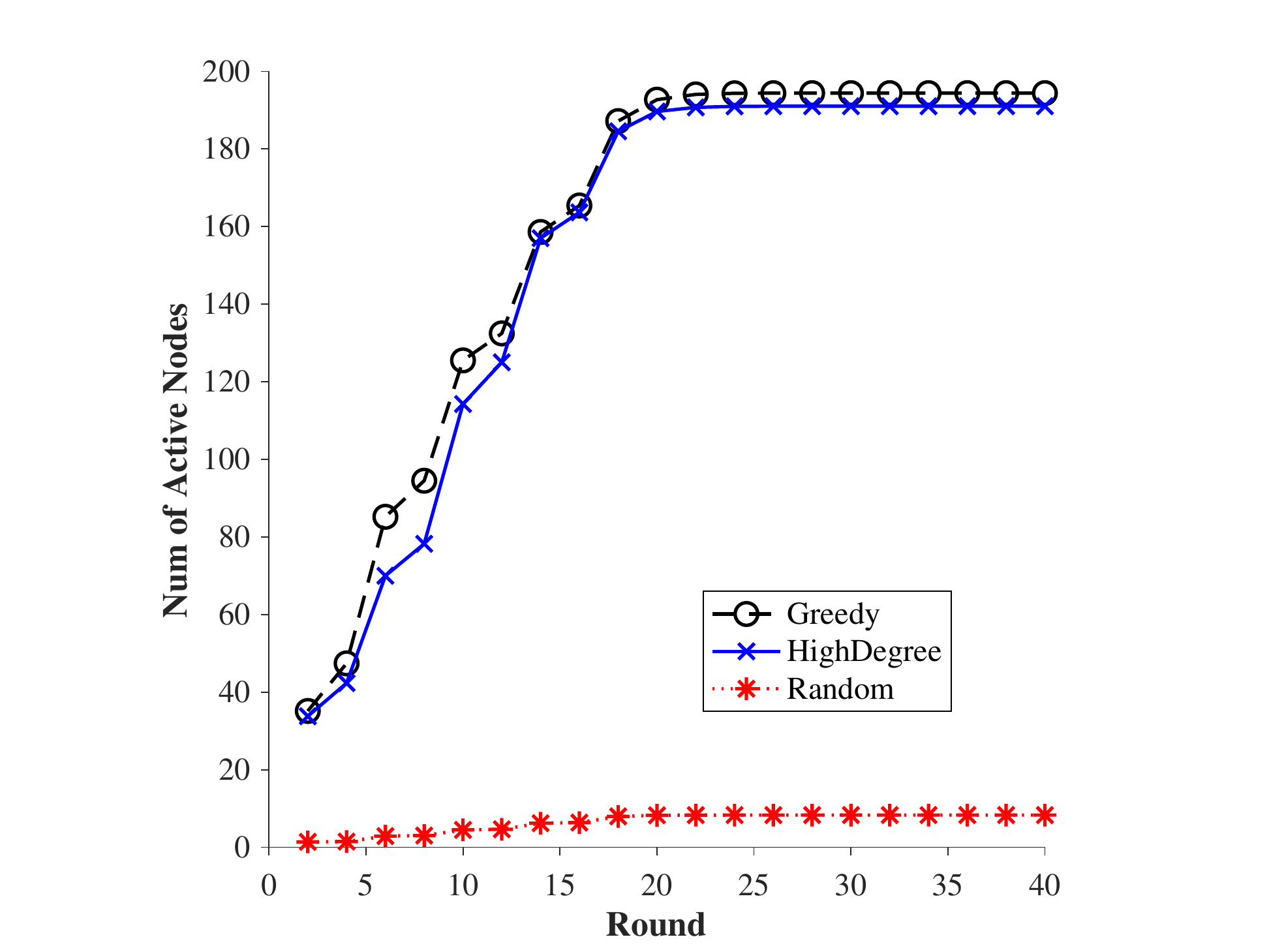}}
	\subfloat[{[Wiki, WC, $k=5$, $d=\infty$]}]{\label{fig: wikiwc5_3_15}\includegraphics[trim = 0.5in 0in 0.5in 0in, clip, width=0.24\textwidth]{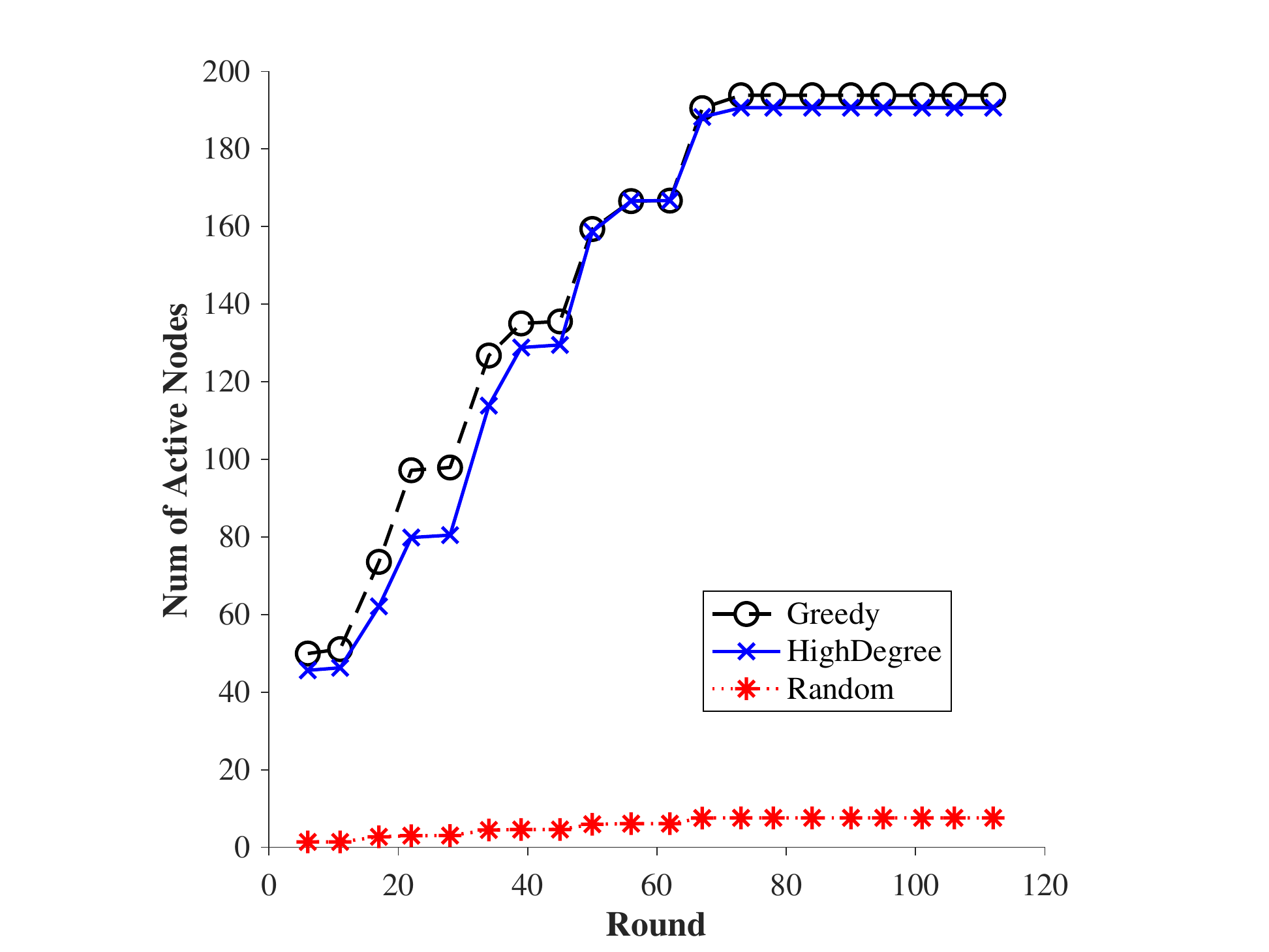}}
	
	\caption{Additional Results of Experiment \RNum{1}. Part 1}
	\vspace{-3mm}
	\label{fig: exp1_more_1}
\end{figure*}

\begin{figure*}[!pt]
\centering
\subfloat[{[Wiki, WC, $k=50$, $d=0$]}]{\label{fig: wikiwc50_3_0}\includegraphics[trim = 0.5in 0in 0.5in 0in, clip, width=0.24\textwidth]{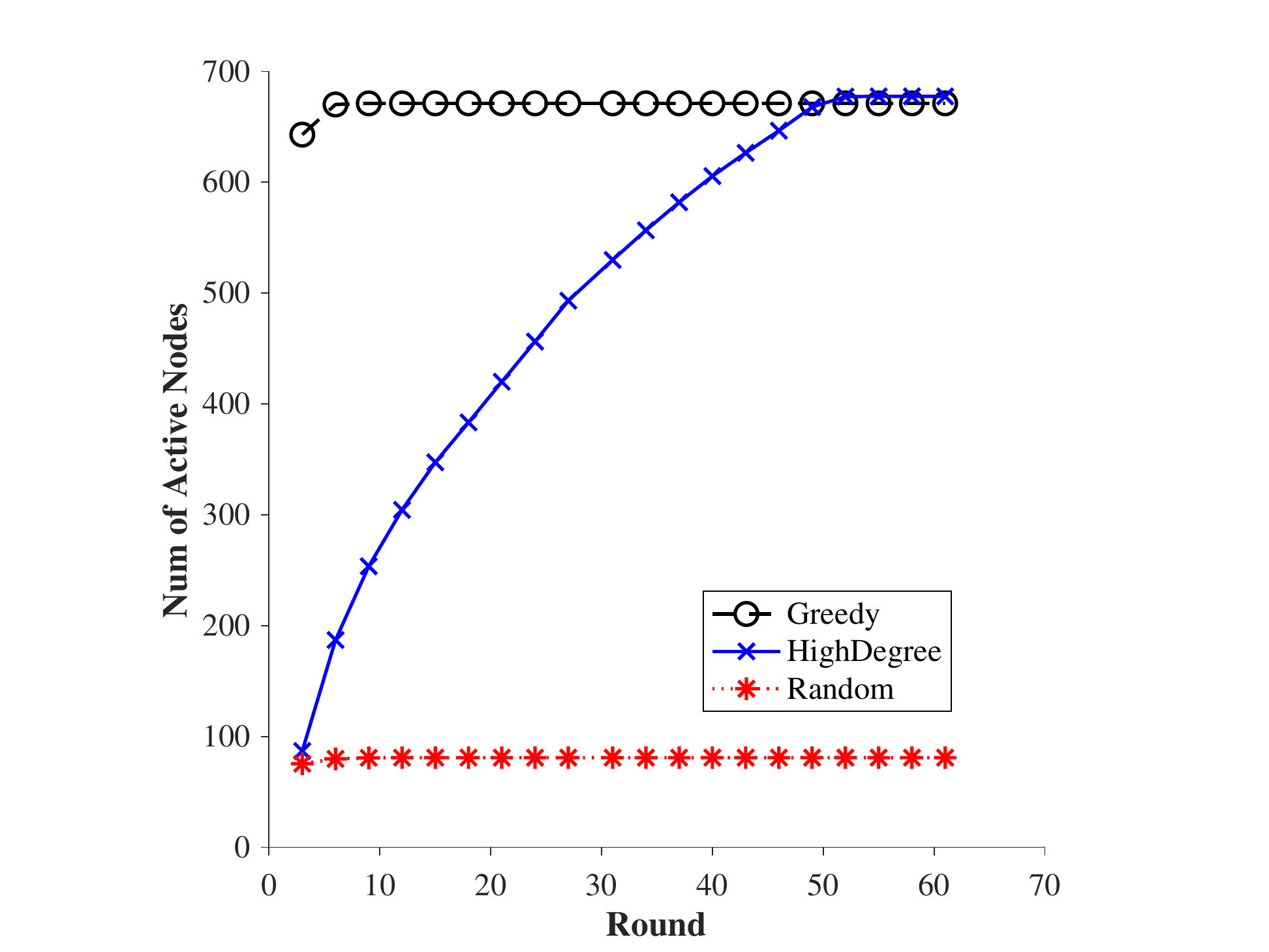}}
\subfloat[{[Wiki, WC, $k=50$, $d=1$]}]{\label{fig: wikiwc50_3_3}\includegraphics[trim = 0.5in 0in 0.5in 0in, clip, width=0.24\textwidth]{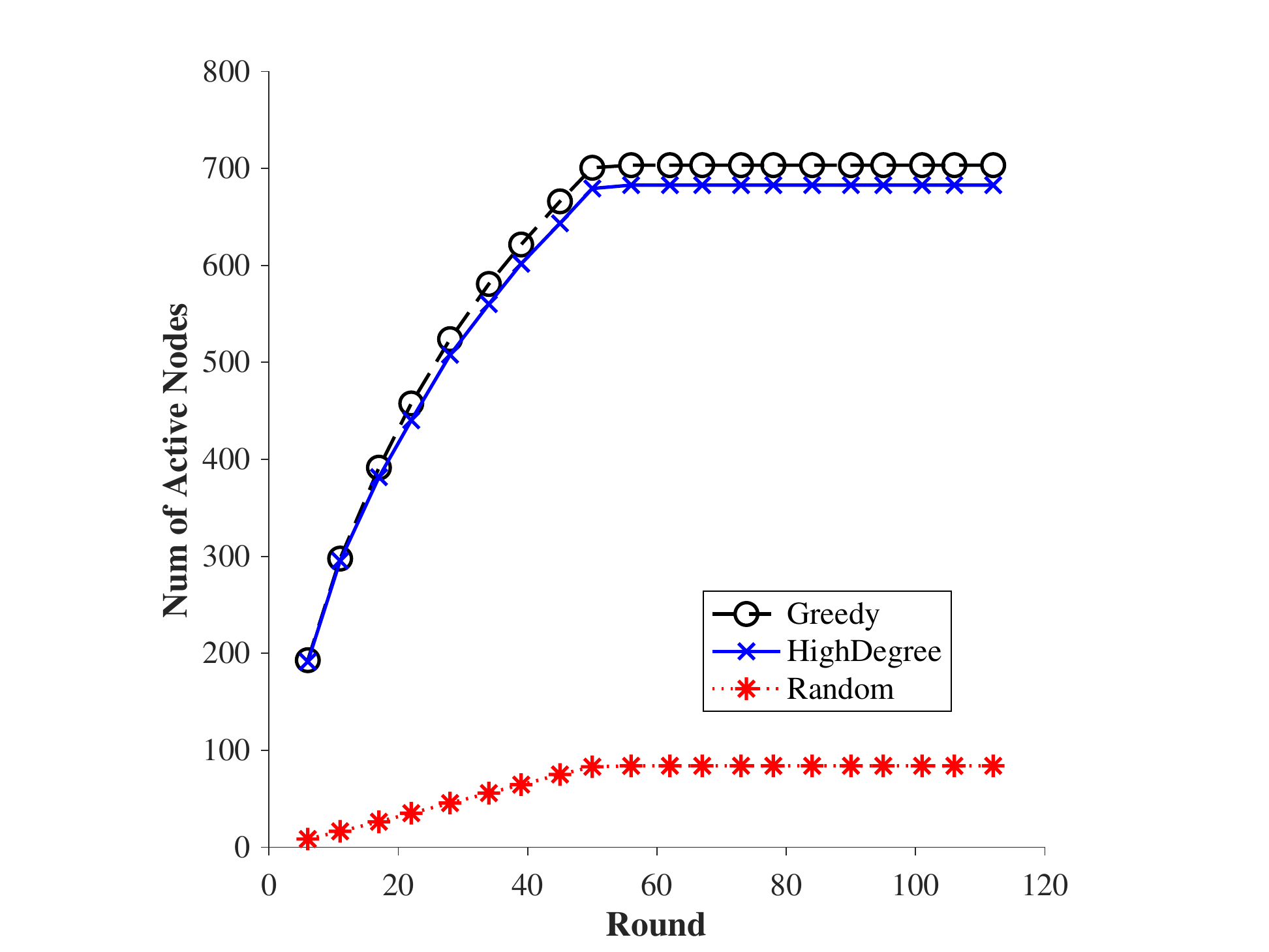}}
\subfloat[{[Wiki, WC, $k=50$, $d=8$]}]{\label{fig: wikiwc50_3_9}\includegraphics[trim = 0.5in 0in 0.5in 0in, clip, width=0.24\textwidth]{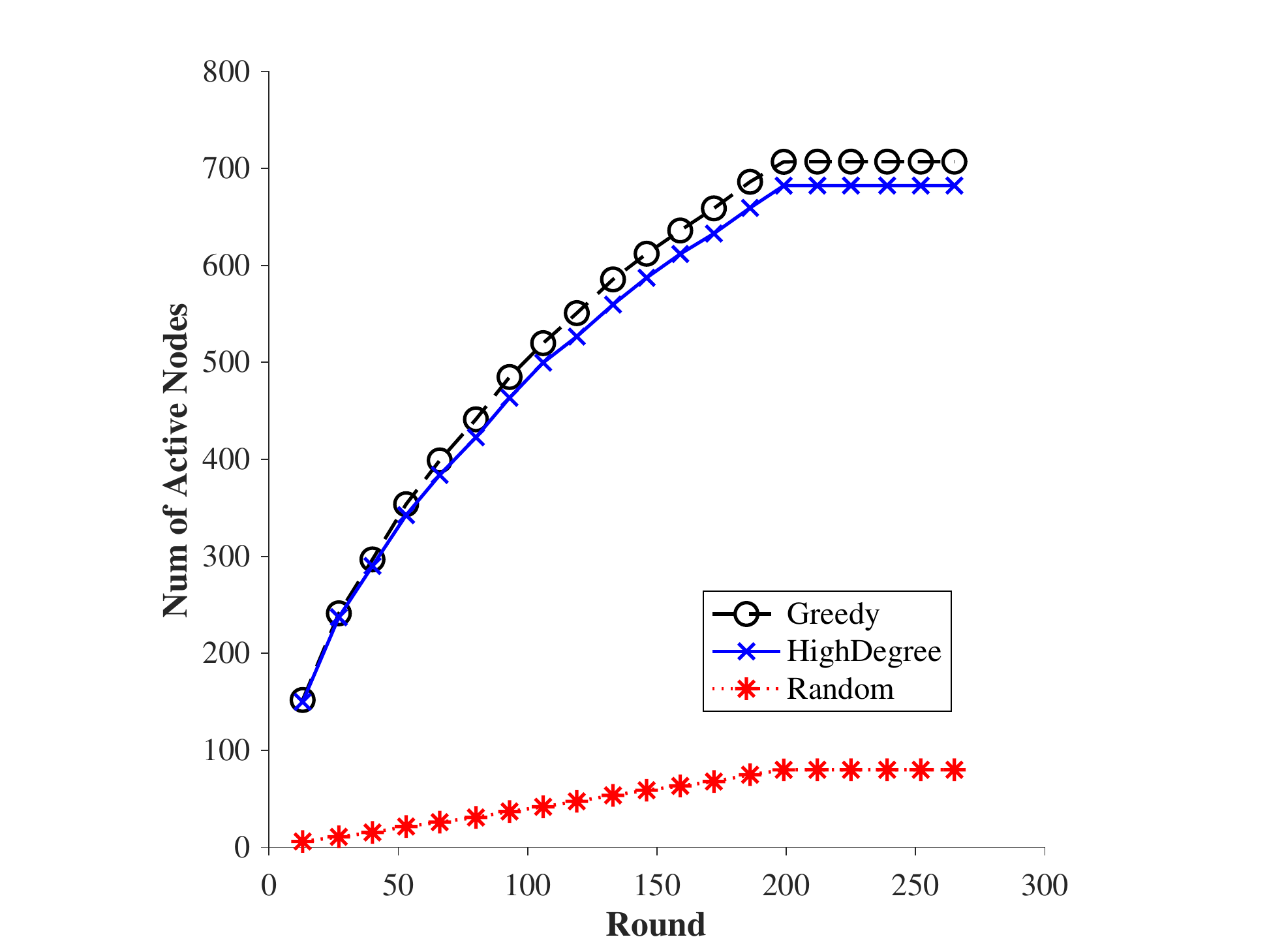}}
\subfloat[{[Wiki, WC, $k=50$, $d=\infty$]}]{\label{fig: wikiwc50_3_15}\includegraphics[trim = 0.5in 0in 0.5in 0in, clip, width=0.24\textwidth]{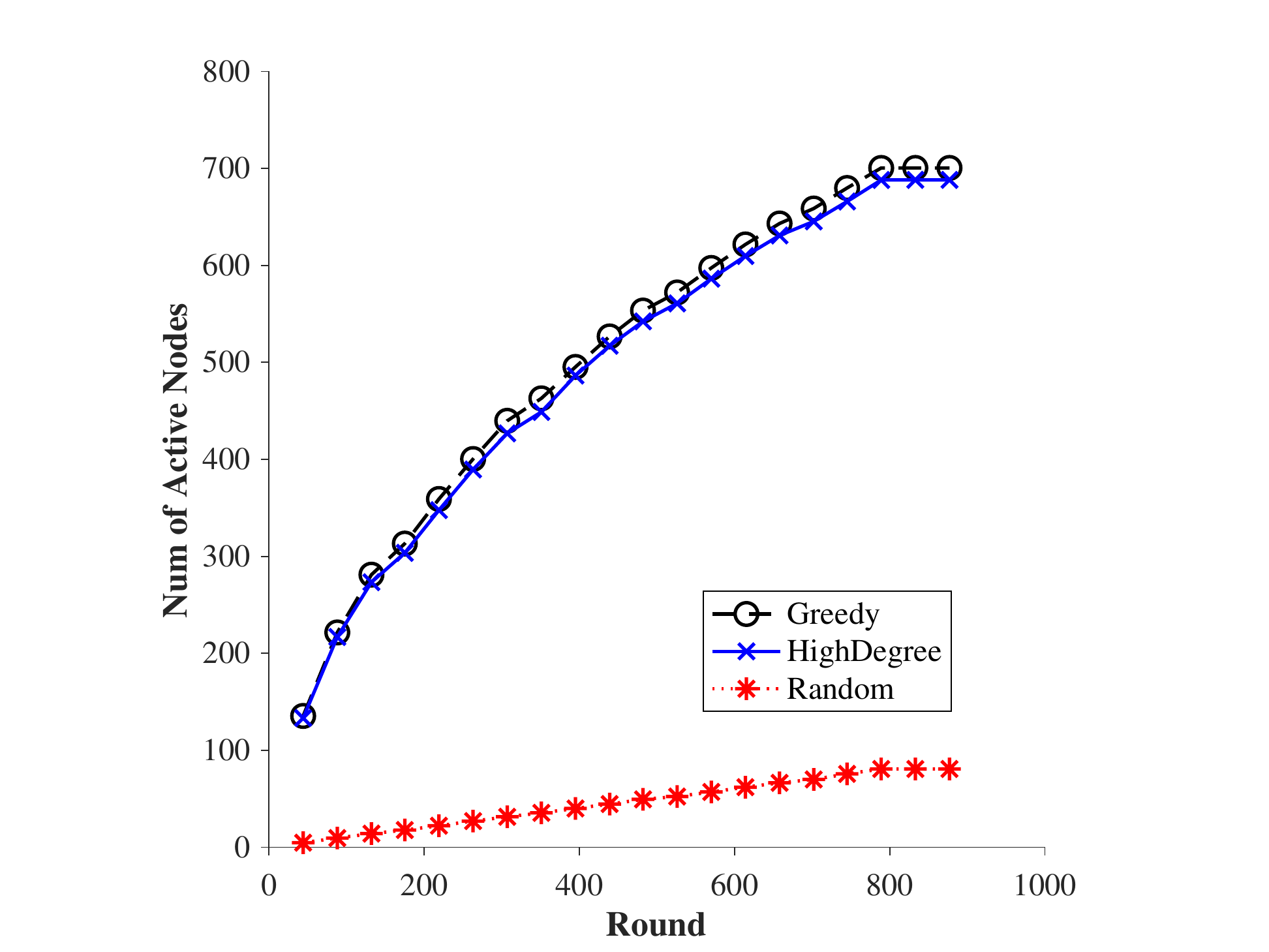}}

\subfloat[{[Wiki, IC, $k=5$, $d=0$]}]{\label{fig: wikiic5_3_0}\includegraphics[trim = 0.5in 0in 0.5in 0in, clip, width=0.24\textwidth]{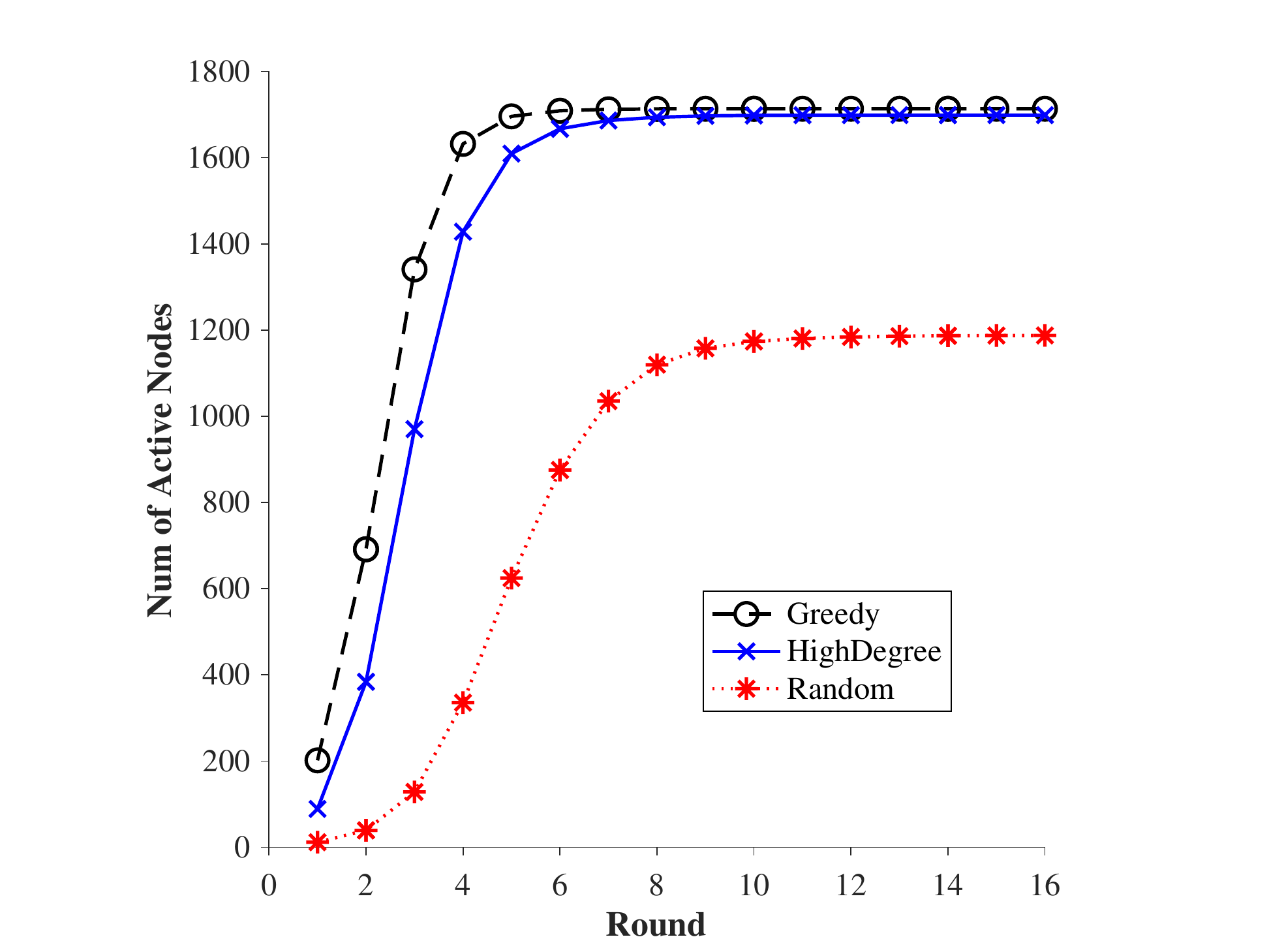}}
\subfloat[{[Wiki, IC, $k=5$, $d=1$]}]{\label{fig: wikiic5_3_3}\includegraphics[trim = 0.5in 0in 0.5in 0in, clip, width=0.24\textwidth]{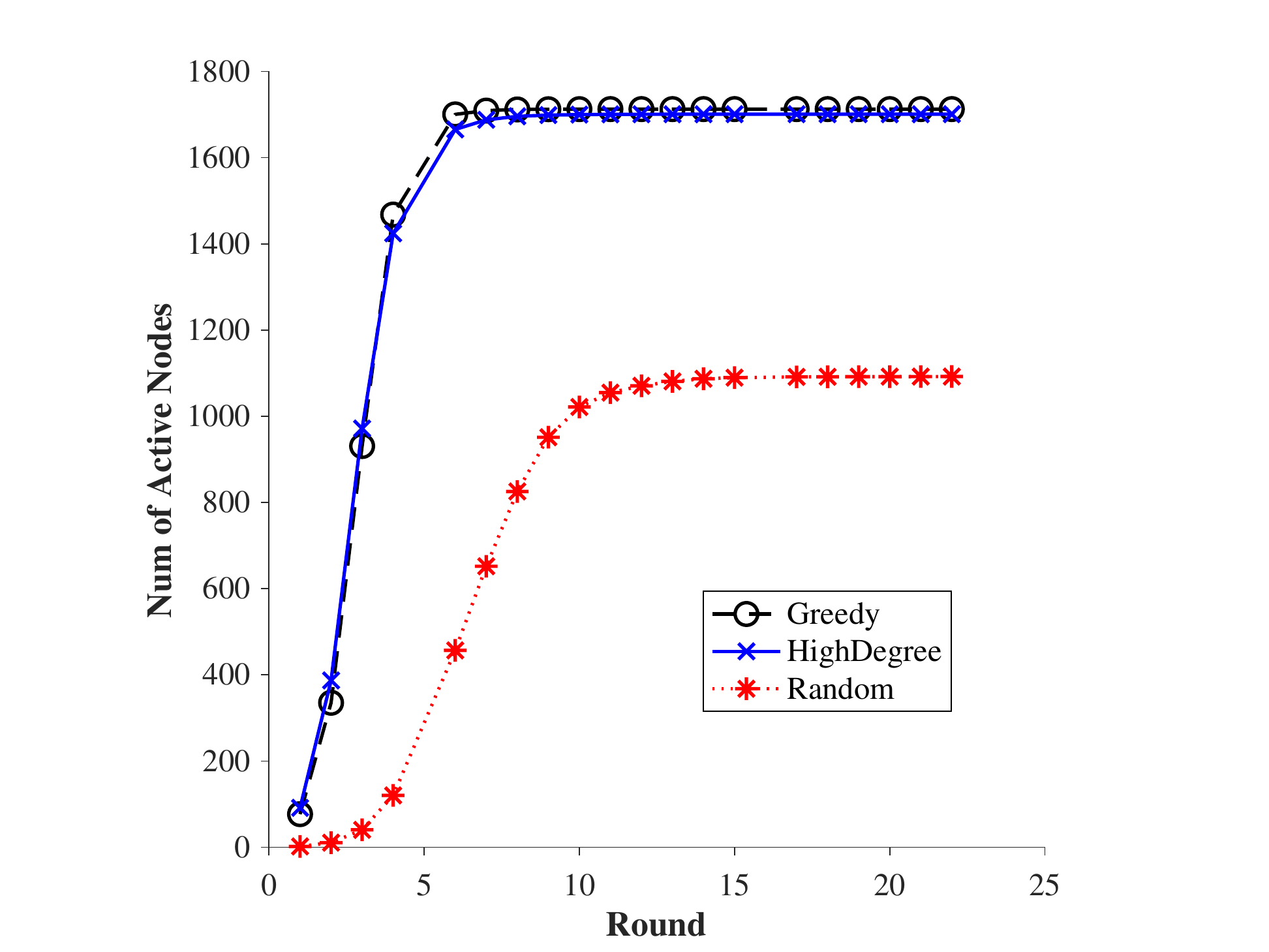}}
\subfloat[{[Wiki, IC, $k=5$, $d=8$]}]{\label{fig: wikiic5_3_9}\includegraphics[trim = 0.5in 0in 0.5in 0in, clip, width=0.24\textwidth]{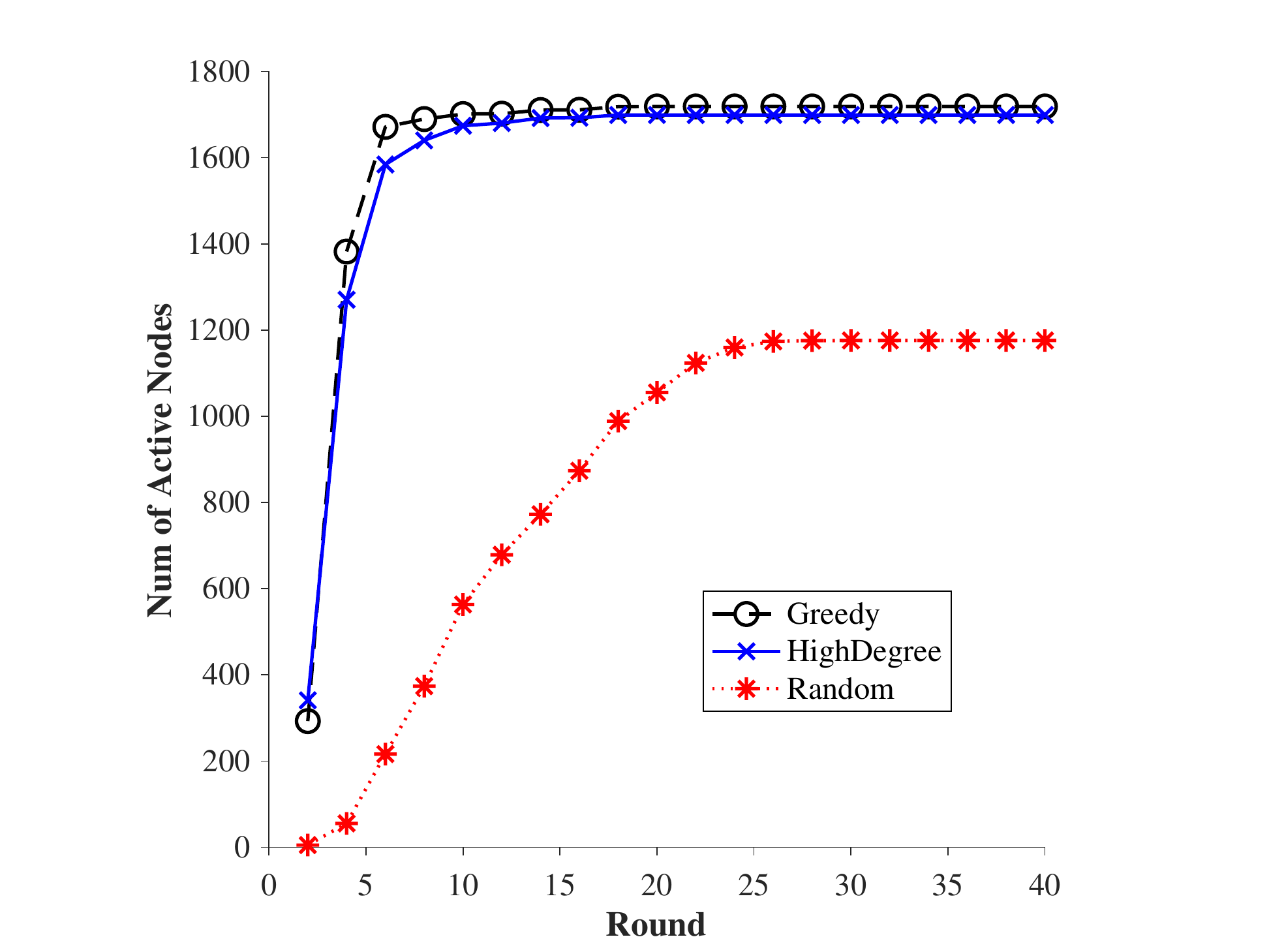}}
\subfloat[{[Wiki, IC, $k=5$, $d=\infty$]}]{\label{fig: wikiic5_3_15}\includegraphics[trim = 0.5in 0in 0.5in 0in, clip, width=0.24\textwidth]{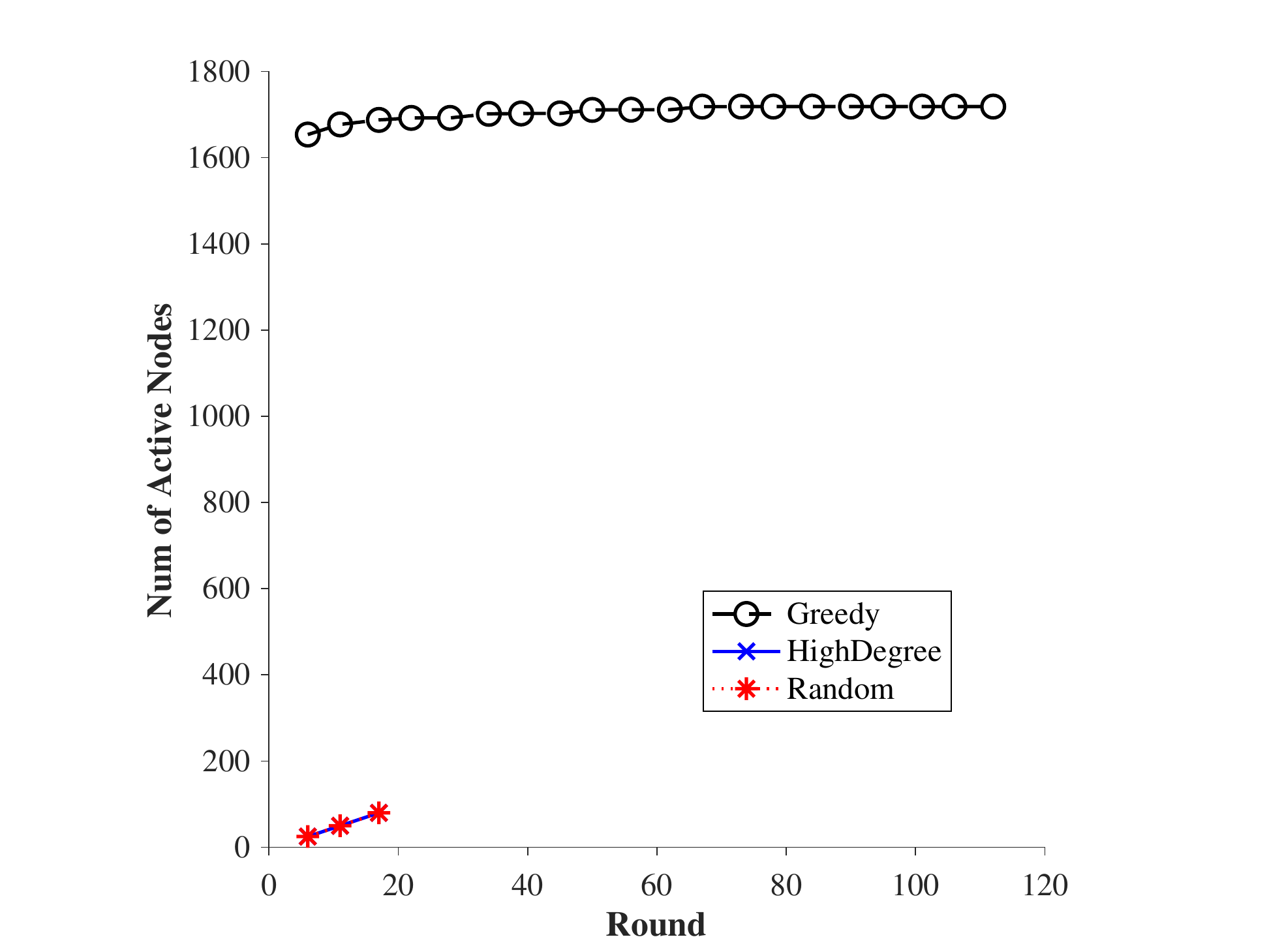}}

\subfloat[{[Wiki, IC, $k=50$, $d=0$]}]{\label{fig: wikiic50_3_0}\includegraphics[trim = 0.5in 0in 0.5in 0in, clip, width=0.24\textwidth]{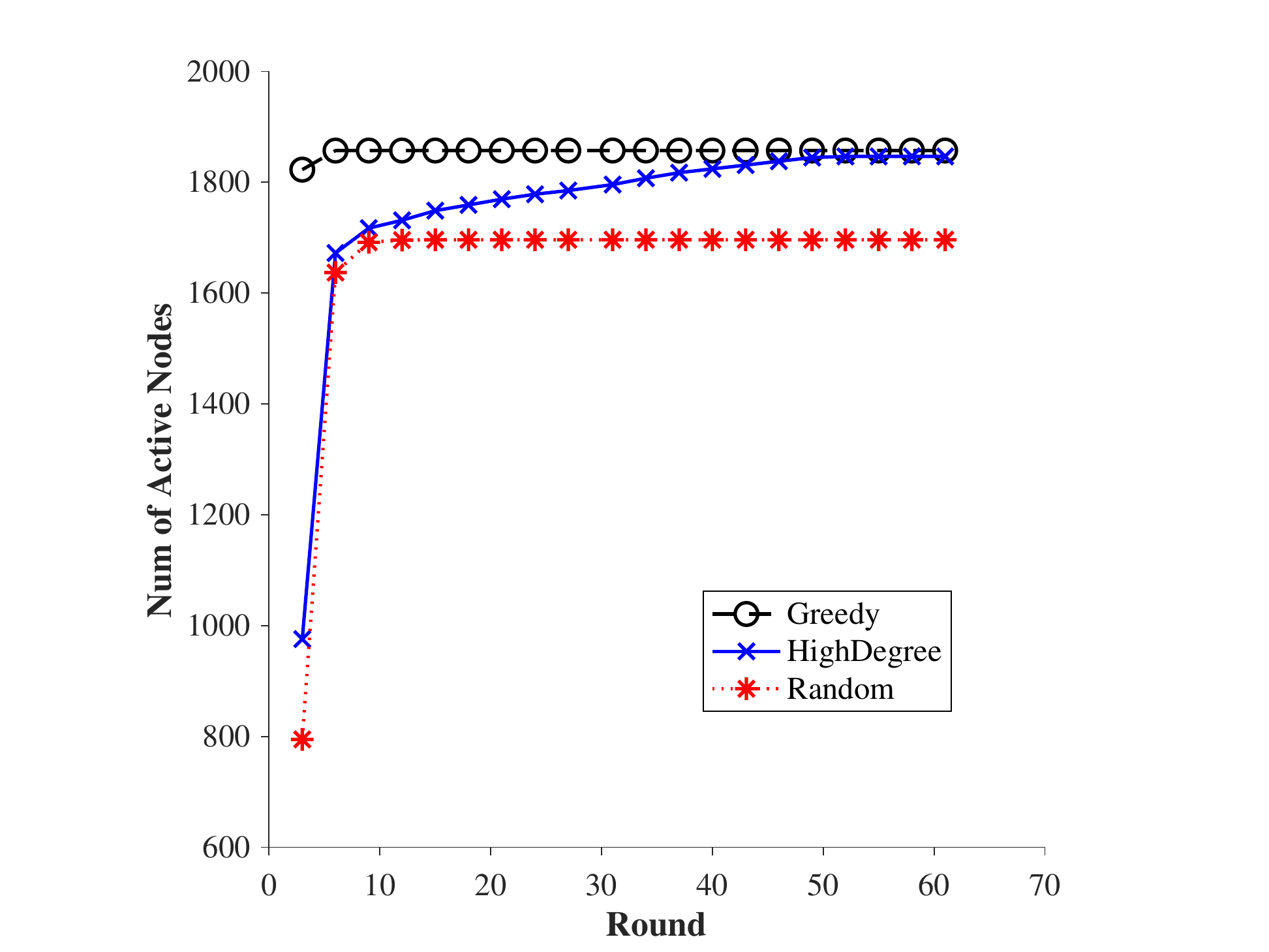}}
\subfloat[{[Wiki, IC, $k=50$, $d=1$]}]{\label{fig: wikiic50_3_3}\includegraphics[trim = 0.5in 0in 0.5in 0in, clip, width=0.24\textwidth]{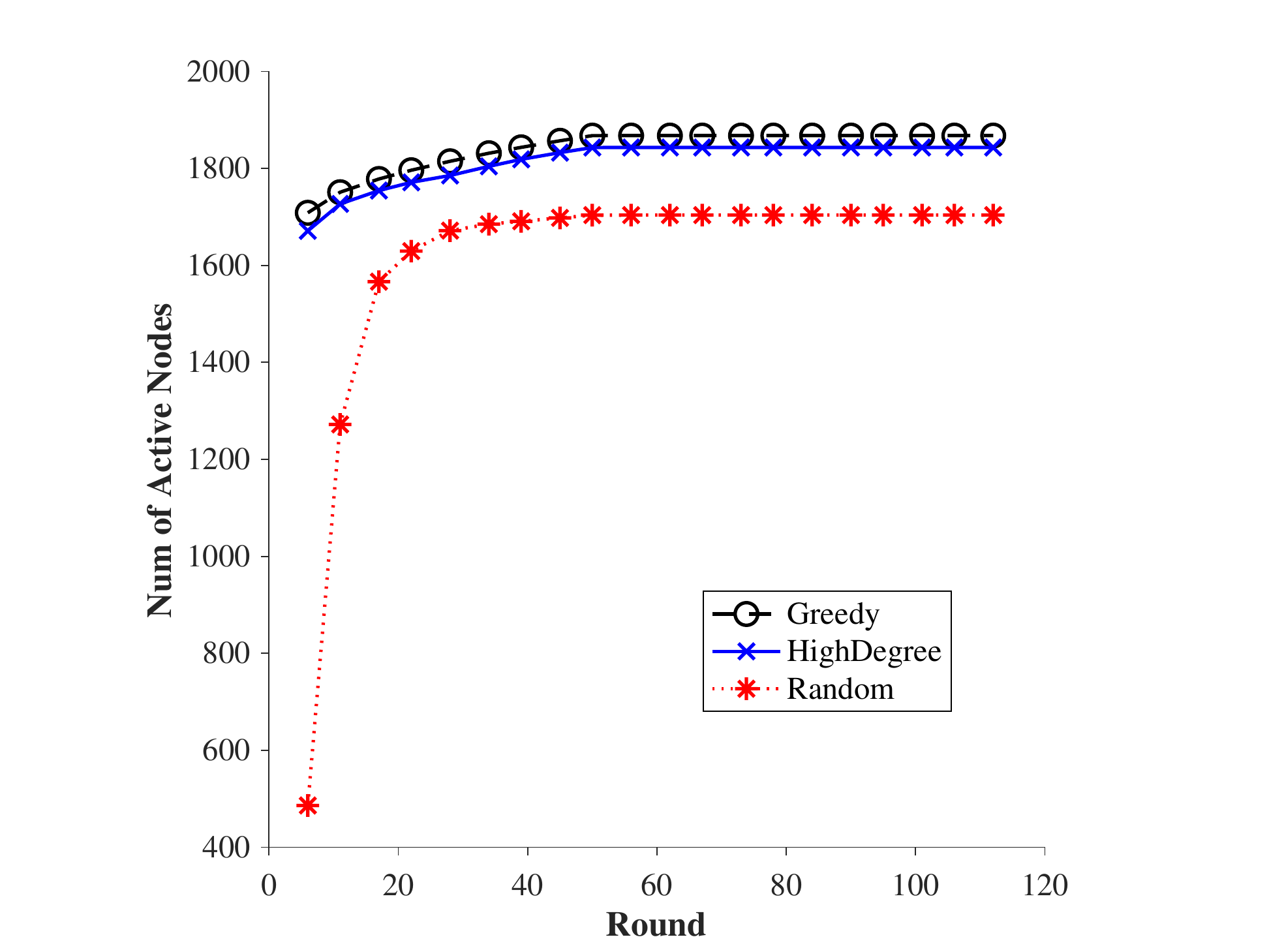}}
\subfloat[{[Wiki, IC, $k=50$, $d=8$]}]{\label{fig: wikiic50_3_9}\includegraphics[trim = 0.5in 0in 0.5in 0in, clip, width=0.24\textwidth]{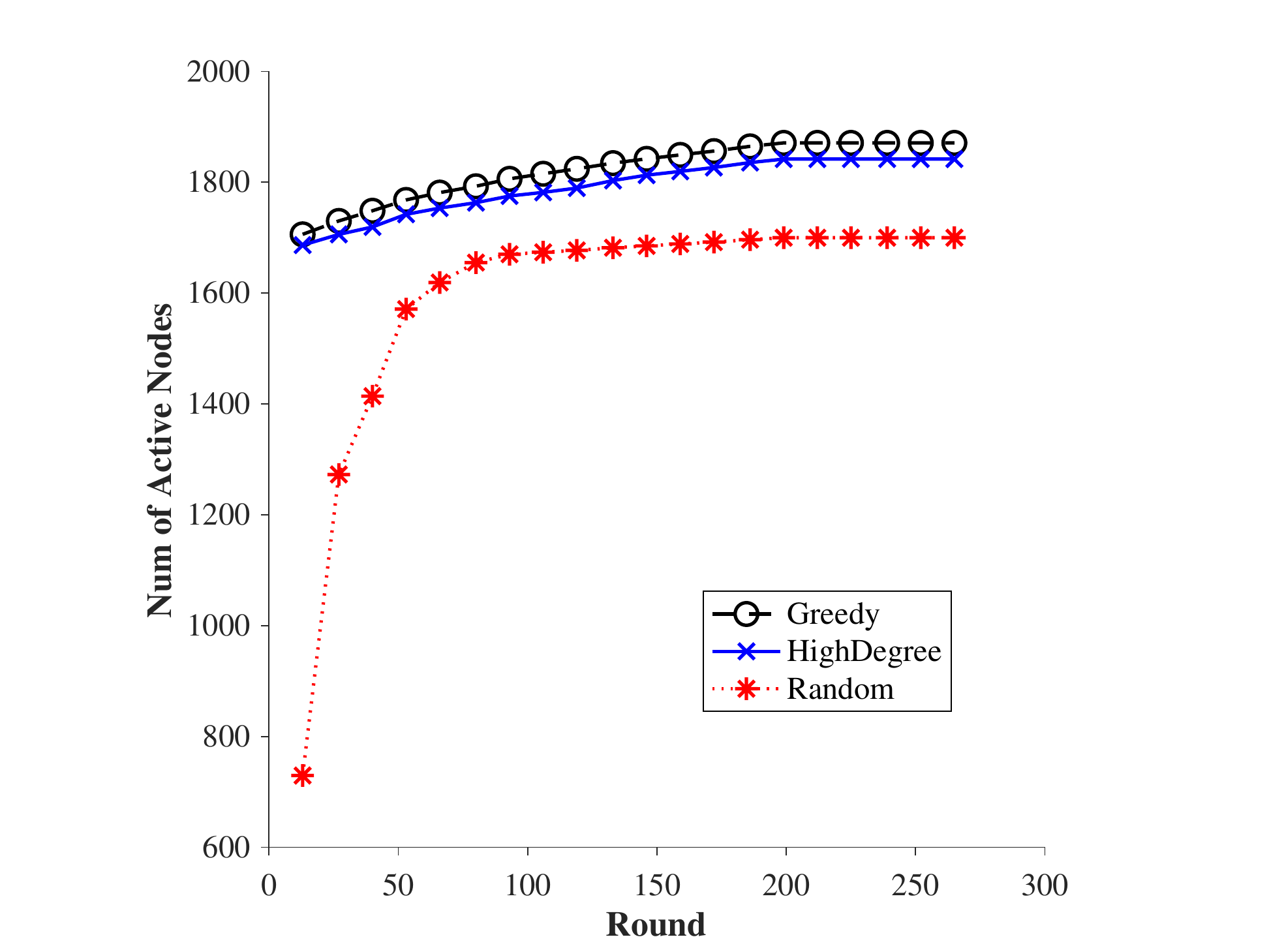}}
\subfloat[{[Wiki, IC, $k=50$, $d=\infty$]}]{\label{fig: wikiic50_3_15}\includegraphics[trim = 0.5in 0in 0.5in 0in, clip, width=0.24\textwidth]{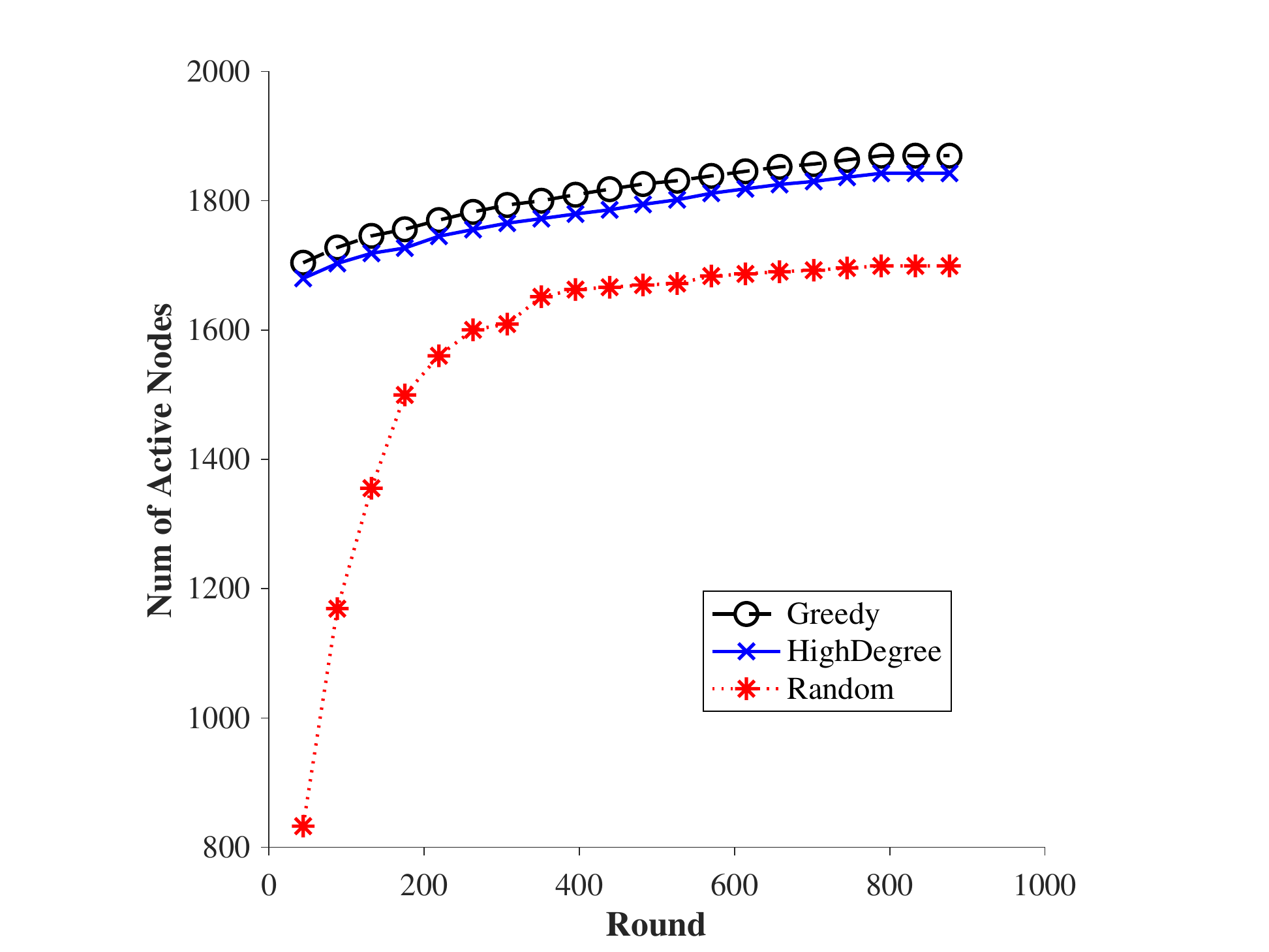}}

\subfloat[{[Higgs, IC, $k=20$, $d=0$]}]{\label{fig: higgs20_3_0}\includegraphics[trim = 0.5in 0in 0.5in 0in, clip, width=0.24\textwidth]{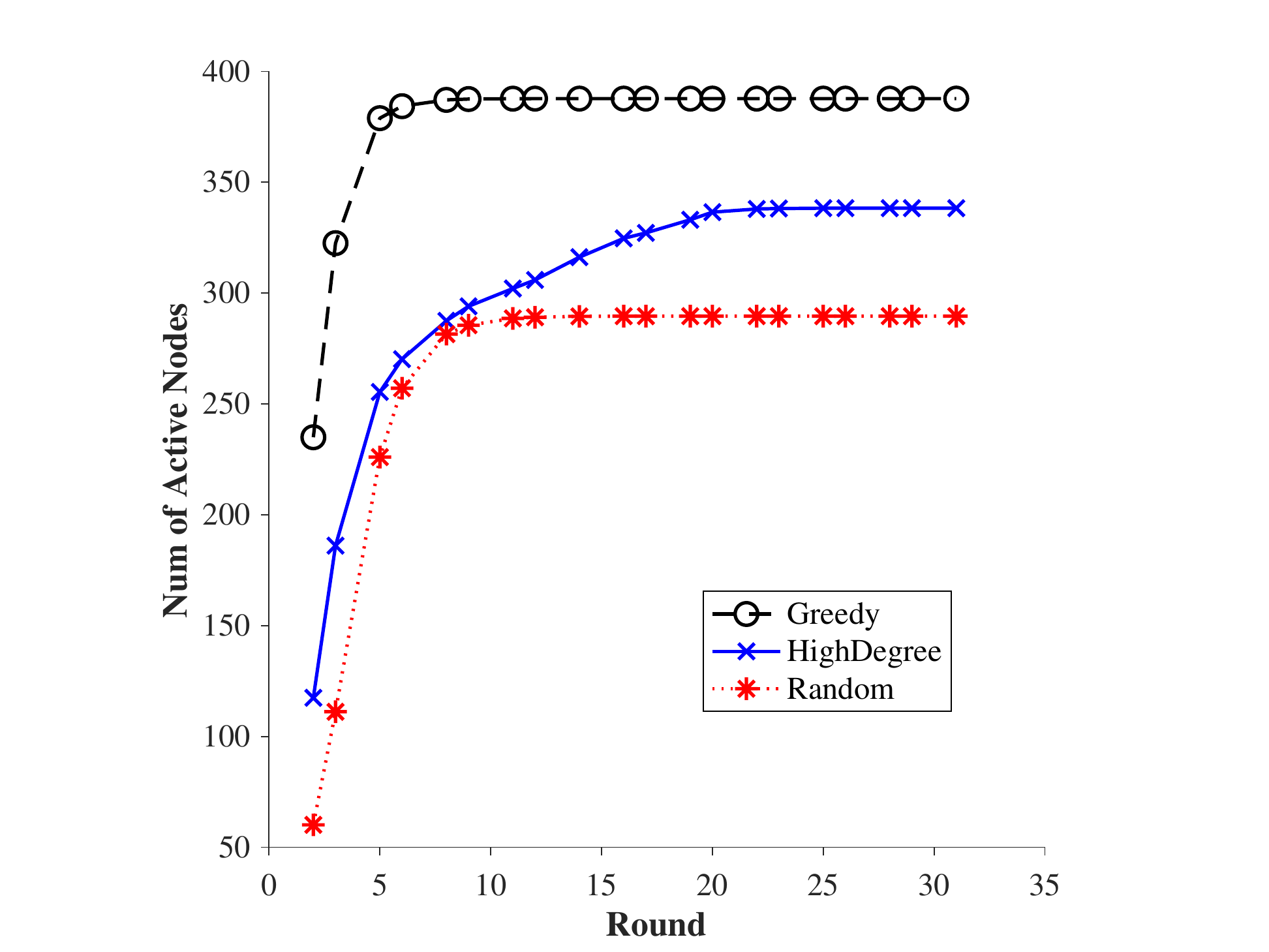}}
\subfloat[{[Higgs, IC, $k=20$, $d=1$]}]{\label{fig: higgs20_3_3}\includegraphics[trim = 0.5in 0in 0.5in 0in, clip,width=0.24\textwidth]{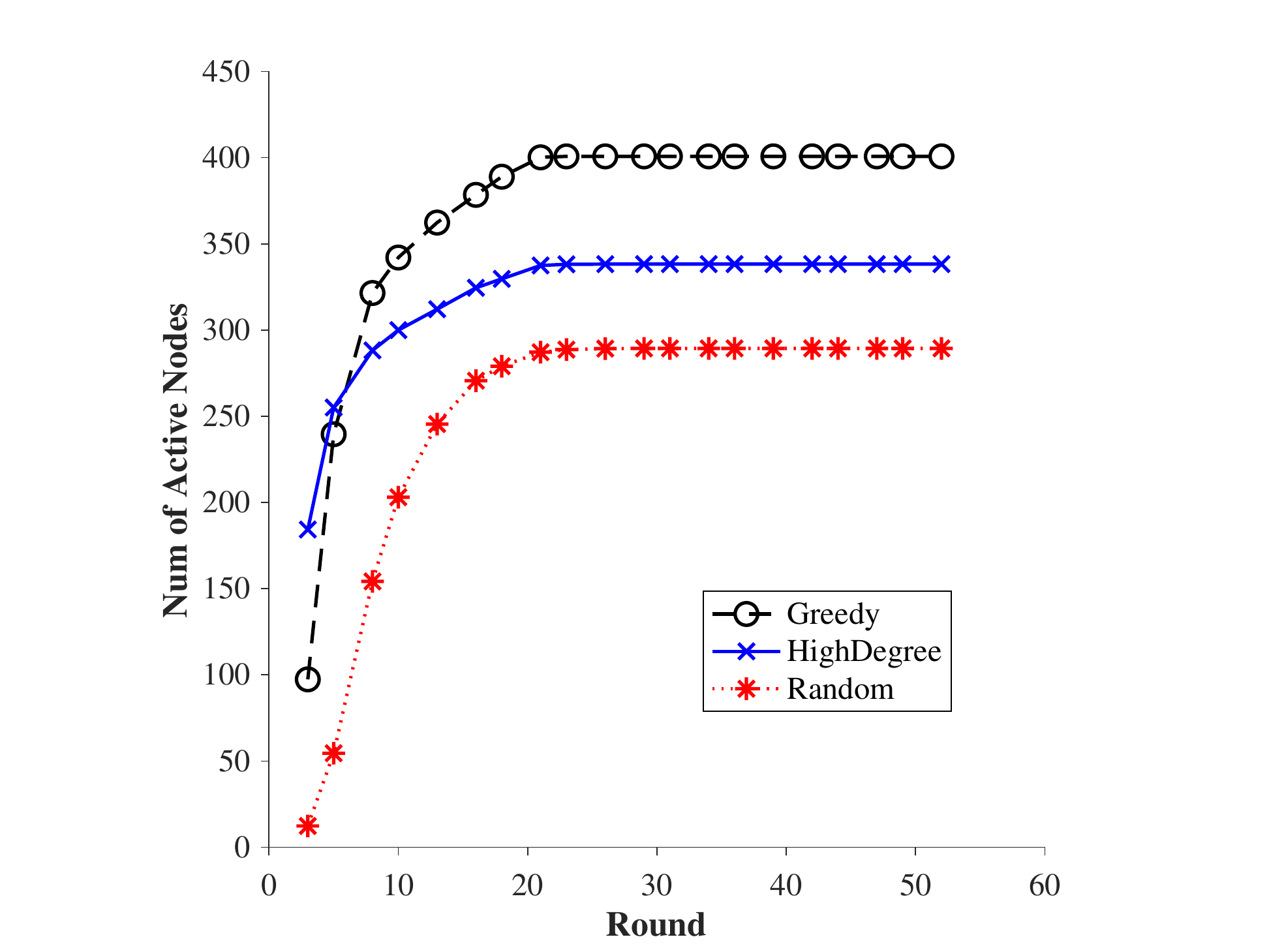}}
\subfloat[{[Higgs, IC, $k=20$, $d=8$]}]{\label{fig: higgs20_3_9}\includegraphics[trim = 0.5in 0in 0.5in 0in, clip,width=0.24\textwidth]{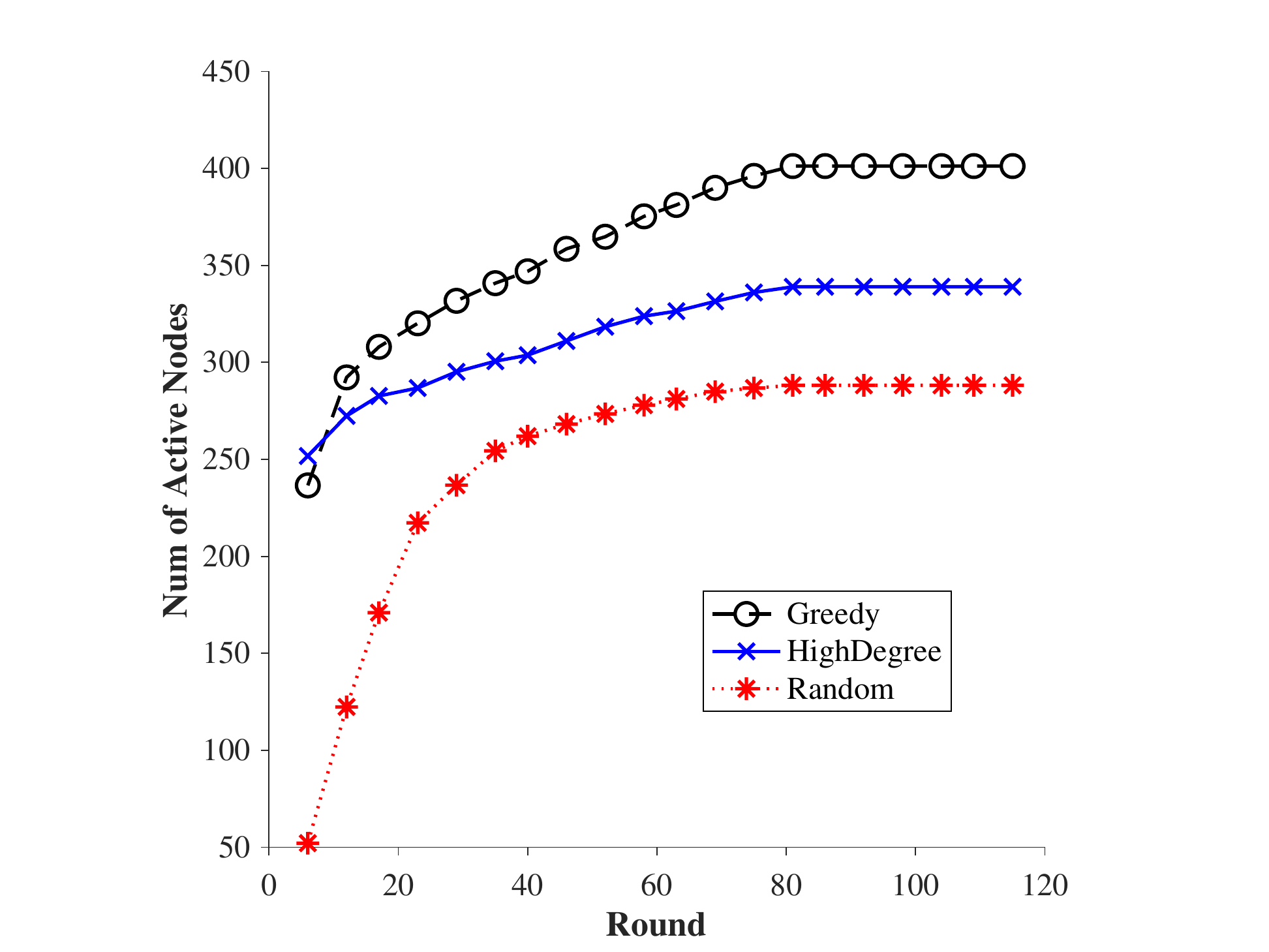}}
\subfloat[{[Higgs, IC, $k=20$, $d=15$]}]{\label{fig: higgs20_3_15}\includegraphics[trim = 0.5in 0in 0.5in 0in, clip,width=0.24\textwidth]{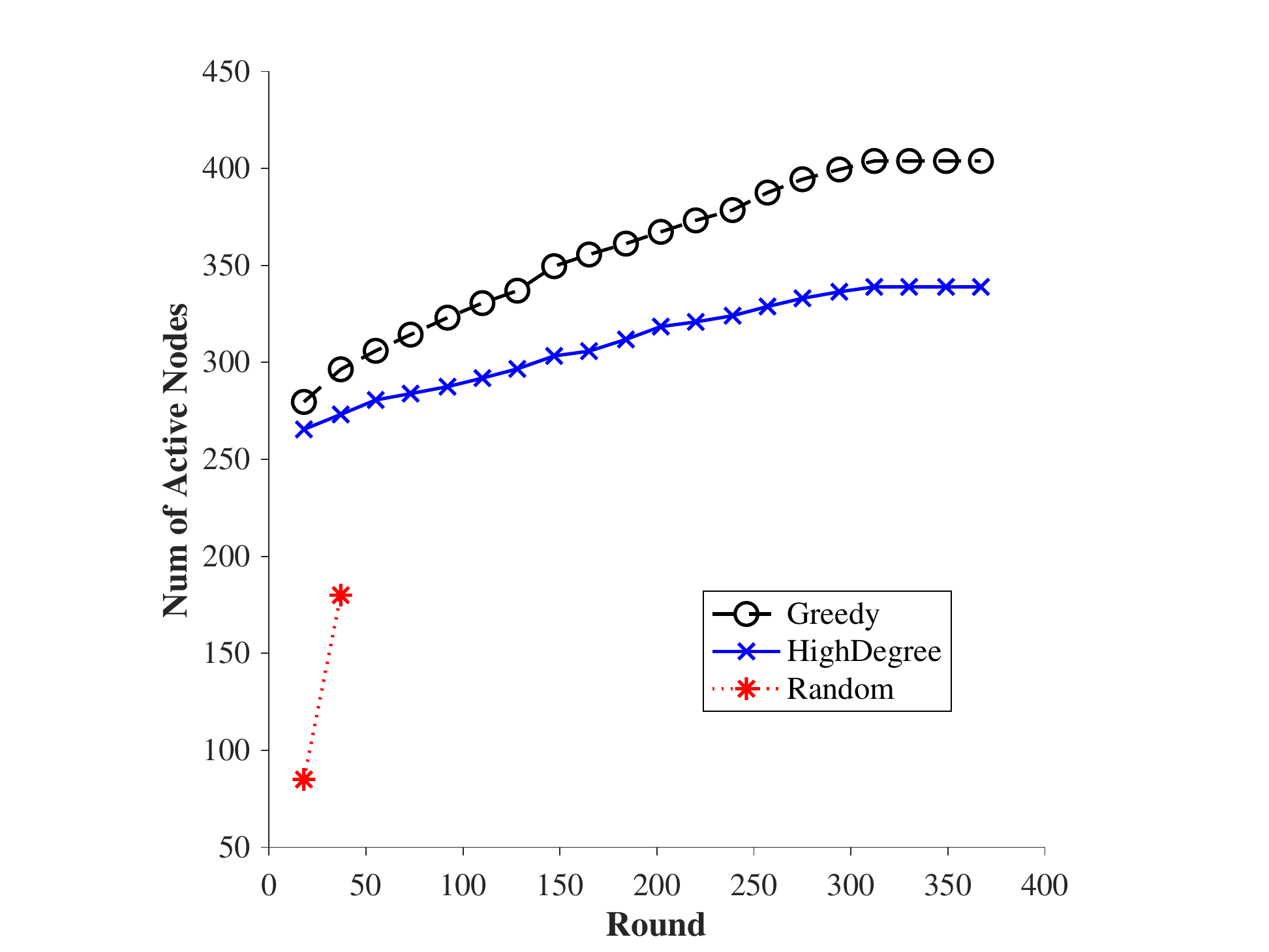}}

\subfloat[{[HepTh, WC, $k=5$, $d=0$]}]{\label{fig: hepthwc5_3_0}\includegraphics[trim = 0.5in 0in 0.5in 0in, clip, width=0.24\textwidth]{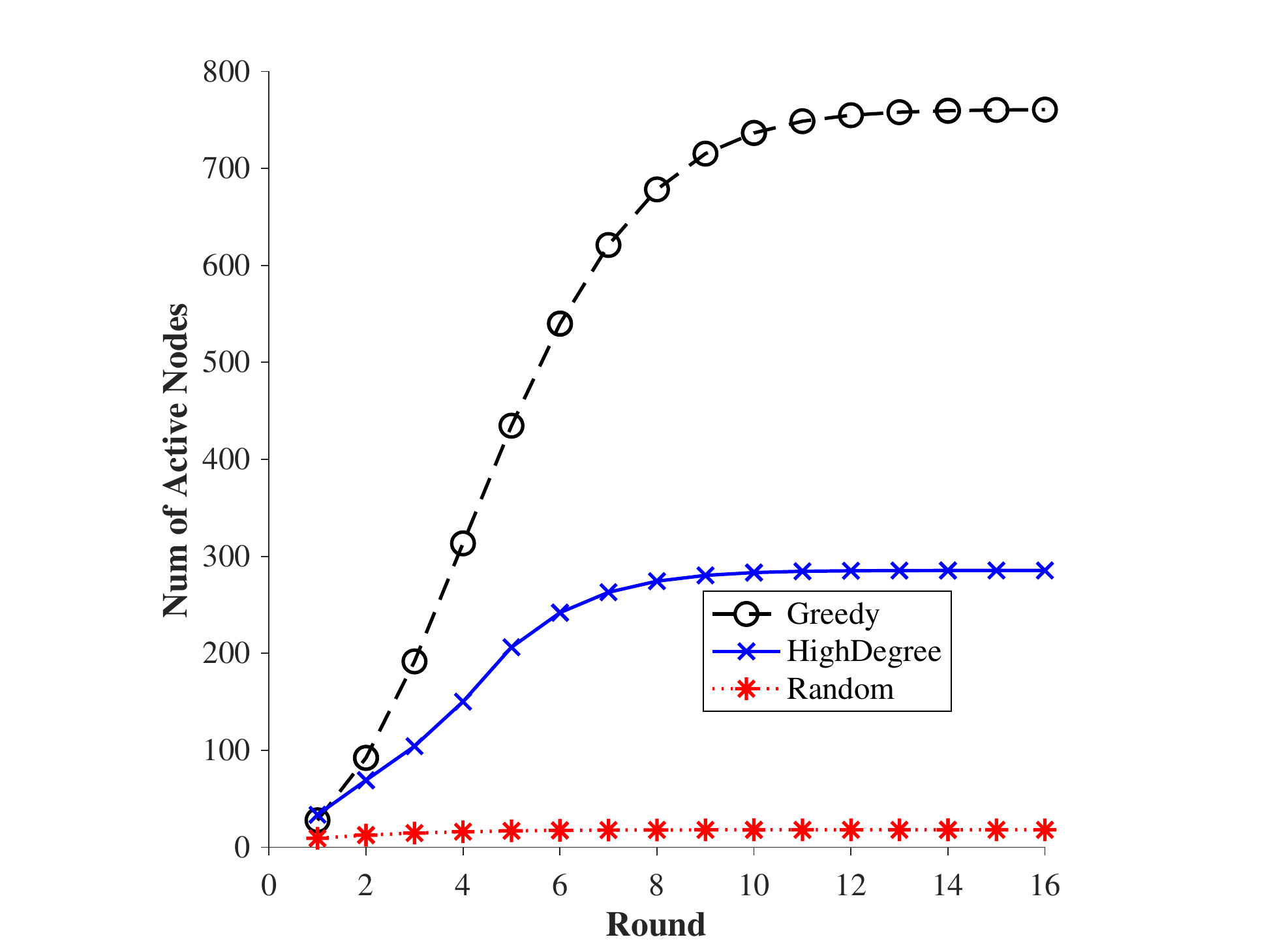}}
\subfloat[{[HepTh, WC, $k=5$, $d=1$]}]{\label{fig: hepthwc5_3_3}\includegraphics[trim = 0.5in 0in 0.5in 0in, clip,width=0.24\textwidth]{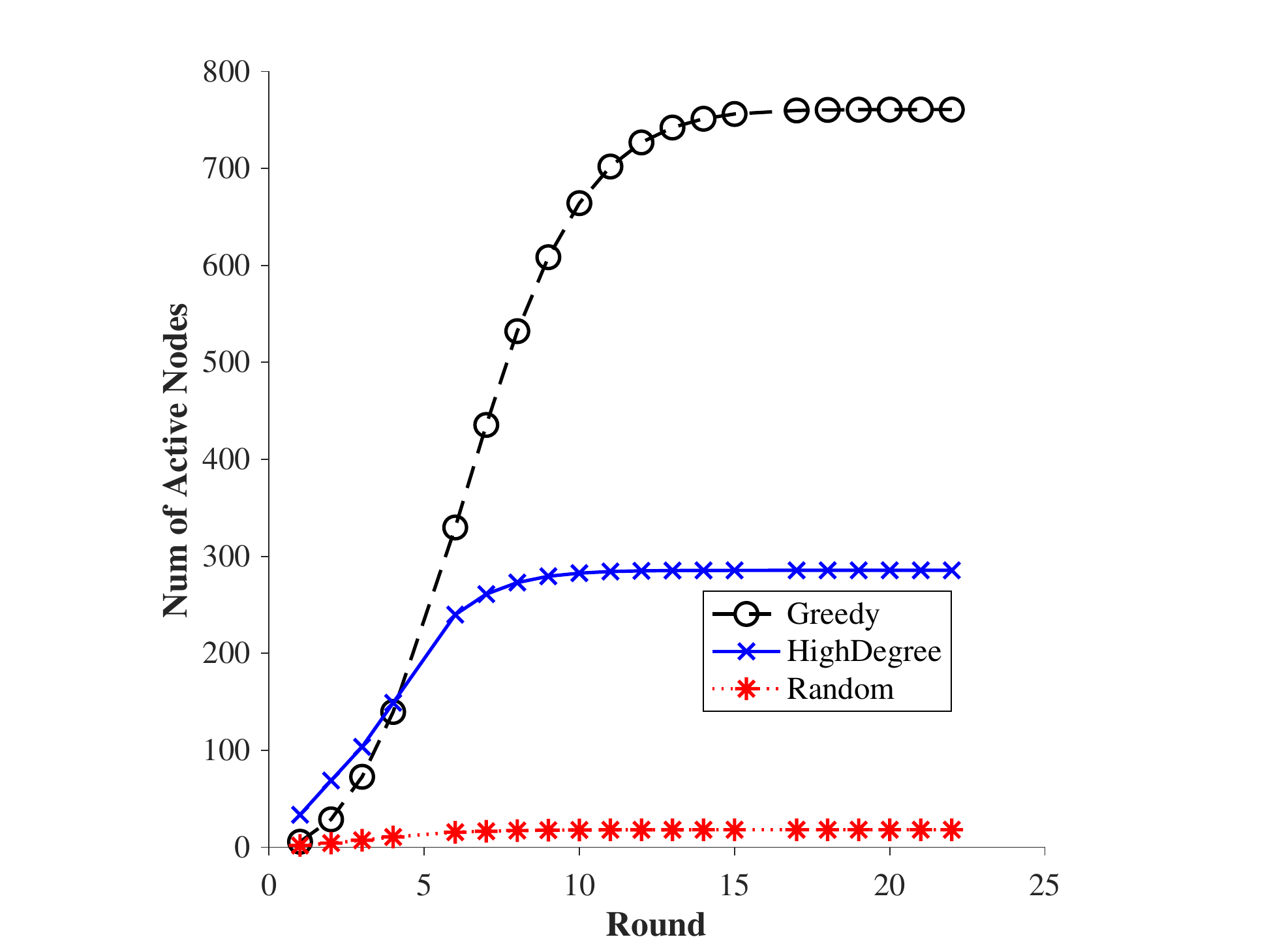}}
\subfloat[{[HepTh, WC, $k=5$, $d=8$]}]{\label{fig: hepthwc5_3_9}\includegraphics[trim = 0.5in 0in 0.5in 0in, clip,width=0.24\textwidth]{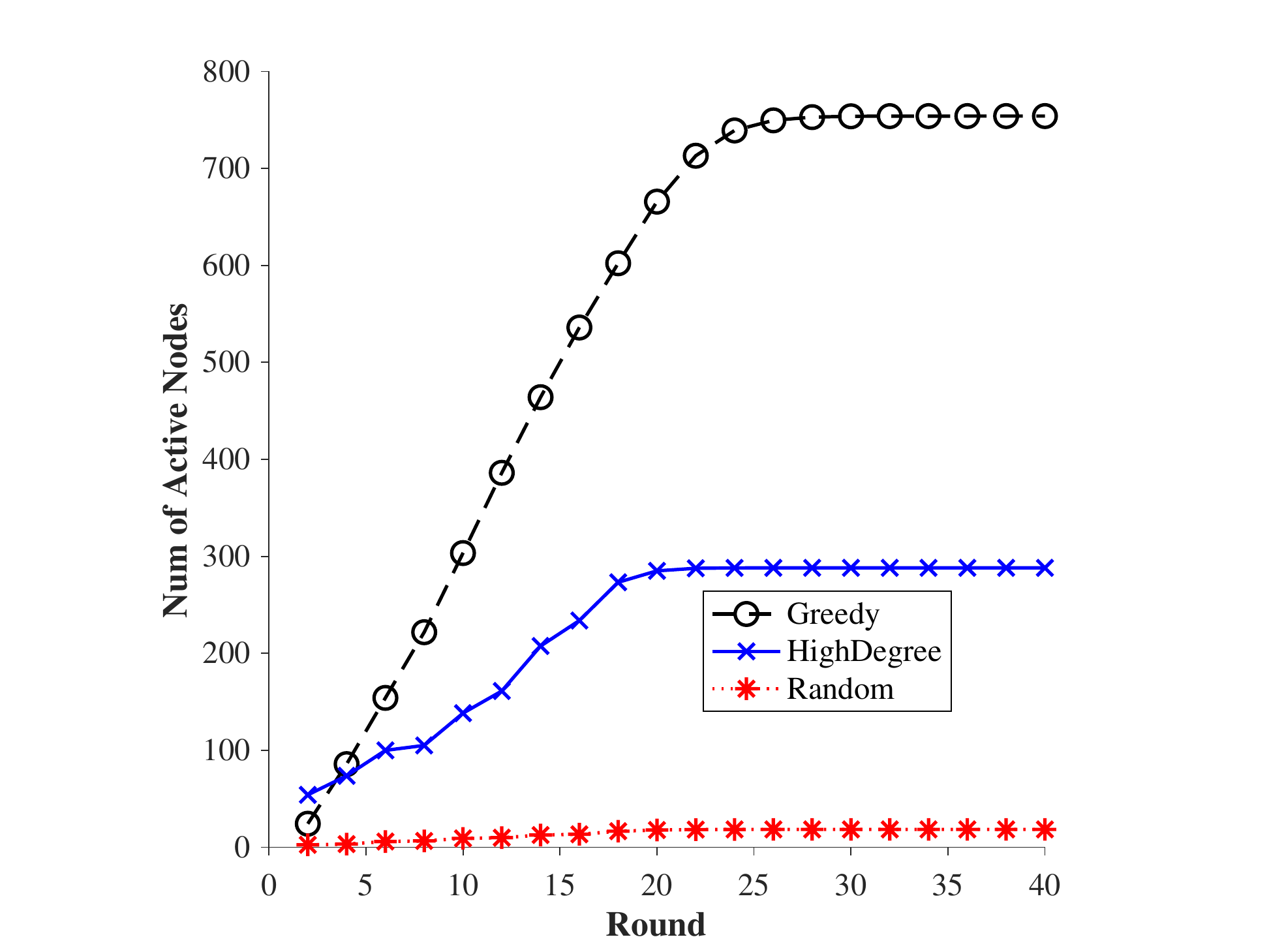}}
\subfloat[{[HepTh, WC, $k=5$, $d=15$]}]{\label{fig: hepthwc5_3_15}\includegraphics[trim = 0.5in 0in 0.5in 0in, clip,width=0.24\textwidth]{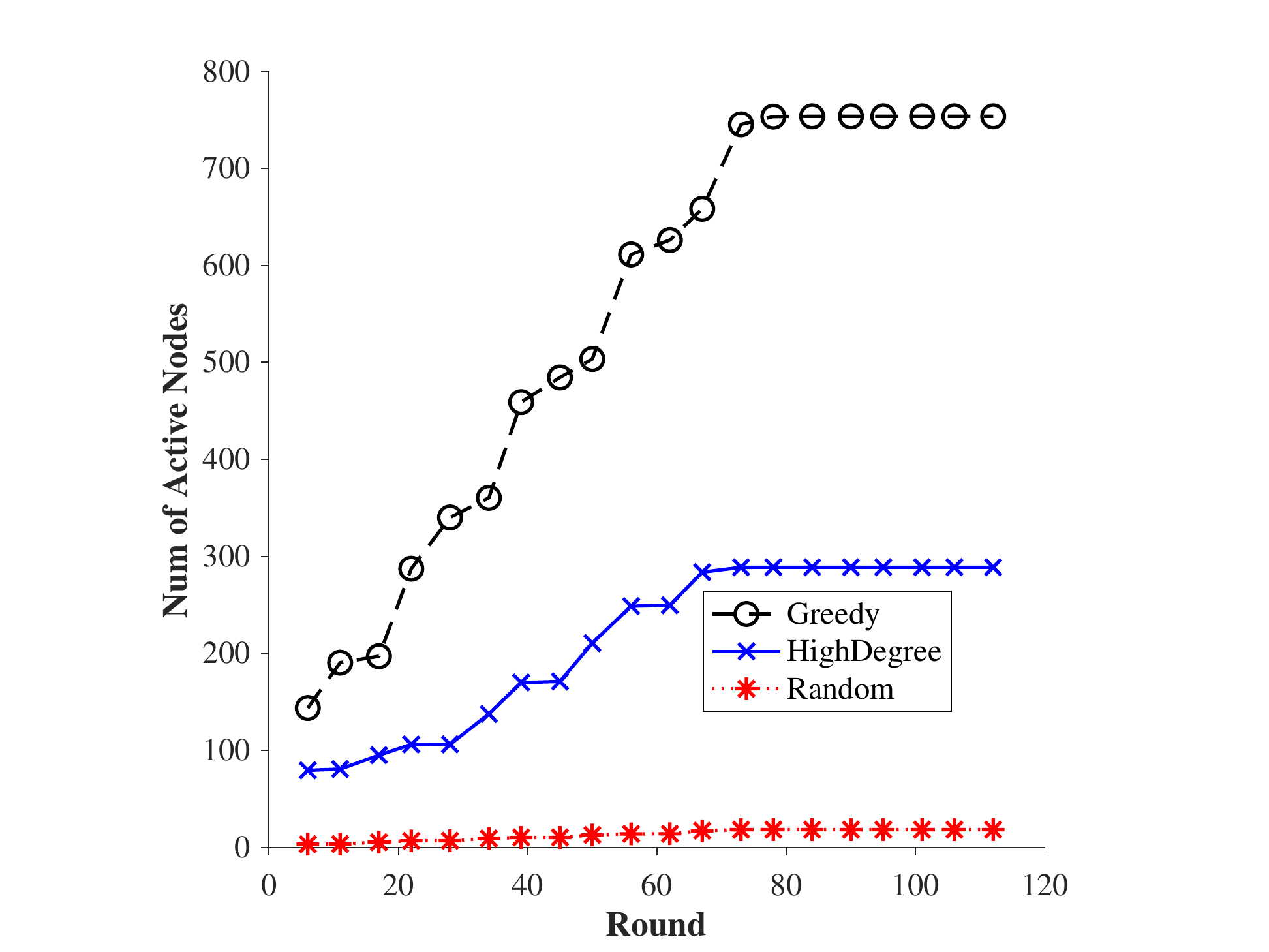}}

\caption{Additional Results of Experiment \RNum{1}. Part 2}
\vspace{-3mm}
\label{fig: exp1_more_2}
\end{figure*}

\begin{figure*}[!pt]
	\centering
	\subfloat[{[Youtube, WC, $k=5$, $d=0$]}]{\label{fig: youtubewc5_3_0}\includegraphics[trim = 0.5in 0in 0.5in 0in, clip, width=0.24\textwidth]{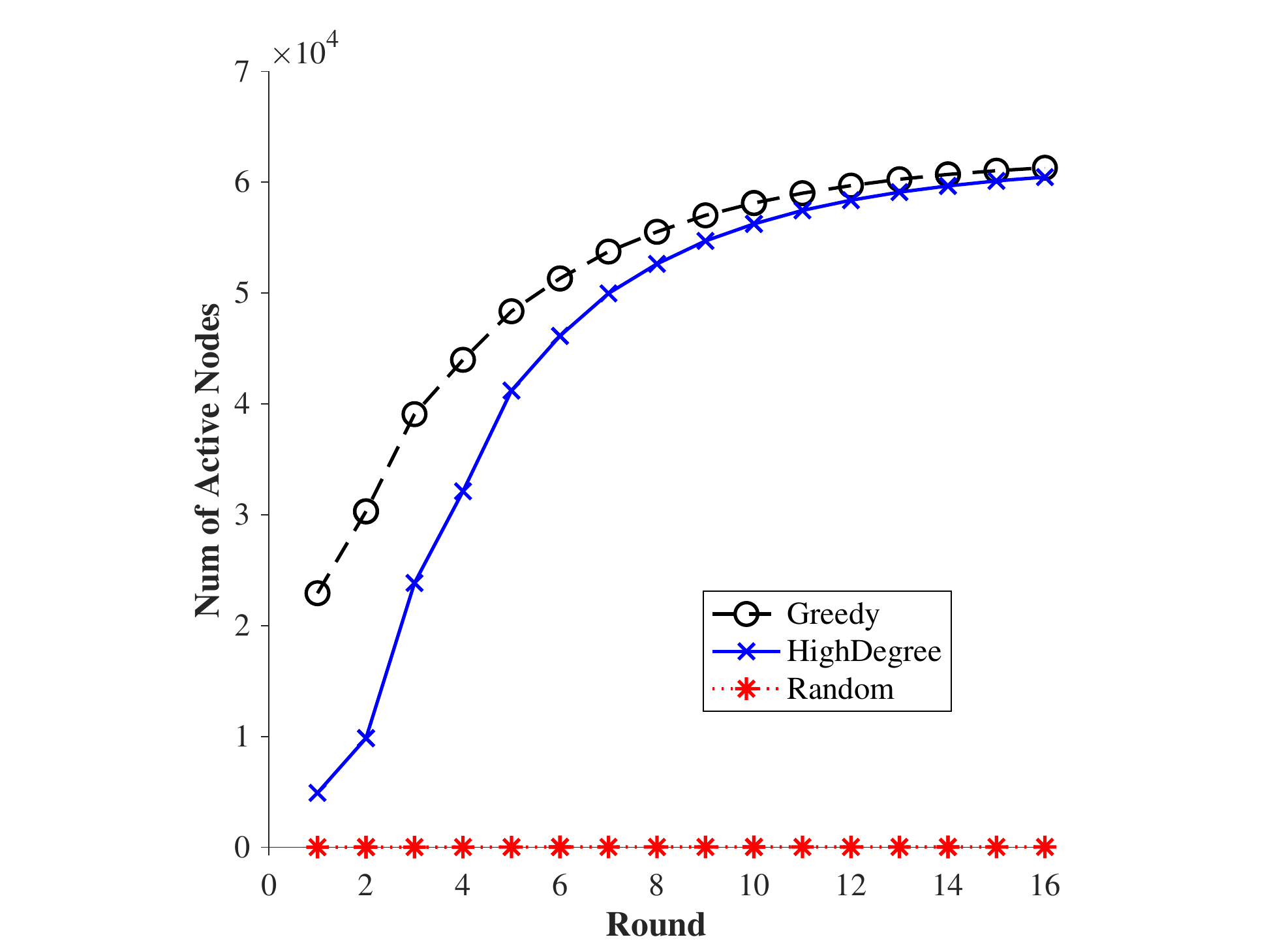}}
	\subfloat[{[Youtube, WC, $k=5$, $d=1$]}]{\label{fig: youtubewc5_3_3}\includegraphics[trim = 0.5in 0in 0.5in 0in, clip, width=0.24\textwidth]{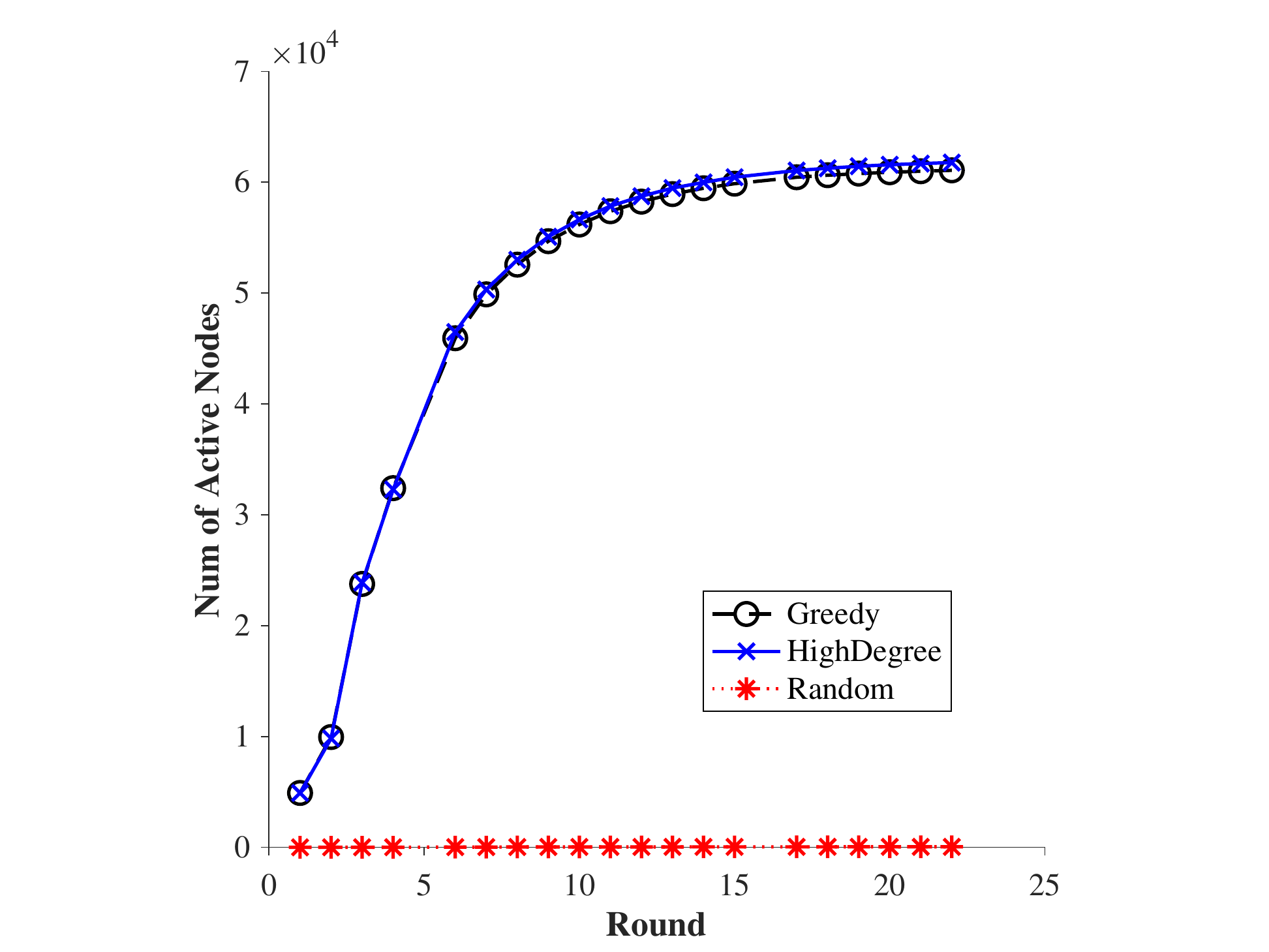}}
	\subfloat[{[Youtube, WC, $k=5$, $d=8$]}]{\label{fig: youtubewc5_3_9}\includegraphics[trim = 0.5in 0in 0.5in 0in, clip, width=0.24\textwidth]{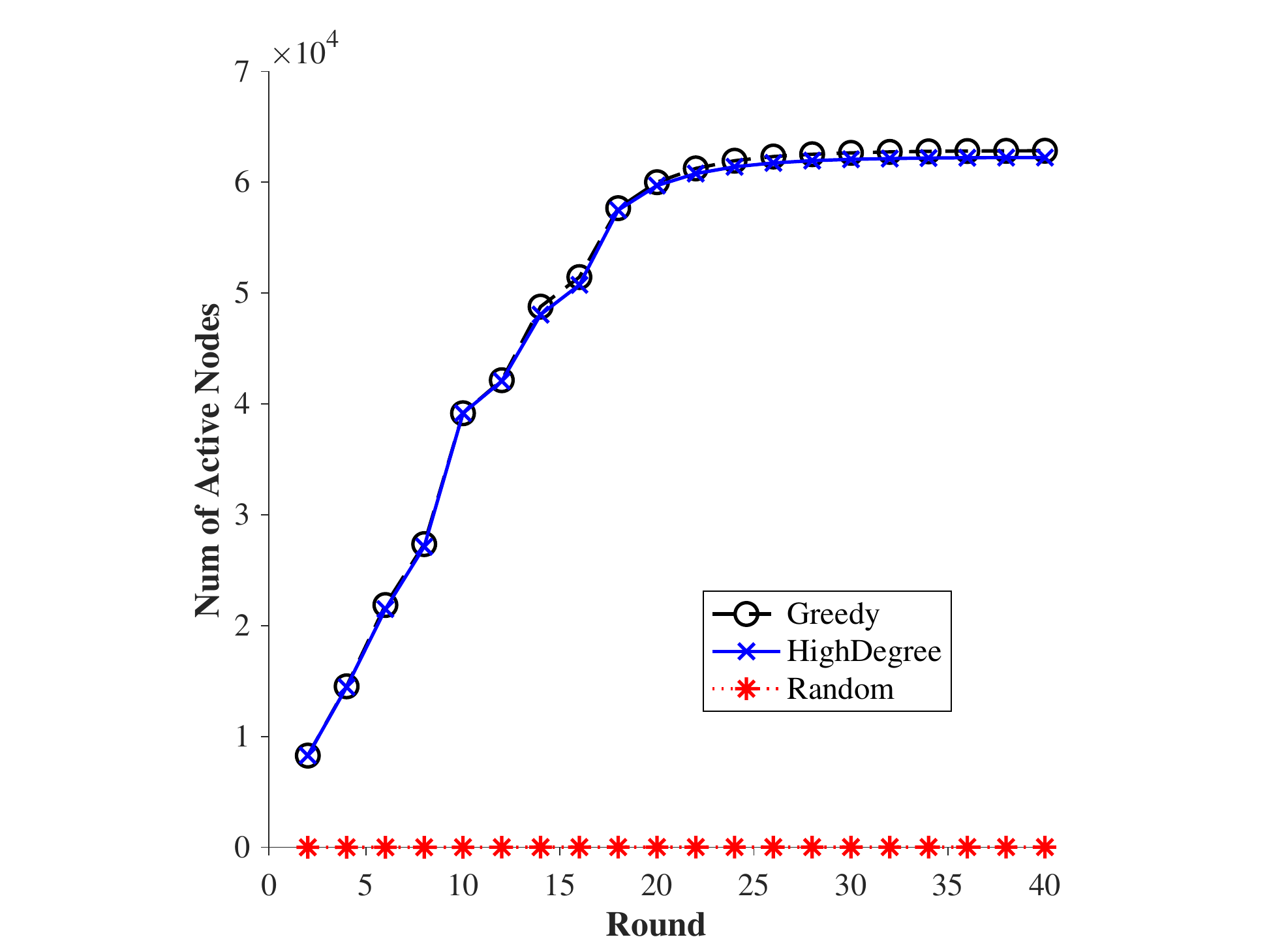}}
	\subfloat[{[Youtube, WC, $k=5$, $d=\infty$]}]{\label{fig: youtubewc5_3_15}\includegraphics[trim = 0.5in 0in 0.5in 0in, clip, width=0.24\textwidth]{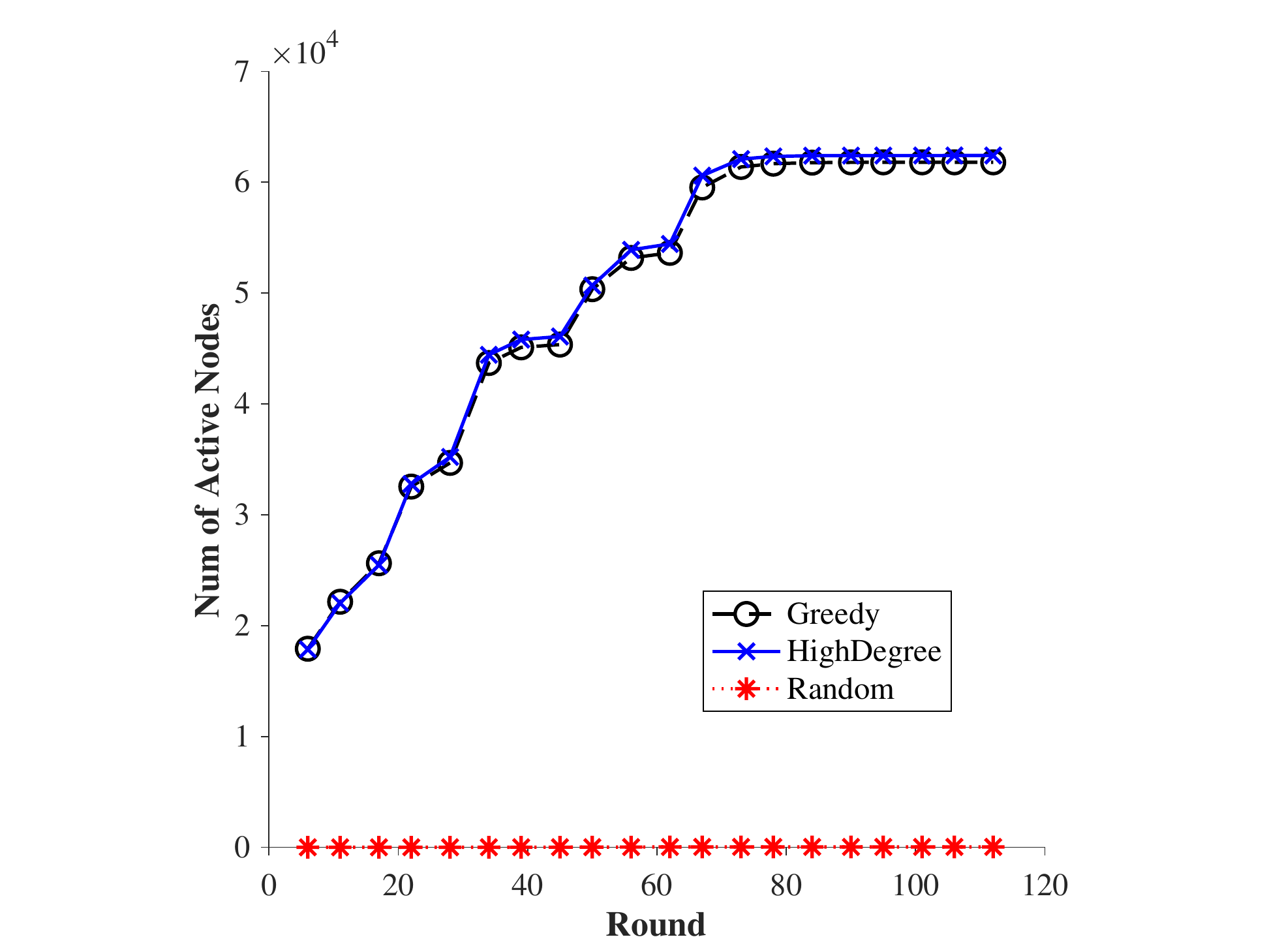}}

	\caption{Additional Results of Experiment \RNum{1}. Part 3}
	\vspace{-8mm}
	\label{fig: exp1_more_3}
\end{figure*}

\begin{figure*}[!t]
	\centering
	\subfloat[ {[Power, IC, $k=5$]}]{\label{fig: poweric5}\includegraphics[width=0.33\textwidth]{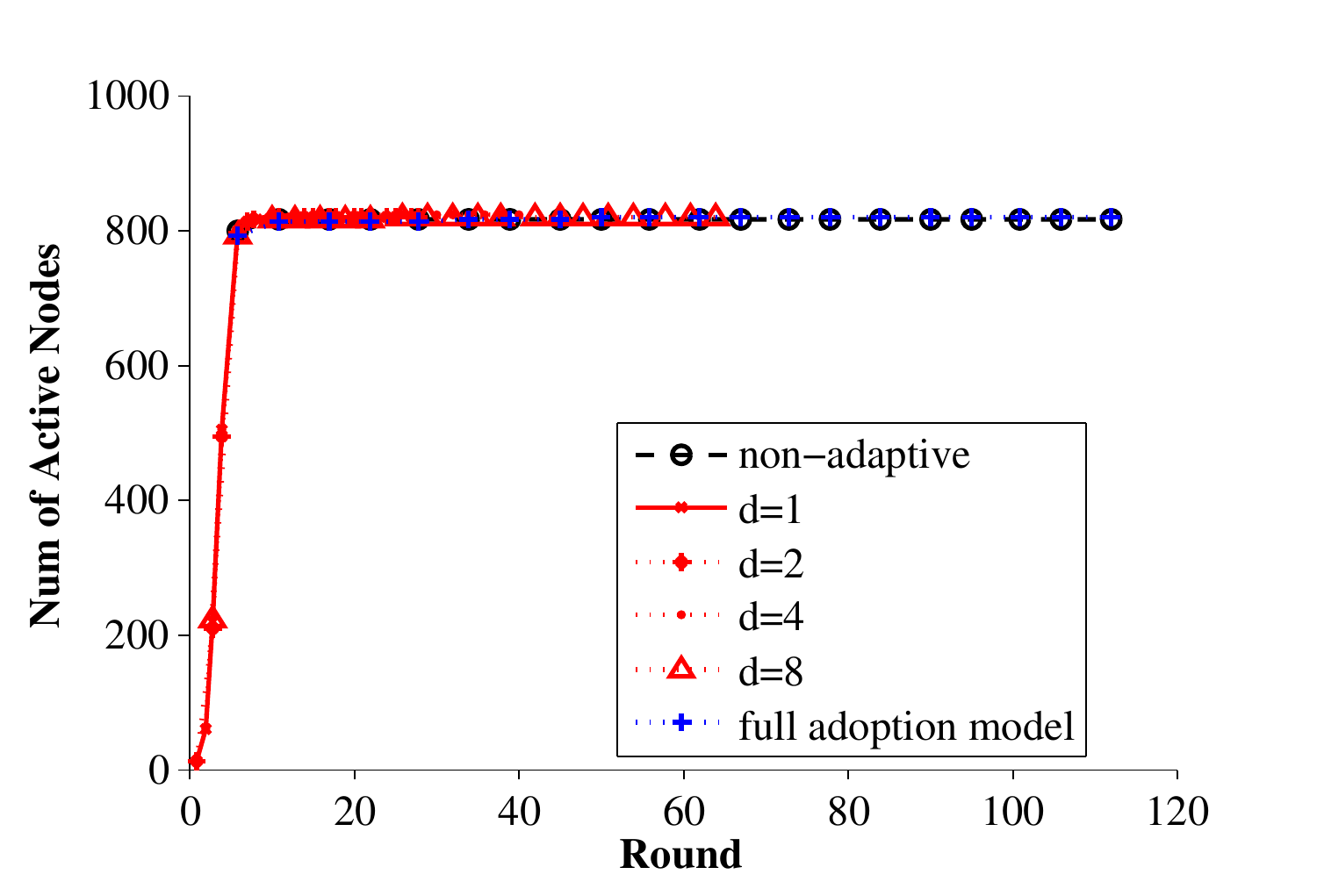}}
	\subfloat[ {[Power, IC, $k=10$]}]{\label{fig: poweric10}\includegraphics[width=0.33\textwidth]{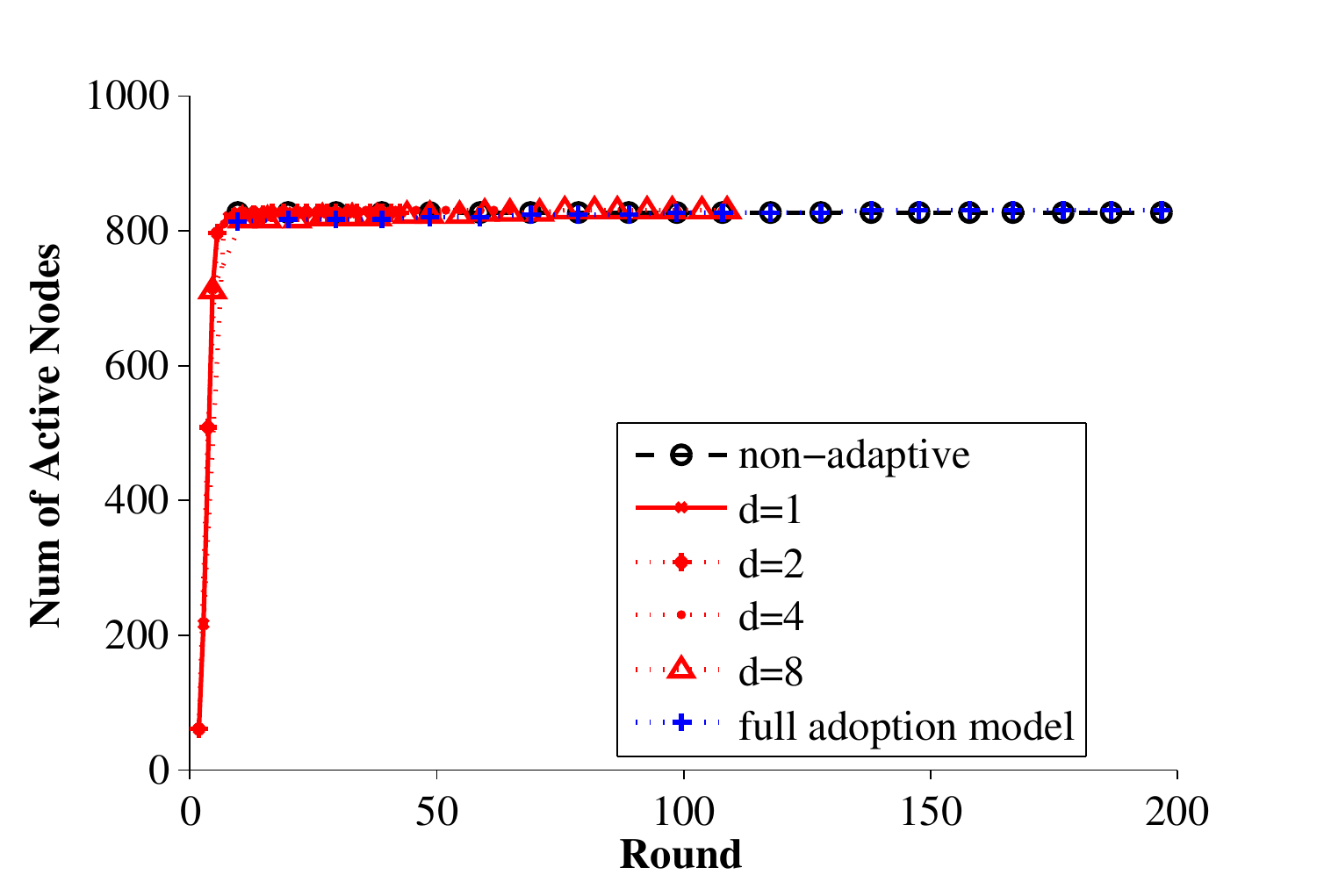}}
	\subfloat[ {[Power, IC, $k=20$]}]{\label{fig: poweric20}\includegraphics[width=0.33\textwidth]{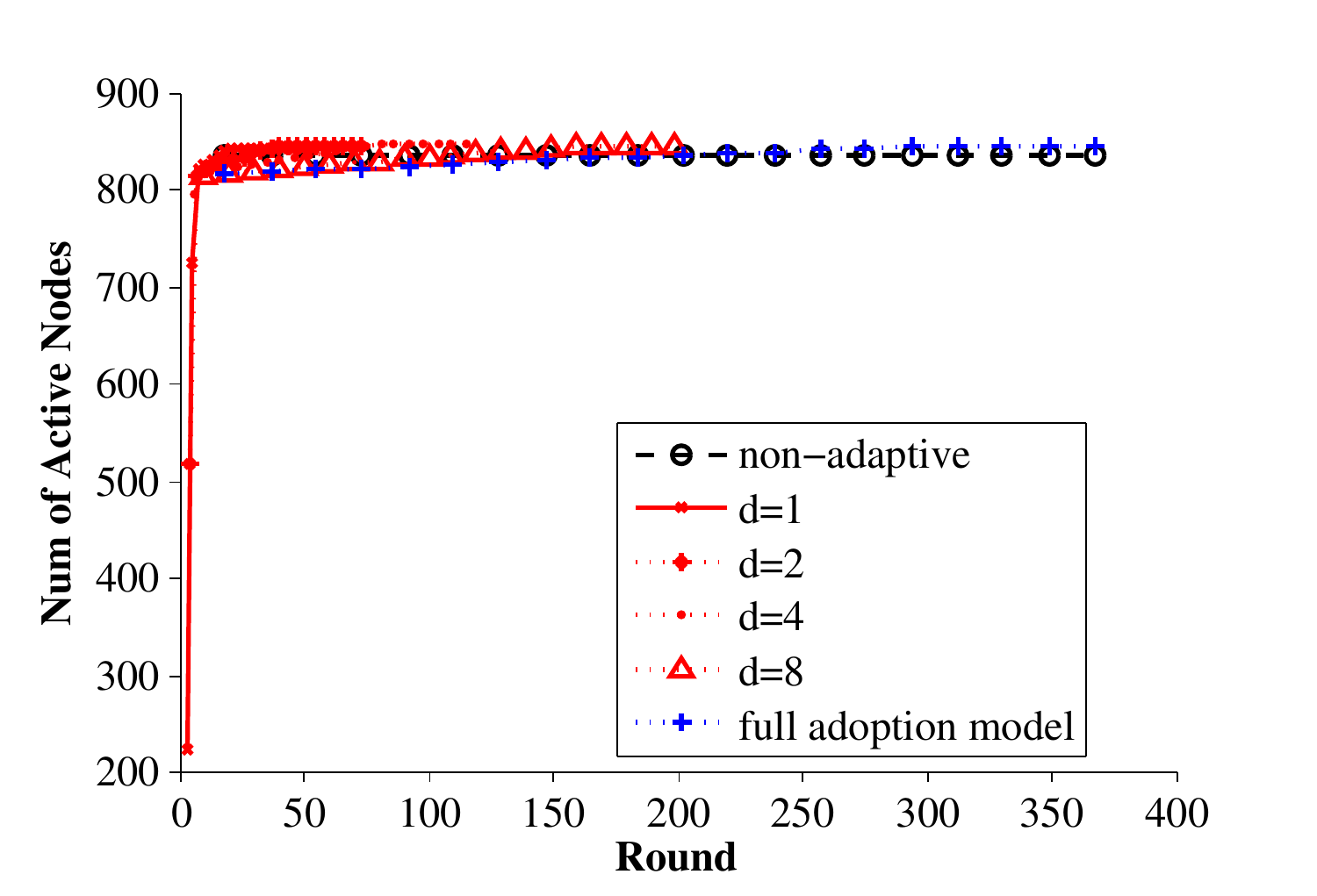}}
	\vspace{-3mm}
	\subfloat[ {[Power, IC, $k=50$]}]{\label{fig: poweric50}\includegraphics[width=0.33\textwidth]{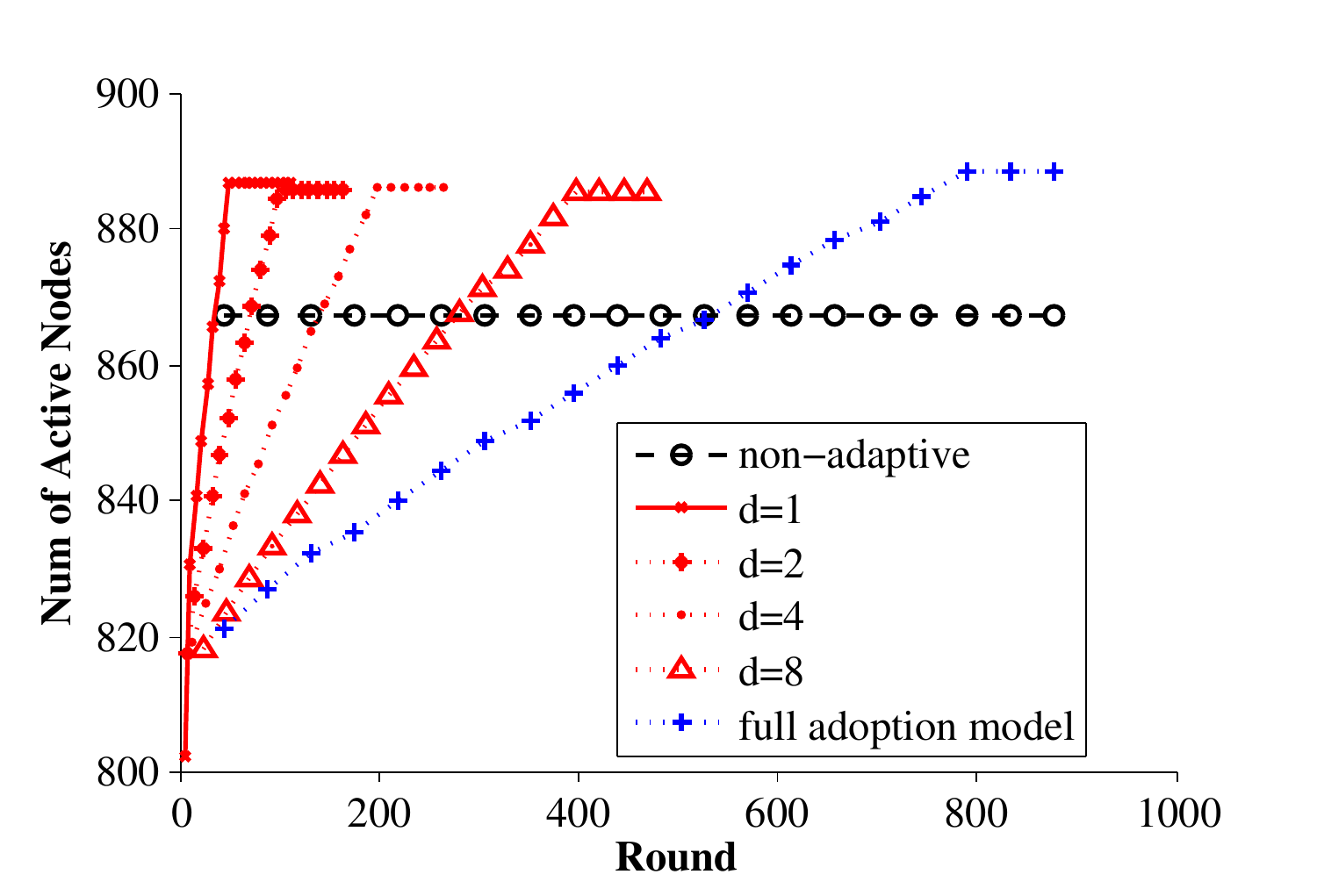}}
	\subfloat[ {[Power, WC, $k=50$]}]{\label{fig: powerwc50}\includegraphics[width=0.33\textwidth]{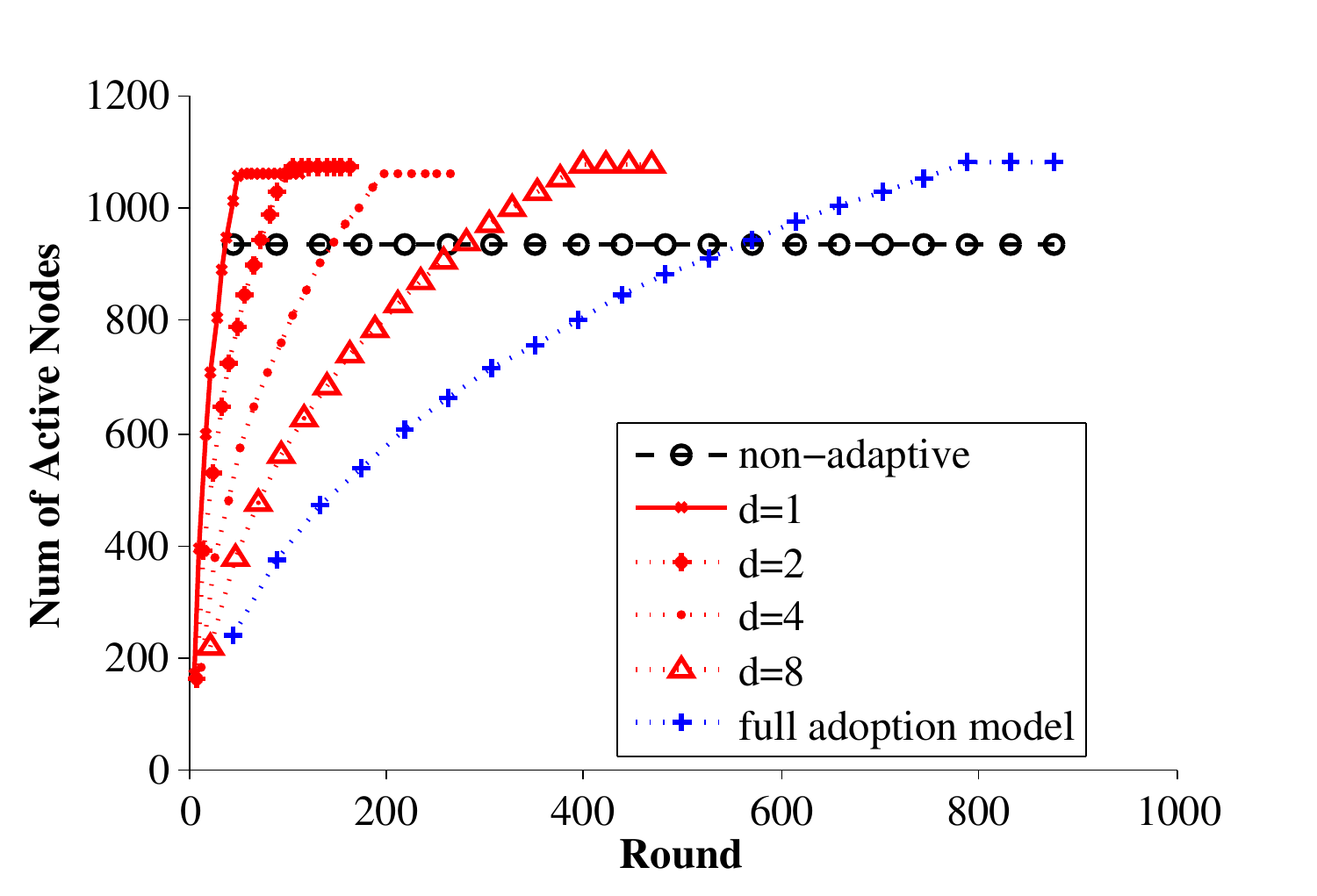}}
	\subfloat[ {[Wiki, WC, $k=5$]}]{\label{fig: wikiwc5}\includegraphics[width=0.33\textwidth]{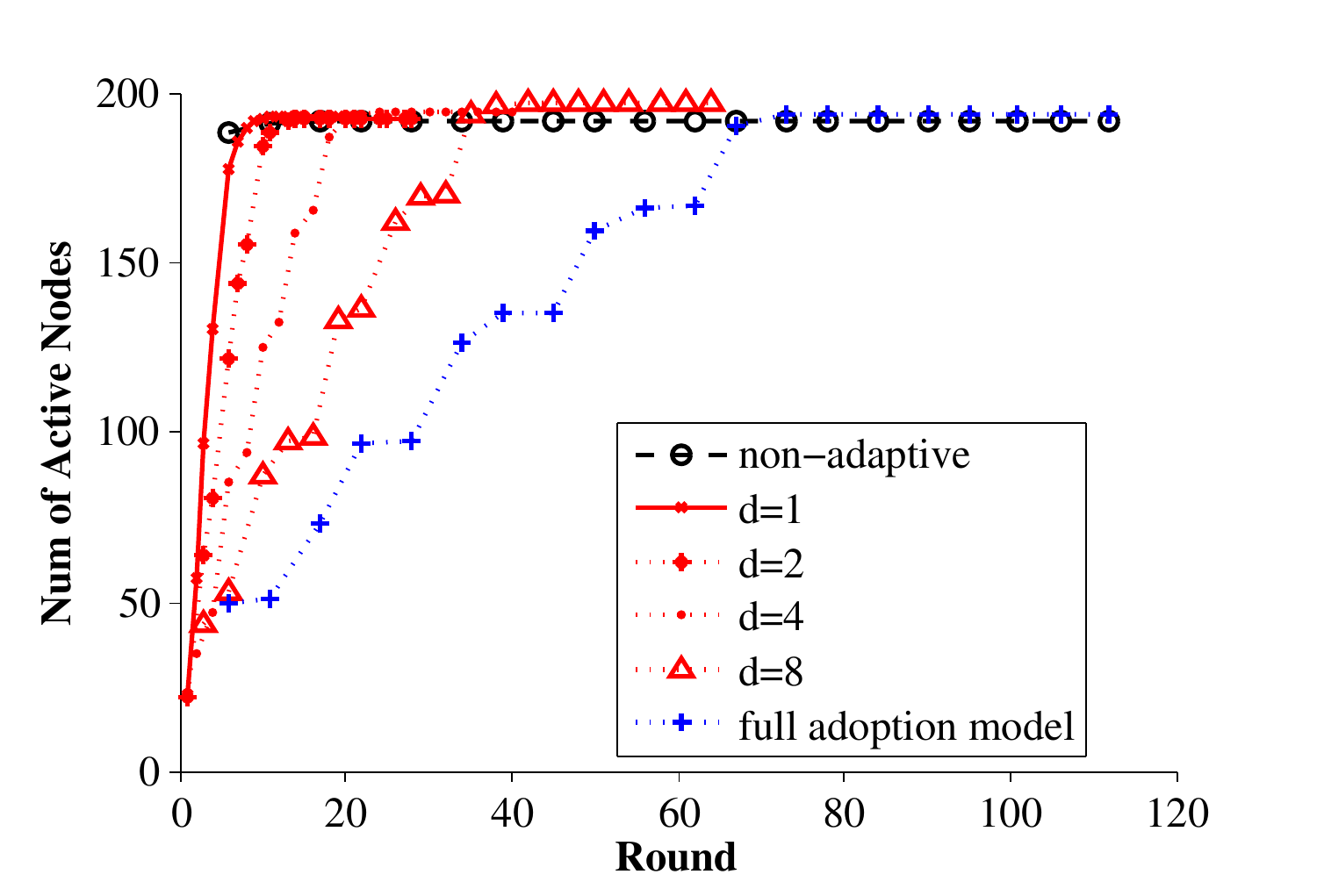}}
	\vspace{-3mm}
	\subfloat[ {[Wiki, IC, $k=5$]}]{\label{fig: wikiic5}\includegraphics[width=0.33\textwidth]{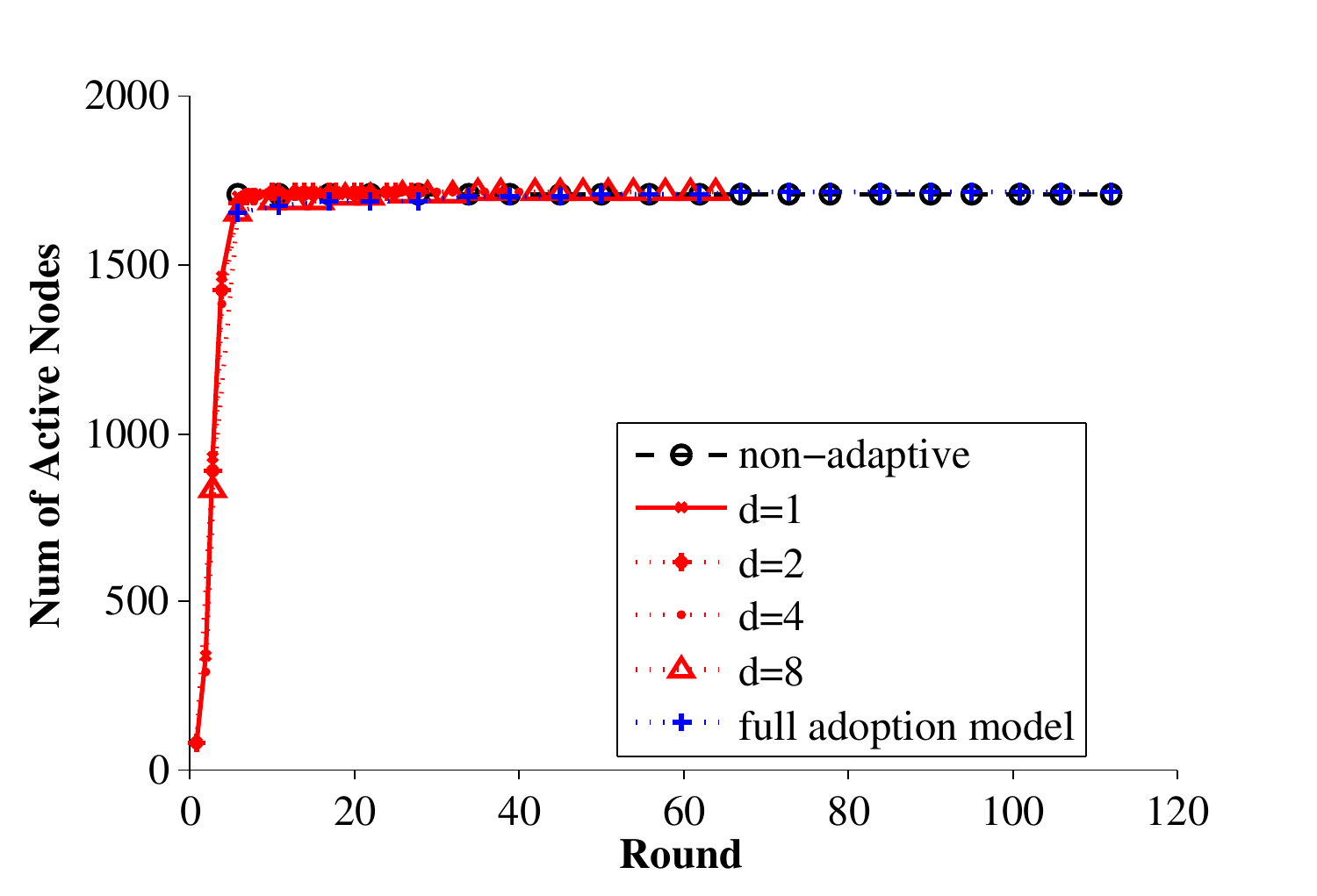}}
	\subfloat[ {[Wiki, IC, $k=50$]}]{\label{fig: wikiic50}\includegraphics[width=0.33\textwidth]{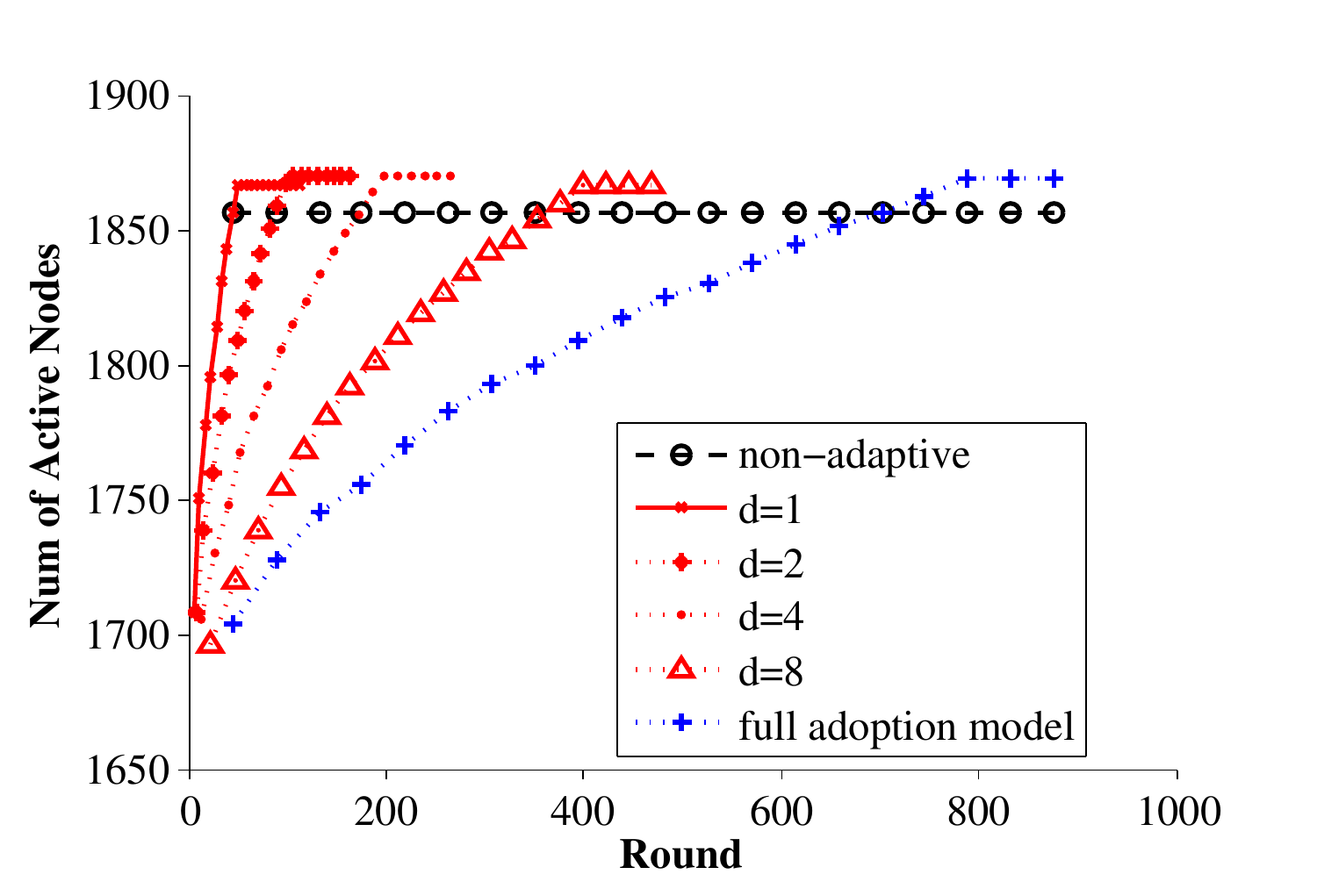}}
	\subfloat[ {[Wiki, WC, $k=50$]}]{\label{fig: wikiwc50}\includegraphics[width=0.33\textwidth]{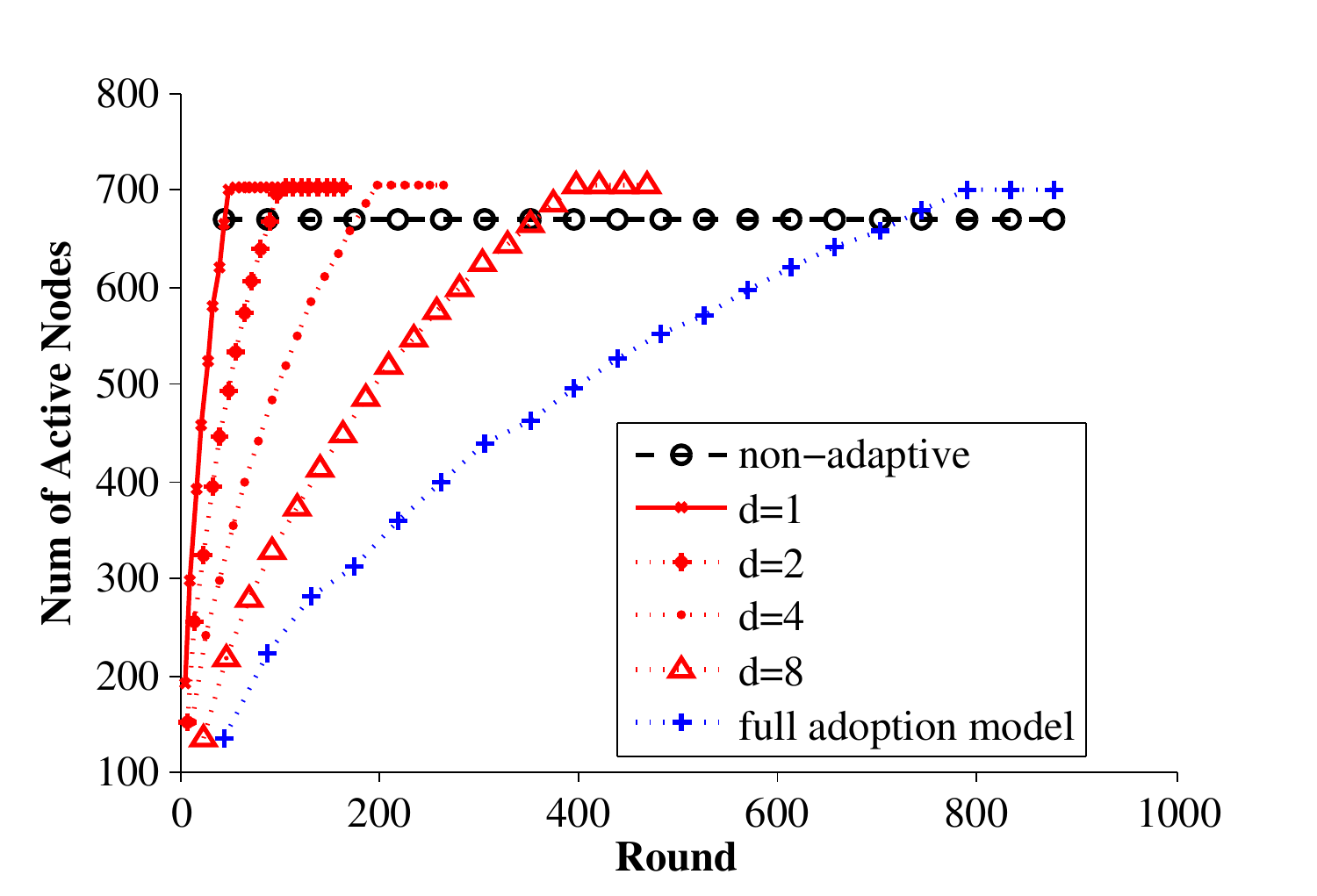}}
	\vspace{-3mm}
	\subfloat[ {[Higgs, IC, $k=20$]}]{\label{fig: higgs20}\includegraphics[width=0.33\textwidth]{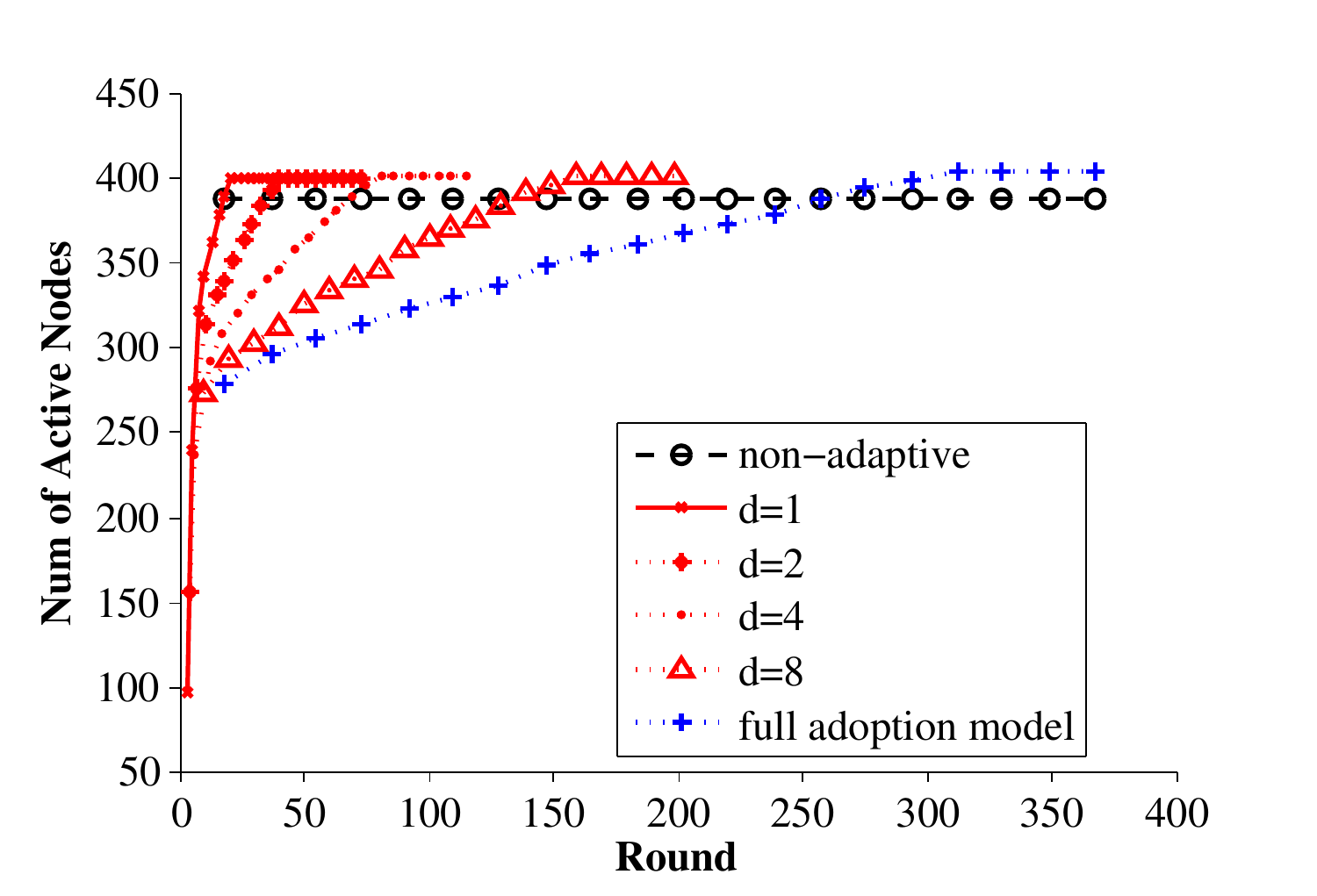}}
	\subfloat[ {[HepTh, WC, $k=5$]}]{\label{fig: hepthwc5}\includegraphics[width=0.33\textwidth]{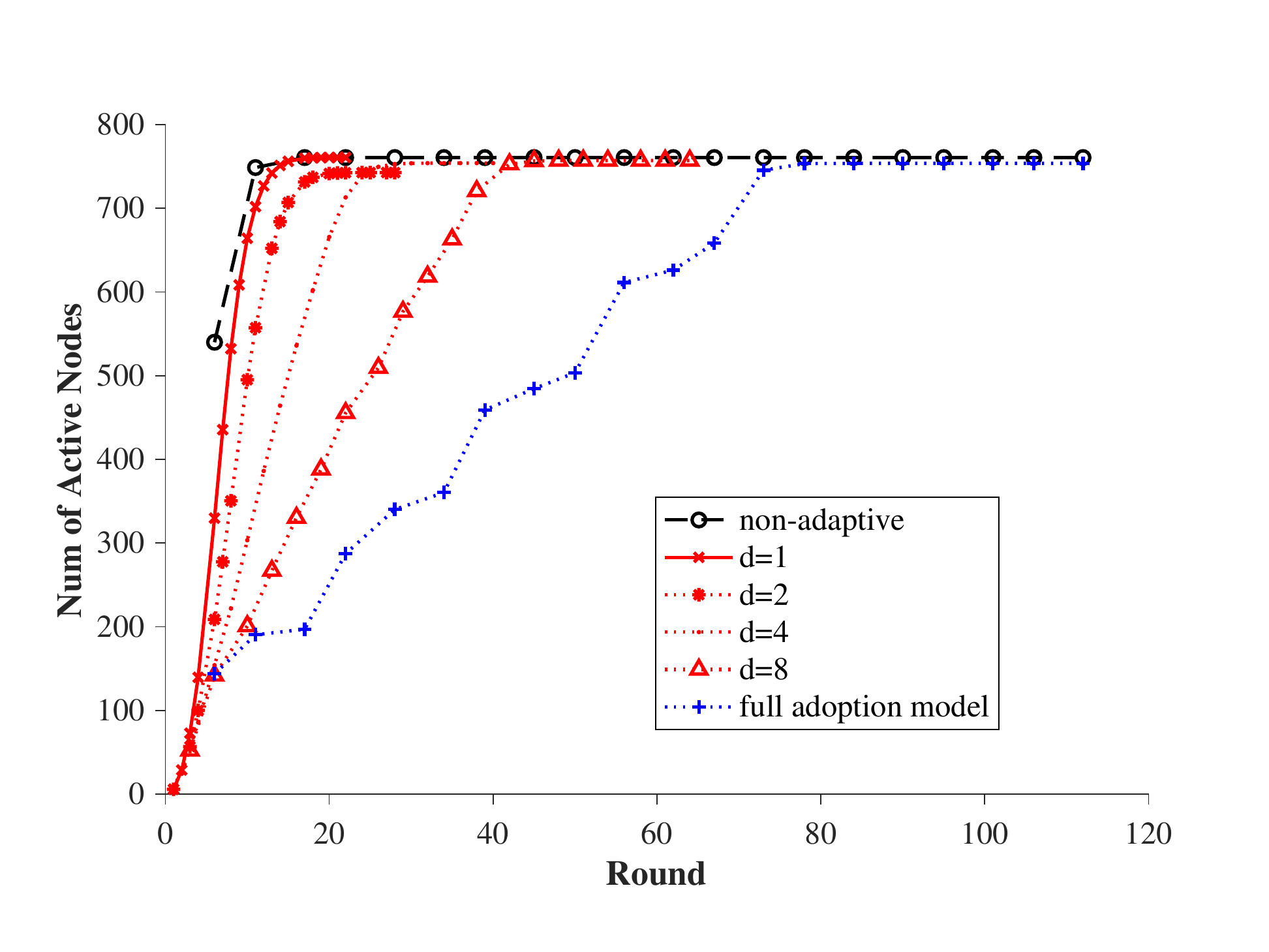}}
	\subfloat[ {[Youtube, WC, $k=5$]}]{\label{fig: youtubewc5}\includegraphics[width=0.33\textwidth]{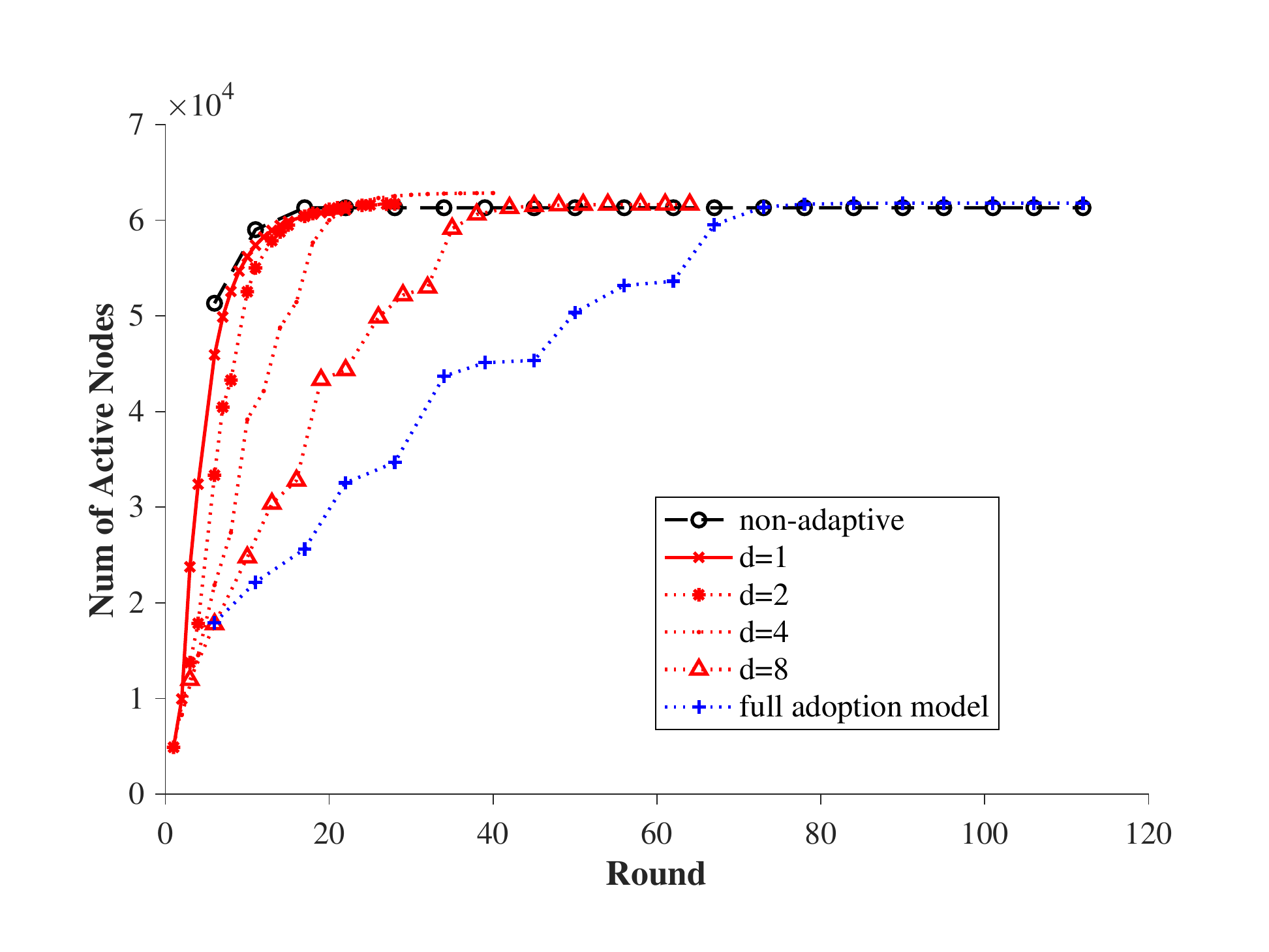}}
	\vspace{-0mm}
	\caption{Additional Results of Experiment \RNum{2}}
	\vspace{-6mm}
	\label{fig: exp2_more}
\end{figure*}

\end{document}